\documentclass[a4paper,notitlepage,fleqn,11pt]{article}
\usepackage[usenames,dvipsnames]{xcolor}

\usepackage[tbtags]{amsmath}
\usepackage{amsfonts}
\usepackage{amsthm}
\usepackage{xcolor}
\usepackage{amsfonts}
\usepackage{graphicx,amssymb}
\usepackage{multirow,url}
\usepackage{rotating}
\usepackage{mathrsfs}
\usepackage{epstopdf}
\usepackage{afterpage}
\usepackage{natbib}
\usepackage{tabularx}
\usepackage{booktabs}
\usepackage{threeparttable}
\usepackage{xr}
\usepackage{subfigure}
\usepackage{caption}
\usepackage{cases}
\usepackage{enumerate}
\usepackage{setspace}
\usepackage[title]{appendix}

\newtheorem{lemma}{Lemma}[section]
\newtheorem{propo}{Proposition}[section]

\newtheorem{assum}{Assumption}
\newtheorem{thm}{Theorem}[section]
\newtheoremstyle{remark}
{2ex}
{2ex}
{}
{}
{\bfseries}
{.}
{.5em}
{}
\theoremstyle{remark}
\newtheorem{remark}{Remark}[section]

\DeclareMathOperator*{\argmin}{argmin}

\DeclareMathOperator*{\diag}{diag}
\DeclareMathOperator*{\Diag}{Diag}
\DeclareMathOperator*{\cov}{cov}

\DeclareMathOperator{\tr}{tr}

\DeclareMathOperator{\rank}{rank}
\DeclareMathOperator{\ovec}{vec}
\DeclareMathOperator{\ovech}{vech}
\DeclareMathOperator{\ovechsec}{vech^{-}}
\newcommand{\bm}[1]{\mbox{\boldmath{$#1$}}}

\renewcommand{\arraystretch}{0.85}

\numberwithin{equation}{section}
\allowdisplaybreaks[4]

\title{An efficient multivariate volatility model for many assets}
\author{Wenyu Li$^{a}$, Yuchang Lin$^{b}$, Qianqian Zhu$^{b}$ and Guodong Li$^{c}$ \\
\emph{$^{a}$Nankai University, China}\\
\emph{$^{b}$Shanghai University of Finance and Economics, China}\\
\emph{$^{c}$University of Hong Kong, Hong Kong, China}\\
\emph{}
}
\date{}

\begin{document}
\begin{spacing}{1.5}
	\maketitle

\begin{abstract}
    This paper develops a flexible and computationally efficient multivariate volatility model that accommodates dynamic conditional correlations and volatility spillover effects among financial assets. 
    The new model has desirable properties such as identifiability and computational tractability for many assets. 
    A sufficient condition for strict stationarity is derived for the process. 
    Two quasi-maximum likelihood estimation methods are proposed for the new model without and with low-rank constraints on the coefficient matrices, respectively, and the asymptotic properties of both estimators are established. 
    Moreover, a selection-consistent Bayesian information criterion is developed for order selection, and testing for volatility spillover effects is discussed. 
    The finite-sample performance of the proposed methodology is evaluated in simulations for small and moderate dimensions. 
    Its usefulness and inference tools are illustrated by two empirical examples for 5 stock markets and 17 industry portfolios. 
\end{abstract}

\section{Introduction} \label{section Introduction}

Multivariate time series are usually characterized by temporal and cross-sectional dependencies. For example, it is well known that financial volatilities move together over time across assets and markets. 
A multivariate modelling framework accounts for both dependencies and facilitates better decision-making in areas such as asset pricing, portfolio selection, option pricing, hedging and risk management \citep{bauwens2006multivariate,caporin2014robust,francq2019garch}.
Since the introduction of multivariate generalized autoregressive conditional heteroscedastic (GARCH) models, multivariate volatility models have attracted considerable interest owing to their popularity and usefulness in modeling conditional covariances and correlations. 
In particular, conditional covariances are usually used to forecast Value-at-Risk (VaR) thresholds, whereas conditional correlations are used to determine portfolios \citep{mcaleer2005automated,caporin2012we}. 

Numerous specifications have been proposed for multivariate GARCH models, including the vector GARCH model \citep{bollerslev1988capital}, the constant conditional correlation (CCC) model \citep{bollerslev1990ModellingTC,ling2003asymptotic}, the Baba-Engle-Kraft-Kroner (BEKK) model \citep{engle1995MultivariateSG}, the dynamic conditional correlation (DCC) model \citep{engle2002dynamic,tse2002multivariate} and its variants \citep{cappiello2006asymmetric,aielli2013dynamic}, and the structural multivariate GARCH model \citep{hafner2022identification} and so on.  
Among them, BEKK and DCC models are the most widely used. 
For BEKK models, the conditional covariance matrix is directly modeled in a specific form such that positivity can be easily imposed, whereas the coefficients are difficult to interpret. 
The DCC model separately models volatilities and conditional correlation dynamics, which enhances its interpretability. 
However, the fully parameterized BEKK and DCC models with general coefficient matrices are seldom used in practice, primarily for two reasons: (a) their parameters are non-identifiable \citep{francq2019garch}; and (b) they are numerically infeasible to estimate for many assets; see \cite{caporin2014robust} for numerical evidence. Specifically, they involve a large number of parameters, and their model structures require high powers of coefficient matrices to be computed during estimation. 
To make BEKK and DCC models feasible, a common solution is to constrain the coefficient matrices to be scalar or diagonal \citep{ding2001large,engle2002dynamic}. Although these parsimonious specifications enable feasible estimation, they may fail to capture complex dynamics and cross-sectional dependence.
Clearly, there is a dilemma between model flexibility and parsimony, and this trade-off is particularly important for multivariate GARCH models. 

An alternative to standard GARCH formulations is provided by logarithmic GARCH (log-GARCH) models. By modeling the logarithm of the conditional variance, this class of models guarantees positivity without imposing non-negativity constraints on parameters, thereby facilitating more flexible estimation and inference. 
Initially developed in univariate form \citep{Geweke01011986,Pantula01011986,milhoj1987multiplicative}, the log-GARCH model has been further refined \citep{francq2013garch, francq2018exponential} and extended to multivariate frameworks. 
In the multivariate setting, log-GARCH specifications have been proposed under both CCC and DCC structures. 
Among these, \cite{francq2017equation} introduced a multivariate specification that accommodates time-varying conditional correlations (with CCC as a special case). Related DCC-type log-GARCH specifications have also been studied by \cite{sucarrat2016estimation} and \cite{escribano2018equation}. 
Nevertheless, the trade-off between model flexibility and parsimony in multivariate GARCH models remains a challenge for their logarithmic counterparts as well. 

In terms of estimation methods, 
the Gaussian quasi-maximum likelihood estimator (QMLE) is well developed for BEKK models and has theoretical guarantees. In particular, the strong consistency was proved by \cite{jeantheau1998strong} under the multivariate log-moment condition, 
and the asymptotic normality was established by \cite{comte2003asymptotic} and \cite{hafner2009asymptotic} under eighth-order and sixth-order moments, respectively. 
Based on these results, \cite{caporin2012we} showed that the BEKK model could be used to obtain consistent estimates of conditional correlation matrices. 
In contrast, a two-step estimation procedure that is computationally feasible for many assets is commonly used for DCC models: volatility parameters are estimated using the QMLE of univariate GARCH models in the first step, and the remaining correlation parameters are estimated by QMLE in the second step.
However, the asymptotic properties of DCC models estimated by this two-step procedure remain unclear \citep{caporin2012we,caporin2013ten}, leaving statistical inference partially unresolved. 
For multivariate log-GARCH models, asymptotic theory remains limited. 
\cite{francq2017equation} established asymptotic properties for an equation-by-equation least squares estimator under a componentwise innovation condition: each innovation is independent of its own past, though cross-sectional correlations may be time-varying (with CCC as a special case).
However, extending these results to models with explicitly parameterized DCC-type correlation structures remains an open challenge, with no complete asymptotic theory currently available.

Despite the popularity of BEKK and DCC models and the structural advantages of log-GARCH specifications, a new model that both inherits their merits and reasonably balances generality and computational efficiency is still needed.  
This paper takes a different approach to balancing flexibility and parsimony by proposing a new identifiable multivariate volatility model for many assets. 
Specifically, using a suitable coefficient-matrix decomposition \citep{huang2024sarma,zheng2022interpretable}, the new model preserves the generality of fully parameterized conditional variance models in capturing temporal and cross-sectional dependence among series while remaining comparable to diagonal DCC models in interpretability, identifiability and parsimony. 
Moreover, by modeling the logarithm of volatility, the new model guarantees positivity without parameter constraints, thereby enhancing the flexibility in estimation and inference. Our main contributions are summarized as follows. 

  First, a new multivariate volatility model is proposed in Section \ref{section SGARCH(r,s)} to capture temporal and cross-sectional dependence among many assets. 
  The new model achieves a better balance between flexibility and parsimony than traditional multivariate GARCH and log-GARCH models. 
  Specifically, it is general enough to allow for dynamic conditional correlations and volatility spillovers among assets, but also parsimonious enough for feasible estimation and efficient computation for many assets. 
  Meanwhile, the proposed model is identifiable without imposing any sophisticated identification constraints (see Proposition \ref{propo Identification}). 
  In addition, a sufficient condition for stationarity of the proposed model is provided in Theorem \ref{thm Stationarity}. 
  The proof technique can also be used to derive the stationarity conditions of other multivariate GARCH or log-GARCH models with DCC. 
  
  Second, two quasi-maximum likelihood estimators are proposed in Section \ref{section estimation} for models without and with low-rank constraints on coefficient matrices, respectively. 
  Their asymptotic properties are established in Theorems \ref{thm Consistency}--\ref{thm Asymptotic normality under low rank}, thereby enabling statistical inference. 
  The full QMLE of the proposed model under low-rank restrictions remains feasible when the dimension $m \geq 20$; see Section \ref{section simulation} for details. However, when $m\geq 10$, the full QMLE of existing multivariate GARCH models such as diagonal CCC models may fail numerically \cite[p.~303]{francq2019garch}. 
  We may argue that the proposed model can accommodate higher dimensions, although the infeasibility of existing methods may also arise from other computational difficulties. 
  To the best of our knowledge, Theorems \ref{thm Consistency}--\ref{thm Asymptotic normality under low rank} provide the first rigorous theoretical results on estimation within the framework of DCC models. The proof technique can be extended to establish asymptotic properties for other multivariate GARCH or log-GARCH models with DCC. 
  
  Third, the Bayesian information criterion (BIC) is considered for order selection in Section \ref{subsection model selection}, and its selection consistency is proved in Theorem \ref{thm BIC}. 
  To our knowledge, the literature has paid little attention to order selection for multivariate GARCH models and even less to selection consistency. 
  Fourth, hypothesis testing for volatility spillover effects is studied in Section \ref{subsection test}.
  Finally, the new framework can successfully capture the conditional covariances and correlations among many assets. 
  The empirical results in Section \ref{section real data} indicate that the proposed framework generally outperforms the commonly used DCC and BEKK models in terms of out-of-sample forecasting. 

The rest of this paper is organized as follows. 
Section \ref{section SGARCH(r,s)} introduces the new multivariate volatility model and studies its probabilistic properties. 
Section \ref{section estimation} gives two quasi-maximum likelihood estimators, a BIC for order selection, and hypothesis tests for volatility spillovers, together with theoretical justifications. 
Simulation studies and empirical analysis are provided in Sections \ref{section simulation} and \ref{section real data}, respectively. 
Section \ref{section conclusion} concludes the paper with a discussion. 
All technical details and additional results for simulation and empirical analysis are relegated to the Appendix. 
Throughout this paper, 
$\|\cdot\|_{r}$ denotes either the $\ell_{r}$ vector norm or the matrix norm induced by the $\ell_{r}$ vector norm. 
For a matrix $A$, $\rho(A)$, $\rank(A)$ and $A^{g}$ denote its spectral radius, rank and generalized inverse, respectively. 
For a matrix $A = [a_{ij}]$,  
$\ovec(A) = (a_{11}, \ldots, a_{m1}, \ldots, a_{1m}, \ldots, a_{mm})^{\prime}$, 
$\ovech(A) = (a_{11}, \ldots, a_{m1}, a_{22}, \ldots, a_{m2}, \ldots, a_{mm})^{\prime}$ 
and $\ovechsec(A) = (a_{21}, \ldots, a_{m1}, a_{32}, \ldots, a_{m2}, \ldots, a_{m,m-1})^{\prime}$. 
The function 
$\ln^{+}(x)$ is defined as $\max\{\ln(x), 0\}$ for $x > 0$. 
Moreover, $\to_{d}$ and $\to_{p}$ denote the convergence in distribution and probability, respectively. 
The dataset in Section \ref{section real data} and computer programs for the analysis are available at \url{https://github.com/wyLI2020/MGARCH}.

\section{An efficient multivariate log-GARCH model} \label{section SGARCH(r,s)}

\subsection{Motivation} \label{section motivation}
  
Consider an $m$-dimensional multivariate time series $\{\mathbf{y}_{t}\}$ that follows a multivariate log-GARCH model \citep{francq2017equation} with dynamic conditional correlation (DCC): 
\begin{align}
  &\mathbf{y}_{t} = H_{t}^{1/2} \bm{\eta}_{t}, \;\;
  H_{t} = D_{t} R_{t} D_{t}, \;\;
  R_{t} = R(\mathbf{y}_{t-1}, \mathbf{y}_{t-2}, \ldots; \bm{\beta}), \label{eq yt, Ht, Rt in DCC}\\
  &\ln\bm{h}_{t} = \bm{\omega} + \sum_{i=1}^{q} A_{i} \ln\mathbf{y}^{2}_{t-i} + \sum_{j=1}^{p} B_{j} \ln\bm{h}_{t-j}, \label{eq volatility in DCC}
\end{align} 
where $\mathbf{y}_{t} = (y_{1t}, \ldots, y_{mt})^{\prime}$, $\{\bm{\eta}_{t}\}$ is an independent and identically distributed ($i.i.d.$) sequence of $m$-dimensional random vectors with zero mean and identity covariance matrix, 
$H_{t}$ is the conditional covariance matrix of $\mathbf{y}_{t}$ given the $\sigma$-field $\mathcal{F}_{t-1} = \sigma\{\mathbf{y}_{t-1}, \mathbf{y}_{t-2}, \ldots\}$, 
and $D_{t} = \Diag\{h_{11,t}^{1/2}, \ldots, h_{mm,t}^{1/2}\}$ with $H_{t} = [h_{ij,t}]$. Here $R_{t}$ is the conditional correlation matrix of $\mathbf{y}_{t}$ given $\mathcal{F}_{t-1}$, which depends on a parameter vector $\bm{\beta}$ and can take the form of the DCC model in \cite{tse2002multivariate} or \cite{engle2002dynamic}. 
For model \eqref{eq volatility in DCC} of order $(p,q)$, 
$\ln\bm{h}_{t} = (\ln h_{11,t}, \ldots, \ln h_{mm,t})^{\prime}$, $\ln\mathbf{y}^{2}_{t} = (\ln y_{1t}^{2}, \ldots, \ln y_{mt}^{2})^{\prime}$, 
$\bm{\omega}$ is an $m$-dimensional vector with positive coefficients, 
and $A_{i}$'s and $B_{j}$'s are $m \times m$ coefficient matrices. 
Note that model \eqref{eq volatility in DCC} extends the univariate log-GARCH($p,q$) specification to allow for volatility spillover effects and reduces to $m$ univariate log-GARCH($p,q$) equations if $A_{i}$'s and $B_{j}$'s are all diagonal. 
However, model \eqref{eq volatility in DCC} in general has two major drawbacks: 
\begin{itemize}
  \item[(i)] Non-identifiability. For illustration, consider model \eqref{eq volatility in DCC} of order $(1,1)$. A unimodular matrix $M$ may exist such that model \eqref{eq volatility in DCC} formulated with matrices $(A_1,B_1)$ corresponds to the same process as the model formulated with matrices $(MA_1,MB_1)$; see \cite{jeantheau1998strong}. 
  \item[(ii)] Computational intractability. 
  The equivalent form \eqref{eq ARCHinfinity for volatility in DCC when pqOne} of model \eqref{eq volatility in DCC} of order $(1,1)$ indicates that the likelihood function involves high-order matrix polynomials, introducing two computational challenges. First, the resulting optimization landscape is highly rugged, often requiring a large number of iterations to achieve convergence. Second, each iteration entails evaluating high-order matrix powers. Since a single matrix multiplication has computational complexity of order $O(m^{3})$, repeated multiplications for high-order powers become prohibitively expensive. Hence the overall estimation procedure is computationally intractable even for moderate values of $m$. 
\end{itemize}
A simple remedy for drawbacks (i) and (ii) is to impose diagonality on $B_j$'s \citep{francq2017equation}, or even on $A_{i}$'s \citep{tse2002multivariate,engle2002dynamic,aielli2013dynamic}. 
However, the price of this parsimony is the exclusion of potential oscillatory dynamics or even volatility spillover effects among assets. 
Instead, we explore the root causes of drawbacks (i) and (ii) and seek a compromise between the model generality and computational efficiency. 
This paper focuses on the DCC-log-GARCH model in \eqref{eq yt, Ht, Rt in DCC}--\eqref{eq volatility in DCC} of order $(1,1)$. It can be rewritten as an equivalent multivariate log-ARCH($\infty$) form with 
\begin{equation} \label{eq ARCHinfinity for volatility in DCC when pqOne}
  \ln\bm{h}_{t} = (I_{m} - B_{1})^{-1} \bm{\omega} + \sum_{i=1}^{\infty} B_{1}^{i-1} A_{1} \ln\mathbf{y}^{2}_{t-i}, 
\end{equation}
where $B_{1}$ is assumed to satisfy $\rho(B_{1}) < 1$. 
Note that the interaction between $A_{1}$ and $B_{1}$ results in the non-identifiability of model \eqref{eq volatility in DCC} or \eqref{eq ARCHinfinity for volatility in DCC when pqOne}. 
Moreover, the high-order matrix powers in the likelihood function make the computation intractable even for a moderate value of $m$. 
Consequently, both drawbacks (i) and (ii) arise from the matrix multiplications in $B_{1}^{i} A_{1}$, and thus eliminating these matrix multiplications is key to overcoming the drawbacks.   

Notably, diagonalizable matrices are dense in the space of all complex matrices \citep[Theorem 9.34]{andreescu2016essential}. Meanwhile, the set of non-diagonalizable matrices forms a proper algebraic subset of the same space and hence has Lebesgue measure zero \citep[Section 0.16]{wonham1979linear}. 
Thus, with only negligible loss of generality, we assume that $B_{1}$ is diagonalizable over the complex field. 
Then we can block-diagonalize $B_{1}$ using the real Jordan decomposition $B_{1} = P J P^{-1}$, and reparameterize model \eqref{eq ARCHinfinity for volatility in DCC when pqOne}, where $P$ is an invertible matrix and $J$ is the real Jordan form containing eigenvalues of $B_{1}$.  
The eigenvalues of $B_{1}$ may be real, complex or zero. 
Let $r$, $s$ and $m-r-2s$ denote the numbers of nonzero real eigenvalues, conjugate pairs of nonzero complex eigenvalues, and zero eigenvalues, respectively, where $0 \leq r, 2s, m-r-2s \leq m$. Then $J = \Diag\{\lambda_{1}, \ldots, \lambda_{r}, C_{1}, \ldots, C_{s}, \bm{0}_{m-r-2s}\}$ is an $m \times m$ real block diagonal matrix with $0 < |\lambda_{j}| < 1$ for $1 \leq j \leq r$, and 
  $$
  C_{k} = \gamma_{k} 
  \left(
  \begin{matrix}
    \cos\varphi_{k} & \sin\varphi_{k} \\
    -\sin\varphi_{k} & \cos\varphi_{k}
  \end{matrix}
  \right), 
  \;\; 0 < \gamma_{k} < 1 \;\; \text{and} \;\; 0 < \varphi_{k} < \pi 
  \;\; \text{for} \;\; 1 \leq k \leq s. 
  $$
Substituting the Jordan decomposition $B_{1} = P J P^{-1}$ into model \eqref{eq ARCHinfinity for volatility in DCC when pqOne}, we have 
\begin{equation} \label{eq decomposition of ht in DCC}
  \ln\bm{h}_{t} = \overline{\bm{\omega}} + \sum_{i=1}^{\infty} \left\{
  \sum_{k=1}^{r} \lambda_{k}^{i-1} \overline{G}_{0,k} 
  + \sum_{k=1}^{s} \gamma_{k}^{i-1} \left[\cos((i-1) \varphi_{k}) \overline{G}_{1,k} + \sin((i-1) \varphi_{k}) \overline{G}_{2,k}\right]
  \right\} \ln\mathbf{y}^{2}_{t-i}, 
\end{equation}
where $\overline{\bm{\omega}} = (I_{m} - B_{1})^{-1} \bm{\omega}$, and the matrices $\overline{G}_{0,k}$, $\overline{G}_{1,k}$ and $\overline{G}_{2,k}$ are defined via the column vectors of $P$ and $(P^{-1} A_{1})^{\prime}$; 
see Section \ref{section connection between DCC and SGARCH} of the Appendix for details. 
Critically, this decomposition replaces the matrix multiplications in $B_{1}^{i} A_{1}$ with scalar powers, thereby addressing both drawbacks of model \eqref{eq volatility in DCC}. First, it ensures identifiability by removing interactions among coefficient matrices that could generate equivalent parameterizations; see Proposition \ref{propo Identification}. 
Second, it reduces computational cost from $O(m^{3})$ per matrix power to $O(1)$ per scalar operation, and smooths the likelihood surface to accelerate convergence; see Section \ref{subsection additional simulation} of the Appendix for numerical evidence. 
\begin{remark}[Additional computational benefit] \label{remark non-normality risk}
A further advantage of the Jordan decomposition concerns matrix non-normality \citep{Trefethen2005}. In model \eqref{eq volatility in DCC} or \eqref{eq ARCHinfinity for volatility in DCC when pqOne}, even with spectral radius $\rho(B_{1})<1$, powers of a non-normal $B_{1}$ can exhibit transient growth, amplifying numerical errors in high-order polynomial evaluations. This increases the cost of stable likelihood evaluation and may further complicate numerical optimization. By reducing matrix powers to scalar powers, model \eqref{eq decomposition of ht in DCC} circumvents this issue entirely.
\end{remark}

\subsection{The proposed model}

Based on the previous discussion, to overcome the aforementioned drawbacks (i) and (ii) of DCC models, 
we propose a new multivariate log-GARCH model with DCC as follows: 
\begin{align}
  &\mathbf{y}_{t} = H_{t}^{1/2} \bm{\eta}_{t}, \;\;
  H_{t} = D_{t} R_{t} D_{t}, \;\;
  R_{t} = (1 - \beta_{1} - \beta_{2}) \underline{R} + \beta_{1} \Psi_{t-1} + \beta_{2} R_{t-1}, \label{model Rt SGARCH(r,s)}\\
  &\ln\bm{h}_{t} = \underline{\bm{\omega}} + \sum_{i=1}^{\infty} \Phi_{i} \ln\mathbf{y}^{2}_{t-i}, \label{model Dt SGARCH(r,s)}
\end{align}
where 
$\mathbf{y}_{t} = (y_{1t}, \ldots, y_{mt})^{\prime}$, 
$H_{t}$ is the conditional covariance matrix of $\mathbf{y}_{t}$ given the $\sigma$-field $\mathcal{F}_{t-1} = \sigma\{\mathbf{y}_{t-1}, \mathbf{y}_{t-2}, \ldots\}$, 
$D_{t} = \Diag\{h_{11,t}^{1/2}, \ldots, h_{mm,t}^{1/2}\}$ with $H_{t} = [h_{ij,t}]$, 
and $\{\bm{\eta}_{t}\}$ is an $i.i.d.$ sequence of $m$-dimensional random vectors with zero mean and identity covariance matrix; these random vectors are nonzero with probability one. 
We employ the DCC specification in \cite{tse2002multivariate} for the conditional correlation matrix $R_{t}$, where $\underline{R}$ is an $m \times m$ positive-definite parameter matrix with unit diagonal and off-diagonal elements in $(-1, 1)$, 
$\Psi_{t-1}$ is the sample correlation matrix of $\{\bm{\varepsilon}_{t-1}, \ldots, \bm{\varepsilon}_{t-\Bbbk}\}$ with $\bm{\varepsilon}_{t} = D_{t}^{-1} \mathbf{y}_{t}$ and $\Bbbk \geq m$ to ensure its positive definiteness, 
and $\beta_{1}$ and $\beta_{2}$ are non-negative parameters satisfying $\beta_{1} + \beta_{2} < 1$. 
We may instead specify $R_{t}$ using other DCC models \citep{engle2002dynamic,aielli2013dynamic,hafner2009generalized}, and the theory established below can be extended accordingly.

For model \eqref{model Dt SGARCH(r,s)}, $\ln\bm{h}_{t} = (\ln h_{11,t}, \ldots, \ln h_{mm,t})^{\prime}$, $\ln \mathbf{y}^{2}_{t} = (\ln y_{1t}^{2}, \ldots, \ln y_{mt}^{2})^{\prime}$, 
$\underline{\bm{\omega}} = (\underline{\omega}_{1}, \ldots, \underline{\omega}_{m})^{\prime}$, 
and the $m \times m$ coefficient matrices $\Phi_{i}$'s are defined as 
\begin{align}\label{model Phii in Dt}
\Phi_{i} = \sum_{k=1}^{r} \lambda_{k}^{i-1} G_{0,k} 
+ \sum_{k=1}^{s} \gamma_{k}^{i-1} \left[\cos((i-1) \varphi_{k}) G_{1,k} + \sin((i-1) \varphi_{k}) G_{2,k}\right], 
\end{align}   
where the orders $r$ and $s$ are non-negative integers such that $r+2s \leq m$, 
$\lambda_{k}$'s, $\gamma_{k}$'s and $\varphi_{k}$'s are scalar parameters that satisfy $0 < |\lambda_{k}| < 1$, $0 < \gamma_{k} < 1$ and $0 < \varphi_{k} < \pi$, 
and $G_{0,k}$'s, $G_{1,k}$'s and $G_{2,k}$'s are $m \times m$ parameter matrices.    
For identification, we assume that $\lambda_1 > \cdots > \lambda_r$ and $\gamma_1 > \cdots > \gamma_s$. 
The order $(r,s)$ is prespecified for estimation in Sections \ref{subsection QMLE}--\ref{subsection QMLE under low rank}, and can be selected via BIC in practice; see Section \ref{subsection model selection}. 
Given $(r,s)$, the fully parameterized model has $d = m + (r + 2s) (1 + m^2) + m(m-1)/2 + 2$ free parameters, including $\underline{\bm{\omega}}$, $\{\lambda_k, G_{0,k}\}_{k=1}^{r}$, $\{\gamma_k, \varphi_k, G_{1,k}, G_{2,k}\}_{k=1}^{s}$, $\ovechsec(\underline{R})$, $\beta_{1}$ and $\beta_{2}$. 

\begin{remark}[Identification condition] \label{remark identification condition}
  Non-identifiability occurs if $\lambda_{k} = \lambda_{k'}$ or $(\gamma_{k},\varphi_{k}) = (\gamma_{k'},\varphi_{k'})$ for some $k \neq k'$, as the corresponding dynamic components can be merged, reducing the model order $r$ or $s$. 
  To avoid this identification issue, we assume that the $\lambda_{k}$'s and the $\gamma_{k}$'s are each distinct, and impose the ordering $\lambda_1 > \cdots > \lambda_r$ and $\gamma_1 > \cdots > \gamma_s$ without loss of generality. 
  This condition is standard in the literature \citep{zheng2022interpretable} and is mild in practice, as it excludes only the case with $\gamma_k = \gamma_{k'}$ but $\varphi_k \neq \varphi_{k'}$ for some $k \neq k'$, which constitutes a set of Lebesgue measure zero in the parameter space. 
\end{remark}

The proposed model in \eqref{model Rt SGARCH(r,s)}--\eqref{model Phii in Dt} belongs to the class of variance-correlation-type models and retains the parsimony of diagonal DCC models. Meanwhile, it preserves the generality of fully parameterized covariance models in capturing temporal and cross-sectional dependence among series, without incurring prohibitive computational cost.
In particular, model \eqref{model Dt SGARCH(r,s)} with $\Phi_{i}$'s defined in \eqref{model Phii in Dt} circumvents the matrix multiplications $B_{1}^{i} A_{1}$ in \eqref{eq ARCHinfinity for volatility in DCC when pqOne}, requiring only scalar multiplications and first-order polynomials of the matrices $G_{0,k}$'s, $G_{1,k}$'s and $G_{2,k}$'s. This makes the proposed model identifiable without imposing any sophisticated constraints (see Proposition \ref{propo Identification}), and substantially enhances computational efficiency through a smoother likelihood surface and reduced per-iteration complexity. 
Hence, the model in \eqref{model Dt SGARCH(r,s)}--\eqref{model Phii in Dt} overcomes the two drawbacks of model \eqref{eq volatility in DCC}. 
Overall, the proposed model provides a computationally tractable alternative to existing multivariate (log-)GARCH models, achieving a superior balance between flexibility and parsimony.

\subsection{The model with low-rank restrictions}

The fully parameterized model in \eqref{model Rt SGARCH(r,s)}--\eqref{model Phii in Dt} involves $d = m + (r + 2s) (1 + m^2) + m(m-1)/2 + 2$ parameters. While this is manageable for small $m$, the $O(m^{2})$ parameter growth in unrestricted matrices $G_{0,k}$, $G_{1,k}$ and $G_{2,k}$ can compromise parsimony when $m$ is moderate. To address this, we impose the following low-rank restrictions on $G_{0,k}$, $G_{1,k}$ and $G_{2,k}$: 
\begin{equation} \label{eq low rank restrictions}
  \rank(G_{0,k}) = 1 \;\; \text{and} \;\; 1 \leq \rank(G_{\ell,k}) \leq 2 \;\; \text{for} \;\; \ell =1,2, 
\end{equation}
and they can be reparameterized as in \eqref{G-decomposition} in Section \ref{subsection QMLE under low rank}. 
Under this low-rank specification, the number of free parameters decreases to $d^{*} = m + r + 2s  + 2m(r + 4s)  + m(m-1)/2 + 2$, reducing the number of volatility parameters from $O(m^{2})$ to $O(m)$. This offers a  dimensional advantage over model \eqref{eq volatility in DCC}, whose volatility parameters also scale as $O(m^{2})$.

\begin{remark}
\label{remark connection with DCC}
  The DCC-log-GARCH$(1,1)$ model in \eqref{eq yt, Ht, Rt in DCC}--\eqref{eq volatility in DCC} with $B_{1}$ diagonalizable over the complex field is a special case of the proposed model in \eqref{model Rt SGARCH(r,s)}--\eqref{model Phii in Dt} under the low-rank constraints in \eqref{eq equivalence condition}. 
  Specifically, when $B_{1}$ is diagonalizable, the DCC-log-GARCH$(1,1)$ model can be rewritten as the proposed model with parameter matrices given by: 
  \begin{equation} \label{eq equivalence condition}
    G_{0,k} = \bm{\zeta}_{k} \bm{\xi}_{k}^{\prime}, \;
    G_{1,k} = \bm{\zeta}_{r+k} \bm{\xi}_{r+k}^{\prime} + \bar{\bm{\zeta}}_{r+k} \bar{\bm{\xi}}_{r+k}^{\prime}, \; \text{and} \; 
    G_{2,k} = \bm{\zeta}_{r+k} \bar{\bm{\xi}}_{r+k}^{\prime} - \bar{\bm{\zeta}}_{r+k} \bm{\xi}_{r+k}^{\prime}. 
  \end{equation}
  Here $\bm{\xi}_{k}$'s, $\bar{\bm{\xi}}_{k}$'s, $\bm{\zeta}_{k}$'s and $\bar{\bm{\zeta}}_{k}$'s are $m \times 1$ vectors, and the representation in \eqref{eq equivalence condition} inherently satisfies the low-rank restrictions in \eqref{eq low rank restrictions}; see Section \ref{section connection between DCC and SGARCH} of the Appendix for notation details. 
  Since the set of non-diagonalizable matrices has Lebesgue measure zero, we may claim that our model with a low-rank specification is more general than the DCC-log-GARCH$(1,1)$ model, and the unrestricted model is much more general.
  Moreover, if $B_{1}$ is further assumed to be diagonal, the DCC-log-GARCH$(1,1)$ model can then be rewritten as the proposed model with order $(r,s) = (\rank(B_1), 0)$. 
  
  The above discussion also implies that the proposed model in \eqref{model Rt SGARCH(r,s)}--\eqref{model Phii in Dt} can be rewritten as a DCC-log-GARCH$(1,1)$ model whenever the matrices $G_{0,k}$, $G_{1,k}$, and $G_{2,k}$ satisfy \eqref{eq equivalence condition}.
  Importantly, when $(r,s)=(1,0)$, this condition is sufficient but not necessary. With $\Phi_{i} = \lambda_{1}^{i-1} G_{0,1}$, our volatility model in \eqref{model Dt SGARCH(r,s)} can be directly rewritten into 
  $$
    \ln\bm{h}_{t} = (1 - \lambda_{1}) \underline{\bm{\omega}} + G_{0,1} \ln\mathbf{y}^{2}_{t-1} + \lambda_{1} \ln\bm{h}_{t-1}, 
  $$
  where $G_{0,1}$ captures ARCH effects and volatility spillovers, and $\lambda_1$ implies common persistence across all volatility components. 
\end{remark}

\subsection{Stationarity conditions}

Next, we derive a stationarity condition for the model in \eqref{model Rt SGARCH(r,s)}--\eqref{model Phii in Dt}. 
Establishing the stationarity of $\{\mathbf{y}_{t}\}$ is challenging because the ARCH($\infty$) form in \eqref{model Dt SGARCH(r,s)} makes it infeasible to rewrite $\{\mathbf{y}_{t}\}$ as a Markov chain; ARCH($\infty$) processes are not Markovian in general \citep{Fryzlewicz_SubbaRao2011}. 
Note that $\bm{\varepsilon}_{t} = D_{t}^{-1} \mathbf{y}_{t} = D_{t}^{-1} H_{t}^{1/2} \bm{\eta}_{t} = R_{t}^{1/2} \bm{\eta}_{t}$, and hence models \eqref{model Rt SGARCH(r,s)}--\eqref{model Dt SGARCH(r,s)} can be rewritten as follows 
\begin{align}
	&\bm{\varepsilon}_{t} = R_{t}^{1/2} \bm{\eta}_{t}, \;\;
	R_{t} = (1 - \beta_{1} - \beta_{2}) \underline{R} + \beta_{1} \Psi_{t-1} + \beta_{2} R_{t-1}, \tag{5$*$} \label{rewritten model Rt SGARCH(r,s)}\\
	&\mathbf{y}_{t} = H_{t}^{1/2} \bm{\eta}_{t}, \;\;
	H_{t} = D_{t} R_{t} D_{t}, \;\;
	\ln\bm{h}_{t} = \underline{\bm{\omega}} + \sum_{i=1}^{\infty} \Phi_{i} \ln\mathbf{y}^{2}_{t-i}. \tag{6$*$} \label{rewritten model Dt SGARCH(r,s)}
\end{align}
We can rewrite \eqref{rewritten model Rt SGARCH(r,s)} into a Markov chain since $\Psi_{t-1}$ is the sample correlation matrix of $\{\bm{\varepsilon}_{t-1}, \ldots, \bm{\varepsilon}_{t-\Bbbk}\}$, making it possible to first establish the stationarity of $\{(R_{t}, \bm{\varepsilon}_{t})\}$.
For $t \geq 1$, $1 \leq k \leq \Bbbk$ and a constant $c$, let $\varpi_{t,k}(c) = {2m \beta_{1}}{[c (1 - \beta_{1} - \beta_{2}) \lambda_{\text{min}}(\underline{R})]^{-1/2}} \|\bm{\eta}_{t-k}\|_{2} + I(k=1)\beta_{2}$. 
Define the $\Bbbk \times \Bbbk$ matrix 
$W_{t}(c) = (\bm{\varpi}_{t}(c), \bar{I}_{\Bbbk-1}^{\prime})^{\prime}$ with $\bm{\varpi}_{t}(c) = (\varpi_{t,1}(c), \ldots, \varpi_{t,\Bbbk}(c))^{\prime}$ and $\bar{I}_{\Bbbk-1} = (I_{\Bbbk-1}, \mathbf{0}_{\Bbbk-1})$. 
Denote by $\varsigma(c) \triangleq \lim_{t \to +\infty} t^{-1} E[\ln(\|W_{1}(c)W_{2}(c) \cdots W_{t}(c)\|)]$ the top Lyapunov exponent of $\{W_{t}(c)\}$, where $\|\cdot\|$ is some matrix norm. 

\begin{lemma}[Stationarity for $\{(R_{t}, \bm{\varepsilon}_{t})\}$] \label{lemma Stationarity of epsilont}
	Suppose that $E\|\bm{\eta}_{t}\|_{2}^{\epsilon} < \infty$ holds for some $\epsilon \geq 1$. 
	Then the process $\{(R_{t}, \bm{\varepsilon}_{t})\}$ 
	defined in model \eqref{rewritten model Rt SGARCH(r,s)} admits a strictly stationary and non-anticipative solution. 
	Moreover, suppose that $\sum_{k=1}^{\Bbbk} \varepsilon_{i, t-k}^{2} \geq \underline{c}$ almost surely for some positive constant $\underline{c}$ and $1 \leq i \leq m$, $E(\ln^{+}\|W_{t}(\underline{c})\|) < \infty$ holds for some matrix norm $\|\cdot\|$, and $\varsigma(\underline{c})<0$. Then this strictly stationary and non-anticipative solution is unique and ergodic. 
\end{lemma}

Lemma \ref{lemma Stationarity of epsilont} provides a sufficient condition for the strict stationarity of $\{(R_{t}, \bm{\varepsilon}_{t})\}$. The conditions on matrices $\{W_{t}\}$ in Lemma \ref{lemma Stationarity of epsilont} are standard in the GARCH literature \citep{francq2019garch}; see also Assumption U2 of \cite{fermanian2017stationarity}. 

Note that $(\mathbf{y}_{t}, H_{t})$ is a measurable function of $(R_{t}, \bm{\varepsilon}_{t})$. Therefore, by Theorem 36.4 of \citet{billingsley1995}, $(\mathbf{y}_{t}, H_{t})$ is strictly stationary whenever $(R_{t}, \bm{\varepsilon}_{t})$ is strictly stationary. 
Hence, using Lemma \ref{lemma Stationarity of epsilont}, the following theorem establishes a sufficient condition for the strict stationarity of $\{\mathbf{y}_{t}\}$.
Let $\ln\bm{\varepsilon}^{2}_{t} = (\ln \varepsilon_{1t}^{2}, \ldots, \ln \varepsilon_{mt}^{2})^{\prime}$ with $\bm{\varepsilon}_{t} = (\varepsilon_{1t}, \ldots, \varepsilon_{mt})^{\prime}$. 

\begin{thm}[Stationarity for $\mathbf{y}_{t}$] \label{thm Stationarity}
	Under the conditions of Lemma \ref{lemma Stationarity of epsilont}, suppose that $E\|\ln\bm{\varepsilon}^{2}_{t}\|_{2} < \infty$ and $\sum_{k=1}^{r} (1 - |\lambda_{k}|)^{-1} \|G_{0,k}\| + \sum_{k=1}^{s} (1 - |\gamma_{k}|)^{-1} (\|G_{1,k}\| + \|G_{2,k}\|) < 1$ holds for some matrix norm $\|\cdot\|$ induced by a vector norm. 
	Then the process $\{\mathbf{y}_{t}\}$ defined in models \eqref{model Rt SGARCH(r,s)}--\eqref{model Phii in Dt} admits a unique, non-anticipative, strictly stationary and ergodic solution. 
\end{thm}

In comparison with the proofs for corrected and standard DCC models in \cite{aielli2013dynamic} and \cite{fermanian2017stationarity}, the proof of stationarity for $\{\mathbf{y}_{t}\}$ is nontrivial due to the ARCH($\infty$) form in model \eqref{model Dt SGARCH(r,s)} or \eqref{rewritten model Dt SGARCH(r,s)}. 
To the best of our knowledge, Theorem \ref{thm Stationarity} provides the first rigorous stationarity result for the specification of $R_{t}$ defined in \eqref{model Rt SGARCH(r,s)}.  
Note that the two-step proof strategy can also be adapted to derive stationarity conditions for DCC models in \cite{tse2002multivariate}. 

\begin{remark}[Extension to a general multivariate log-GARCH model] \label{remark general SGARCH} 
  The proposed models in Sections 2.2-2.4 are based on the DCC-log-GARCH$(1,1)$ specification, and we can extend the same idea to the case with general orders of $(p,q)$. 
  Specifically, 
  the proposed model in \eqref{model Rt SGARCH(r,s)}--\eqref{model Phii in Dt} can be generalized to a model of order $(q,r,s)$ with $\Phi_{i}$ defined by 
  \begin{footnotesize}
  \begin{align*}
    \Phi_{i} =  
    \sum_{k=1}^{q-1} I(i=k) G_{k} 
    + \sum_{k=1}^{r} I(i\geq q) \lambda_{k}^{i-q} G_{0,k}  
    + \sum_{k=1}^{s} I(i\geq q) \gamma_{k}^{i-q} 
    \left[\cos((i-q) \varphi_{k}) G_{1,k} 
    + \sin((i-q) \varphi_{k}) G_{2,k}\right], 
  \end{align*}
  \end{footnotesize}
  where $r$ and $s$ satisfy $r+2s \leq mp$, 
  $G_{k}$'s are $m \times m$ parameter matrices, 
  and the parameters $\lambda_{k}$'s, $\gamma_{k}$'s, $\varphi_{k}$'s, $G_{0,k}$'s, $G_{1,k}$'s and $G_{2,k}$'s are defined as in \eqref{model Phii in Dt}.
  Under suitable low-rank restrictions on $G_{0,k}$'s, $G_{1,k}$'s and $G_{2,k}$'s, this general model coincides with the model in \eqref{eq yt, Ht, Rt in DCC}--\eqref{eq volatility in DCC} under appropriate diagonalizability conditions on the matrices $B_{j}$'s; see Section \ref{section SGARCH(q,r*,s*)} of the Appendix for details. Furthermore, a sufficient condition for strict stationarity of this general model can be established as in Theorem \ref{thm Stationarity}, with the matrix norm condition replaced by 
  $\sum_{k=1}^{q-1} \|G_{k}\|_{2} + \sum_{k=1}^{r} (1 - |\lambda_{k}|)^{-1} \|G_{0,k}\| + \sum_{k=1}^{s} (1 - |\gamma_{k}|)^{-1} \left(\|G_{1,k}\| + \|G_{2,k}\|\right) < 1$. 
\end{remark}

\section{Statistical inference}\label{section estimation}

This section introduces the one-step Gaussian quasi-maximum likelihood estimators (QMLEs) for the proposed model in \eqref{model Rt SGARCH(r,s)}--\eqref{model Phii in Dt} without and with the low-rank restrictions, respectively.
Order selection using the BIC is also investigated for both QMLEs, and hypothesis tests are constructed to detect volatility spillover effects.

\subsection{QMLE without low-rank restrictions} \label{subsection QMLE}

Let $\bm{g}_{0,k} = \ovec(G_{0,k})$ for $1 \leq k \leq r$, 
$\bm{g}_{\ell,k} = \ovec(G_{\ell,k})$ for $\ell = 1, 2$ and $1 \leq k \leq s$, 
and $\underline{\bm{r}} = \ovechsec(\underline{R})$. 
Denote $\bm{\delta} = (\underline{\bm{\omega}}^{\prime}, \bm{\kappa}^{\prime})^{\prime}$ with 
$\bm{\kappa} = (\bm{\lambda}^{\prime}, \bm{\gamma}^{\prime}, \bm{\varphi}^{\prime}, \bm{g}_{0}^{\prime}, \bm{g}_{1}^{\prime}, \bm{g}_{2}^{\prime})^{\prime}$, 
and $\bm{\beta} = (\beta_{1}, \beta_{2}, \underline{\bm{r}}^{\prime})^{\prime}$, where 
$\bm{\lambda} = (\lambda_{1}, \ldots, \lambda_{r})^{\prime}$, 
$\bm{\gamma} = (\gamma_{1}, \ldots, \gamma_{s})^{\prime}$, 
$\bm{\varphi} = (\varphi_{1}, \ldots, \varphi_{s})^{\prime}$, 
$\bm{g}_{0} = (\bm{g}_{0,1}^{\prime}, \ldots, \bm{g}_{0,r}^{\prime})^{\prime}$, and 
$\bm{g}_{\ell} = (\bm{g}_{\ell,1}^{\prime}, \ldots, \bm{g}_{\ell,s}^{\prime})^{\prime}$ for $\ell = 1, 2$. 
Let $\bm{\theta} = (\bm{\delta}^{\prime}, \bm{\beta}^{\prime})^{\prime}$ be the parameter vector, with true value $\bm{\theta}_{0} = (\bm{\delta}_{0}^{\prime}, \bm{\beta}_{0}^{\prime})^{\prime}$. 
Denote the parameter space of $\bm{\theta}$ by $\Theta$, which is a compact subset of $\mathbb{R}^{m} \times \{(-1,0) \cup (0,1)\}^{r} \times (0,1)^{s} \times (0,\pi)^{s} \times \mathbb{R}^{(r+2s)m^{2}} \times (0,1)^{2} \times (-1,1)^{m(m-1)/2}$.

Let $\{\mathbf{y}_{1}, \ldots, \mathbf{y}_{n}\}$ be observations of $\{\mathbf{y}_{t}\}$ generated by the model in \eqref{model Rt SGARCH(r,s)}--\eqref{model Phii in Dt}.
Below we write 
$D_{t}(\bm{\delta})$, $R_{t}(\bm{\theta})$, $H_{t}(\bm{\theta})$, $\ln\bm{h}_{t}(\bm{\delta})$ and $\bm{\varepsilon}_{t}(\bm{\delta}) = D_{t}^{-1}(\bm{\delta}) \mathbf{y}_{t}$ to emphasize their dependence on parameters. 
Since these functions depend on observations in the infinite past, initial values are required in practice. 
Without loss of generality, we simply set the initial values of $\{\mathbf{y}_{s}, s \leq 0\}$ to vectors of ones, 
and denote the resulting functions by $\widetilde{D}_{t}(\bm{\delta})$, $\widetilde{R}_{t}(\bm{\theta})$, $\widetilde{H}_{t}(\bm{\theta})$, $\ln\widetilde{\bm{h}}_{t}(\bm{\delta})$ and $\widetilde{\bm{\varepsilon}}_{t}(\bm{\delta})$, respectively.   
These initial values can be shown to have asymptotically negligible effects on the inference. 
Given $\{\mathbf{y}_{1}, \ldots, \mathbf{y}_{n}\}$ and initial values of $\{\mathbf{y}_{s}, s \leq 0\}$, the negative conditional Gaussian log-likelihood function (ignoring a constant) can be written as 
\begin{equation} \label{eq loss Ltilde} 
  \widetilde{L}_{n}(\bm{\theta})=\sum_{t=1}^{n} \widetilde{\ell}_{t}(\bm{\theta}) \quad \text{with} \quad \widetilde{\ell}_{t}(\bm{\theta}) = \mathbf{y}_{t}^{\prime} \widetilde{H}_{t}^{-1}(\bm{\theta}) \mathbf{y}_{t}/2 + \ln |\widetilde{H}_{t}(\bm{\theta})|/2.
\end{equation}
Then the QMLE without any low-rank restrictions can be defined as 
\begin{equation} \label{est general QMLE} 
  \widehat{\bm{\theta}}_{\text{G}} = \argmin_{\bm{\theta} \in \Theta} \widetilde{L}_{n}(\bm{\theta}).
\end{equation}

To prove the asymptotic properties for $\widehat{\bm{\theta}}_{\text{G}}$, we introduce the assumptions below. 

\begin{assum} \label{assum stationarity}
  The process $\{\mathbf{y}_{t}\}$ is strictly stationary and ergodic. 
\end{assum}

\begin{assum} \label{assum parameters}
  ($\romannumeral1$) $\Theta$ is compact, and $\bm{\theta}_{0}$ lies in the interior of $\Theta$; 
  ($\romannumeral2$) for some $0 < \varrho < 1$, $0 < \sup_{\bm{\theta} \in \Theta} \max\{|\lambda_{1}|, \ldots, |\lambda_{r}|, \gamma_{1}, \ldots, \gamma_{s}\} \leq \varrho$ holds; 
  ($\romannumeral3$) $\underline{R}$ is positive definite for $\bm{\theta} \in \Theta$. 
\end{assum}

Assumptions \ref{assum stationarity}--\ref{assum parameters} impose basic conditions on the process $\{\mathbf{y}_{t}\}$ and the parameters. 
Theorem \ref{thm Stationarity} provides a sufficient condition for the strict stationarity required in Assumption \ref{assum stationarity}.
Compactness of the parameter space $\Theta$ in Assumption \ref{assum parameters}($\romannumeral1$) is standard for consistency, and the interior condition on the true value $\bm{\theta}_{0}$ is general and necessary for establishing asymptotic normality. 
Assumption \ref{assum parameters}($\romannumeral2$) ensures that the moment conditions on $\ln\bm{h}_{t}(\bm{\delta})$ and its derivatives hold uniformly on $\Theta$; see details in Lemma \ref{lemma moments of derivatives of lnhtunderline} of the Appendix. 
Finally, Assumption \ref{assum parameters}($\romannumeral3$) guarantees that $R_{t}(\bm{\theta})$ is positive definite for all $\bm{\theta} \in \Theta$. 

To establish the consistency of $\widehat{\bm{\theta}}_{\text{G}}$, we also need to prove that the true parameter vector $\bm{\theta}_{0}$ is the unique minimizer of the population loss $E[\ell_{t}(\bm{\theta})]$. 
Proposition \ref{propo Identification} states that the model in \eqref{model Rt SGARCH(r,s)}--\eqref{model Phii in Dt} is identifiable, which guarantees that $\bm{\theta}_{0}=\argmin_{\bm{\theta} \in \Theta}E[\ell_{t}(\bm{\theta})]$ is unique. 

\begin{propo}[Identification] \label{propo Identification}
  Suppose that $G_{\ell,k0} \neq 0_{m}$ for $1 \leq k \leq r$ when $\ell = 0$ and for $1 \leq k \leq s$ when $\ell = 1,2$, where $G_{\ell,k0}$ is the true value of $G_{\ell,k}$. 
  If Assumption \ref{assum stationarity} holds, 
  then 
  (\romannumeral1) the order $(r,s)$ is identifiable, i.e., no other order $(r',s')$ yields an equivalent representation of the model in \eqref{model Rt SGARCH(r,s)}--\eqref{model Phii in Dt}; 
  (\romannumeral2) the true value $\bm{\theta}_{0}$ is identifiable, i.e., the model in \eqref{model Rt SGARCH(r,s)}--\eqref{model Phii in Dt} does not hold true when $\bm{\theta}_{0}$ is replaced by $\bm{\theta} \neq \bm{\theta}_{0}$. 
\end{propo}     

Let $d = m + (r + 2s) (1 + m^2) + m(m-1)/2 + 2$ be the dimension of the parameter vector $\bm{\theta}$. Define the $d \times d$ matrices 
$\Sigma = E({\partial\ell_{t}(\bm{\theta}_{0})}/{\partial\bm{\theta}} {\partial\ell_{t}(\bm{\theta}_{0})}/{\partial\bm{\theta}^{\prime}})$ and $\Sigma_{*} = E({\partial^{2}\ell_{t}(\bm{\theta}_{0})}/{\partial\bm{\theta} \partial\bm{\theta}^{\prime}})$, 
where the derivatives in $\Sigma$ and $\Sigma_{*}$ are provided in Section \ref{section Derivatives} of the Appendix. 

\begin{thm} \label{thm Consistency}
  Suppose that Assumptions \ref{assum stationarity}--\ref{assum parameters} hold. 
  If $E\|\ln\mathbf{y}^{2}_{t}\|_{2} < \infty$, then $\widehat{\bm{\theta}}_{\text{G}} \to \bm{\theta}_{0}$ almost surely as $n \to \infty$. 
\end{thm}

\begin{thm} \label{thm Asymptotic normality}
  Suppose that Assumptions \ref{assum stationarity}--\ref{assum parameters} hold. 
  If $E\|\ln\mathbf{y}^{2}_{t}\|_{2}^{2+\epsilon} < \infty$ for some $\epsilon > 0$ and $E \|\bm{\eta}_{t} \bm{\eta}_{t}^{\prime}\|_{2}^{2} < \infty$, then $\sqrt{n} (\widehat{\bm{\theta}}_{\text{G}} - \bm{\theta}_{0}) \to_{d} N(\bm{0}, \Sigma_{\text{G}})$ as $n \to \infty$, where $\Sigma_{\text{G}}=\Sigma_{*}^{-1} \Sigma \Sigma_{*}^{-1}$. 
\end{thm}

Theorems \ref{thm Consistency} and \ref{thm Asymptotic normality} establish the strong consistency and asymptotic normality of the QMLE $\widehat{\bm{\theta}}_{\text{G}}$ in \eqref{est general QMLE}, respectively. 
If $\bm{\eta}_{t}$ is multivariate normal, then the QMLE reduces to the MLE, and the asymptotic normality in Theorem \ref{thm Asymptotic normality} can be simplified to $\sqrt{n} (\widehat{\bm{\theta}}_{\text{G}} - \bm{\theta}_{0}) \to_{d} N(\bm{0}, \Sigma_{*}^{-1})$ as $n \to \infty$. 
In practice, the asymptotic covariance matrix $\Sigma_{\text{G}}$ can be approximated by $\widehat{\Sigma}_{\text{G}} = \widehat{\Sigma}_{*}^{-1} \widehat{\Sigma} \widehat{\Sigma}_{*}^{-1}$, where $\widehat{\Sigma} = 1/n \sum_{t=1}^{n} {\partial\widetilde{\ell}_{t}(\widehat{\bm{\theta}}_{\text{G}})}/{\partial\bm{\theta}} {\partial\widetilde{\ell}_{t}(\widehat{\bm{\theta}}_{\text{G}})}/{\partial\bm{\theta}^{\prime}}$ and $\widehat{\Sigma}_{*} = 1/n \sum_{t=1}^{n} {\partial^{2}\widetilde{\ell}_{t}(\widehat{\bm{\theta}}_{\text{G}})}/{\partial\bm{\theta} \partial\bm{\theta}^{\prime}}$.
The following proposition establishes its consistency. 

\begin{propo} \label{propo approximates of asymptotic variances}
  If the conditions of Theorem \ref{thm Asymptotic normality} hold, then $\widehat{\Sigma}_{\text{G}} \to_{p} \Sigma_{\text{G}}$ as $n \to \infty$. 
\end{propo}

\subsection{QMLE with low-rank restrictions} \label{subsection QMLE under low rank}

When $m$ is moderate, we consider the model in \eqref{model Rt SGARCH(r,s)}--\eqref{model Phii in Dt} with the low-rank constraints in \eqref{eq low rank restrictions}, i.e., $\rank(G_{0,k}) = 1$ and $1 \leq \rank(G_{\ell,k}) \leq 2$ for $\ell = 1,2$.  
To impose these low-rank restrictions, we decompose the $m \times m$ parameter matrices $G_{\ell, k}$'s as follows: 
\begin{equation}\label{G-decomposition}
  G_{0,k} = \bm{g}_{0,k,1} \bm{g}_{0,k,2}^{\prime}, \;
  G_{1,k} = \bm{g}_{1,k,1} \bm{g}_{1,k,2}^{\prime} + \bm{g}_{1,k,3} \bm{g}_{1,k,4}^{\prime} \; \text{and} \; 
  G_{2,k} = \bm{g}_{2,k,1} \bm{g}_{2,k,2}^{\prime} + \bm{g}_{2,k,3} \bm{g}_{2,k,4}^{\prime},
\end{equation}
where $\bm{g}_{0,k,i}$, $\bm{g}_{1,k,j}$ and $\bm{g}_{2,k,j}$ for $i = 1,2$ and $j = 1,2,3,4$ are $m \times 1$ parameter vectors. 
Note that the decomposition in \eqref{G-decomposition} includes \eqref{eq equivalence condition} in Remark \ref{remark connection with DCC} as a special case. 

Denote 
$\bm{g}_{0,k}^{*} = (\bm{g}_{0,k,1}^{\prime}, \bm{g}_{0,k,2}^{\prime})^{\prime}$ for $1 \leq k \leq r$, and 
$\bm{g}_{\ell,k}^{*} = (\bm{g}_{\ell,k,1}^{\prime}, \bm{g}_{\ell,k,2}^{\prime}, \bm{g}_{\ell,k,3}^{\prime}, \bm{g}_{\ell,k,4}^{\prime})^{\prime}$ for $\ell = 1,2$ and $1 \leq k \leq s$. 
Let $\bm{\vartheta} = (\bm{\delta}^{* \prime}, \bm{\beta}^{\prime})^{\prime}$ be the parameter vector of the model in \eqref{model Rt SGARCH(r,s)}--\eqref{model Phii in Dt} under the low-rank decomposition \eqref{G-decomposition}, where 
$\bm{\delta}^{*} = (\underline{\bm{\omega}}^{\prime}, \bm{\kappa}^{* \prime})^{\prime}$ and 
$\bm{\kappa}^{*} = (\bm{\lambda}^{\prime}, \bm{\gamma}^{\prime}, \bm{\varphi}^{\prime}, \bm{g}_{0}^{* \prime}, \bm{g}_{1}^{* \prime}, \bm{g}_{2}^{* \prime})^{\prime}$ with 
$\bm{g}_{0}^{*} = (\bm{g}_{0,1}^{* \prime}, \ldots, \bm{g}_{0,r}^{* \prime})^{\prime}$ and $\bm{g}_{\ell}^{*} = (\bm{g}_{\ell,1}^{* \prime}, \ldots, \bm{g}_{\ell,s}^{* \prime})^{\prime}$ for $\ell = 1, 2$. 
Let $\Theta^{*}$ denote the parameter space of $\bm{\vartheta}$, which is a compact subset of $\mathbb{R}^{m} \times \{(-1,0) \cup (0,1)\}^{r} \times (0,1)^{s} \times (0,\pi)^{s} \times \mathbb{R}^{(2r+8s)m} \times (0,1)^{2} \times (-1,1)^{m(m-1)/2}$.
Let $\bm{\theta}(\bm{\vartheta})$ be the mapping from $\bm{\vartheta}$ to $\bm{\theta}$. 
Then the QMLE under the low-rank decomposition \eqref{G-decomposition} can be defined as 
\begin{equation} \label{est lowrank QMLE} 
  \widehat{\bm{\theta}}_{\text{LR}} = \argmin_{\bm{\theta} = \bm{\theta}(\bm{\vartheta}), \bm{\vartheta} \in \Theta^{*}} \widetilde{L}_{n}(\bm{\theta}), 
\end{equation}
where $\widetilde{L}_{n}(\bm{\theta})$ is the negative log-likelihood function defined in \eqref{eq loss Ltilde}. 
Note that $d-d^{*}=(r+2s)m^{2} - 2(r+4s)m$ is non-negative for $m\geq 4$ under the order constraint that $r+2s \leq m$, and is non-negative for $m\geq 2$ if the order $s=0$. 
As a result, for moderate dimensions the model in \eqref{model Rt SGARCH(r,s)}--\eqref{model Phii in Dt} under the low-rank decomposition \eqref{G-decomposition} usually has fewer parameters than the model without any low-rank restrictions. 

Let $\Delta = {\partial \bm{\theta}(\bm{\vartheta})} / {\partial \bm{\vartheta}^{\prime}}$ be the Jacobian matrix of $\bm{\theta}$ with respect to $\bm{\vartheta}$, which is known under the low-rank decomposition \eqref{G-decomposition}. 
Let $P_{\vartheta} = \Delta (\Delta^{\prime} \Sigma_{*} \Delta)^{g} \Delta^{\prime} \Sigma_{*}$ be the projection matrix. 

\begin{thm} \label{thm Asymptotic normality under low rank}
  If the conditions in Theorem \ref{thm Asymptotic normality} hold, then $\sqrt{n} (\widehat{\bm{\theta}}_{\text{LR}} - \bm{\theta}_{0}) \to_{d} N(\bm{0}, \Sigma_{\text{LR}})$ as $n \to \infty$, 
  where $\Sigma_{\text{LR}} = P_{\vartheta} \Sigma_{\text{G}} P_{\vartheta}^{\prime}$. 
\end{thm}

Theorem \ref{thm Asymptotic normality under low rank} establishes the asymptotic normality of the QMLE $\widehat{\bm{\theta}}_{\text{LR}}$ in \eqref{est lowrank QMLE}, and it is derived using the technique for overparameterized models in \cite{shapiro1986asymptotic} without requiring identification of $\bm{\vartheta}$. 
The consistency of $\widehat{\bm{\theta}}_{\text{LR}}$ therefore follows. 
Moreover, if the innovations $\{\bm{\eta}_{t}\}$ are multivariate normal, then the asymptotic normality in Theorem \ref{thm Asymptotic normality under low rank} can be simplified to $\sqrt{n} (\widehat{\bm{\theta}}_{\text{LR}} - \bm{\theta}_{0}) \to_{d} N(\bm{0}, \Delta (\Delta^{\prime} \Sigma_{*} \Delta)^{g} \Delta^{\prime})$ as $n \to \infty$. 
Finally, the covariance matrix $\Sigma_{\text{LR}}$ of $\widehat{\bm{\theta}}_{\text{LR}}$ can be consistently estimated by $\widehat{\Sigma}_{\text{LR}} = \widehat{P}_{\vartheta} \widehat{\Sigma}_{\text{G}} \widehat{P}_{\vartheta}^{\prime}$, where $\widehat{P}_{\vartheta} = \Delta (\Delta^{\prime} \widehat{\Sigma}_{*} \Delta)^{g} \Delta^{\prime} \widehat{\Sigma}_{*}$, and both $\widehat{\Sigma}_{\text{G}}$ and $\widehat{\Sigma}_{*}$ are defined before Proposition \ref{propo approximates of asymptotic variances}.

\subsection{Model selection} \label{subsection model selection}

To select the order $(r,s)$ for the proposed model in \eqref{model Rt SGARCH(r,s)}--\eqref{model Phii in Dt}, we introduce the following Bayesian information criterion (BIC): 
\begin{equation} \label{BIC}
  \text{BIC}(r,s) =  2\widetilde{L}_{n}\left(\widehat{\bm{\theta}}^{(r,s)}\right) + \bar{d} \ln(n), 
\end{equation}
where the QMLE $\widehat{\bm{\theta}}^{(r,s)}$ is  $\widehat{\bm{\theta}}_{\text{G}}$ defined in \eqref{est general QMLE} without the low-rank restrictions or $\widehat{\bm{\theta}}_{\text{LR}}$ in \eqref{est lowrank QMLE} with the low-rank restrictions when the order is set to $(r,s)$, $\bar{d}$ is the corresponding dimension of $\widehat{\bm{\theta}}^{(r,s)}$ (i.e. $d$ for $\widehat{\bm{\theta}}_{\text{G}}$ or $d^{*}$ for $\widehat{\bm{\theta}}_{\text{LR}}$), and $\widetilde{L}_{n}(\widehat{\bm{\theta}}^{(r,s)})$ is the negative log-likelihood function in \eqref{eq loss Ltilde} evaluated at $\widehat{\bm{\theta}}^{(r,s)}$. 
Denote $\Pi = \{(r,s): 1 \leq r+2s \leq K_{\text{max}}\}$, where $K_{\text{max}}$ is a predetermined positive integer satisfying $K_{\text{max}} \leq m$ to guarantee the order constraint that $r+2s \leq m$. 
Here, we exclude the case of $(r,s) = (0,0)$ from the space $\Pi$ to prevent the degenerate specification $\ln\bm{h}_{t} = \underline{\bm{\omega}}$. 
Let $(\widehat{r}, \widehat{s}) = \argmin_{(r,s) \in \Pi} \text{BIC}(r,s)$. 
The theorem below verifies the selection consistency of the proposed BIC. 

\begin{thm} \label{thm BIC}
  Suppose that the conditions in Theorem \ref{thm Asymptotic normality} hold. 
  If the true order $(r_{0}, s_{0}) \in \Pi$, then $(\widehat{r}, \widehat{s})$ converges to $(r_{0}, s_{0})$ in probability as $n \to \infty$. 
\end{thm}

\subsection{Testing for volatility spillover effects}\label{subsection test}

In this subsection, we construct significance tests to detect volatility spillover effects based on the proposed model in \eqref{model Rt SGARCH(r,s)}--\eqref{model Phii in Dt}. 
Model \eqref{model Dt SGARCH(r,s)} shows that volatility spillover effects exist whenever the off-diagonal elements of the coefficient matrices $\Phi_{i}$ are nonzero. 
This motivates us to detect volatility spillovers by testing the significance of $\Phi_{i}$'s off-diagonal elements. 
Note that the elements of $\Phi_{i}$ in \eqref{model Phii in Dt} decay exponentially as the lag $i$ increases, and thus $\Phi_{1}$ is the leading term for investigating volatility spillover effects.  
As a result, it is natural to focus on the significance of $\Phi_{1}$'s off-diagonal elements $\Phi_{1,ij}$ for $i\neq j$. 

Consider the hypothesis testing problem $H_0: \Phi_{10,ij}=0$ versus $H_1: \Phi_{10,ij}\neq 0$, where $\Phi_{10,ij}$ is the true value of $\Phi_{1,ij}$. 
Note that $\Phi_{1} = \sum_{k=1}^{r} G_{0,k} + \sum_{k=1}^{s} G_{1,k}$ is a linear combination of $G_{0,k}$ and $G_{1,k}$ by \eqref{model Phii in Dt}. 
Then $\Phi_{1,ij}=\bm{c}_{ij}^{\prime} \bm{\theta}$, where $\bm{c}_{ij}$ is a constant vector with ones in the $[d_{ij} + (k_{1}-1)m^{2}]$-th and $[d_{ij} + rm^2 + (k_{2}-1)m^{2}]$-th positions for $1 \leq k_{1} \leq r$ and $1 \leq k_{2} \leq s$, and zeros elsewhere, with $d_{ij} = m+r+2s + (j-1)m+i$. 
By Theorems \ref{thm Asymptotic normality}--\ref{thm Asymptotic normality under low rank}, we obtain $\sqrt{n} (\widehat{\Phi}_{1,ij} - \Phi_{10,ij}) \to_{d} N(0, \bm{c}_{ij}^{\prime}\Sigma_{\theta}\bm{c}_{ij})$ as $n \to \infty$, where $\widehat{\Phi}_{1,ij}=\bm{c}_{ij}^{\prime}\widehat{\bm{\theta}}$, $\widehat{\bm{\theta}}$ is either $\widehat{\bm{\theta}}_{\text{G}}$ or $\widehat{\bm{\theta}}_{\text{LR}}$, and $\Sigma_{\theta}$ is its asymptotic covariance matrix. 
To test the null hypothesis $H_0$, we construct the statistic $z_{ij}=\widehat{\Phi}_{1,ij}/\text{s.e.}(\widehat{\Phi}_{1,ij})$, where $\text{s.e.}(\widehat{\Phi}_{1,ij})=(\bm{c}_{ij}^{\prime} \widehat{\Sigma}_{\theta} \bm{c}_{ij}/n)^{1/2}$ is the standard error of $\widehat{\Phi}_{1,ij}$ and $\widehat{\Sigma}_{\theta}$ estimates $\Sigma_{\theta}$.  

Under $H_0$, we have $z_{ij}\to_d N(0,1)$ as $n \to \infty$. Thus, we can reject $H_0$ at the significance level $\alpha$ if $|z_{ij}|>z_{1-\alpha/2}$, where $z_{\alpha}$ is the $\alpha$-th quantile of the standard normal distribution. A rejection indicates significant volatility transmission from the $j$-th asset to the $i$-th asset.

\section{Simulation studies} \label{section simulation}

\subsection{Model interpretation} \label{subsection model interpretation}

The first experiment illustrates the model structure by linking the coefficient matrices $\Phi_{i}$'s in \eqref{model Phii in Dt} to the autocorrelation of squared returns (or volatilities). 
To isolate the roles of parameters $\lambda_{k}$'s, $\gamma_{k}$'s and $\varphi_{k}$'s, we consider three specifications of models \eqref{model Rt SGARCH(r,s)}--\eqref{model Dt SGARCH(r,s)} with $m=2$: 
(a) $\Phi_{i} = \lambda_{1}^{i-1} G_{0,1}$, with $\lambda_{1}$ varying; 
(b) $\Phi_{i} = \gamma_{1}^{i-1} [\cos((i-1) \pi/7) G_{1,1} + \sin((i-1) \pi/7) G_{2,1}]$, with $\gamma_{1}$ varying; 
and (c) $\Phi_{i} = 0.8^{i-1} [\cos((i-1) \varphi_{1}) G_{1,1} + \sin((i-1) \varphi_{1}) G_{2,1}]$, with $\varphi_{1}$ varying. 

Figures \ref{figure ACF_squared_return}--\ref{figure ACF_squared_volatility} in the Appendix display the autocorrelation functions (ACFs) of squared returns $\{y_{1t}^2\}$ and squared volatilities $\{h_{11,t}\}$ under models (a)--(c), respectively. 
We have the following findings:  
(i) the term $\lambda_{1}^{i-1} G_{0,1}$ generates an exponential decay in the volatility process, while the term $\gamma_{1}^{i-1} [\cos((i-1) \varphi_{1}) G_{1,1} + \sin((i-1) \varphi_{1}) G_{2,1}]$ produces damped oscillatory patterns (i.e., sine and cosine waves); 
(ii) $\lambda_{1}$ governs the persistence of volatility, with a larger $|\lambda_{1}|$ implying stronger persistence; 
(iii) $\gamma_{1}$ controls the amplitude and damping rate of the oscillatory volatility cycles, with a larger $\gamma_{1}$ yielding a larger amplitude and a slower damping rate; 
and (iv) $\varphi_{1}$ determines the frequency of the damped volatility oscillations, with a larger $\varphi_{1}$ implying a higher frequency.

\subsection{Model estimation} \label{subsection model estimation}

The second experiment examines the finite-sample performance of the proposed QMLEs $\widehat{\bm{\theta}}_{\text{G}}$ in \eqref{est general QMLE} and $\widehat{\bm{\theta}}_{\text{LR}}$ in \eqref{est lowrank QMLE}.  
We consider five data-generating processes (DGPs) from the proposed model in \eqref{model Rt SGARCH(r,s)}--\eqref{model Phii in Dt}; see Section \ref{subsection additional results for simulation} in the Appendix for details. 
$\text{DGP}1: m=2,(r_0,s_0) = (1,0)$; $\text{DGP}2: m=2,(r_0,s_0) = (2,0)$; $\text{DGP}3: m=2,(r_0,s_0) = (0,1)$; $\text{DGP4}: m=5, (r_0,s_0) = (1,0)$; and $\text{DGP5}: m=20, (r_0,s_0) = (1,0)$. 
Here DGP1--DGP4 and DGP5 represent the data-generating processes of small and moderate dimensions, respectively.
All coefficient matrices in these DGPs satisfy \eqref{eq equivalence condition}, ensuring that each DGP admits a DCC-log-GARCH$(1,1)$ representation, as discussed in Remark \ref{remark connection with DCC}.
Notably, DGP4 and DGP5 correspond to the models with low-rank coefficient matrices. 
Accordingly, we apply $\widehat{\bm{\theta}}_{\text{G}}$ to DGP1--DGP4 and $\widehat{\bm{\theta}}_{\text{LR}}$ to DGP4--DGP5. 
In addition, the innovations $\{\bm{\eta}_{t}\}$ are $i.i.d.$ random vectors following either the multivariate normal distribution or Student's $t_5$ distribution, both with zero mean and identity covariance matrix. For DGP1–DGP4, we consider $n = 1000$ and $2000$; for DGP5, we additionally include $n = 3000$ to accommodate the higher dimension. Each setting is replicated $1000$ times.

Table \ref{table general QMLE m2r1s0Bkm} and Tables \ref{table general QMLE m2r2s0Bkm}--\ref{table general QMLE m5r1s0Bkm} in the Appendix 
report the biases, empirical standard deviations (ESDs), and asymptotic standard deviations (ASDs) of $\widehat{\bm{\theta}}_{\text{G}}$ under DGP1--DGP4, and Table \ref{table lowrank QMLE m5r1s0Bkm} in the Appendix reports those of $\widehat{\bm{\theta}}_{\text{LR}}$ under DGP4. 
Here, bias is the average absolute difference between the true value and the estimates across 1000 replications, and the ASD is calculated using Theorems \ref{thm Asymptotic normality}--\ref{thm Asymptotic normality under low rank}. 
Moreover, Figure \ref{figure simulation_m20} in the Appendix shows box plots of  $\widehat{\bm{\theta}}_{\text{LR}}-\bm{\theta}_{0}$ under DGP5. 

\begin{table}[htbp]
\renewcommand\arraystretch{1.2}
\caption{\label{table general QMLE m2r1s0Bkm}
Biases, ESDs, and ASDs of $\widehat{\bm{\theta}}_{\text{G}}$ in \eqref{est general QMLE} under DGP1 when $m=2$, $(r,s)=(1,0)$, and the innovations $\{\bm{\eta}_{t}\}$ follow a multivariate normal or Student's $t_5$ distribution.}
\centering
\scalebox{0.68}{
\begin{tabular}{ccccc@{\hspace{0.8em}}ccccc@{\hspace{1em}}ccccc@{\hspace{0.8em}}ccccc}
\hline\hline
\multicolumn{2}{c}{} & \multicolumn{3}{c}{Normal} & \multicolumn{3}{c}{$t_5$}
& \multicolumn{2}{c}{} & \multicolumn{3}{c}{Normal} & \multicolumn{3}{c}{$t_5$} \\
\cmidrule(lr){3-5}\cmidrule(lr){6-8}\cmidrule(lr){11-13}\cmidrule(lr){14-16}
 & $n$ & Bias & ESD & ASD & Bias & ESD & ASD
&  & $n$ & Bias & ESD & ASD & Bias & ESD & ASD \\
\midrule

\multirow{2}{*}{$\underline{\omega}_{1}$}
& 1000 & 0.0417 & 0.0516 & 0.0498 & 0.0524 & 0.0679 & 0.0731
& \multirow{2}{*}{$(G_{0,1})_{11}$}
& 1000 & 0.0100 & 0.0126 & 0.0120 & 0.0154 & 0.0195 & 0.0188 \\
& 2000 & 0.0280 & 0.0351 & 0.0351 & 0.0317 & 0.0408 & 0.0529
& & 2000 & 0.0068 & 0.0084 & 0.0082 & 0.0111 & 0.0143 & 0.0133 \\

\multirow{2}{*}{$\underline{\omega}_{2}$}
& 1000 & 0.0415 & 0.0521 & 0.0500 & 0.0530 & 0.0674 & 0.0732
& \multirow{2}{*}{$(G_{0,1})_{21}$}
& 1000 & 0.0103 & 0.0130 & 0.0119 & 0.0157 & 0.0200 & 0.0188 \\
& 2000 & 0.0270 & 0.0341 & 0.0351 & 0.0329 & 0.0418 & 0.0528
& & 2000 & 0.0067 & 0.0084 & 0.0082 & 0.0112 & 0.0141 & 0.0131 \\

\multirow{2}{*}{$\lambda_{1}$}
& 1000 & 0.0343 & 0.0434 & 0.0419 & 0.0538 & 0.0714 & 0.0698
& \multirow{2}{*}{$(G_{0,1})_{12}$}
& 1000 & 0.0099 & 0.0124 & 0.0119 & 0.0153 & 0.0195 & 0.0184 \\
& 2000 & 0.0225 & 0.0281 & 0.0275 & 0.0371 & 0.0482 & 0.0461
& & 2000 & 0.0069 & 0.0086 & 0.0083 & 0.0112 & 0.0146 & 0.0132 \\

\multirow{2}{*}{$\beta_{1}$}
& 1000 & 0.0194 & 0.0243 & 0.0234 & 0.0205 & 0.0262 & 0.0257
& \multirow{2}{*}{$(G_{0,1})_{22}$}
& 1000 & 0.0099 & 0.0126 & 0.0120 & 0.0161 & 0.0211 & 0.0188 \\
& 2000 & 0.0129 & 0.0164 & 0.0161 & 0.0143 & 0.0182 & 0.0178
& & 2000 & 0.0064 & 0.0081 & 0.0083 & 0.0109 & 0.0137 & 0.0132 \\

\multirow{2}{*}{$\beta_{2}$}
& 1000 & 0.0500 & 0.0622 & 0.0593 & 0.0553 & 0.0825 & 0.0665
& \multirow{2}{*}{$(\underline{R})_{21}$}
& 1000 & 0.0390 & 0.0475 & 0.0459 & 0.0400 & 0.0506 & 0.0467 \\
& 2000 & 0.0317 & 0.0404 & 0.0389 & 0.0353 & 0.0469 & 0.0438
& & 2000 & 0.0275 & 0.0344 & 0.0332 & 0.0280 & 0.0357 & 0.0332 \\

\bottomrule
\end{tabular}
}
\end{table}

We have the following findings from Table \ref{table general QMLE m2r1s0Bkm} and Tables \ref{table general QMLE m2r2s0Bkm}--\ref{table lowrank QMLE m5r1s0Bkm}. 
First, as the sample size increases, most of the biases, ESDs and ASDs become smaller, and the ESDs approach their corresponding ASDs. 
Second, when the distribution of $\{\bm{\eta}_{t}\}$ becomes heavier-tailed, the ESDs and ASDs increase. 
This is expected because heavier-tailed innovations make the Gaussian QMLE less efficient.
Finally, a comparison of the results for DGP4 fitted by the QMLE $\widehat{\bm{\theta}}_{\text{G}}$ in Table \ref{table general QMLE m5r1s0Bkm} and those obtained by the QMLE $\widehat{\bm{\theta}}_{\text{LR}}$ in Table \ref{table lowrank QMLE m5r1s0Bkm} shows that $\widehat{\bm{\theta}}_{\text{LR}}$ has smaller biases, ESDs and ASDs than $\widehat{\bm{\theta}}_{\text{G}}$. This indicates that the QMLE $\widehat{\bm{\theta}}_{\text{LR}}$ is more efficient when a moderate-dimensional model has a low-rank structure.  

Figure \ref{figure simulation_m20} shows that the medians are near the centers of the boxes, indicating that the distribution of the QMLE $\widehat{\bm{\theta}}_{\text{LR}}$ is symmetric. 
As the sample size $n$ increases, the medians approach zero and the interquartile ranges (IQRs) narrow, indicating improved accuracy and efficiency. 
Moreover, the IQRs increase when the distribution of $\{\bm{\eta}_{t}\}$ becomes heavier-tailed because the Gaussian QMLE $\widehat{\bm{\theta}}_{\text{LR}}$ loses efficiency for heavy-tailed innovations. 
These findings for the moderate dimension $m=20$ are consistent with those in Table \ref{table lowrank QMLE m5r1s0Bkm} for the small dimension $m=5$.

In sum, the reasonable finite-sample performance of the proposed QMLEs $\widehat{\bm{\theta}}_{\text{G}}$ and $\widehat{\bm{\theta}}_{\text{LR}}$ supports the asymptotic properties established in Sections \ref{subsection QMLE}--\ref{subsection QMLE under low rank}.

\subsection{Model selection} \label{subsection model selection in simulation}

In the third experiment, we assess the proposed model selection method in Section \ref{subsection model selection}. 
The data are generated from DGP1--DGP4 in Section \ref{subsection model estimation}, and both QMLEs $\widehat{\bm{\theta}}_{\text{G}}$ and $\widehat{\bm{\theta}}_{\text{LR}}$ are considered for estimation. 
The BIC in \eqref{BIC} is employed to select the order $(r,s)$ with $K_{\text{max}} = 2$, since $r_{0}+2s_{0}\leq 2$ holds for DGP1--DGP4. 
Then the models selected by the BIC as underfitted, correct, and overfitted correspond to $(\widehat{r}, \widehat{s}) \in \Pi_{\text{under}} = \{(r,s) \in \Pi: r < r_{0} \; \text{or} \; s < s_{0}\}$, 
$(\widehat{r}, \widehat{s}) \in \Pi_{\text{true}} = \{(r,s) \in \Pi: r = r_{0} \; \text{and} \; s = s_{0}\}$, and
$(\widehat{r}, \widehat{s}) \in \Pi_{\text{over}} = \{(r,s) \in \Pi: r \geq r_{0} \; \text{and} \; s \geq s_{0}\} \setminus \Pi_{\text{true}}$, respectively.
Note that under the constraint $r+2s \leq m$, there are no overfitted cases for DGP2--DGP3 with $m=2$ and $(r_0,s_0) = (2,0)$ or $(0,1)$. 

Table \ref{table BIC} in the Appendix reports the percentages of BIC selections that are underfitted, correct and overfitted when the QMLEs $\widehat{\bm{\theta}}_{\text{G}}$ and $\widehat{\bm{\theta}}_{\text{LR}}$ are used. 
The BIC performs better as the sample size increases, supporting its selection consistency. 
It performs slightly worse when the distribution of $\{\bm{\eta}_{t}\}$ becomes heavier-tailed, and this tendency is consistent with the findings for the QMLEs $\widehat{\bm{\theta}}_{\text{G}}$ and $\widehat{\bm{\theta}}_{\text{LR}}$ because Gaussian QMLEs lose efficiency under heavy tails.   
Moreover, the overfitting rate for the BIC using $\widehat{\bm{\theta}}_{\text{LR}}$ is slightly higher than that using $\widehat{\bm{\theta}}_{\text{G}}$ for DGP4 with $m=5$. 
This is expected because the QMLE $\widehat{\bm{\theta}}_{\text{LR}}$ involves fewer parameters to optimize than the QMLE $\widehat{\bm{\theta}}_{\text{G}}$ when $m=5$. 
In addition, the BIC using $\widehat{\bm{\theta}}_{\text{LR}}$ performs poorly when $m=2$ and $(r,s) = (0,1)$.  
This is because the QMLE $\widehat{\bm{\theta}}_{\text{LR}}$ under the low-rank decomposition \eqref{G-decomposition} has more redundant optimization parameters than the unrestricted QMLE $\widehat{\bm{\theta}}_{\text{G}}$, causing the associated BIC to favor underfitted models. 
Thus, we suggest using the QMLE $\widehat{\bm{\theta}}_{\text{G}}$ when the dimension $m$ is small, especially when $m \leq 3$ and $s > 0$.

\section{Empirical examples} \label{section real data}

We illustrate the proposed model and its inference tools through two empirical applications. 
We first analyze five major stock market indices using the proposed model without low-rank constraints, and compare its out-of-sample performance in forecasting portfolio VaR with that of competing models. 
We also analyze seventeen industry portfolios using our model with low-rank constraints; see details in Section \ref{subsection example 2} of the Appendix. 

We analyze daily log returns for five major stock market indices from January 4, 2013, to December 30, 2022, using data downloaded from Yahoo Finance (\url{https://finance.yahoo.com/}). 
We focus on the centered percentage log returns of the Deutsche Aktien Index (DAX), French Cotation Automatique Continue Index (CAC), Financial Times 100 Stock Index (FTSE), Swiss Market Index (SMI), and Spanish IBEX 35 (IBEX). 
Because holidays are not shared across all stock markets, we remove dates on which all five returns are missing and replace the remaining missing values with $10^{-4}$ \citep{billio2005multivariate}. 
The resulting time series $\{\mathbf{y}_{t}\}_{t=1}^n$ has $m=5$ and $n=2456$, where $\mathbf{y}_{t} = (y_{1t}, \ldots, y_{mt})^{\prime}$.
The time plot of $\{\mathbf{y}_{t}\}$ in Figure \ref{figure time plot} of the Appendix shows very similar volatility clustering patterns across components. Moreover, the summary statistics for $\{\mathbf{y}_{t}\}$ in Table \ref{table major stocks summary statistics} of the Appendix indicate that each component is skewed and heavy-tailed. 
These findings motivate the use of the proposed multivariate log-GARCH model in \eqref{model Rt SGARCH(r,s)}--\eqref{model Phii in Dt} and its inference tools.

\subsection{Model fitting} \label{subsection major stocks estimation}

We first employ the QMLE without any low-rank restrictions in Section \ref{subsection QMLE} to fit the entire dataset. To calculate the QMLE $\widehat{\bm{\theta}}_{\text{G}}$ in \eqref{est general QMLE}, we randomly generate initial parameter vectors and choose the one with the smallest negative log-likelihood. 
Based on $K_{\text{max}} = 2$ and the QMLE $\widehat{\bm{\theta}}_{\text{G}}$, the proposed BIC in \eqref{BIC} selects $(r,s) = (1,0)$, and the fitted model is 
\begin{equation} \label{fitted model general}
\left\{
\begin{aligned}
\mathbf{y}_{t} &= H_{t}^{1/2} \bm{\eta}_{t}, \;\;
  H_{t} = D_{t} R_{t} D_{t}, \;\;
  R_{t} = 0.125 \widehat{\underline{R}} + 0.028 \Psi_{t-1} + 0.847 R_{t-1}, \\
  \ln\bm{h}_{t} &= \widehat{\underline{\bm{\omega}}} + \sum_{i=1}^{\infty} \widehat{\Phi}_{i} \ln\mathbf{y}^{2}_{t-i} 
  \;\; \text{with} \;\; 
  \widehat{\Phi}_{i} = 0.867^{i-1} \widehat{G}_{0,1}, 
\end{aligned}
\right. 
\end{equation}
where the QMLEs of $G_{0,1}$ and $\underline{R}$, together with their statistical significance, are presented in Figure \ref{figure major stocks Ghat and Rulinehat}, and details for the other parameters are provided in Table \ref{table major stocks fitted coefficients} of the Appendix. 
According to Remark \ref{remark connection with DCC}, the fitted volatility equation in \eqref{fitted model general} admits an equivalent DCC-log-GARCH$(1,1)$ representation as follows: 
$$
  \ln\bm{h}_{t} = (1 - 0.867) \widehat{\underline{\bm{\omega}}} + \widehat{G}_{0,1} \ln\mathbf{y}^{2}_{t-1} + 0.867 \ln\bm{h}_{t-1}. 
$$

The significant estimate $\widehat{\lambda}_{1}=0.867$ indicates high persistence in the volatility dynamics. 
The left panel of Figure \ref{figure major stocks Ghat and Rulinehat} shows that many off-diagonal elements of $\widehat{G}_{0,1}$ are statistically significant. 
As the significance test for volatility spillover effects in Section \ref{subsection test} concerns the off-diagonal elements of $\Phi_{1}$ and $\widehat{\Phi}_{1} = \widehat{G}_{0,1}$ in \eqref{fitted model general}, this suggests the presence of volatility spillovers across these stock markets. 
In particular, the IBEX market exhibits significant volatility transmission to each of the other four stock markets. 
Moreover, the right panel of Figure \ref{figure major stocks Ghat and Rulinehat} shows that the estimated long-run correlation coefficients in $\underline{R}$ are all positive and relatively large, indicating a highly integrated market structure among these stock markets. 
Meanwhile, the significant estimates $\widehat{\beta}_{1} = 0.028$ and $\widehat{\beta}_{2} = 0.847$ for $R_{t}$ in Table \ref{table major stocks fitted coefficients} imply that the conditional correlation process is time-varying and highly persistent. 
Additionally, to assess empirical adequacy, we compare the conditional covariance matrices $H_{t}$ and conditional correlation matrices $R_{t}$ fitted by \eqref{fitted model general} with those obtained from a benchmark diagonal BEKK model of order $(1,1)$. 
As shown in Figure \ref{figure_comparison} of the Appendix, these two models produce highly similar fitted paths for both $H_{t}$ and $R_{t}$. 
This close agreement suggests that our proposed model captures the multivariate dependence structure well while retaining a parsimonious parametrization. 

\begin{figure}
  \centering
  \includegraphics[width=5in]{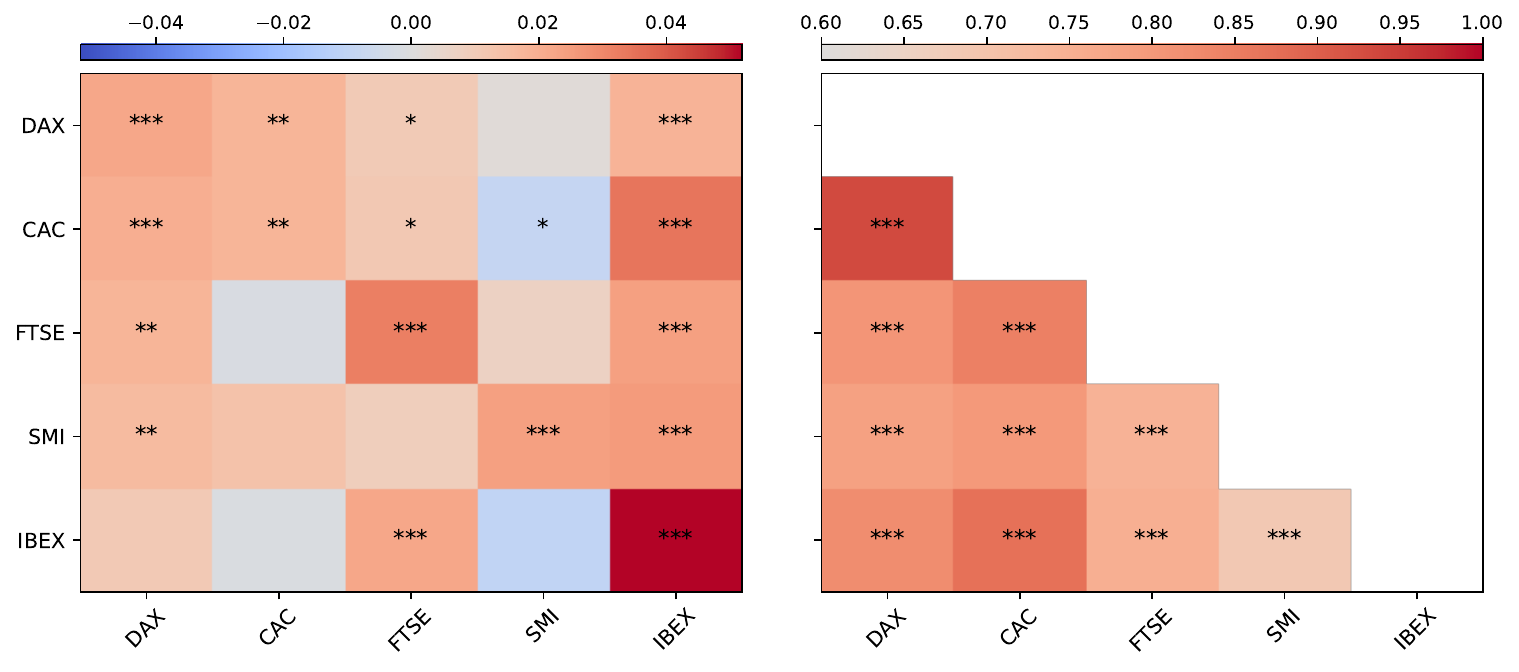}
  \caption{\label{figure major stocks Ghat and Rulinehat} 
  The estimates for $G_{0,1}$ (left panel) and $\underline{R}$ (right panel) based on the QMLE $\widehat{\bm{\theta}}_{\text{G}}$ in \eqref{est general QMLE}. Statistical significance is indicated by asterisks: $*$, $**$, and $***$ denote significance at the $10\%$, $5\%$, and $1\%$ levels, respectively.}
\end{figure}

\subsection{Out-of-sample forecasting} \label{subsection major stocks forecasting}

Next, we examine the forecasting performance of the fitted multivariate log-GARCH model by evaluating covariance matrix forecasts and portfolio VaR forecasts. We implement a one-step-ahead rolling forecasting procedure with a fixed moving window of $n_0=1970$ days (8 years), leaving $n-n_0=486$ days (2 years) for out-of-sample evaluation. At each step, the proposed model in \eqref{model Rt SGARCH(r,s)}--\eqref{model Phii in Dt} is fitted with order $(r,s)=(1,0)$. 
For comparison, we also consider other multivariate volatility models, including CCC-type \citep{bollerslev1990ModellingTC}, DCC-type \citep{tse2002multivariate,engle2002dynamic,aielli2013dynamic}, and BEKK-type \citep{ding2001large} specifications. 
For numerical feasibility with many assets, we focus on scalar or diagonal specifications for these benchmarks. The CCC and DCC models are fitted by a two-step method, while the BEKK models are fitted by either the full QMLE or the variance targeting (VT) method \citep{engle1996garch}. 
Table \ref{table major stocks comparison methods} of the Appendix summarizes all models used for comparison, together with their abbreviations.

For covariance matrix forecasts, we assess the accuracy of the one-step-ahead predictions $\widehat{H}_{t}$ obtained from the fitted multivariate volatility models. 
We calculate the mean squared error (MSE) and the quasi-likelihood (QLIKE) loss, defined by $\text{MSE}_{t} = m^{-2} \tr[ (\widehat{H}_{t} - \mathbf{y}_{t} \mathbf{y}_{t}^{\prime})^{\prime} (\widehat{H}_{t} - \mathbf{y}_{t} \mathbf{y}_{t}^{\prime})]$ and $\text{QLIKE}_{t} = \ln |\widehat{H}_{t}| + \mathbf{y}_{t}^\prime \widehat{H}_{t}^{-1} \mathbf{y}_{t}$. 
Smaller values of both measures indicate better forecasting performance. 
In addition, based on each loss function, we construct the 90\% model confidence set (MCS) proposed by \cite{hansen2011model}, which consists of the models whose predictive ability cannot be statistically distinguished from that of the best-performing model at the 10\% significance level. 
Table \ref{table stock stat loss} in the Appendix reports the average MSEs and QLIKE losses for the one-step-ahead covariance matrix forecasts produced by the competing models, and also indicates the models included in the corresponding 90\% MCS. 
We have the following findings. 
First, the proposed model performs best under the QLIKE criterion, attaining the smallest QLIKE among all competing methods. 
Second, under the MSE criterion, the proposed model remains competitive and is included in the 90\% MCS, although it does not achieve the smallest MSE. 
Notably, the MSEs are fairly close across most competing models, whereas the QLIKEs show clearer differences. This may be because MSE measures elementwise forecast errors, while QLIKE depends on both $\widehat{H}_t^{-1}$ and $\log|\widehat{H}_t|$ and is therefore more sensitive to differences in the overall covariance structure. 

For portfolio VaR forecasting, we construct portfolios based on the fitted multivariate volatility models and evaluate the resulting one-step-ahead VaR forecasts. 
Specifically, we consider the equally weighted (EW) and minimum variance (MV) portfolios, whose returns are given by $z_{t} = \bm{\iota}^{\prime} \mathbf{y}_{t}$, with $\bm{\iota} = m^{-1} \bm{1}_{m}$ for the EW portfolio and $\bm{\iota} = (\bm{1}_{m}^{\prime} H_{t}^{-1} \bm{1}_{m})^{-1} H_{t}^{-1} \bm{1}_{m}$ for the MV portfolio. 
The $\tau$-th VaR of $z_t$ is defined as its negative $\tau$-th conditional quantile. 
We assume a volatility model for $z_t$ such that $Q_{\tau}(z_{t}\mid \mathcal{F}_{t-1})=\sigma_{t} b_{\tau}$, where $\sigma_{t}$ is the conditional standard deviation of $z_{t}$ with $\sigma_{t}^{2} = \bm{\iota}^{\prime} H_{t} \bm{\iota}$, and $b_{\tau}$ is the $\tau$-th quantile of the $i.i.d.$ innovations $\{e_t\}$ with $e_t=z_t/\sigma_{t}$. 
Based on the one-step-ahead forecast $\widehat{H}_{t}$, we construct the feasible EW or MV portfolio return $\widehat{z}_{t}$ together with its associated conditional standard deviation $\widehat{\sigma}_{t}$, and then obtain the corresponding VaR forecasts. 
We consider the lower and upper $1\%, 2.5\%$ and $5\%$ quantile levels, that is the $1\%$, $2.5\%$ and $5\%$ VaRs for long and short positions, and apply this rolling procedure to all competing methods. 

To evaluate the performance of VaR forecasts, we compute the empirical coverage rate (ECR) and prediction error (PE), report the average quantile loss ($\text{Loss}_{\mathrm{Q}}$) and the calibration loss ($\text{Loss}_{\mathrm{C}}$), construct the 90\% model confidence set (MCS) based on $\text{Loss}_{\mathrm{Q}}$ and $\text{Loss}_{\mathrm{C}}$, and conduct two standard VaR backtests. 
Let $\hbar_{t} = I(\widehat{z}_t<\widehat{Q}_{\tau}(\widehat{z}_{t} \mid \mathcal{F}_{t-1}))$ be the hit series. 
We then define $\text{ECR} = (n-n_0)^{-1} \sum_{t=n_0+1}^{n} \hbar_{t}$, $\text{PE}=[\tau(1-\tau)/(n-n_0)]^{-1/2}|(n-n_0)^{-1} \sum_{t=n_0+1}^{n} \hbar_{t} - \tau|$, $\text{Loss}_{\mathrm{Q}} = (n-n_0)^{-1} \sum_{t=n_0+1}^{n}(\tau-\hbar_t) (\widehat{z}_t-\widehat{Q}_{\tau}(\widehat{z}_{t}\mid\mathcal{F}_{t-1}))$ and $\text{Loss}_{\mathrm{C}} = (n-n_0)^{-1} \sum_{t=n_0+1}^{n}(\hbar_t-\tau)^2$. 
An ECR closer to $\tau$ leads to a smaller PE. 
Smaller values of PE, $\text{Loss}_{\mathrm{Q}}$, and $\text{Loss}_{\mathrm{C}}$ indicate better VaR forecasting performance. 
In addition, for each of $\text{Loss}_{\mathrm{Q}}$ and $\text{Loss}_{\mathrm{C}}$, we construct the 90\% MCS to identify the models whose predictive ability cannot be statistically distinguished from that of the best-performing model. 
Moreover, we consider two standard VaR backtests: the likelihood ratio test for correct conditional coverage (CC) of \cite{Christoffersen_1998} and the dynamic quantile (DQ) test of \cite{Engle2004}. 
If the null hypothesis of a backtest cannot be rejected, then the corresponding VaR forecasts are regarded as satisfactory. 
Further details of the backtests are provided in Section \ref{subsection additive results for example 1} of the Appendix. 

\begin{table}[htbp]
    \centering
    \renewcommand\arraystretch{1.2}
    \caption{\label{table major stocks VaR MV} Out-of-sample VaR forecasting performance for the MV portfolio in the first example, where MGARCH-G denotes the proposed model in \eqref{model Rt SGARCH(r,s)}--\eqref{model Phii in Dt} fitted by the QMLE $\widehat{\bm{\theta}}_{\text{G}}$ in \eqref{est general QMLE}. The smallest values of PE, $\text{Loss}_{\mathrm{Q}}$, and $\text{Loss}_{\mathrm{C}}$ in each column are shown in boldface. For ECR, $\dagger$ and $\ddagger$ indicate the models that pass one or both of the CC and DQ tests at the 5\% significance level, respectively. For $\text{Loss}_{\mathrm{Q}}$ and $\text{Loss}_{\mathrm{C}}$, $\dagger$ indicates that the corresponding model is included in the 90\% MCS at the 10\% significance level.}
    \scalebox{0.8}{
    \resizebox{\textwidth}{!}{
    \begin{tabular}{l *{12}{r}}
        \toprule\toprule
        & \multicolumn{4}{c}{$\tau=1\%$} & \multicolumn{4}{c}{$\tau=2.5\%$} & \multicolumn{4}{c}{$\tau=5\%$} \\
        \cmidrule(lr){2-5} \cmidrule(lr){6-9} \cmidrule(lr){10-13} 
        Method & ECR & PE & $\text{Loss}_{\mathrm{Q}}$ & $\text{Loss}_{\mathrm{C}}$ & ECR & PE & $\text{Loss}_{\mathrm{Q}}$ & $\text{Loss}_{\mathrm{C}}$ & ECR & PE & $\text{Loss}_{\mathrm{Q}}$ & $\text{Loss}_{\mathrm{C}}$ \\
        \midrule
        MGARCH-G & $0.62^{\ddagger}$ & $0.85$ & $\mathbf{2.77}^{\dagger}$ & $\mathbf{0.61}^{\dagger}$ & $1.85^{\ddagger}$ & $0.92$ & $\mathbf{5.82}^{\dagger}$ & $\mathbf{1.82}^{\dagger}$ & $5.35^{\ddagger}$ & $\mathbf{0.35}$ & $\mathbf{9.86}^{\dagger}$ & $\mathbf{5.06}^{\dagger}$ \\
        CCC & $1.23^{\ddagger}$ & $0.52$ & $3.11^{\dagger}$ & $1.22^{\dagger}$ & $3.09^{\ddagger}$ & $0.83$ & $6.12^{\dagger}$ & $2.99^{\dagger}$ & $7.61$ & $2.64$ & $10.49^{\dagger}$ & $7.10^{\dagger}$ \\
        DCC-T & $1.23^{\ddagger}$ & $0.52$ & $3.09^{\dagger}$ & $1.22^{\dagger}$ & $2.67^{\ddagger}$ & $\mathbf{0.25}$ & $6.05^{\dagger}$ & $2.60^{\dagger}$ & $7.20$ & $2.23$ & $10.29^{\dagger}$ & $6.73^{\dagger}$ \\
        D-DCC-T & $1.03^{\ddagger}$ & $\mathbf{0.06}$ & $3.19^{\dagger}$ & $1.02^{\dagger}$ & $2.67^{\ddagger}$ & $\mathbf{0.25}$ & $6.13^{\dagger}$ & $2.60^{\dagger}$ & $7.00$ & $2.02$ & $10.35^{\dagger}$ & $6.55^{\dagger}$ \\
        DCC-E & $1.03^{\ddagger}$ & $\mathbf{0.06}$ & $3.22^{\dagger}$ & $1.02^{\dagger}$ & $3.29^{\ddagger}$ & $1.12$ & $6.18^{\dagger}$ & $3.19^{\dagger}$ & $7.00$ & $2.02$ & $10.35^{\dagger}$ & $6.55^{\dagger}$ \\
        D-DCC-E & $1.23^{\ddagger}$ & $0.52$ & $3.28^{\dagger}$ & $1.22^{\dagger}$ & $3.09^{\ddagger}$ & $0.83$ & $6.20^{\dagger}$ & $2.99^{\dagger}$ & $6.79$ & $1.81$ & $10.53^{\dagger}$ & $6.36^{\dagger}$ \\
        DCC-A & $1.44^{\ddagger}$ & $0.98$ & $3.32^{\dagger}$ & $1.42^{\dagger}$ & $2.88^{\ddagger}$ & $0.54$ & $6.35^{\dagger}$ & $2.80^{\dagger}$ & $7.41$ & $2.44$ & $10.69^{\dagger}$ & $6.92^{\dagger}$ \\
        D-DCC-A & $1.44^{\ddagger}$ & $0.98$ & $3.20^{\dagger}$ & $1.42^{\dagger}$ & $3.50^{\ddagger}$ & $1.41$ & $6.42^{\dagger}$ & $3.39^{\dagger}$ & $7.61$ & $2.64$ & $10.80^{\dagger}$ & $7.10^{\dagger}$ \\
        S-BEKK-VT & $0.82^{\ddagger}$ & $0.39$ & $3.08^{\dagger}$ & $0.82^{\dagger}$ & $2.26^{\ddagger}$ & $0.33$ & $6.28^{\dagger}$ & $2.21^{\dagger}$ & $5.97^{\ddagger}$ & $0.98$ & $10.22^{\dagger}$ & $5.62^{\dagger}$ \\
        S-BEKK-F & $1.44^{\ddagger}$ & $0.98$ & $4.31$ & $1.42^{\dagger}$ & $3.50^{\dagger}$ & $1.41$ & $8.27$ & $3.39^{\dagger}$ & $6.58^{\ddagger}$ & $1.60$ & $13.61$ & $6.18^{\dagger}$ \\
        D-BEKK-VT & $1.23^{\ddagger}$ & $0.52$ & $3.16^{\dagger}$ & $1.22^{\dagger}$ & $2.26^{\ddagger}$ & $0.33$ & $6.33^{\dagger}$ & $2.21^{\dagger}$ & $5.97^{\ddagger}$ & $0.98$ & $10.31^{\dagger}$ & $5.62^{\dagger}$ \\
        D-BEKK-F & $1.03^{\ddagger}$ & $\mathbf{0.06}$ & $3.21^{\dagger}$ & $1.02^{\dagger}$ & $3.29^{\ddagger}$ & $1.12$ & $6.31^{\dagger}$ & $3.19^{\dagger}$ & $6.58^{\ddagger}$ & $1.60$ & $10.66^{\dagger}$ & $6.18^{\dagger}$ \\
        \midrule
        & \multicolumn{4}{c}{$\tau=95\%$} & \multicolumn{4}{c}{$\tau=97.5\%$} & \multicolumn{4}{c}{$\tau=99\%$} \\
        \cmidrule(lr){2-5} \cmidrule(lr){6-9} \cmidrule(lr){10-13} 
        Method & ECR & PE & $\text{Loss}_{\mathrm{Q}}$ & $\text{Loss}_{\mathrm{C}}$ & ECR & PE & $\text{Loss}_{\mathrm{Q}}$ & $\text{Loss}_{\mathrm{C}}$ & ECR & PE & $\text{Loss}_{\mathrm{Q}}$ & $\text{Loss}_{\mathrm{C}}$ \\
        \midrule
        MGARCH-G & $95.88^{\ddagger}$ & $0.90$ & $\mathbf{8.05}^{\dagger}$ & $\mathbf{3.95}^{\dagger}$ & $96.91^{\ddagger}$ & $0.83$ & $\mathbf{4.75}^{\dagger}$ & $2.99^{\dagger}$ & $99.18^{\ddagger}$ & $\mathbf{0.39}$ & $\mathbf{2.24}^{\dagger}$ & $\mathbf{0.82}^{\dagger}$ \\
        CCC & $94.03^{\ddagger}$ & $0.98$ & $8.24^{\dagger}$ & $5.62^{\dagger}$ & $96.71^{\ddagger}$ & $1.12$ & $5.01^{\dagger}$ & $3.19^{\dagger}$ & $97.53$ & $3.26$ & $2.48^{\dagger}$ & $2.43^{\dagger}$ \\
        DCC-T & $94.44^{\ddagger}$ & $0.56$ & $8.26^{\dagger}$ & $5.25^{\dagger}$ & $96.50^{\ddagger}$ & $1.41$ & $4.99^{\dagger}$ & $3.39^{\dagger}$ & $97.74$ & $2.80$ & $2.45^{\dagger}$ & $2.23^{\dagger}$ \\
        D-DCC-T & $94.44^{\ddagger}$ & $0.56$ & $8.37^{\dagger}$ & $5.25^{\dagger}$ & $96.09^{\ddagger}$ & $1.99$ & $5.06^{\dagger}$ & $3.78^{\dagger}$ & $97.53$ & $3.26$ & $2.50^{\dagger}$ & $2.43^{\dagger}$ \\
        DCC-E & $94.86^{\ddagger}$ & $\mathbf{0.15}$ & $8.27^{\dagger}$ & $4.88^{\dagger}$ & $96.50^{\ddagger}$ & $1.41$ & $5.06^{\dagger}$ & $3.39^{\dagger}$ & $97.94^{\ddagger}$ & $2.34$ & $2.48^{\dagger}$ & $2.03^{\dagger}$ \\
        D-DCC-E & $94.44^{\ddagger}$ & $0.56$ & $8.29^{\dagger}$ & $5.25^{\dagger}$ & $96.50^{\ddagger}$ & $1.41$ & $5.05^{\dagger}$ & $3.39^{\dagger}$ & $98.15^{\ddagger}$ & $1.89$ & $2.45^{\dagger}$ & $1.82^{\dagger}$ \\
        DCC-A & $92.59$ & $2.44$ & $8.31^{\dagger}$ & $6.92$ & $96.91^{\ddagger}$ & $0.83$ & $4.96^{\dagger}$ & $2.99^{\dagger}$ & $98.15^{\dagger}$ & $1.89$ & $2.62^{\dagger}$ & $1.82^{\dagger}$ \\
        D-DCC-A & $92.18$ & $2.85$ & $8.80^{\dagger}$ & $7.29$ & $95.06$ & $3.44$ & $5.39^{\dagger}$ & $4.75$ & $97.33$ & $3.71$ & $2.89^{\dagger}$ & $2.63^{\dagger}$ \\
        S-BEKK-VT & $95.68^{\ddagger}$ & $0.69$ & $8.71^{\dagger}$ & $4.14^{\dagger}$ & $97.12^{\ddagger}$ & $\mathbf{0.54}$ & $5.21^{\dagger}$ & $\mathbf{2.80}^{\dagger}$ & $98.35^{\ddagger}$ & $1.43$ & $2.47^{\dagger}$ & $1.62^{\dagger}$ \\
        S-BEKK-F & $93.00$ & $2.02$ & $10.94$ & $6.55$ & $96.09^{\ddagger}$ & $1.99$ & $6.25$ & $3.78^{\dagger}$ & $98.15^{\ddagger}$ & $1.89$ & $3.02^{\dagger}$ & $1.82^{\dagger}$ \\
        D-BEKK-VT & $95.88^{\ddagger}$ & $0.90$ & $8.52^{\dagger}$ & $\mathbf{3.95}^{\dagger}$ & $96.71^{\ddagger}$ & $1.12$ & $5.27^{\dagger}$ & $3.19^{\dagger}$ & $98.35^{\ddagger}$ & $1.43$ & $2.59^{\dagger}$ & $1.62^{\dagger}$ \\
        D-BEKK-F & $94.86^{\ddagger}$ & $\mathbf{0.15}$ & $8.81^{\dagger}$ & $4.88^{\dagger}$ & $96.71^{\ddagger}$ & $1.12$ & $5.35^{\dagger}$ & $3.19^{\dagger}$ & $97.94^{\dagger}$ & $2.34$ & $2.66^{\dagger}$ & $2.03^{\dagger}$ \\
        \bottomrule
    \end{tabular}
    }}
\end{table}

Table \ref{table major stocks VaR MV} and Table \ref{table major stocks VaR EW} in the Appendix summarize the one-step-ahead VaR forecasts of all methods for the MV and EW portfolios, respectively, using ECR, PE, $\text{Loss}_{\mathrm{Q}}$, $\text{Loss}_{\mathrm{C}}$, and backtests at the three lower- and upper-tail levels. 
Table \ref{table major stocks VaR MV} shows that the proposed model performs best for the MV portfolio. It passes both backtests at all quantile levels, achieves the smallest $\text{Loss}_{\mathrm{Q}}$ throughout, and yields the smallest $\text{Loss}_{\mathrm{C}}$ except at one quantile level. It also achieves the smallest PE at two quantile levels, tied with three competing models.
For the EW portfolio in Table \ref{table major stocks VaR EW}, the advantage of the proposed model is less pronounced. Although it still passes both backtests at all quantile levels and is included in the 90\% MCS under both $\text{Loss}_{\mathrm{Q}}$ and $\text{Loss}_{\mathrm{C}}$, its overall performance is broadly comparable to that of most competing methods. This suggests that, with equal portfolio weights, differences in covariance matrix forecasts translate less strongly into differences in VaR forecasts. 
By contrast, for MV portfolios, small differences in the estimated covariance structure can have a larger impact on VaR forecasting performance, as MV portfolio weights depend on $H_{t}^{-1}$. 
The proposed model also performs best for both the EW and MV portfolios in the second dataset of dimension 17. 
This further suggests that its advantage in forecasting conditional covariance matrices and VaR becomes more pronounced as the dimension increases. 
Overall, these results indicate that the proposed multivariate log-GARCH model estimated by QMLE generally outperforms the competing methods in VaR forecasting, highlighting its favorable trade-off between generality and efficiency especially for many assets.

\section{Conclusion and discussion}\label{section conclusion}

This paper proposes an efficient multivariate log-GARCH model that flexibly captures dynamic conditional correlations and spillover effects among financial assets.
The new model avoids the non-identification issue that affects other multivariate volatility models and, more importantly, is computationally tractable for a moderate number of assets. 
A sufficient condition is derived for its strict stationarity. 
To facilitate the application of the new model, we propose inference tools with theoretical guarantees, including two QMLEs for model estimation with and without low-rank constraints, a BIC for order selection, and tests for volatility spillovers. 
Finally, we employ the proposed methodology to analyze 5 major stock markets and 17 industry portfolios, 
and our method provides more accurate VaR forecasts than competing methods. 

The proposed approach can be improved and extended in the following directions. 
First, zero returns are frequently observed in financial data, especially for illiquid assets or data affected by price discreteness and trading frictions \citep{sucarrat2022risk}. While Section 4.3 of \cite{francq2019garch} considered an extended univariate log-GARCH model that explicitly handles such observations, our current multivariate framework rules out this case by assuming innovations are nonzero with probability one. 
Future work could extend the methodology to accommodate zero returns. 
Second, the idea in this paper can be adapted to other multivariate volatility models such as BEKK models. 
Future work could complement the commonly used scalar and diagonal BEKK models \citep{engle1995MultivariateSG} with a general but computationally tractable BEKK model for many assets. 
Finally, we may extend the proposed multivariate log-GARCH model to high-dimensional settings by imposing appropriate sparsity \citep{zheng2022interpretable} or low-rank \citep{wang2024high} assumptions on the coefficient matrices. 
Such an extension would require inference tools with non-asymptotic guarantees for these settings.


\clearpage
\appendix
\setcounter{figure}{0}
\setcounter{table}{0}
\renewcommand{\theassum}{A.\arabic{assum}}
\renewcommand{\thethm}{A.\arabic{thm}}
\renewcommand{\thelemma}{A.\arabic{lemma}}
\renewcommand{\thefigure}{A.\arabic{figure}}
\renewcommand{\thetable}{A.\arabic{table}}

\section*{Appendix}

This Appendix includes additional results for simulation and empirical analysis, as well as technical details for Remarks \ref{remark connection with DCC}--\ref{remark general SGARCH}, Lemma \ref{lemma Stationarity of epsilont}, Theorem \ref{thm Stationarity}, Propositions \ref{propo Identification}--\ref{propo approximates of asymptotic variances} and Theorems \ref{thm Consistency}--\ref{thm BIC}. 
	It also provides the detailed structures of the derivatives in the matrices $\Sigma$ and $\Sigma_{*}$ in Theorem \ref{thm Asymptotic normality}, and introduces Lemmas \ref{lemma for identification}--\ref{lemma positivity of Sigma*} which give some preliminary results for proving the aforementioned propositions and theorems.

	Throughout the supplement, 
	$|\cdot|$ denotes either the absolute value of a scalar/vector or the determinant of a matrix, 
	$\|\cdot\|_{r}$ denotes either the $\ell_{r}$ vector norm or the matrix norm induced by $\ell_{r}$ vector norm, and 
	$\|\cdot\|_{F}$ denotes the Frobenius norm of a matrix. 
	Denote by $\rho(\cdot)$ and $\rank(\cdot)$ the spectral radius and the rank of a matrix, respectively. 
	For positive integers $m$ and $m'$, 
	$I_{m}$ (or $0_{m}$) denotes the $m \times m$ identity matrix (or matrix of zeros), 
	$0_{m \times m'}$ is an $m \times m'$ matrix of zeros, 
	and $\bm{0}_{m}$ is an $m \times 1$ vector of zeros. 
	Denote by $\bm{e}_{\ell}$ the vector with the $\ell$-th element being one and the others being zeros. 
	Denote $(A)_{i \cdot}$ and $(A)_{\cdot j}$ as the $i$-th row and $j$-th column of matrix $A$, respectively; 
	denote $(A)_{ij}$ as the $(i,j)$-th element of matrix $A$; 
	and denote $A = [a_{ij}]$ when $(A)_{ij} = a_{ij}$. 
	For a matrix $A = [a_{ij}]$, 
	$\ovec(A)$ transforms $A$ into a column vector by vertically stacking the columns of $A$, i.e. $\ovec(A) = (a_{11}, \ldots, a_{m1}, a_{12}, \ldots, a_{m2}, \ldots, a_{1m}, \ldots, a_{mm})^{\prime}$; 
	$\ovech(A)$ transforms $A$ into a column vector by vertically stacking the columns of the lower triangular part of A, i.e. $\ovech(A) = (a_{11}, \ldots, a_{m1}, a_{22}, \ldots, a_{m2}, \ldots, a_{mm})^{\prime}$; and 
	$\ovechsec(A) = (a_{21}, \ldots, a_{m1}, a_{32}, \ldots, a_{m2}, \ldots, a_{m,m-1})^{\prime}$ is a subvector of $\ovech(A)$. 
	For square matrices $A$ and $A_{j}$'s, 
	$\lambda_{\text{min}}(A)$ and $\lambda_{\text{max}}(A)$ denote the smallest and largest eigenvalues of $A$, respectively; 
	$\diag(A)$ denotes the vector of diagonal elements of $A$, 
	$\Diag(A)$ denotes the diagonal matrix whose main diagonal is $\diag(A)$, 
	and $\Diag\{A_{1}, \ldots, A_{m}\}$ denotes the block diagonal matrix whose main diagonal consists of $A_{1}, \ldots, A_{m}$. 
	And denote $\Diag\{\bm{a}\} = \Diag\{a_{1}, \ldots, a_{m}\}$ as the diagonal matrix whose main diagonal is $\bm{a} = (a_{1}, \ldots, a_{m})^{\prime}$. 
	For a positive semi-definite matrix $A$, 
	$A^{1/2}$ is the square root of $A$ which satisfies that $A^{1/2} (A^{1/2})^{\prime} = A$. 
	The operator $\otimes$ denotes the Kronecker product of two matrices, 
	and $A^{\otimes k}$ is the Kronecker product of the $k$ matrices $A$. 
	The function $\min\{\}$ (or $\max\{\}$) gives the minimum (or maximum) value of its arguments, and $\ln^{+}(x)$ is defined as $\max\{\ln(x), 0\}$ for $x > 0$. 
	Moreover, 
	$\to_{p}$ and $\to_{d}$ denote the convergence in probability and in distribution, respectively. 
	$E$ denotes the expectation with respect to the probability measure $P$.

	\section{Additional results for the simulation studies}

	\subsection{Additional results for Section \ref{section simulation}} \label{subsection additional results for simulation}
     
      Section \ref{subsection model interpretation} of the manuscript aims to illustrate the model structure by linking the coefficient matrices $\Phi_{i}$'s in \eqref{model Phii in Dt} to the autocorrelation of squared returns (or square volatilities). 
		We consider three data generating processes (DGPs) from the proposed model in \eqref{model Rt SGARCH(r,s)}--\eqref{model Phii in Dt} with the following settings: 
		\begin{align*}
			\text{Model (a)}: \; & m=2,(r_0,s_0) = (1,0), \underline{\bm{\omega}} = \bm{0}_{m}, \underline{R} = I_{m}, \beta_1=0, \beta_2=0, \\
			& \Phi_{i} = \lambda_{1}^{i-1} G_{0,1} \;\;\text{with}\;\; \text{varying} \;\; \lambda_{1} \;\; \text{and} \;\;
			G_{0,1} = 
			\left(
			\begin{matrix}
				0.1 & 0.15 \\
				0.15 & -0.1
			\end{matrix}
			\right); \\
			\text{Model (b)}: \; & m=2,(r_0,s_0) = (0,1), \underline{\bm{\omega}} = \bm{0}_{m}, \underline{R} = I_{m}, \beta_1=0, \beta_2=0, \\
			& \Phi_{i} = \gamma_{1}^{i-1} \left[\cos((i-1) \pi/7) G_{1,1} + \sin((i-1) \pi/7) G_{2,1}\right] \;\;\text{with}\;\; \text{varying} \;\; \gamma_{1}, \\
			& G_{1,1} = 
			\left(
			\begin{matrix}
				0.1 & 0.15 \\
				0.15 & -0.1
			\end{matrix}
			\right) \;\; \text{and} \;\; 
			G_{2,1} = 
			\left(
			\begin{matrix}
				0.01 & -0.01 \\
				0.01 & 0.01
			\end{matrix}
			\right); \\
			\text{Model (c)}: \; & m=2,(r_0,s_0) = (0,1), \underline{\bm{\omega}} = \bm{0}_{m}, \underline{R} = I_{m}, \beta_1=0, \beta_2=0, \\
			& \Phi_{i} = 0.8^{i-1} \left[\cos((i-1) \varphi_{1}) G_{1,1} + \sin((i-1) \varphi_{1}) G_{2,1}\right] \;\;\text{with}\;\; \text{varying} \;\; \varphi_{1}, \\
			& G_{1,1} = 
			\left(
			\begin{matrix}
				0.1 & 0.15 \\
				0.15 & -0.1
			\end{matrix}
			\right) \;\; \text{and} \;\; 
			G_{2,1} = 
			\left(
			\begin{matrix}
				0.01 & -0.01 \\
				0.01 & 0.01
			\end{matrix}
			\right). 
		\end{align*}
		Moreover, the innovations $\{\bm{\eta}_{t}\}$ are $i.i.d.$ random vectors following the multivariate normal distribution with zero mean and identity covariance matrix. 
		It can be verified that these DGPs satisfy the strict stationarity condition in Theorem \ref{thm Stationarity}. 
		Figures \ref{figure ACF_squared_return}--\ref{figure ACF_squared_volatility} display the autocorrelation functions (ACFs) of squared returns $\{y_{1t}^2\}$ and squared volatilities $\{h_{11,t}\}$ under models (a)--(c), respectively. The simulation findings are summarized in Section \ref{subsection model interpretation}.
	
		Sections \ref{subsection model estimation}--\ref{subsection model selection in simulation} aim to examine the finite-sample performance of the proposed QMLEs and model selection method in Section 3. We consider five DGPs from the proposed model in \eqref{model Rt SGARCH(r,s)}--\eqref{model Phii in Dt} with the following settings: 
		\begin{align*}
			\text{DGP}1: \; & m=2,(r_0,s_0) = (1,0), \underline{\bm{\omega}} = \underline{\bm{\omega}}_{m}, \underline{R} = \underline{R}_{m}, \beta_1=0.1, \beta_2=0.8, \\
			& \Phi_{i} = 0.8^{i-1} G_{0,1} \;\;\text{with}\;\; 
			G_{0,1} = (1, 1)^{\prime} (0.045, 0.045); \\
			\text{DGP}2: \; & m=2,(r_0,s_0) = (2,0), \underline{\bm{\omega}} = \underline{\bm{\omega}}_{m}, \underline{R} = \underline{R}_{m}, \beta_1=0.1, \beta_2=0.8, \\
			& \Phi_{i} = 0.8^{i-1} G_{0,1} + (-0.8)^{i-1} G_{0,2} \;\;\text{with}\;\; \\
			& G_{0,1} = (1, 1)^{\prime} (0.045, 0.045), G_{0,2} = (1, -1)^{\prime} (0.045, -0.045); \\
			\text{DGP}3: \; & m=2,(r_0,s_0) = (0,1), \underline{\bm{\omega}} = \underline{\bm{\omega}}_{m}, \underline{R} = \underline{R}_{m}, \beta_1=0.1, \beta_2=0.8, \\
			& \Phi_{i} = 0.8^{i-1} \left[\cos(0.7(i-1)) G_{1,1} + \sin(0.7(i-1)) G_{2,1}\right] \;\;\text{with}\;\; \\
			& G_{1,1} = (0.8, 0.6)^{\prime} (0.064, 0.062) + (-0.6, 0.8)^{\prime} (0.002, 0.016), \\
			& G_{2,1} = (0.8, 0.6)^{\prime} (0.002, 0.016) + (0.6, -0.8)^{\prime} (0.064, 0.062); \\
			\text{DGP4}: \; & m=5, (r_0,s_0) = (1,0), \underline{\bm{\omega}} = \underline{\bm{\omega}}_{m}, \underline{R} = \underline{R}_{m}, \beta_1=0.1, \beta_2=0.8, \\
			& \Phi_{i} = 0.8^{i-1} G_{0,1} \;\;\text{with}\;\; G_{0,1} = \bm{g}_{0,1,1} \bm{g}_{0,1,2}^{\prime}, \\
			& \bm{g}_{0,1,1}=(1.00, 0.96, 0.92, 0.88, 0.86)^{\prime}, \bm{g}_{0,1,2}= (0.025, 0.0255, 0.0265, 0.028, 0.03)^{\prime}; \\
			\text{DGP5}: \; & m=20, (r_0,s_0) = (1,0), \underline{\bm{\omega}} = \underline{\bm{\omega}}_{m}, \underline{R} = \underline{R}_{m}, \beta_1=0.1, \beta_2=0.8, \\
			& \Phi_{i} = 0.5^{i-1} G_{0,1} \;\;\text{with}\;\; G_{0,1} = \bm{g}_{0,1,1} \bm{g}_{0,1,2}^{\prime}, \\
			& \bm{g}_{0,1,1} \;\text{and}\; \bm{g}_{0,1,2} \;\text{are randomly generated from}\; U(0.5, 0.6) \;\text{and}\; U(0.03, 0.05), 
		\end{align*}
		where $\underline{\bm{\omega}}_{m}$ is an $m$-dimensional vector with all elements being 1.45 for $m=2$ and $5$, and 1.3 for $m=20$, 
		$\underline{R}_{m}$ is an $m \times m$ matrix with ones on the main diagonal and $0.5$ elsewhere, 
		and $U(a,b)$ is the uniform distribution on $(a,b)$. 
		Moreover, the innovations $\{\bm{\eta}_{t}\}$ are $i.i.d.$ random vectors following the multivariate normal or Student's $t_5$ distribution with zero mean and identity covariance matrix. 
		It can be verified that DGP1--DGP5 satisfy the strict stationarity condition in Theorem \ref{thm Stationarity}. 
		In addition, we use the \texttt{constrOptim} function in \textsf{R} with quasi-Newton method to calculate the QMLEs defined by \eqref{est general QMLE} and \eqref{est lowrank QMLE} in the manuscript. Particularly, an exact algorithm based on the fast Fourier transform is employed to efficiently handle the infinite sums involved in the computation of the QMLEs; see \cite{Nielsen_Noel2021} for details.
		
		To save space, Section \ref{subsection model estimation} of the manuscript only reports the biases, empirical standard deviations (ESDs), and asymptotic standard deviations (ASDs) of $\widehat{\bm{\theta}}_{\text{G}}$ under DGP1. 
		This section provides the simulation results for other settings in Tables \ref{table general QMLE m2r2s0Bkm}--\ref{table lowrank QMLE m5r1s0Bkm} and Figure \ref{figure simulation_m20}. 
		Note that $G_{0,1}$ and $\underline{R}$ are $5 \times 5$ (or $20 \times 20$) matrices in DGP4 (or DGP5). 
		To save space, Tables \ref{table general QMLE m5r1s0Bkm}--\ref{table lowrank QMLE m5r1s0Bkm} only report the results of the elements in the first row of $G_{0,1}$ and the off-diagonal elements in the first column of $\underline{R}$, and Figure \ref{figure simulation_m20} only presents the first six elements in the first row of $G_{0,1}$ and the first six off-diagonal elements in the first column of $\underline{R}$.
		The simulation findings from Tables \ref{table general QMLE m2r2s0Bkm}--\ref{table lowrank QMLE m5r1s0Bkm} and Figure \ref{figure simulation_m20} are summarized in the manuscript. 

		Moreover, to save space, the manuscript relegates the simulation results in Section \ref{subsection model selection in simulation} for the proposed model selection method to this section.
		Table \ref{table BIC} provides the percentages of underfitted, correct selected, and overfitted cases by the BIC using the QMLEs $\widehat{\bm{\theta}}_{\text{G}}$ and $\widehat{\bm{\theta}}_{\text{LR}}$. 
		It can be seen that the BIC performs better as the sample size increases, which confirms the selection consistency of BIC. The BIC performs slightly worse when the distribution of $\{\bm{\eta}_{t}\}$ gets more heavy-tailed, and this tendency is consistent to the findings for the QMLEs $\widehat{\bm{\theta}}_{\text{G}}$ and $\widehat{\bm{\theta}}_{\text{LR}}$ due to the efficiency loss of Gaussian QMLEs.   
		Moreover, the overfitted percentage of BIC using $\widehat{\bm{\theta}}_{\text{LR}}$ is slightly greater than that using $\widehat{\bm{\theta}}_{\text{G}}$ for DGP4 with $m=5$. This is as expected because there are less parameters to optimize for the QMLE $\widehat{\bm{\theta}}_{\text{LR}}$ than for the QMLE $\widehat{\bm{\theta}}_{\text{G}}$ when $m=5$. 
		In addition, it can be observed that BIC using $\widehat{\bm{\theta}}_{\text{LR}}$ performs poorly when $m=2$ and $(r,s) = (0,1)$.  
		This is because the QMLE $\widehat{\bm{\theta}}_{\text{LR}}$ under the low-rank decomposition \eqref{G-decomposition} has more redundant parameters to optimize than those in the QMLE $\widehat{\bm{\theta}}_{\text{G}}$ without low-rank restrictions, which makes the BIC using $\widehat{\bm{\theta}}_{\text{LR}}$ tend to select underfitted models. 
		Thus we suggest to use the QMLE $\widehat{\bm{\theta}}_{\text{G}}$ when the dimension $m$ is small, especially when $m \leq 3$ and $s > 0$. 

	\subsection{Simulation studies for computational advantages} \label{subsection additional simulation}	

		This subsection illustrates the computational advantages of the proposed multivariate log-GARCH model through simulation experiments. 
		As discussed in Section \ref{section motivation}, the proposed model in \eqref{model Rt SGARCH(r,s)}--\eqref{model Phii in Dt} is designed to overcome the computational intractability of model \eqref{eq volatility in DCC}. 
		In particular, it reduces the computational cost from $O(m^{3})$ per matrix power to $O(1)$ per scalar operation, smooths the likelihood surface to accelerate convergence, and avoids numerical instabilities arising from potential non-normality of the coefficient matrix $B_{1}$ in \eqref{eq volatility in DCC}. 
		The simulation results below provide numerical evidence for these advantages.

		To focus on the comparison of computational cost due to the conditional variance part $\ln\bm{h}_{t}$, we consider a constant conditional correlation (CCC) setting with $R_t\equiv\underline{R}$ for the proposed model and two benchmark specifications: the multivariate log-GARCH(1,1) model in \eqref{eq volatility in DCC} and the multivariate log-ARCH(1) model. Specifically, let $\bm{y}_{t} = H_{t}^{1/2} \bm{\eta}_{t}$ with $H_{t} = D_{t} \underline{R} D_{t}$, where the diagonal elements of $\underline{R}$ are equal to one and the off-diagonal elements are set to 0.1. 
		Consider the following conditional variance dynamics:
		\begin{align}
			& \text{log-MARCH}(1): \quad \quad \;\; \ln\bm{h}_{t} = \bm{\omega} + A_{1} \ln\mathbf{y}^{2}_{t-1}; \label{eq ccc-log-arch(1)} \\ 
			& \text{log-MGARCH}(1,1): \;\;\; \ln\bm{h}_{t} = \bm{\omega} + A_{1} \ln\mathbf{y}^{2}_{t-1} + B_{1} \ln\bm{h}_{t-1}; \label{eq ccc-log-garch(1,1)} \\ 
			& \text{our model}:  \quad \quad \quad \quad \;\; \ln\bm{h}_{t} = \underline{\bm{\omega}} + \sum_{i=1}^{\infty} \Phi_{i} \ln\mathbf{y}^{2}_{t-i}, \label{eq ccc-log-arch(infty)} 
		\end{align}
		where $\Phi_{i} = \sum_{k=1}^{r} \lambda_{k}^{i-1} G_{0,k} + \sum_{k=1}^{s} \gamma_{k}^{i-1} [\cos((i-1) \varphi_{k}) G_{1,k} + \sin((i-1) \varphi_{k}) G_{2,k}]$. 

		To generate $B_{1}$ and $\Phi_{i}$, we consider two choices of the $m \times m$ matrix $J$ in Jordan form: 
		\[J^{(a)} = \Diag\{0.6, \bm{0}_{m-1}\}\quad\text{and}\quad J^{(b)} = \Diag\{C_{1}, \bm{0}_{m-2}\}\quad\text{with}\quad C_{1} = 0.6 \begin{pmatrix}
			\cos(\pi/3) & \sin(\pi/3) \\
			-\sin(\pi/3) & \cos(\pi/3)
		\end{pmatrix}.\]
		For each choice, we first randomly generate an orthogonal matrix $P$ and then set $B_{1} = PJP^{-1}$. 
		Next we randomly generate $A_{1}$ such that model \eqref{eq ccc-log-garch(1,1)} satisfies the stationarity condition $\rho(A_{1}+B_{1})<1$. 
		Then by \eqref{eq equivalence condition}, we construct the parameters $\lambda_{1}$, $\gamma_{1}$, $\varphi_{1}$, $G_{0,1}$, $G_{1,1}$ and $G_{2,1}$ in model \eqref{eq ccc-log-arch(infty)} from $J$, $P$ and $A_{1}$. 
		Consequently, model \eqref{eq ccc-log-arch(1)} is a special case of model \eqref{eq ccc-log-garch(1,1)}, while \eqref{eq ccc-log-arch(infty)} and \eqref{eq ccc-log-garch(1,1)} correspond to the same DGP by Remark \ref{remark connection with DCC}. 
		Moreover, the numbers of parameters to be estimated are $m^2 + n_{d}$ for \eqref{eq ccc-log-arch(1)}, $2 m^2 + n_{d}$ for \eqref{eq ccc-log-garch(1,1)}, and for \eqref{eq ccc-log-arch(infty)}, $1 + m^2 + n_{d}$ when $J = J^{(a)}$ and $2 + 2 m^2 + n_{d}$ when $J = J^{(b)}$, where $n_{d} = m + m(m-1)/2$. 
		For each model, we implement the full QMLE using the \texttt{optim} function in \textsf{R}, and record the numbers of iterations required for convergence over 200 replications. We consider dimensions $m=2,5$, and $10$.

		Figures \ref{figure convergence_iterations_r1s0}--\ref{figure convergence_iterations_r0s1} display box plots of the numbers of iterations required for convergence over 200 replications, under $J=J^{(a)}$ and $J=J^{(b)}$, respectively.
		We have the following findings. First, the proposed model requires substantially fewer iterations than the multivariate log-GARCH(1,1) model, and this efficiency gain becomes more pronounced as the dimension $m$ increases. Additionally, the multivariate log-GARCH(1,1) model becomes time-consuming for moderately large dimensions (in our experiments, for $m>10$). Second, the multivariate log-ARCH(1) model also requires far fewer iterations than the multivariate log-GARCH(1,1) model. Given that the latter incorporates the additional autoregressive term $B_{1}\ln \bm{h}_{t-1}$, this finding indicates that $B_{1}$ is the main source of computational burden. Finally, the proposed model not only achieves a substantial reduction in iterations relative to the multivariate log-GARCH(1,1) model but also remains competitive with the log-ARCH(1) benchmark. These results suggest that the proposed model achieves an effective balance between model generality and computational efficiency.

	\section{Additional results for the empirical analysis}
	\subsection{Additional results for five major stock market indices} \label{subsection additive results for example 1}
		The first empirical example analyzes the daily log returns of five major stock market indices. 
		To save space, the manuscript relegates the descriptive analysis for $\{\mathbf{y}_{t}\}_{t=1}^n$ as well as details of the fitted model \eqref{fitted model general} in Section \ref{section real data} to this section. 
		Figure \ref{figure time plot} shows the time plot of $\{\mathbf{y}_{t}\}$, which implies very similar volatility clustering pattern among each component. Moreover, the summary statistics of $\{\mathbf{y}_{t}\}$ in Table \ref{table major stocks summary statistics} indicate that each component is skewed and heavy-tailed. 
		These findings motivate us to investigate $\{\mathbf{y}_{t}\}$ using the proposed multivariate log-GARCH model in \eqref{model Rt SGARCH(r,s)}--\eqref{model Phii in Dt} and its inference tools. 
		
		For the fitted model \eqref{fitted model general} in the manuscript, Table \ref{table major stocks fitted coefficients} reports summary information on the fitted intercepts and scalar coefficients. 
		In addition, Figure \ref{figure_comparison} compares the  entries of the conditional covariance matrices $H_{t}$ (lower-left triangle) and the conditional correlation matrices $R_{t}$ (upper-right triangle) fitted by our proposed model and the benchmark diagonal BEKK model. The first 50 observations are excluded to reduce initialization bias. The fitted paths under the two models are highly similar for both $H_{t}$ and $R_{t}$, suggesting that our proposed model captures the multivariate dependence structure well while retaining a parsimonious parametrization. 

		We examine the forecasting performance of the fitted multivariate log-GARCH model by evaluating covariance matrix forecasts and portfolio Value-at-Risk (VaR) forecasts. 
		For covariance matrix prediction, we assess predictive accuracy using the mean squared error (MSE) and the quasi-likelihood (QLIKE) loss. 
		To provide a comprehensive comparison, we include several competing multivariate volatility models, summarized along with their abbreviations in Table \ref{table major stocks comparison methods}.
		Table \ref{table stock stat loss} reports the average MSEs and QLIKE losses for one-step-ahead covariance matrix forecasts across all models, and also indicates which models belong to the corresponding 90\% model confidence set (MCS) proposed by \cite{hansen2011model}. 
		The findings from Table \ref{table stock stat loss} are summarized in Section \ref{section real data} of the manuscript. 

		For portfolio VaR forecasts, we construct portfolios based on the fitted multivariate volatility models and evaluate the resulting one-step-ahead VaR forecasts. 
		Particularly, we consider both the equally weighted (EW) and minimum variance (MV) portfolios, with portfolio return $z_{t} = \bm{\iota}^{\prime} \mathbf{y}_{t}$. For the EW portfolio, the weight vector is given by $\bm{\iota} = m^{-1} \bm{1}_{m}$. For the MV portfolio with short selling, the weight vector $\bm{\iota}$ is chosen to minimize the conditional variance of $z_{t}$ denoted by $\sigma_{t}^{2} = \bm{\iota}^{\prime} H_{t} \bm{\iota}$, where $H_{t}$ is the conditional covariance matrix of $\mathbf{y}_{t}$ given $\mathcal{F}_{t-1}$. The corresponding MV weight has the form of $\bm{\iota} = (\bm{1}_{m}^\prime H_{t}^{-1} \bm{1}_{m})^{-1} H_{t}^{-1} \bm{1}_{m}$. 
		Note that the $\tau$-th VaR of $z_t$ is its negative $\tau$-th conditional quantile. We assume a volatility model for $z_t$ such that the $\tau$-th negative VaR of $z_t$ is $Q_{\tau}(z_{t}\mid \mathcal{F}_{t-1})=\sigma_{t}b_{\tau}$, where $b_{\tau}$ is the $\tau$-th quantile of the $i.i.d.$ innovations $\{e_t\}$ with $e_t=z_t/\sigma_{t}$. 
		Moreover, the feasible EW portfolio return is given by $\widehat{z}_{t}^{\mathrm{EW}} = (\bm{\iota}^{\mathrm{EW}})^{\prime} \mathbf{y}_{t} = m^{-1} \bm{1}_{m}^\prime \mathbf{y}_t$, and the conditional variance is $(\widehat{\sigma}_{t}^{\mathrm{EW}})^{2} = (\widehat{\bm{\iota}}^{\mathrm{EW}})^\prime \widehat{H}_{t} \widehat{\bm{\iota}}^{\mathrm{EW}} = m^{-2} \bm{1}_{m}^{\prime} \widehat{H}_{t} \bm{1}_{m}$. 
		Based on the estimate $\widehat{H}_{t}=\widehat{D}_{t} \widehat{R}_{t} \widehat{D}_{t}$ from the fitted multivariate log-GARCH model, we can calculate the feasible MV weight as $\widehat{\bm{\iota}}^{\mathrm{MV}} = (\bm{1}_{m}^\prime \widehat{H}_{t}^{-1} \bm{1}_{m})^{-1} \widehat{H}_{t}^{-1} \bm{1}_{m}$, the feasible MV portfolio return as $\widehat{z}_{t}^{\mathrm{MV}} = (\widehat{\bm{\iota}}^{\mathrm{MV}})^{\prime} \mathbf{y}_{t}$, and the minimum conditional variance as $(\widehat{\sigma}_{t}^{\mathrm{MV}})^{2} = (\widehat{\bm{\iota}}^{\mathrm{MV}})^\prime \widehat{H}_{t} \widehat{\bm{\iota}}^{\mathrm{MV}}$. 
		For either portfolio strategy, let $\widehat{z}_{t}$ denote the corresponding feasible portfolio return and let $\widehat{\sigma}_{t}$ denote its associated conditional standard deviation. 
		To examine the VaR forecasts of the portfolio $\widehat{z}_t$, we conduct one-step-ahead predictions using a rolling forecasting procedure with a fixed moving window. 
		Specifically, we begin with the forecast origin $t_{0} = n_{0} + 1 = 1971$, and fit the multivariate log-GARCH model with order $(r,s) = (1,0)$ by the QMLE $\widehat{\bm{\theta}}_{\text{G}}$ using the data from the beginning to $n_{0} = 1970$ which covers 8 years' data observations. 
		We then obtain the one-day-ahead forecast $\widehat{H}_{t_{0}}$ and construct the feasible EW or MV portfolio return $\widehat{z}_{t_{0}}$ together with its associated conditional standard deviation $\widehat{\sigma}_{t_{0}}$.  
		Thus the $\tau$-th negative VaR of $\widehat{z}_{t_{0}}$ can be calculated by $\widehat{Q}_{\tau}( \widehat{z}_{t_{0}}\mid \mathcal{F}_{t_{0}-1})=\widehat{\sigma}_{t_{0}}\widehat{b}_{\tau}$, where $\widehat{b}_{\tau}$ is the $\tau$-th sample quantile obtained via filtered historical simulation (FHS) based on the most recent 1000 historical standardized residuals $\{\widehat{e}_{t_{0}-k}\}_{k=1}^{1000}$, with $\widehat{e}_{t} = \widehat{z}_t/\widehat{\sigma}_t$. 

		To evaluate the VaR forecast performance, we compute the empirical coverage rate (ECR) and prediction error (PE), report the average quantile loss ($\text{Loss}_{\mathrm{Q}}$) and calibration loss ($\text{Loss}_{\mathrm{C}}$), construct the 90\% model confidence set (MCS) based on $\text{Loss}_{\mathrm{Q}}$ and $\text{Loss}_{\mathrm{C}}$, and conduct two standard VaR backtests. 
		Particularly, the two VaR backtests include the likelihood ratio test for correct conditional coverage (CC) in \cite{Christoffersen_1998} and the dynamic quantile (DQ) test in \cite{Engle2004}. 
		Let $\hbar_{t} = I(\widehat{z}_t<\widehat{Q}_{\tau}(\widehat{z}_{t} \mid \mathcal{F}_{t-1}))$ be the hit series. 
		Under the null hypothesis of the CC test, conditional on $\mathcal{F}_{t-1}$, $\{\hbar_{t}\}$ are $i.i.d.$ Bernoulli random variables with success probability $\tau$. 
		For the DQ test, its null hypothesis is that all coefficients in the regression of $\hbar_{t}$ are zero and the intercept equals $\tau$, where the regressors includes four lagged hits $\hbar_{t-i}$ ($1\le i \le 4$) and the contemporaneous VaR forecast \citep{Engle2004}. 
		If the null hypothesis of VaR backtests cannot be rejected, then the VaR forecasts are satisfactory. 
		
		Table \ref{table major stocks VaR MV} in the manuscript and Table \ref{table major stocks VaR EW} summarize the one-step-ahead VaR forecasting results of all methods for the MV and EW portfolios, respectively, in terms of ECR, PE, $\text{Loss}_{\mathrm{Q}}$, $\text{Loss}_{\mathrm{C}}$, and backtesting results at the lower and upper $1\%$, $2.5\%$, and $5\%$ quantile levels. 
		The findings from Table \ref{table major stocks VaR MV} and Table \ref{table major stocks VaR EW} are summarized in Section \ref{section real data} of the manuscript.

	\subsection{Seventeen industry portfolios} \label{subsection example 2}
		To save space, the manuscript relegates the second empirical example to this section. 
		This section analyzes the 17 centered daily Fama-French industry portfolios, from January 2, 2009, to December 30, 2022, 
		which is downloaded from Kenneth French's database (\url{http://mba.tuck.dartmouth.edu/pages/faculty/ken.french/data_library.html}). 
		The 17 industries include food, mines, oil, clothes, consumer durables, chemicals, consumer goods, construction, steel, fabricated products, machine, cars, transportation, utilities, retail stores, financial, and other. Moreover, the 17 industry portfolios are built based on NYSE, AMEX, and NASDAQ stocks using equal weights, and they are widely used for portfolio research in the literature \citep{fama1997industry,behr2012using,hirshleifer2020mood}. 
		Finally, a dataset of the centered time series $\{\mathbf{y}_{t}\}_{t=1}^n$ is recorded with $m=17$ and $n=3524$ for $\mathbf{y}_{t} = (y_{1t}, \ldots, y_{mt})^{\prime}$.

		We employ the QMLE with low-rank restrictions on coefficient matrices in Section \ref{subsection QMLE under low rank} to fit the entire dateset of moderate dimension,
		and the BIC in \eqref{BIC} selects $(r,s) = (0,1)$ based on $K_{\text{max}} = 2$. 
		As in Section \ref{subsection major stocks forecasting} of the manuscript, we examine the forecasting performance of the fitted multivariate log-GARCH model by evaluating covariance matrix forecasts and portfolio VaR forecasts.
		For VaR forecasting, we consider both the equally weighted (EW) and minimum variance (MV) portfolios. 
		We conduct one-step-ahead predictions using a rolling forecasting procedure with a fixed moving window $n_{0} = 3021$ that covers 12 years' data observations. 
		Particularly, we adopt the common convention that 5 consecutive trading days constitute one ``week'', and the MV weight $\bm{\iota}$ is updated ``weekly'' in the rolling forecasting procedure to alleviate computational costs for all methods \citep{de2018mgarch}. 

		Table \ref{table portfolio stat loss} reports the average MSEs and QLIKE losses for one-step-ahead covariance matrix forecasts across all models, and also indicates which models belong to the corresponding 90\% model confidence set (MCS). 
		We can see that the proposed model remains highly competitive in covariance matrix forecasting, although it no longer dominates all competitors. 
		In particular, it is included in the 90\% MCS under both MSE and QLIKE. The smallest QLIKE is achieved by D-BEKK-VT, while the proposed model still performs better than most competing methods. Under the MSE criterion, the differences across most methods remain quite limited. As in the first empirical example, this is likely because MSE evaluates forecast errors in an elementwise manner, while QLIKE depends on both $\widehat{H}_t^{-1}$ and $\log|\widehat{H}_t|$ and is therefore more sensitive to differences in the overall covariance structure.

		Tables \ref{table portfolio VaR MV}--\ref{table portfolio VaR EW} summarize the one-step-ahead VaR forecasting results of all methods for the MV and EW portfolios, respectively, in terms of ECR, PE, $\text{Loss}_{\mathrm{Q}}$, $\text{Loss}_{\mathrm{C}}$, and backtesting results at the lower and upper $1\%$, $2.5\%$, and $5\%$ quantile levels. 
		From Table \ref{table portfolio VaR MV}, it can be seen that the proposed model still performs best for the MV portfolio. 
		In particular, it achieves the smallest PE and $\text{Loss}_{\mathrm{C}}$ at five quantile levels, and is included in the 90\% MCS under both $\text{Loss}_{\mathrm{Q}}$ and $\text{Loss}_{\mathrm{C}}$ at all quantile levels. 
		Moreover, it passes both the CC and DQ tests at four quantile levels, which is the highest number jointly attained with the CCC model.		
		For the EW portfolio in Table \ref{table portfolio VaR EW}, the proposed model remains highly competitive. 
		It attains the smallest $\text{Loss}_{\mathrm{C}}$ at all quantile levels, and is included in the 90\% MCS under both $\text{Loss}_{\mathrm{Q}}$ and $\text{Loss}_{\mathrm{C}}$ throughout. 
		In addition, it passes both backtests at four quantile levels, and achieves the smallest PE at three quantile levels. 
		
		Compared with the first empirical example, the proposed model now performs best for both the MV and EW portfolios. This suggests that, as the dimension increases, the proposed multivariate log-GARCH model becomes more advantageous in portfolio VaR forecasting. 
		Overall, the results from the two empirical examples indicate that the proposed multivariate log-GARCH model estimated by QMLE generally outperforms the competing methods in forecasting VaRs, highlighting its favorable trade-off between generality and efficiency especially for many assets.

	\section{The connection with DCC-log-GARCH models} \label{section connection between DCC and SGARCH}
		We consider the following DCC-log-GARCH model: 
		\begin{align}
			&\mathbf{y}_{t} = H_{t}^{1/2} \bm{\eta}_{t}, \;\;
			H_{t} = D_{t} R_{t} D_{t}, \;\;
			R_{t} = (1 - \beta_{1} - \beta_{2}) \underline{R} + \beta_{1} \Psi_{t-1} + \beta_{2} R_{t-1}, \notag\\
			&\ln\bm{h}_{t} = \bm{\omega} + A_{1} \ln\mathbf{y}^{2}_{t-1} + B_{1} \ln\bm{h}_{t-1}, \label{eq volatility in DCC(1,1)}
		\end{align}
		where $\mathbf{y}_{t} = (y_{1t}, \ldots, y_{mt})^{\prime}$, 
		$\bm{\eta}_{t}$ is an $m$-dimensional innovation, 
		$H_{t} = [h_{ij,t}]$ is the conditional covariance matrix of $\mathbf{y}_{t}$ given $\mathcal{F}_{t-1} = \sigma\{\mathbf{y}_{t-1}, \mathbf{y}_{t-2}, \ldots\}$, 
		$D_{t} = [\Diag(H_{t})]^{1/2}$, 
		and $R_{t}$ is the conditional correlation matrix of $\mathbf{y}_{t}$ that is generated from a commonly used model \citep{tse2002multivariate}. 
		In \eqref{eq volatility in DCC(1,1)}, 
		$\ln\bm{h}_{t} = (\ln h_{11,t}, \ldots, \ln h_{mm,t})^{\prime}$, 
		$\ln\mathbf{y}^{2}_{t} = (\ln y_{1t}^{2}, \ldots, \ln y_{mt}^{2})^{\prime}$, 
		$\bm{\omega}$ is an $m$-dimensional parameter vector, 
		and both $A_{1}$ and $B_{1}$ are $m \times m$ parameter matrices. 

		We first show that the DCC-log-GARCH model can be rewritten in the form of the proposed efficient multivariate log-GARCH model in \eqref{model Dt SGARCH(r,s)}--\eqref{model Phii in Dt}. 
		Suppose that $\rho(B_{1}) < 1$, then model \eqref{eq volatility in DCC(1,1)} can be rewritten as a multivariate log-ARCH($\infty$) form below, 
		\begin{align} \label{eq ARCHinfinity for volatility in DCC(1,1)}
			\ln\bm{h}_{t} = (I_{m} - B_{1})^{-1} \bm{\omega} + \sum_{i=0}^{\infty} B_{1}^{i} A_{1} \ln\mathbf{y}^{2}_{t-i-1}. 
		\end{align}
		Moreover, suppose that $B_{1}$ is diagonalizable over the complex field, and it has 
		$r$ nonzero real eigenvalues, $\lambda_{1}, \ldots, \lambda_{r}$, 
		and $s$ conjugate pairs of nonzero complex eigenvalues, $\lambda_{r+1}, \ldots, \lambda_{r+2s}$, with $(\lambda_{r+2k-1}, \lambda_{r+2k}) = (\gamma_{k} e^{i \varphi_{k}}, \gamma_{k} e^{-i \varphi_{k}}) = (\gamma_{k}(\cos\varphi_{k} + i \sin\varphi_{k}), \gamma_{k}(\cos\varphi_{k} - i \sin\varphi_{k}))$, $\gamma_{k} > 0$, $\varphi_{k} \in (0,\pi)$, $1 \leq k \leq s$, $i$ being the imaginary unit 
		and $r+2s \leq m$. 
		Note that since $\rho(B_{1}) < 1$, it holds that $0 < |\lambda_{j}| < 1$ for $1 \leq j \leq r$ and $0 < \gamma_{k} < 1$ for $1 \leq k \leq s$. 
		By the block diagonalization theorem, there exists an $m \times m$ invertible matrix $P$ such that $B_{1} = P J P^{-1}$, 
		where $J = \Diag\{\lambda_{1}, \ldots, \lambda_{r}, C_{1}, \ldots, C_{s}, \bm{0}_{m-r-2s}\}$ is an $m \times m$ real block diagonal matrix with 
		$$
			C_{k} = \gamma_{k} 
			\left(
			\begin{matrix}
				\cos\varphi_{k} & \sin\varphi_{k} \\
				-\sin\varphi_{k} & \cos\varphi_{k}
			\end{matrix}
			\right) 
			\;\; \text{for} \;\; 1 \leq k \leq s. 
		$$
		Thus \eqref{eq ARCHinfinity for volatility in DCC(1,1)} can be rewritten as follows, 
		\begin{align} \label{eq using B=PJP-1}
			\ln\bm{h}_{t} = (I_{m} - B_{1})^{-1} \bm{\omega} + \sum_{i=0}^{\infty} P J^{i} P^{-1} A_{1} \ln\mathbf{y}^{2}_{t-i-1}. 
		\end{align}
		Let $\underline{\bm{\omega}} = (I_{m} - B_{1})^{-1} \bm{\omega}$, $\underline{A} = P^{-1} A_{1}$ and $\underline{B} = P$. 
		It follows that 
		\begin{align}
			\ln\bm{h}_{t} = \underline{\bm{\omega}} + \sum_{i=0}^{\infty} \underline{B} J^{i} \underline{A} \ln\mathbf{y}^{2}_{t-i-1}. 
		\end{align}
		Furthermore, denote 
		$$
			\underline{A} = (\underline{A}_{1}, \ldots, \underline{A}_{r}, \underline{A}_{r+1}, \ldots, \underline{A}_{r+s}, \underline{A}_{r+s+1})^{\prime} 
			\;\; \text{and} \;\; 
			\underline{B} = (\underline{B}_{1}, \ldots, \underline{B}_{r}, \underline{B}_{r+1}, \ldots, \underline{B}_{r+s}, \underline{B}_{r+s+1}), 
		$$
		with $\underline{A}_{k} = \bm{\xi}_{k}$ and $\underline{B}_{k} = \bm{\zeta}_{k}$ being $m \times 1$ matrices for $1 \leq k \leq r$, 
		$\underline{A}_{r+k'} = (\bm{\xi}_{r+k'}, \bar{\bm{\xi}}_{r+k'})$ and $\underline{B}_{r+k'} = (\bm{\zeta}_{r+k'}, \bar{\bm{\zeta}}_{r+k'})$ being $m \times 2$ matrices for $1 \leq k' \leq s$, 
		and $\underline{A}_{r+s+1}$ and $\underline{B}_{r+s+1}$ being $m \times (m-r-2s)$ matrices. 
		It holds that 
		\begin{equation} \label{eq rewritten ARCHinfinity for volatility when pqOne}
		\begin{aligned}
			\ln\bm{h}_{t} 
			&= \underline{\bm{\omega}} 
			+ \sum_{i=0}^{\infty} \left[ 
			\sum_{k=1}^{r} \lambda_{k}^{i} \bm{\zeta}_{k} \bm{\xi}_{k}^{\prime} 
			+ \sum_{k=1}^{s} \left(\bm{\zeta}_{r+k}, \bar{\bm{\zeta}}_{r+k}\right) C_{k}^{i} \left(\bm{\xi}_{r+k}, \bar{\bm{\xi}}_{r+k}\right)^{\prime}
			\right] \ln\mathbf{y}^{2}_{t-i-1} \\
			&= \underline{\bm{\omega}} 
			+ \sum_{i=0}^{\infty} \Bigg\{ 
			\sum_{k=1}^{r} \lambda_{k}^{i} \bm{\zeta}_{k} \bm{\xi}_{k}^{\prime} 
			+ \sum_{k=1}^{s} \gamma_{k}^{i} \Big[\cos(i \varphi_{k}) \left(\bm{\zeta}_{r+k} \bm{\xi}_{r+k}^{\prime} + \bar{\bm{\zeta}}_{r+k} \bar{\bm{\xi}}_{r+k}^{\prime}\right) \\
			&\mathrel{\phantom{= \underline{\bm{\omega}} + \sum_{i=0}^{\infty} \Bigg\{ \sum_{k=1}^{r} \lambda_{k}^{i} \bm{\zeta}_{k} \bm{\xi}_{k}^{\prime} + \sum_{k=1}^{s} \gamma_{k}^{i} \Big[}}
			+ \sin(i \varphi_{k}) \left(\bm{\zeta}_{r+k} \bar{\bm{\xi}}_{r+k}^{\prime} - \bar{\bm{\zeta}}_{r+k} \bm{\xi}_{r+k}^{\prime}\right)\Big]
			\Bigg\} \ln\mathbf{y}^{2}_{t-i-1}. 
		\end{aligned}
		\end{equation}
		Then let matrices 
		\begin{align} \label{eq G in DCC}
			G_{0,k} &= \bm{\zeta}_{k} \bm{\xi}_{k}^{\prime} 
			\;\; \text{for} \;\; 1 \leq k \leq r, \;\; \text{and} \notag\\
			G_{1,k} &= \bm{\zeta}_{r+k} \bm{\xi}_{r+k}^{\prime} + \bar{\bm{\zeta}}_{r+k} \bar{\bm{\xi}}_{r+k}^{\prime} \;\; \text{and} \;\; 
			G_{2,k} = \bm{\zeta}_{r+k} \bar{\bm{\xi}}_{r+k}^{\prime} - \bar{\bm{\zeta}}_{r+k} \bm{\xi}_{r+k}^{\prime} 
			\;\; \text{for} \;\; 1 \leq k \leq s. 
		\end{align}
		We can obtain that 
		\begin{align*}
			\ln\bm{h}_{t} 
			&= \underline{\bm{\omega}} 
			+ \sum_{i=0}^{\infty} \left\{
			\sum_{k=1}^{r} \lambda_{k}^{i} G_{0,k} 
			+ \sum_{k=1}^{s} \gamma_{k}^{i} \left[\cos(i \varphi_{k}) G_{1,k} + \sin(i \varphi_{k}) G_{2,k}\right]
			\right\} \ln\mathbf{y}^{2}_{t-i-1} \\
			&= \underline{\bm{\omega}} 
			+ \sum_{i=1}^{\infty} \left\{
			\sum_{k=1}^{r} \lambda_{k}^{i-1} G_{0,k} 
			+ \sum_{k=1}^{s} \gamma_{k}^{i-1} \left[\cos((i-1) \varphi_{k}) G_{1,k} + \sin((i-1) \varphi_{k}) G_{2,k}\right]
			\right\} \ln\mathbf{y}^{2}_{t-i} \\
			&= \underline{\bm{\omega}} + \sum_{i=1}^{\infty} \Phi_{i} \ln\mathbf{y}^{2}_{t-i}, 
		\end{align*}
		where $\Phi_{i} = \sum_{k=1}^{r} \lambda_{k}^{i-1} G_{0,k} + \sum_{k=1}^{s} \gamma_{k}^{i-1} \left[\cos((i-1) \varphi_{k}) G_{1,k} + \sin((i-1) \varphi_{k}) G_{2,k}\right]$. 
		As a result, we rewrite the DCC-log-GARCH model in the form of the proposed efficient multivariate log-GARCH model. 

		On the other hand, from the above process, we know that if $B_{1}$ is diagonalizable over the complex field and \eqref{eq G in DCC} holds, then the proposed efficient multivariate log-GARCH model in \eqref{model Dt SGARCH(r,s)}--\eqref{model Phii in Dt} will be equivalent to the DCC-log-GARCH model.

	\section{Model interpretation} \label{section model interpretation}

	Denote $f_{I,k} = \sum_{i=1}^{\infty} \lambda_{k}^{i-1} G_{0,k} \ln\mathbf{y}^{2}_{t-i}$ for $1 \leq k \leq r$, and $f_{II,k} = \sum_{i=1}^{\infty} \gamma_{k}^{i-1} [\cos((i-1) \varphi_{k}) G_{1,k} + \sin((i-1) \varphi_{k}) G_{2,k}] \ln\mathbf{y}^{2}_{t-i}$ for $1 \leq k \leq s$. Then the new model defined by \eqref{model Dt SGARCH(r,s)}–\eqref{model Phii in Dt} admits an additive representation consisting of $r+s$ components as follows: 
	$$
	  \ln\bm{h}_{t} = \underline{\bm{\omega}} + \sum_{k=1}^{r} f_{I,k} + \sum_{k=1}^{s} f_{II,k}.
	$$
	Each component exhibits a distinct decaying pattern and contributes additively to the dynamics of $\ln\bm{h}_{t}$. To better understand the dependence structure, we consider the dynamics generated by a single component of each type:
	\begin{align}
	  & \text{Type I:} \quad \ln\bm{h}_{t} = \underline{\bm{\omega}} + \sum_{i=1}^{\infty} \lambda_{k}^{i-1} G_{0,k} \ln\mathbf{y}^{2}_{t-i}; \label{eq component 1}\\
	  & \text{Type II:} \quad \ln\bm{h}_{t} = \underline{\bm{\omega}} + \sum_{i=1}^{\infty} \gamma_{k}^{i-1} \left[\cos((i-1) \varphi_{k}) G_{1,k} + \sin((i-1) \varphi_{k}) G_{2,k}\right] \ln\mathbf{y}^{2}_{t-i}. \label{eq component 2}
	\end{align}
	Let $\mathbf{z}_{t} = \ln\mathbf{y}^{2}_{t}$ and $\mathbf{v}_{t} = \ln\bm{\varepsilon}^{2}_{t}$. From $\bm{\varepsilon}_{t} = D_{t}^{-1} \mathbf{y}_{t}$, we have $\mathbf{z}_{t} = \ln\bm{h}_{t} + \mathbf{v}_{t}$. Then the specifications in \eqref{eq component 1} and \eqref{eq component 2} can be rewritten as vector autoregressive moving average (VARMA) models of order $(1,1)$ and $(2,2)$, respectively: 
	\begin{equation} \label{eq rewritten as VARMA}
	\begin{aligned}
	  & \text{Type I}\Leftrightarrow \text{VARMA}(1,1): \quad \mathbf{z}_{t} = \mathbf{c}^{*} + A_{1}^{*} \mathbf{z}_{t-1} + \mathbf{v}_{t} - B_{1}^{*} \mathbf{v}_{t-1}; \\ 
	  & \text{Type II} \Leftrightarrow \text{VARMA}(2,2): \quad \mathbf{z}_{t} = \mathbf{c}^{\circ} + A_{1}^{\circ} \mathbf{z}_{t-1} + A_{2}^{\circ} \mathbf{z}_{t-2} + \mathbf{v}_{t} - B_{1}^{\circ} \mathbf{v}_{t-1} - B_{2}^{\circ} \mathbf{v}_{t-2}, 
	\end{aligned}
	\end{equation}
	where $\mathbf{c}^{*} = (1 - \lambda_{1}) \underline{\bm{\omega}}$, $A_{1}^{*} = \lambda_{1} I_{m} + G_{0,1}$, $B_{1}^{*} = \lambda_{1} I_{m}$, $\mathbf{c}^{\circ} = (1 - 2 \gamma_{1} \cos\varphi_{1} + \gamma_{1}^{2}) \underline{\bm{\omega}}$, $A_{1}^{\circ} = 2 \gamma_{1} \cos\varphi_{1} I_{m} + G_{1,1}$, $A_{2}^{\circ} = - \gamma_{1}^{2} I_{m} - \gamma_{1} \cos\varphi_{1} G_{1,1} + \gamma_{1} \sin\varphi_{1} G_{2,1}$, $B_{1}^{\circ} = 2 \gamma_{1} \cos\varphi_{1} I_{m}$, and $B_{2}^{\circ} = - \gamma_{1}^{2} I_{m}$; see the proof at the end of this section for details. 

	Let $\Gamma_{z}(s) = \cov(\mathbf{z}_{t}, \mathbf{z}_{t-s})$ denote the autocovariance function of $\mathbf{z}_{t}$. 
	From the VARMA representations, the autocovariances satisfy the following recursive equations: 
	\begin{align*}
	  & \text{Type I:} \quad \Gamma_{z}(s) = A_{1}^{*} \Gamma_{z}(s-1) + \cov(\mathbf{v}_{t} - B_{1}^{*} \mathbf{v}_{t-1}, \mathbf{z}_{t-s}); \\ 
	  & \text{Type II:} \quad \Gamma_{z}(s) = A_{1}^{\circ} \Gamma_{z}(s-1) + A_{2}^{\circ} \Gamma_{z}(s-2) + \cov(\mathbf{v}_{t} - B_{1}^{\circ} \mathbf{v}_{t-1} - B_{2}^{\circ} \mathbf{v}_{t-2}, \mathbf{z}_{t-s}). 
	\end{align*}
	These recursions indicate that the autocorrelation structure of $\ln\mathbf{y}^{2}_{t}$ (and hence of $\mathbf{y}^{2}_{t}$) is governed by the parameters $\lambda_{k}$'s, $\gamma_{k}$'s and $\varphi_{k}$'s. 
	Although a simple closed-form expression for the autocorrelation function is not available, the recursions provide valuable insight: the $\lambda_{k}$'s are associated with volatility persistence, while the pairs $(\gamma_k,\varphi_k)$ capture oscillatory dynamics. 
	Simulations in Section \ref{subsection model interpretation} of the manuscript further illustrate and support these relationships. 

	\begin{proof}[\textbf{Proof of \eqref{eq rewritten as VARMA}}]
		Recall that $\bm{\varepsilon}_{t} = D_{t}^{-1} \mathbf{y}_{t}$ with $\mathbf{y}_{t} = (y_{1t}, \ldots, y_{mt})^{\prime}$, $\bm{\varepsilon}_{t} = (\varepsilon_{1t}, \ldots, \varepsilon_{mt})^{\prime}$ and $D_{t} = \Diag\{h_{11,t}^{1/2}, \ldots, h_{mm,t}^{1/2}\}$. We have that $\ln y_{it}^{2} = \ln h_{ii,t} + \ln \varepsilon_{it}^{2}$ for $i = 1, \ldots, m$, i.e. 
		$$
			\ln\mathbf{y}^{2}_{t} = \ln\bm{h}_{t} + \ln\bm{\varepsilon}^{2}_{t}, 
		$$
		where $\ln\mathbf{y}^{2}_{t} = (\ln y_{1t}^{2}, \ldots, \ln y_{mt}^{2})^{\prime}$, $\ln\bm{h}_{t} = (\ln h_{11,t}, \ldots, \ln h_{mm,t})^{\prime}$ and $\ln\bm{\varepsilon}^{2}_{t} = (\ln \varepsilon_{1t}^{2}, \ldots, \ln \varepsilon_{mt}^{2})^{\prime}$. 
		This together with $\ln\bm{h}_{t} = \underline{\bm{\omega}} + \sum_{i=1}^{\infty} \Phi_{i} \ln\mathbf{y}^{2}_{t-i}$ by model \eqref{model Dt SGARCH(r,s)}, implies that 
		$$
			\ln\mathbf{y}^{2}_{t} = \underline{\bm{\omega}} + \sum_{i=1}^{\infty} \Phi_{i} \ln\mathbf{y}^{2}_{t-i} + \ln\bm{\varepsilon}^{2}_{t}. 
		$$
		Denote $\mathbf{z}_{t} = \ln\mathbf{y}^{2}_{t}$ and $\mathbf{v}_{t} = \ln\bm{\varepsilon}^{2}_{t}$. It then follows that 
		\begin{equation} \label{rewritten model zt}
			\mathbf{z}_{t} = \underline{\bm{\omega}} + \sum_{i=1}^{\infty} \Phi_{i} \mathbf{z}_{t-i} + \mathbf{v}_{t}. 
		\end{equation}

		Firstly, we show that if $\Phi_{i} = \lambda_{1}^{i-1} G_{0,1}$, then model \eqref{rewritten model zt} can be rewritten as a VARMA($1,1$) model. 
		Denote $L$ as the backshift operator. 
		By \eqref{rewritten model zt}, it can be shown that 
		\begin{align*}
			\mathbf{z}_{t} &= \underline{\bm{\omega}} + \sum_{i=1}^{\infty} \lambda_{1}^{i-1} G_{0,1} \mathbf{z}_{t-i} + \mathbf{v}_{t} 
			= \underline{\bm{\omega}} + G_{0,1} \sum_{i=0}^{\infty} \lambda_{1}^{i} \mathbf{z}_{t-i-1} + \mathbf{v}_{t} \\
			& = \underline{\bm{\omega}} + G_{0,1} \sum_{i=0}^{\infty} (\lambda_{1} L)^{i} \mathbf{z}_{t-1} + \mathbf{v}_{t} 
			= \underline{\bm{\omega}} + G_{0,1} \frac{1}{1 - \lambda_{1} L} \mathbf{z}_{t-1} + \mathbf{v}_{t}. 
		\end{align*}
		It follows that 
		\begin{align*}
			(1 - \lambda_{1} L) \mathbf{z}_{t} &= (1 - \lambda_{1} L) \underline{\bm{\omega}} + G_{0,1} \mathbf{z}_{t-1} + (1 - \lambda_{1} L) \mathbf{v}_{t}. 
		\end{align*}
		Then we obtain that  
		\begin{align*}
			\mathbf{z}_{t} = (1 - \lambda_{1}) \underline{\bm{\omega}} + (\lambda_{1} I_{m} + G_{0,1}) \mathbf{z}_{t-1} + \mathbf{v}_{t} - \lambda_{1} I_{m} \mathbf{v}_{t-1} 
			\equiv \mathbf{c}^{*} + A_{1}^{*} \mathbf{z}_{t-1} + \mathbf{v}_{t} - B_{1}^{*} \mathbf{v}_{t-1}, 
		\end{align*}
		and 
		\begin{align*}
			\cov(\mathbf{z}_{t}, \mathbf{z}_{t-s}) = A_{1}^{*} \cov(\mathbf{z}_{t-1}, \mathbf{z}_{t-s}) + \cov(\mathbf{v}_{t} - B_{1}^{*} \mathbf{v}_{t-1}, \mathbf{z}_{t-s}), 
		\end{align*}
		where $A_{1}^{*} = \lambda_{1} I_{m} + G_{0,1}$, $B_{1}^{*} = \lambda_{1} I_{m}$ and $\mathbf{c}^{*} = (1 - \lambda_{1}) \underline{\bm{\omega}}$. 

		Secondly, we prove that if $\Phi_{i} = \gamma_{1}^{i-1} \left[\cos((i-1) \varphi_{1}) G_{1,1} + \sin((i-1) \varphi_{1}) G_{2,1}\right]$, then model \eqref{rewritten model zt} can be rewritten as a VARMA($2,2$) model. 
		Denote 
		$$
			\Phi^{*}(L) = \sum_{i=1}^{\infty} \gamma_{1}^{i-1} \left[\cos((i-1) \varphi_{1}) G_{1,1} + \sin((i-1) \varphi_{1}) G_{2,1}\right] L^{i}. 
		$$ 
		Then by \eqref{rewritten model zt} we have 
		\begin{align} \label{eq zt case II}
			\mathbf{z}_{t} = \underline{\bm{\omega}} + \Phi^{*}(L) \mathbf{z}_{t} + \mathbf{v}_{t},
			\;\; \text{or} \;\;
			[I_{m} - \Phi^{*}(L)] \mathbf{z}_{t} = \underline{\bm{\omega}} + \mathbf{v}_{t}. 
		\end{align}
		By Euler's formula, it holds that $\cos(j \varphi_{1}) = (e^{i j \varphi_{1}} + e^{-i j \varphi_{1}}) / 2$ and $\sin(j \varphi_{1}) = (e^{i j \varphi_{1}} - e^{-i j \varphi_{1}}) / (2 i)$, where $i$ is the imaginary unit and $j$ is an integer. It can be shown that 
		\begin{align*}
			\Phi^{*}(L) &= \sum_{j=0}^{\infty} \gamma_{1}^{j} \left[\cos(j \varphi_{1}) G_{1,1} + \sin(j \varphi_{1}) G_{2,1}\right] L^{j+1} \\
			&= \sum_{j=0}^{\infty} \gamma_{1}^{j} \left(\frac{e^{i j \varphi_{1}} + e^{-i j \varphi_{1}}}{2} G_{1,1} + \frac{e^{i j \varphi_{1}} - e^{-i j \varphi_{1}}}{2 i} G_{2,1}\right) L^{j+1} \\
			&= \frac{1}{2} G_{1,1} \left[\sum_{j=0}^{\infty} \left(\gamma_{1} e^{i \varphi_{1}} L\right)^{j} + \sum_{j=0}^{\infty} \left(\gamma_{1} e^{- i \varphi_{1}} L\right)^{j}\right] L 
			+ \frac{1}{2 i} G_{2,1} \left[\sum_{j=0}^{\infty} \left(\gamma_{1} e^{i \varphi_{1}} L\right)^{j} - \sum_{j=0}^{\infty} \left(\gamma_{1} e^{- i \varphi_{1}} L\right)^{j}\right] L \\
			&= \frac{1}{2} G_{1,1} \left(\frac{1}{1 - \gamma_{1} e^{i \varphi_{1}} L} + \frac{1}{1 - \gamma_{1} e^{- i \varphi_{1}} L}\right) L 
			+ \frac{1}{2 i} G_{2,1} \left(\frac{1}{1 - \gamma_{1} e^{i \varphi_{1}} L} - \frac{1}{1 - \gamma_{1} e^{- i \varphi_{1}} L}\right) L \\
			&= \frac{1}{2} G_{1,1} \frac{2 - \gamma_{1} (e^{i \varphi_{1}} + e^{- i \varphi_{1}}) L}{(1 - \gamma_{1} e^{i \varphi_{1}} L)(1 - \gamma_{1} e^{- i \varphi_{1}} L)} L 
			+ \frac{1}{2 i} G_{2,1} \frac{\gamma_{1} (e^{i \varphi_{1}} - e^{- i \varphi_{1}}) L}{(1 - \gamma_{1} e^{i \varphi_{1}} L)(1 - \gamma_{1} e^{- i \varphi_{1}} L)} L \\
			&= \frac{1}{2} G_{1,1} \frac{2 - \gamma_{1} (e^{i \varphi_{1}} + e^{- i \varphi_{1}}) L}{1 - \gamma_{1} (e^{i \varphi_{1}} + e^{- i \varphi_{1}}) L + \gamma_{1}^{2} L^{2}} L 
			+ \frac{1}{2 i} G_{2,1} \frac{\gamma_{1} (e^{i \varphi_{1}} - e^{- i \varphi_{1}}) L}{1 - \gamma_{1} (e^{i \varphi_{1}} + e^{- i \varphi_{1}}) L + \gamma_{1}^{2} L^{2}} L \\
			&= G_{1,1} \frac{1 - \gamma_{1} \cos\varphi_{1} L}{1 - 2 \gamma_{1} \cos\varphi_{1} L + \gamma_{1}^{2} L^{2}} L 
			+ G_{2,1} \frac{\gamma_{1} \sin\varphi_{1} L}{1 - 2 \gamma_{1} \cos\varphi_{1} L + \gamma_{1}^{2} L^{2}} L \\
			&\equiv P^{-1}(L) \left[(L - \gamma_{1} \cos\varphi_{1} L^{2}) G_{1,1} + \gamma_{1} \sin\varphi_{1} L^{2} G_{2,1}\right], 
		\end{align*}
		where $P(L) = 1 - 2 \gamma_{1} \cos\varphi_{1} L + \gamma_{1}^{2} L^{2}$. 
		This together with \eqref{eq zt case II} implies that 
		\begin{align*}
			P(L) [I_{m} - \Phi^{*}(L)] \mathbf{z}_{t} = P(L) \underline{\bm{\omega}} + P(L) \mathbf{v}_{t}, 
		\end{align*}
		that is 
		\begin{align*}
			\mathbf{z}_{t} =& \left(1 - 2 \gamma_{1} \cos\varphi_{1} + \gamma_{1}^{2}\right) \underline{\bm{\omega}} 
			+ \left(2 \gamma_{1} \cos\varphi_{1} I_{m} + G_{1,1}\right) \mathbf{z}_{t-1} \\
			&+ \left(- \gamma_{1}^{2} I_{m} - \gamma_{1} \cos\varphi_{1} G_{1,1} + \gamma_{1} \sin\varphi_{1} G_{2,1}\right) \mathbf{z}_{t-2} 
			+ \mathbf{v}_{t} - 2 \gamma_{1} \cos\varphi_{1} I_{m} \mathbf{v}_{t-1} - \left(- \gamma_{1}^{2}\right) I_{m} \mathbf{v}_{t-2} \\
			\equiv& \mathbf{c}^{*} + A_{1}^{*} \mathbf{z}_{t-1} + A_{2}^{*} \mathbf{z}_{t-2} + \mathbf{v}_{t} - B_{1}^{*} \mathbf{v}_{t-1} - B_{2}^{*} \mathbf{v}_{t-2}, 
		\end{align*}
	    where $A_{1}^{*} = 2 \gamma_{1} \cos\varphi_{1} I_{m} + G_{1,1}$, $A_{2}^{*} = - \gamma_{1}^{2} I_{m} - \gamma_{1} \cos\varphi_{1} G_{1,1} + \gamma_{1} \sin\varphi_{1} G_{2,1}$, $B_{1}^{*} = 2 \gamma_{1} \cos\varphi_{1} I_{m}$, $B_{2}^{*} = - \gamma_{1}^{2} I_{m}$ and $\mathbf{c}^{*} = (1 - 2 \gamma_{1} \cos\varphi_{1} + \gamma_{1}^{2}) \underline{\bm{\omega}}$. 
	    It then follows that 
		\begin{align*}
			\cov(\mathbf{z}_{t}, \mathbf{z}_{t-s}) = A_{1}^{*} \cov(\mathbf{z}_{t-1}, \mathbf{z}_{t-s}) + A_{2}^{*} \cov(\mathbf{z}_{t-2}, \mathbf{z}_{t-s}) + \cov(\mathbf{v}_{t} - B_{1}^{*} \mathbf{v}_{t-1} - B_{2}^{*} \mathbf{v}_{t-2}, \mathbf{z}_{t-s}). 
		\end{align*}
	\end{proof}

	\section{Theory extension to modeling non-negative variables} \label{section theory extension}
		The moment conditions in Theorems 3.1--3.2 involve only second-order quantities $\ln \mathbf{y}_{t}^{2}$ and $\bm{\eta}_{t}\bm{\eta}_{t}^{\prime}$, suggesting that the asymptotic theory may extend beyond GARCH-type models. Following the insight of \cite{engle1998autoregressive} on the equivalence between GARCH and autoregressive conditional duration (ACD) models, our asymptotic theory naturally admits an extension to models for positive semi-definite matrix-valued processes. One prominent example is the matrix multiplicative error model \citep[Matrix-MEM]{sheppard2007positive}: $$X_{t} = H_{t}^{1/2} E_{t} (H_{t}^{1/2})',$$ 
		where $X_{t}$ is an observed positive semi-definite matrix-valued process, $E_{t}$ is a positive semi-definite innovation matrix with conditional mean equal to identity matrix, and $H_{t}$ is the conditional mean matrix of $X_{t}$. 
		Under a suitable conditional likelihood, we anticipate that the consistency and asymptotic normality results established in this paper can be adapted to this wider class.

	\section{Derivatives} \label{section Derivatives}
	\subsection{Derivatives of $\ln\bm{h}_{t}(\bm{\delta})$}
		Recall that $\ln\bm{h}_{t}(\bm{\delta}) = \underline{\bm{\omega}} + \sum_{i=1}^{\infty} \Phi_{i}(\bm{\kappa}) \ln\mathbf{y}^{2}_{t-i}$ with $\Phi_{i}(\bm{\kappa}) = \sum_{k=1}^{r} \lambda_{k}^{i-1} G_{0,k} + \sum_{k=1}^{s} \gamma_{k}^{i-1} [\cos((i-1) \varphi_{k}) G_{1,k} + \sin((i-1) \varphi_{k}) G_{2,k}]$ by model \eqref{model Dt SGARCH(r,s)}. 
		And $\bm{\delta} = (\underline{\bm{\omega}}^{\prime}, \bm{\lambda}^{\prime}, \bm{\gamma}^{\prime}, \bm{\varphi}^{\prime}, \bm{g}_{0}^{\prime}, \bm{g}_{1}^{\prime}, \bm{g}_{2}^{\prime})^{\prime}$, where 
		$\underline{\bm{\omega}} = (\underline{\omega}_{1}, \ldots, \underline{\omega}_{m})^{\prime}$, 
		$\bm{\lambda} = (\lambda_{1}, \ldots, \lambda_{r})^{\prime}$, 
		$\bm{\gamma} = (\gamma_{1}, \ldots, \gamma_{s})^{\prime}$, 
		$\bm{\varphi} = (\varphi_{1}, \ldots, \varphi_{s})^{\prime}$, 
		$\bm{g}_{0} = (\bm{g}_{0,1}^{\prime}, \ldots, \bm{g}_{0,r}^{\prime})^{\prime}$, 
		$\bm{g}_{1} = (\bm{g}_{1,1}^{\prime}, \ldots, \bm{g}_{1,s}^{\prime})^{\prime}$ and 
		$\bm{g}_{2} = (\bm{g}_{2,1}^{\prime}, \ldots, \bm{g}_{2,s}^{\prime})^{\prime}$, with 
		$\bm{g}_{0,k} = \ovec(G_{0,k})$, $\bm{g}_{1,k} = \ovec(G_{1,k})$ and $\bm{g}_{2,k} = \ovec(G_{2,k})$. 
		Then the first derivatives of $\ln\bm{h}_{t}(\bm{\delta})$ are: 
		\begin{align} \label{eq first derivative of lnhtunderline}
			&\frac{\partial\ln\bm{h}_{t}(\bm{\delta})}{\partial\underline{\omega}_{\ell}} = \bm{e}_{\ell}; \notag\\
			&\frac{\partial\ln\bm{h}_{t}(\bm{\delta})}{\partial\lambda_{\ell}} = \sum_{i=2}^{\infty} (i-1) \lambda_{\ell}^{i-2} G_{0,\ell} \ln\mathbf{y}^{2}_{t-i}; \notag\\
			&\frac{\partial\ln\bm{h}_{t}(\bm{\delta})}{\partial\gamma_{\ell}} = \sum_{i=2}^{\infty} (i-1) \gamma_{\ell}^{i-2} [\cos((i-1) \varphi_{\ell}) G_{1,\ell} + \sin((i-1) \varphi_{\ell}) G_{2,\ell}] \ln\mathbf{y}^{2}_{t-i}; \notag\\
			&\frac{\partial\ln\bm{h}_{t}(\bm{\delta})}{\partial\varphi_{\ell}} = \sum_{i=2}^{\infty} (i-1) \gamma_{\ell}^{i-1} [-\sin((i-1) \varphi_{\ell}) G_{1,\ell} + \cos((i-1) \varphi_{\ell}) G_{2,\ell}] \ln\mathbf{y}^{2}_{t-i}; \notag\\
			&\frac{\partial\ln\bm{h}_{t}(\bm{\delta})}{\partial\bm{g}_{0,\ell}^{\prime}} = \left(\sum_{i=1}^{\infty} \lambda_{\ell}^{i-1} \ln\mathbf{y}^{2}_{t-i}\right)^{\prime} \otimes I_{m}; \notag\\
			&\frac{\partial\ln\bm{h}_{t}(\bm{\delta})}{\partial\bm{g}_{1,\ell}^{\prime}} = \left[\sum_{i=1}^{\infty} \gamma_{\ell}^{i-1} \cos((i-1) \varphi_{\ell}) \ln\mathbf{y}^{2}_{t-i}\right]^{\prime} \otimes I_{m}; \notag\\
			&\frac{\partial\ln\bm{h}_{t}(\bm{\delta})}{\partial\bm{g}_{2,\ell}^{\prime}} = \left[\sum_{i=1}^{\infty} \gamma_{\ell}^{i-1} \sin((i-1) \varphi_{\ell}) \ln\mathbf{y}^{2}_{t-i}\right]^{\prime} \otimes I_{m}, 
		\end{align}
		where the last three derivatives are obtained by using the fact $\ovec(A B C) = (C^{\prime} \otimes A) \ovec(B)$ for any matrices $A$, $B$ and $C$, and then rewriting $\ln\bm{h}_{t}(\bm{\delta})$ as 
		\begin{small}
		\begin{align*}
			\ln\bm{h}_{t}(\bm{\delta}) 
			=& \sum_{k=1}^{r}\left[\left(\sum_{i=1}^{\infty} \lambda_{k}^{i-1} \ln\mathbf{y}^{2}_{t-i}\right)^{\prime} \otimes I_{m}\right] \bm{g}_{0,k} 
			+ \sum_{k=1}^{s}\left\{\left[\sum_{i=1}^{\infty} \gamma_{k}^{i-1} \cos((i-1) \varphi_{k}) \ln\mathbf{y}^{2}_{t-i}\right]^{\prime} \otimes I_{m}\right\} \bm{g}_{1,k} \\
			&+ \sum_{k=1}^{s}\left\{\left[\sum_{i=1}^{\infty} \gamma_{k}^{i-1} \sin((i-1) \varphi_{k}) \ln\mathbf{y}^{2}_{t-i}\right]^{\prime} \otimes I_{m}\right\} \bm{g}_{2,k}. 
		\end{align*}
		\end{small}
		Furthermore, the second derivatives of $\ln\bm{h}_{t}(\bm{\delta})$ are: 
		\begin{align} \label{eq second derivative of lnhtunderline}
			&\frac{\partial^{2}\ln\bm{h}_{t}(\bm{\delta})}{\partial\lambda_{\ell} \partial\lambda_{\ell}} = \sum_{i=3}^{\infty} (i-1)(i-2) \lambda_{\ell}^{i-3} G_{0,\ell} \ln\mathbf{y}^{2}_{t-i}; \notag\\
			&\frac{\partial^{2}\ln\bm{h}_{t}(\bm{\delta})}{\partial\bm{g}_{0,\ell}^{\prime} \partial\lambda_{\ell}} = \left[\sum_{i=2}^{\infty} (i-1) \lambda_{\ell}^{i-2} \ln\mathbf{y}^{2}_{t-i}\right]^{\prime} \otimes I_{m}; \notag\\
			&\frac{\partial^{2}\ln\bm{h}_{t}(\bm{\delta})}{\partial\gamma_{\ell} \partial\gamma_{\ell}} = \sum_{i=3}^{\infty} (i-1)(i-2) \gamma_{\ell}^{i-3} \left[\cos((i-1) \varphi_{\ell}) G_{1,\ell} + \sin((i-1) \varphi_{\ell}) G_{2,\ell}\right] \ln\mathbf{y}^{2}_{t-i}; \notag\\
			&\frac{\partial^{2}\ln\bm{h}_{t}(\bm{\delta})}{\partial\varphi_{\ell} \partial\gamma_{\ell}} = \sum_{i=2}^{\infty} (i-1)^{2} \gamma_{\ell}^{i-2} \left[- \sin((i-1) \varphi_{\ell}) G_{1,\ell} + \cos((i-1) \varphi_{\ell}) G_{2,\ell}\right] \ln\mathbf{y}^{2}_{t-i}; \notag\\
			&\frac{\partial^{2}\ln\bm{h}_{t}(\bm{\delta})}{\partial\bm{g}_{1,\ell}^{\prime} \partial\gamma_{\ell}} = \left[\sum_{i=2}^{\infty} (i-1) \gamma_{\ell}^{i-2} \cos((i-1) \varphi_{\ell}) \ln\mathbf{y}^{2}_{t-i}\right]^{\prime} \otimes I_{m}; \notag\\
			&\frac{\partial^{2}\ln\bm{h}_{t}(\bm{\delta})}{\partial\bm{g}_{2,\ell}^{\prime} \partial\gamma_{\ell}} = \left[\sum_{i=2}^{\infty} (i-1) \gamma_{\ell}^{i-2} \sin((i-1) \varphi_{\ell}) \ln\mathbf{y}^{2}_{t-i}\right]^{\prime} \otimes I_{m}; \notag\\
			&\frac{\partial^{2}\ln\bm{h}_{t}(\bm{\delta})}{\partial\varphi_{\ell} \partial\varphi_{\ell}} = - \sum_{i=2}^{\infty} (i-1)^{2} \gamma_{\ell}^{i-1} \left[\cos((i-1) \varphi_{\ell}) G_{1,\ell} + \sin((i-1) \varphi_{\ell}) G_{2,\ell}\right] \ln\mathbf{y}^{2}_{t-i}; \notag\\
			&\frac{\partial^{2}\ln\bm{h}_{t}(\bm{\delta})}{\partial\bm{g}_{1,\ell}^{\prime} \partial\varphi_{\ell}} = - \left[\sum_{i=2}^{\infty} (i-1) \gamma_{\ell}^{i-1} \sin((i-1) \varphi_{\ell}) \ln\mathbf{y}^{2}_{t-i}\right]^{\prime} \otimes I_{m}; \notag\\
			&\frac{\partial^{2}\ln\bm{h}_{t}(\bm{\delta})}{\partial\bm{g}_{2,\ell}^{\prime} \partial\varphi_{\ell}} = \left[\sum_{i=2}^{\infty} (i-1) \gamma_{\ell}^{i-1} \cos((i-1) \varphi_{\ell}) \ln\mathbf{y}^{2}_{t-i}\right]^{\prime} \otimes I_{m}. 
		\end{align}
		Additionally, the other second derivatives of $\ln\bm{h}_{t}(\bm{\delta})$ are zeros. 

	\subsection{Derivatives of $D_{t}(\bm{\delta})$}
		Recall that $D_{t}(\bm{\delta}) = \Diag\{h_{11,t}^{1/2}(\bm{\delta}), \ldots, h_{mm,t}^{1/2}(\bm{\delta})\}$. 
		Let $\delta_{\ell}$ be the $\ell$-th element of $\bm{\delta}$. 
		Then the first derivative of $D_{t}(\bm{\delta})$ is 
		\begin{align} \label{eq first derivative of Dt}
			\frac{\partial D_{t}(\bm{\delta})}{\partial\delta_{\ell}} 
			= \frac{1}{2} D_{t}(\bm{\delta}) \Diag\left\{\frac{\partial\ln\bm{h}_{t}(\bm{\delta})}{\partial\delta_{\ell}}\right\}, 
		\end{align}
		and the second derivative of $D_{t}(\bm{\delta})$ is 
		\begin{align} \label{eq second derivative of Dt}
			\frac{\partial^{2} D_{t}(\bm{\delta})}{\partial\delta_{k} \partial\delta_{\ell}} 
			= \frac{1}{4} D_{t}(\bm{\delta}) \Diag\left\{\frac{\partial\ln\bm{h}_{t}(\bm{\delta})}{\partial\delta_{k}}\right\} \Diag\left\{\frac{\partial\ln\bm{h}_{t}(\bm{\delta})}{\partial\delta_{\ell}}\right\} 
			+ \frac{1}{2} D_{t}(\bm{\delta}) \Diag\left\{\frac{\partial^{2}\ln\bm{h}_{t}(\bm{\delta})}{\partial\delta_{k} \partial\delta_{\ell}}\right\}. 
		\end{align}

	\subsection{Derivatives of $\Psi_{t-1}(\bm{\delta})$}
		Recall that $\Psi_{t-1}(\bm{\delta})$ is the sample correlation matrix of $\{\bm{\varepsilon}_{t-1}(\bm{\delta}), \ldots, \bm{\varepsilon}_{t-\Bbbk}(\bm{\delta})\}$ with $\bm{\varepsilon}_{t}(\bm{\delta}) = (\varepsilon_{1t}(\bm{\delta}), \ldots, \varepsilon_{mt}(\bm{\delta}))^{\prime} = D_{t}^{-1}(\bm{\delta}) \mathbf{y}_{t}$. 
		Denote $\Psi_{t-1}(\bm{\delta}) = [\Psi_{t-1,ij}(\bm{\delta})]$, then we have 
		\begin{align*}
			\Psi_{t-1,ij}(\bm{\delta}) = \frac{\sum_{k=1}^{\Bbbk} \varepsilon_{i,t-k}(\bm{\delta}) \varepsilon_{j,t-k}(\bm{\delta})}{\left[\left(\sum_{k=1}^{\Bbbk} \varepsilon_{i,t-k}^{2}(\bm{\delta})\right) \left(\sum_{k=1}^{\Bbbk} \varepsilon_{j,t-k}^{2}(\bm{\delta})\right)\right]^{1/2}}. 
		\end{align*}
		Let $\delta_{\ell}$ be the $\ell$-th element of $\bm{\delta}$. 
		The first and second derivatives of $\bm{\varepsilon}_{t}(\bm{\delta})$ are 
		\begin{align}
			\frac{\partial\bm{\varepsilon}_{t}(\bm{\delta})}{\partial\delta_{\ell}} 
			&= - \frac{1}{2} \Diag\left\{\frac{\partial\ln\bm{h}_{t}(\bm{\delta})}{\partial\delta_{\ell}}\right\} \bm{\varepsilon}_{t}(\bm{\delta}) \;\; \text{and} \label{eq first derivative of varepsilon}\\
			\frac{\partial^{2}\bm{\varepsilon}_{t}(\bm{\delta})}{\partial\delta_{k} \partial\delta_{\ell}} 
			&= - \frac{1}{2} \Diag\left\{\frac{\partial^{2}\ln\bm{h}_{t}(\bm{\delta})}{\partial\delta_{k} \partial\delta_{\ell}}\right\} \bm{\varepsilon}_{t}(\bm{\delta}) 
			+ \frac{1}{4} \Diag\left\{\frac{\partial\ln\bm{h}_{t}(\bm{\delta})}{\partial\delta_{\ell}}\right\} \Diag\left\{\frac{\partial\ln\bm{h}_{t}(\bm{\delta})}{\partial\delta_{k}}\right\} \bm{\varepsilon}_{t}(\bm{\delta}), \label{eq second derivative of varepsilon}
		\end{align}
		respectively. 
		Denote $\dot{\varepsilon}_{i,t-h}^{(\ell)}(\bm{\delta}) = {\partial\varepsilon_{i,t-h}(\bm{\delta})}/{\partial\delta_{\ell}}$ and $\ddot{\varepsilon}_{i,t-h}^{(k,\ell)}(\bm{\delta}) = {\partial^{2}\varepsilon_{i,t-h}(\bm{\delta})}/{\partial\delta_{k} \partial\delta_{\ell}}$. 
		Then the first derivative of $\Psi_{t-1,ij}(\bm{\delta})$ is 
		\begin{align} \label{eq first derivative of Psitminus1}
			\frac{\partial\Psi_{t-1,ij}(\bm{\delta})}{\partial\delta_{\ell}} 
			=& \frac
			{\sum_{h=1}^{\Bbbk} \left(\dot{\varepsilon}_{i,t-h}^{(\ell)}(\bm{\delta}) \varepsilon_{j,t-h}(\bm{\delta}) + \varepsilon_{i,t-h}(\bm{\delta}) \dot{\varepsilon}_{j,t-h}^{(\ell)}(\bm{\delta})\right)}
			{\left[\left(\sum_{h=1}^{\Bbbk} \varepsilon_{i,t-h}^{2}(\bm{\delta})\right) \left(\sum_{h=1}^{\Bbbk} \varepsilon_{j,t-h}^{2}(\bm{\delta})\right)\right]^{1/2}} \notag\\
			&- \Psi_{t-1,ij}(\bm{\delta}) 
			\left(
			\frac
			{\sum_{h=1}^{\Bbbk} \varepsilon_{i,t-h}(\bm{\delta}) \dot{\varepsilon}_{i,t-h}^{(\ell)}(\bm{\delta})}
			{\sum_{h=1}^{\Bbbk} \varepsilon_{i,t-h}^{2}(\bm{\delta})} 
			+ \frac
			{\sum_{h=1}^{\Bbbk} \varepsilon_{j,t-h}(\bm{\delta}) \dot{\varepsilon}_{j,t-h}^{(\ell)}(\bm{\delta})}
			{\sum_{h=1}^{\Bbbk} \varepsilon_{j,t-h}^{2}(\bm{\delta})}
			\right), 
		\end{align}
		and the second derivative of $\Psi_{t-1,ij}(\bm{\delta})$ is 
		\begin{footnotesize}
		\begin{align} \label{eq second derivative of Psitminus1}
			&\frac{\partial^{2}\Psi_{t-1,ij}(\bm{\delta})}{\partial\delta_{k} \partial\delta_{\ell}} 
			= \frac
			{\sum_{h=1}^{\Bbbk} \left(\ddot{\varepsilon}_{i,t-h}^{(k,\ell)}(\bm{\delta}) \varepsilon_{j,t-h}(\bm{\delta}) + \dot{\varepsilon}_{i,t-h}^{(\ell)}(\bm{\delta}) \dot{\varepsilon}_{j,t-h}^{(k)}(\bm{\delta}) + \dot{\varepsilon}_{i,t-h}^{(k)}(\bm{\delta}) \dot{\varepsilon}_{j,t-h}^{(\ell)}(\bm{\delta}) + \varepsilon_{i,t-h}(\bm{\delta}) \ddot{\varepsilon}_{j,t-h}^{(k,\ell)}(\bm{\delta})\right)}
			{\left[\left(\sum_{h=1}^{\Bbbk} \varepsilon_{i,t-h}^{2}(\bm{\delta})\right) \left(\sum_{h=1}^{\Bbbk} \varepsilon_{j,t-h}^{2}(\bm{\delta})\right)\right]^{1/2}} \notag\\
			-& \frac
			{\sum_{h=1}^{\Bbbk} \left(\dot{\varepsilon}_{i,t-h}^{(\ell)}(\bm{\delta}) \varepsilon_{j,t-h}(\bm{\delta}) + \varepsilon_{i,t-h}(\bm{\delta}) \dot{\varepsilon}_{j,t-h}^{(\ell)}(\bm{\delta})\right)}
			{\left[\left(\sum_{h=1}^{\Bbbk} \varepsilon_{i,t-h}^{2}(\bm{\delta})\right) \left(\sum_{h=1}^{\Bbbk} \varepsilon_{j,t-h}^{2}(\bm{\delta})\right)\right]^{1/2}} 
			\left(\frac
			{\sum_{h=1}^{\Bbbk} \varepsilon_{i,t-h}(\bm{\delta}) \dot{\varepsilon}_{i,t-h}^{(k)}(\bm{\delta})}
			{\sum_{h=1}^{\Bbbk} \varepsilon_{i,t-h}^{2}(\bm{\delta})} 
			+ \frac
			{\sum_{h=1}^{\Bbbk} \varepsilon_{j,t-h}(\bm{\delta}) \dot{\varepsilon}_{j,t-h}^{(k)}(\bm{\delta})}
			{\sum_{h=1}^{\Bbbk} \varepsilon_{j,t-h}^{2}(\bm{\delta})}\right) \notag\\
			-& \frac
			{\partial\Psi_{t-1,ij}(\bm{\delta})}
			{\partial\delta_{k}} 
			\left(\frac
			{\sum_{h=1}^{\Bbbk} \varepsilon_{i,t-h}(\bm{\delta}) \dot{\varepsilon}_{i,t-h}^{(\ell)}(\bm{\delta})}
			{\sum_{h=1}^{\Bbbk} \varepsilon_{i,t-h}^{2}(\bm{\delta})} 
			+ \frac
			{\sum_{h=1}^{\Bbbk} \varepsilon_{j,t-h}(\bm{\delta}) \dot{\varepsilon}_{j,t-h}^{(\ell)}(\bm{\delta})}
			{\sum_{h=1}^{\Bbbk} \varepsilon_{j,t-h}^{2}(\bm{\delta})}\right) \notag\\
			-& \Psi_{t-1,ij}(\bm{\delta}) 
			\left[\frac
			{\sum_{h=1}^{\Bbbk} \left(\dot{\varepsilon}_{i,t-h}^{(k)}(\bm{\delta}) \dot{\varepsilon}_{i,t-h}^{(\ell)}(\bm{\delta}) + \varepsilon_{i,t-h}(\bm{\delta}) \ddot{\varepsilon}_{i,t-h}^{(k,\ell)}(\bm{\delta})\right)}
			{\sum_{h=1}^{\Bbbk} \varepsilon_{i,t-h}^{2}(\bm{\delta})} \right.\notag\\
			&\mathrel{\phantom{\Psi_{t-1,ij}(\bm{\delta})}}\left.+ \frac
			{\sum_{h=1}^{\Bbbk} \left(\dot{\varepsilon}_{j,t-h}^{(k)}(\bm{\delta}) \dot{\varepsilon}_{j,t-h}^{(\ell)}(\bm{\delta}) + \varepsilon_{j,t-h}(\bm{\delta}) \ddot{\varepsilon}_{j,t-h}^{(k,\ell)}(\bm{\delta})\right)}
			{\sum_{h=1}^{\Bbbk} \varepsilon_{j,t-h}^{2}(\bm{\delta})}\right] \notag\\
			+& 2 \Psi_{t-1,ij}(\bm{\delta}) 
			\left[\frac
			{\left(\sum_{h=1}^{\Bbbk} \varepsilon_{i,t-h}(\bm{\delta}) \dot{\varepsilon}_{i,t-h}^{(\ell)}(\bm{\delta})\right) \left(\sum_{h=1}^{\Bbbk} \varepsilon_{i,t-h}(\bm{\delta}) \dot{\varepsilon}_{i,t-h}^{(k)}(\bm{\delta})\right)}
			{\left(\sum_{h=1}^{\Bbbk} \varepsilon_{i,t-h}^{2}(\bm{\delta})\right)^{2}} \right.\notag\\
			&\mathrel{\phantom{2 \Psi_{t-1,ij}(\bm{\delta})}}\left.+ \frac
			{\left(\sum_{h=1}^{\Bbbk} \varepsilon_{j,t-h}(\bm{\delta}) \dot{\varepsilon}_{j,t-h}^{(\ell)}(\bm{\delta})\right) \left(\sum_{h=1}^{\Bbbk} \varepsilon_{j,t-h}(\bm{\delta}) \dot{\varepsilon}_{j,t-h}^{(k)}(\bm{\delta})\right)}
			{\left(\sum_{h=1}^{\Bbbk} \varepsilon_{j,t-h}^{2}(\bm{\delta})\right)^{2}}\right]. 
		\end{align}
		\end{footnotesize}

	\subsection{Derivatives of $R_{t}(\bm{\theta})$}
		Recall that $R_{t}(\bm{\theta}) = (1 - \beta_{1} - \beta_{2}) \underline{R} + \beta_{1} \Psi_{t-1}(\bm{\delta}) + \beta_{2} R_{t-1}(\bm{\theta})$ with $\bm{\theta} = (\bm{\delta}^{\prime}, \bm{\beta}^{\prime})^{\prime}$, $\bm{\beta} = (\beta_{1}, \beta_{2}, \underline{\bm{r}}^{\prime})^{\prime}$ and $\underline{\bm{r}} = \ovechsec(\underline{R})$. 
		Let $\delta_{\ell}$ be the $\ell$-th element of $\bm{\delta}$, and $\underline{R} = [\underline{R}_{ij}]$. 
		Then the first derivatives of $R_{t}(\bm{\theta})$ are: 
		\begin{align} \label{eq first derivative of Rt}
			&\frac{\partial R_{t}(\bm{\theta})}{\partial\delta_{\ell}} = \beta_{1} \frac{\partial\Psi_{t-1}(\bm{\delta})}{\partial\delta_{\ell}} + \beta_{2} \frac{\partial R_{t-1}(\bm{\theta})}{\partial\delta_{\ell}} = \beta_{1} \sum_{h=0}^{\infty} \beta_{2}^{h} \frac{\partial\Psi_{t-h-1}(\bm{\delta})}{\partial\delta_{\ell}}; \notag\\
			&\frac{\partial R_{t}(\bm{\theta})}{\partial\beta_{1}} = - \underline{R} + \Psi_{t-1}(\bm{\delta}) + \beta_{2} \frac{\partial R_{t-1}(\bm{\theta})}{\partial\beta_{1}} = - \frac{1}{1 - \beta_{2}} \underline{R} + \sum_{h=0}^{\infty} \beta_{2}^{h} \Psi_{t-h-1}(\bm{\delta}); \notag\\
			&\frac{\partial R_{t}(\bm{\theta})}{\partial\beta_{2}} = - \underline{R} + R_{t-1}(\bm{\theta}) + \beta_{2} \frac{\partial R_{t-1}(\bm{\theta})}{\partial\beta_{2}} = - \frac{\beta_{1}}{(1 - \beta_{2})^{2}} \underline{R} + \beta_{1} \sum_{h=1}^{\infty} h \beta_{2}^{h-1} \Psi_{t-h-1}(\bm{\delta}); \notag\\
			&\frac{\partial R_{t}(\bm{\theta})}{\partial\underline{R}_{ij}} = (1 - \beta_{1} - \beta_{2}) \dot{\underline{R}}^{(i,j)} + \beta_{2} \frac{\partial R_{t-1}(\bm{\theta})}{\partial\underline{R}_{ij}} = \left(1 - \frac{\beta_{1}}{1 - \beta_{2}}\right) \dot{\underline{R}}^{(i,j)}, 
		\end{align}
		where $\dot{\underline{R}}^{(i,j)}$ is an $m \times m$ matrix with the $(i,j)$-th and $(j,i)$-th elements being ones and others being zeroes. 
		Furthermore, the second derivatives of $R_{t}(\bm{\theta})$ are: 
		\begin{align} \label{eq second derivative of Rt}
			&\frac{\partial^{2} R_{t}(\bm{\theta})}{\partial\delta_{k} \partial\delta_{\ell}} = \beta_{1} \frac{\partial^{2}\Psi_{t-1}(\bm{\delta})}{\partial\delta_{k} \partial\delta_{\ell}} + \beta_{2} \frac{\partial^{2} R_{t-1}(\bm{\theta})}{\partial\delta_{k} \partial\delta_{\ell}} = \beta_{1} \sum_{h=0}^{\infty} \beta_{2}^{h} \frac{\partial^{2}\Psi_{t-h-1}(\bm{\delta})}{\partial\delta_{k} \partial\delta_{\ell}}; \notag\\
			&\frac{\partial^{2} R_{t}(\bm{\theta})}{\partial\beta_{1} \partial\delta_{\ell}} = \frac{\partial\Psi_{t-1}(\bm{\delta})}{\partial\delta_{\ell}} + \beta_{2} \frac{\partial^{2} R_{t-1}(\bm{\theta})}{\partial\beta_{1} \partial\delta_{\ell}} = \sum_{h=0}^{\infty} \beta_{2}^{h} \frac{\partial\Psi_{t-h-1}(\bm{\delta})}{\partial\delta_{\ell}}; \notag\\
			&\frac{\partial^{2} R_{t}(\bm{\theta})}{\partial\beta_{2} \partial\delta_{\ell}} = \frac{\partial R_{t-1}(\bm{\theta})}{\partial\delta_{\ell}} + \beta_{2} \frac{\partial^{2} R_{t-1}(\bm{\theta})}{\partial\beta_{2} \partial\delta_{\ell}} = \beta_{1} \sum_{h=1}^{\infty} h \beta_{2}^{h-1} \frac{\partial\Psi_{t-h-1}(\bm{\delta})}{\partial\delta_{\ell}}; \notag\\
			&\frac{\partial^{2} R_{t}(\bm{\theta})}{\partial\beta_{2} \partial\beta_{1}} = \frac{\partial R_{t-1}(\bm{\theta})}{\partial\beta_{1}} + \beta_{2} \frac{\partial^{2} R_{t-1}(\bm{\theta})}{\partial\beta_{2} \partial\beta_{1}} = - \frac{1}{(1 - \beta_{2})^{2}} \underline{R} + \sum_{h=1}^{\infty} h \beta_{2}^{h-1} \Psi_{t-h-1}(\bm{\delta}); \notag\\
			&\frac{\partial^{2} R_{t}(\bm{\theta})}{\partial\underline{R}_{ij} \partial\beta_{1}} = - \dot{\underline{R}}^{(i,j)} + \beta_{2} \frac{\partial^{2} R_{t-1}(\bm{\theta})}{\partial\underline{R}_{ij} \partial\beta_{1}} = - \frac{1}{1 - \beta_{2}} \dot{\underline{R}}^{(i,j)}; \notag\\
			&\frac{\partial^{2} R_{t}(\bm{\theta})}{\partial\beta_{2} \partial\beta_{2}} = 2 \frac{\partial R_{t-1}(\bm{\theta})}{\partial\beta_{2}} + \beta_{2} \frac{\partial^{2} R_{t-1}(\bm{\theta})}{\partial\beta_{2} \partial\beta_{2}} = - \frac{2 \beta_{1}}{(1 - \beta_{2})^{3}} \underline{R} + \beta_{1} \sum_{h=2}^{\infty} h(h-1) \beta_{2}^{h-2} \Psi_{t-h-1}(\bm{\delta}); \notag\\
			&\frac{\partial^{2} R_{t}(\bm{\theta})}{\partial\underline{R}_{ij} \partial\beta_{2}} = - \dot{\underline{R}}^{(i,j)} + \frac{\partial R_{t-1}(\bm{\theta})}{\partial\underline{R}_{ij}} + \beta_{2} \frac{\partial^{2} R_{t-1}(\bm{\theta})}{\partial\underline{R}_{ij} \partial\beta_{2}} = - \frac{\beta_{1}}{(1 - \beta_{2})^{2}} \dot{\underline{R}}^{(i,j)}; \notag\\
			&\frac{\partial^{2} R_{t}(\bm{\theta})}{\partial\underline{R}_{ij} \partial\delta_{\ell}} = 0; \;\; 
			\frac{\partial^{2} R_{t}(\bm{\theta})}{\partial\beta_{1} \partial\beta_{1}} = 0; \;\;
			\frac{\partial^{2} R_{t}(\bm{\theta})}{\partial\underline{R}_{ij} \partial\underline{R}_{ij}} = 0. 
		\end{align}

	\subsection{Derivatives of $H_{t}(\bm{\theta})$}
		Recall that $H_{t}(\bm{\theta}) = D_{t}(\bm{\delta}) R_{t}(\bm{\theta}) D_{t}(\bm{\delta})$ with $\bm{\theta} = (\bm{\delta}^{\prime}, \bm{\beta}^{\prime})^{\prime}$. 
		Let $\delta_{\ell}$ be the $\ell$-th element of $\bm{\delta}$ and $\beta_{\ell}$ be the $\ell$-th element of $\bm{\beta}$. 
		Then the first derivatives of $H_{t}(\bm{\theta})$ are: 
		\begin{align} \label{eq first derivative of Ht}
			&\frac{\partial H_{t}(\bm{\theta})}{\partial\delta_{\ell}} = 
			\frac{\partial D_{t}(\bm{\delta})}{\partial\delta_{\ell}} R_{t}(\bm{\theta}) D_{t}(\bm{\delta})
			+ D_{t}(\bm{\delta}) \frac{\partial R_{t}(\bm{\theta})}{\partial\delta_{\ell}} D_{t}(\bm{\delta}) 
			+ D_{t}(\bm{\delta}) R_{t}(\bm{\theta}) \frac{\partial D_{t}(\bm{\delta})}{\partial\delta_{\ell}}; \notag\\
			&\frac{\partial H_{t}(\bm{\theta})}{\partial\beta_{\ell}} = 
			D_{t}(\bm{\delta}) \frac{\partial R_{t}(\bm{\theta})}{\partial\beta_{\ell}} D_{t}(\bm{\delta}). 
		\end{align}
		Furthermore, the second derivatives of $H_{t}(\bm{\theta})$ are: 
		\begin{align} \label{eq second derivative of Ht}
			\frac{\partial^{2} H_{t}(\bm{\theta})}{\partial\delta_{k} \partial\delta_{\ell}} &= 
			\frac{\partial^{2} D_{t}(\bm{\delta})}{\partial\delta_{k} \partial\delta_{\ell}} R_{t}(\bm{\theta}) D_{t}(\bm{\delta})
			+ \frac{\partial D_{t}(\bm{\delta})}{\partial\delta_{\ell}} \frac{\partial R_{t}(\bm{\theta})}{\partial\delta_{k}} D_{t}(\bm{\delta}) 
			+ \frac{\partial D_{t}(\bm{\delta})}{\partial\delta_{\ell}} R_{t}(\bm{\theta}) \frac{\partial D_{t}(\bm{\delta})}{\partial\delta_{k}} \notag\\
			&+ \frac{\partial D_{t}(\bm{\delta})}{\partial\delta_{k}} \frac{\partial R_{t}(\bm{\theta})}{\partial\delta_{\ell}} D_{t}(\bm{\delta})
			+ D_{t}(\bm{\delta}) \frac{\partial^{2} R_{t}(\bm{\theta})}{\partial\delta_{k} \partial\delta_{\ell}} D_{t}(\bm{\delta}) 
			+ D_{t}(\bm{\delta}) \frac{\partial R_{t}(\bm{\theta})}{\partial\delta_{\ell}} \frac{\partial D_{t}(\bm{\delta})}{\partial\delta_{k}} \notag\\
			&+ \frac{\partial D_{t}(\bm{\delta})}{\partial\delta_{k}} R_{t}(\bm{\theta}) \frac{\partial D_{t}(\bm{\delta})}{\partial\delta_{\ell}}
			+ D_{t}(\bm{\delta}) \frac{\partial R_{t}(\bm{\theta})}{\partial\delta_{k}} \frac{\partial D_{t}(\bm{\delta})}{\partial\delta_{\ell}} 
			+ D_{t}(\bm{\delta}) R_{t}(\bm{\theta}) \frac{\partial^{2} D_{t}(\bm{\delta})}{\partial\delta_{k} \partial\delta_{\ell}}; \notag\\
			\frac{\partial^{2} H_{t}(\bm{\theta})}{\partial\beta_{k} \partial\delta_{\ell}} &= 
			\frac{\partial D_{t}(\bm{\delta})}{\partial\delta_{\ell}} \frac{\partial R_{t}(\bm{\theta})}{\partial\beta_{k}} D_{t}(\bm{\delta})
			+ D_{t}(\bm{\delta}) \frac{\partial^{2} R_{t}(\bm{\theta})}{\partial\beta_{k} \partial\delta_{\ell}} D_{t}(\bm{\delta}) 
			+ D_{t}(\bm{\delta}) \frac{\partial R_{t}(\bm{\theta})}{\partial\beta_{k}} \frac{\partial D_{t}(\bm{\delta})}{\partial\delta_{\ell}}; \notag\\
			\frac{\partial^{2} H_{t}(\bm{\theta})}{\partial\beta_{k} \partial\beta_{\ell}} &= 
			D_{t}(\bm{\delta}) \frac{\partial^{2} R_{t}(\bm{\theta})}{\partial\beta_{k} \partial\beta_{\ell}} D_{t}(\bm{\delta}). 
		\end{align}

	\subsection{Derivatives of $\ell_{t}(\bm{\theta})$}
		Denote $\ell_{t}(\bm{\theta}) = \frac{1}{2} \mathbf{y}_{t}^{\prime} H_{t}^{-1}(\bm{\theta}) \mathbf{y}_{t} + \frac{1}{2} \ln |H_{t}(\bm{\theta})|$, and let $\theta_{\ell}$ be the $\ell$-th element of $\bm{\theta}$. 
		Then the first derivative of $\ell_{t}(\bm{\theta})$ is 
		\begin{align} \label{eq first derivative of lt}
			\frac{\partial\ell_{t}(\bm{\theta})}{\partial\theta_{\ell}} 
			=& - \frac{1}{2} \mathbf{y}_{t}^{\prime} H_{t}^{-1}(\bm{\theta}) \frac{\partial H_{t}(\bm{\theta})}{\partial\theta_{\ell}} H_{t}^{-1}(\bm{\theta}) \mathbf{y}_{t} + \frac{1}{2} \tr\left(H_{t}^{-1}(\bm{\theta}) \frac{\partial H_{t}(\bm{\theta})}{\partial\theta_{\ell}}\right) \notag\\
			=& \frac{1}{2} \tr\left[\left(I_{m} - H_{t}^{-1}(\bm{\theta}) \mathbf{y}_{t} \mathbf{y}_{t}^{\prime}\right) H_{t}^{-1}(\bm{\theta}) \frac{\partial H_{t}(\bm{\theta})}{\partial\theta_{\ell}}\right], 
		\end{align}
		and the second derivative of $\ell_{t}(\bm{\theta})$ is 
		\begin{align} \label{eq second derivative of lt}
			\frac{\partial^{2}\ell_{t}(\bm{\theta})}{\partial\theta_{k} \partial\theta_{\ell}} 
			=& \frac{1}{2} \mathbf{y}_{t}^{\prime} \left(2 H_{t}^{-1}(\bm{\theta}) \frac{\partial H_{t}(\bm{\theta})}{\partial\theta_{k}} H_{t}^{-1}(\bm{\theta}) \frac{\partial H_{t}(\bm{\theta})}{\partial\theta_{\ell}} H_{t}^{-1}(\bm{\theta}) - H_{t}^{-1}(\bm{\theta}) \frac{\partial^{2} H_{t}(\bm{\theta})}{\partial\theta_{k} \partial\theta_{\ell}} H_{t}^{-1}(\bm{\theta})\right) \mathbf{y}_{t} \notag\\
			&+ \frac{1}{2} \tr\left(- H_{t}^{-1}(\bm{\theta}) \frac{\partial H_{t}(\bm{\theta})}{\partial\theta_{k}} H_{t}^{-1}(\bm{\theta}) \frac{\partial H_{t}(\bm{\theta})}{\partial\theta_{\ell}} + H_{t}^{-1}(\bm{\theta}) \frac{\partial^{2} H_{t}(\bm{\theta})}{\partial\theta_{k} \partial\theta_{\ell}}\right) \notag\\
			=& - \frac{1}{2} \tr\left[\left(I_{m} - 2 H_{t}^{-1}(\bm{\theta}) \mathbf{y}_{t} \mathbf{y}_{t}^{\prime}\right) H_{t}^{-1}(\bm{\theta}) \frac{\partial H_{t}(\bm{\theta})}{\partial\theta_{k}} H_{t}^{-1}(\bm{\theta}) \frac{\partial H_{t}(\bm{\theta})}{\partial\theta_{\ell}}\right] \notag\\
			&+ \frac{1}{2} \tr\left[\left(I_{m} - H_{t}^{-1}(\bm{\theta}) \mathbf{y}_{t} \mathbf{y}_{t}^{\prime}\right) H_{t}^{-1}(\bm{\theta}) \frac{\partial^{2} H_{t}(\bm{\theta})}{\partial\theta_{k} \partial\theta_{\ell}}\right]. 
		\end{align}

	\section{Technical proofs}
		Corresponding to the unknown parameter vector $\bm{\theta} = (\bm{\delta}^{\prime}, \bm{\beta}^{\prime})^{\prime}$ with 
		$\bm{\delta} = (\underline{\bm{\omega}}^{\prime}, \bm{\kappa}^{\prime})^{\prime}$, 
		$\bm{\beta} = (\beta_{1}, \beta_{2}, \underline{\bm{r}}^{\prime})^{\prime}$ and 
		$\underline{\bm{r}} = \ovechsec(\underline{R})$, 
		denote by $\bm{\theta}_{0} = (\bm{\delta}_{0}^{\prime}, \bm{\beta}_{0}^{\prime})^{\prime}$ with 
		$\bm{\delta}_{0} = (\underline{\bm{\omega}}_{0}^{\prime}, \bm{\kappa}_{0}^{\prime})^{\prime}$, 
		$\bm{\beta}_{0} = (\beta_{10}, \beta_{20}, \underline{\bm{r}}_{0}^{\prime})^{\prime}$ and 
		$\underline{\bm{r}}_{0} = \ovechsec(\underline{R}_{0})$ 
		the true value of $\bm{\theta}$. 
		Here 
		$\underline{\bm{\omega}}_{0} = (\underline{\omega}_{10}, \ldots, \underline{\omega}_{m0})^{\prime}$, and 
		$\bm{\kappa}_{0} = (\bm{\lambda}_{0}^{\prime}, \bm{\gamma}_{0}^{\prime}, \bm{\varphi}_{0}^{\prime}, \bm{g}_{00}^{\prime}, \bm{g}_{10}^{\prime}, \bm{g}_{20}^{\prime})^{\prime}$ with 
		$\bm{\lambda}_{0} = (\lambda_{10}, \ldots, \lambda_{r0})^{\prime}$, 
		$\bm{\gamma}_{0} = (\gamma_{10}, \ldots, \gamma_{s0})^{\prime}$, 
		$\bm{\varphi}_{0} = (\varphi_{10}, \ldots, \varphi_{s0})^{\prime}$, 
		$\bm{g}_{0,k0} = \ovec(G_{0,k0})$ for $1 \leq k \leq r$, 
		$\bm{g}_{1,k0} = \ovec(G_{1,k0})$ and $\bm{g}_{2,k0} = \ovec(G_{2,k0})$ for $1 \leq k \leq s$, 
		$\bm{g}_{00} = (\bm{g}_{0,10}^{\prime}, \ldots, \bm{g}_{0,r0}^{\prime})^{\prime}$, and 
		$\bm{g}_{\ell0} = (\bm{g}_{\ell,10}^{\prime}, \ldots, \bm{g}_{\ell,s0}^{\prime})^{\prime}$ for $\ell = 1$ and $2$.

	\subsection{Proof of Theorem \ref{thm Stationarity}}
		By model \eqref{rewritten model Rt SGARCH(r,s)}, we have that 
		\begin{align}
		  \bm{\varepsilon}_{t} = R_{t}^{1/2} \bm{\eta}_{t} 
		  \;\; \text{and} \;\; 
		  R_{t} = \underline{\beta} \underline{R} + \beta_{10} \Psi_{t-1} + \beta_{20} R_{t-1}, \label{eq rewritten model Rt SGARCH(r,s)}
		\end{align}
		where $\{\bm{\eta}_{t}\}$ is a sequence of $i.i.d.$ variables, $\Psi_{t}$ is the sample correlation matrix of $\{\bm{\varepsilon}_{t}, \ldots, \bm{\varepsilon}_{t-\Bbbk+1}\}$, and $\underline{\beta} = 1 - \beta_{10} - \beta_{20}$. 
		Let $\underline{m} = m(m-1)/2$. 
		Define two $(m^{2} + \Bbbk m - 2 m)$-dimensional random vectors and an $(m^{2} + \Bbbk m - 2 m) \times (m^{2} + \Bbbk m - 2 m)$ constant matrix as follows: 
		\begin{small}
		$$
		\bm{x}_{t} = 
		\left(
		\begin{matrix}
		\ovechsec(R_{t}) \\ \ovechsec(\Psi_{t}) \\ \bm{\varepsilon}_{t} \\ \bm{\varepsilon}_{t-1} \\ \vdots \\ \bm{\varepsilon}_{t-\Bbbk+3} \\ \bm{\varepsilon}_{t-\Bbbk+2}
		\end{matrix}
		\right), 
		\bm{\xi}_{t} = 
		\left(
		\begin{matrix}
		\underline{\beta} \ovechsec(\underline{R}_{0}) \\ \ovechsec(\Psi_{t}) \\ \bm{\varepsilon}_{t} \\ \bm{0}_{m} \\ \vdots \\ \bm{0}_{m} \\ \bm{0}_{m} 
		\end{matrix}
		\right), 
		\Upsilon = 
		\left(
		\begin{matrix}
			\beta_{20} I_{\underline{m}} & \beta_{10} I_{\underline{m}} & 0_{\underline{m} \times m} & 0_{\underline{m} \times m} & \cdots & 0_{\underline{m} \times m} & 0_{\underline{m} \times m} \\
			0_{\underline{m}} & 0_{\underline{m}} & 0_{\underline{m} \times m} & 0_{\underline{m} \times m} & \cdots & 0_{\underline{m} \times m} & 0_{\underline{m} \times m} \\
			0_{m \times \underline{m}} & 0_{m \times \underline{m}} & 0_{m} & 0_{m} & \cdots & 0_{m} & 0_{m} \\
			0_{m \times \underline{m}} & 0_{m \times \underline{m}} & I_{m} & 0_{m} & \cdots & 0_{m} & 0_{m} \\
			\vdots & \vdots & \vdots & \vdots & & \vdots & \vdots \\
			0_{m \times \underline{m}} & 0_{m \times \underline{m}} & 0_{m} & 0_{m} & \cdots & 0_{m} & 0_{m} \\
			0_{m \times \underline{m}} & 0_{m \times \underline{m}} & 0_{m} & 0_{m} & \cdots & I_{m} & 0_{m} \\
		\end{matrix}
		\right). 
		$$
		\end{small}
		Then it holds that 
		\begin{align} \label{eq Markov chain}
			\bm{x}_{t} = \Upsilon \bm{x}_{t-1} + \bm{\xi}_{t}. 
		\end{align}

		\begin{lemma} \label{lemma for stationarity}
			The following results hold. 
			\begin{enumerate}[($\romannumeral1$)]
				\item Suppose that $\rho(\Upsilon^{\otimes k}) < 1$ for some $k \geq 1$. Then there exists a vector $\bm{\nu} > 0$ such that $(I_{a} - \Upsilon^{\otimes k})^{\prime} \bm{\nu} > 0$, where $a = (m^{2} + \Bbbk m - 2 m)^{k}$, and $\bm{\nu} > 0$ means that each element of $\bm{\nu}$ is positive. 
				\item Furthermore, suppose that $E\|\bm{\eta}_{t}\|_{2}^{k} < \infty$. Then there exists a compact set $\mathcal{A} = \{\bm{x} \in \mathbb{R}^{a}: \|\bm{x}\|_{1}^{k} \leq \Delta\}$ with $\Delta > 0$, a uniformly bounded function $g_{1}(\bm{x})$ on $\mathcal{A}$, and a constant $c_{0} > 0$, such that the function $g(\bm{x}) = 1 + |\bm{x}^{\otimes k}|^{\prime} \bm{\nu}$ satisfies the following inequalities: 
				\begin{align*}
					&E\left(g(\bm{x}_{t}) \mid \bm{x}_{t-1} = \bm{x}\right) \leq g(\bm{x}) + g_{1}(\bm{x}), \;\; \bm{x} \in \mathbb{R}^{a}, \;\; \text{and} \\
					&E\left(g(\bm{x}_{t}) \mid \bm{x}_{t-1} = \bm{x}\right) \leq (1 - c_{0}) g(\bm{x}), \;\; \bm{x} \in \mathcal{A}^{c} = \mathbb{R}^{a} - \mathcal{A}. 
				\end{align*}
			\end{enumerate}
		\end{lemma}

		\begin{proof}[\textbf{Proof of Lemma \ref{lemma Stationarity of epsilont}}]
			By \eqref{eq rewritten model Rt SGARCH(r,s)}, $\bm{\varepsilon}_{t} = R_{t}^{1/2} \bm{\eta}_{t}$ with $R_{t} = \underline{\beta} \underline{R}_{0} + \beta_{10} \Psi_{t-1} + \beta_{20} R_{t-1}$, $\Psi_{t}$ being the sample correlation matrix of $\{\bm{\varepsilon}_{t}, \ldots, \bm{\varepsilon}_{t-\Bbbk+1}\}$, and $\{\bm{\eta}_{t}\}$ being a sequence of $i.i.d.$ variables. 
			Recall that $\bm{x}_{t} = (\ovechsec(R_{t})^{\prime}, \ovechsec(\Psi_{t})^{\prime}, \bm{\varepsilon}_{t}^{\prime}, \ldots, \bm{\varepsilon}_{t-\Bbbk+2}^{\prime})^{\prime}$. 
			To prove that the process $\{\bm{\varepsilon}_{t}\}$ admits a unique, non-anticipative, strictly stationary and ergodic solution, 
			it suffices to establish the result for $\{[\ovechsec(R_{t})^{\prime}, \bm{\varepsilon}_{t}^{\prime}]^{\prime}\}$, or equivalently to show that the process $\{\bm{x}_{t}\}$ admits a unique, non-anticipative, strictly stationary and ergodic solution. 
			The arguments below adopt the technical tools in \cite{fermanian2017stationarity} and \citet{ling2003asymptotic}. 

			(Non-anticipation and strict stationarity:) 
			Note that 
			($\romannumeral1$) $R_{t}$ is a measurable function of $\Psi_{t-1}$ and $R_{t-1}$, by $R_{t} = \underline{\beta} \underline{R}_{0} + \beta_{10} \Psi_{t-1} + \beta_{20} R_{t-1}$, 
			($\romannumeral2$) $\Psi_{t}$ is a measurable function of $\{\bm{\varepsilon}_{t}, \ldots, \bm{\varepsilon}_{t-\Bbbk+1}\}$, since $\Psi_{t}$ is the sample correlation matrix of $\{\bm{\varepsilon}_{t}, \ldots, \bm{\varepsilon}_{t-\Bbbk+1}\}$, and 
			($\romannumeral3$) $\bm{\varepsilon}_{t}$ is a measurable function of $R_{t}$ and $\bm{\eta}_{t}$, as $\bm{\varepsilon}_{t} = R_{t}^{1/2} \bm{\eta}_{t}$. 
			As a result, $\bm{x}_{t}$ is a measurable function of $\bm{x}_{t-1}$ and $\bm{\eta}_{t}$ only. 
			This together with $\{\bm{\eta}_{t}\}$ being a sequence of $i.i.d.$ variables, implies that $\{\bm{x}_{t}\}$ is a Markov chain with state space $\mathbb{R}^{m^{2} + \Bbbk m - 2 m}$. 
			Moreover, it is straightforward to verify that, $E(g(\bm{x}_{t}) \mid \bm{x}_{t-1} = \bm{x})$ is continuous in $\bm{x}$ for any bounded continuous function $g$ on $\mathbb{R}^{m^{2} + \Bbbk m - 2 m}$, that is, $\{\bm{x}_{t}\}$ is a Feller chain. 
			For the $\Upsilon$ defined before \eqref{eq Markov chain}, it can be verified that $\rho(\Upsilon^{\otimes k}) < 1$ holds for all $k \geq 1$ if $|\beta_{2}| < 1$. 
			Recall that $\bm{x}_{t} = \Upsilon \bm{x}_{t-1} + \bm{\xi}_{t}$ by \eqref{eq Markov chain}. Then under the condition $E\|\bm{\eta}_{t}\|_{2}^{k} < \infty$ for some $k \geq 1$, by Lemma \ref{lemma for stationarity}, there exists a vector $\bm{\nu} > 0$, a compact set $\mathcal{A} = \{\bm{x} \in \mathbb{R}^{a}: \|\bm{x}\|_{1}^{k} \leq \Delta\}$ with $a = k (m^{2} + \Bbbk m - 2 m)$ and $\Delta > 0$, a uniformly bounded function $g_{1}(\bm{x})$ on $\mathcal{A}$, and a constant $c_{0} > 0$, such that the function $g(\bm{x}) = 1 + |\bm{x}^{\otimes k}|^{\prime} \bm{\nu}$ satisfies 
			\begin{align}
				&E\left(g(\bm{x}_{t}) \mid \bm{x}_{t-1} = \bm{x}\right) \leq g(\bm{x}) + g_{1}(\bm{x}), \;\; \bm{x} \in \mathbb{R}^{a}, \;\; \text{and} \label{eq1 in lemma for stationarity}\\
				&E\left(g(\bm{x}_{t}) \mid \bm{x}_{t-1} = \bm{x}\right) \leq (1 - c_{0}) g(\bm{x}), \;\; \bm{x} \in \mathcal{A}^{c} = \mathbb{R}^{a} - \mathcal{A}. \label{eq2 in lemma for stationarity}
			\end{align}
			Note that $g(\bm{x}) \geq 1$. This together with \eqref{eq2 in lemma for stationarity} implies that 
			\begin{align*}
				E\left(g(\bm{x}_{t}) \mid \bm{x}_{t-1} = \bm{x}\right) \leq g(\bm{x}) - c_{0}, \;\; \bm{x} \in \mathcal{A}^{c}. 
			\end{align*}
			Then by (1) in Lemma A.1 of \citet{ling2003asymptotic} or Theorem 2 in \citet{tweedie1988invariant}, there exists a $\sigma$-finite invariant measure $\mu$ for $P$ with $0 < \mu(\mathcal{A}) < \infty$. 
			Furthermore, since $\mathcal{A}$ is a compact set and $g_{1}(\bm{x})$ is uniformly bounded, \eqref{eq1 in lemma for stationarity} implies that when $\bm{x} \in \mathcal{A}$, $E\left(g(\bm{x}_{t}) \mid \bm{x}_{t-1} = \bm{x}\right) \leq c_{1}$ holds for some constant $c_{1} > 0$. 
			Hence we have that 
			\begin{align*}
				\int_{\mathcal{A}} \mu(d \bm{x}) \left[\int_{\mathcal{A}^{c}} P(\bm{x}, d \mathbf{y}) g(\mathbf{y})\right] 
				\leq \int_{\mathcal{A}} \mu(d \bm{x}) E\left(g(\bm{x}_{t}) \mid \bm{x}_{t-1} = \bm{x}\right) \leq c_{1} \mu(\mathcal{A}) < \infty. 
			\end{align*}
			Then by (2) in Lemma A.1 of \citet{ling2003asymptotic} or Theorem 2 in \citet{tweedie1988invariant}, $\{\bm{x}_{t}\}$ has a finite invariant measure $\mu$, and hence $\pi = \mu / \mu(\mathbb{R}^{m^{2} + \Bbbk m - 2 m})$ is an invariant probability measure of $\{\bm{x}_{t}\}$; that is, there exists a strictly stationary solution satisfying \eqref{eq Markov chain}, which is also non-anticipative. 

			(Uniqueness:) 
			Suppose that $\{\bm{x}_{t}\}$ and $\{\widetilde{\bm{x}}_{t}\}$ are strictly stationary solutions to \eqref{eq Markov chain}, where $\bm{x}_{t} = (\ovechsec(R_{t})^{\prime}, \ovechsec(\Psi_{t})^{\prime}, \bm{\varepsilon}_{t}^{\prime}, \ldots, \bm{\varepsilon}_{t-\Bbbk+2}^{\prime})^{\prime}$ and $\widetilde{\bm{x}}_{t} = (\ovechsec(\widetilde{R}_{t})^{\prime}, \ovechsec(\widetilde{\Psi}_{t})^{\prime}, \widetilde{\bm{\varepsilon}}_{t}^{\prime}, \ldots, \widetilde{\bm{\varepsilon}}_{t-\Bbbk+2}^{\prime})^{\prime}$. 
			Note that $\bm{\varepsilon}_{t} = R_{t}^{1/2} \bm{\eta}_{t}$, $\widetilde{\bm{\varepsilon}}_{t} = \widetilde{R}_{t}^{1/2} \bm{\eta}_{t}$, $\Psi_{t}$ is the sample correlation matrix of $\{\bm{\varepsilon}_{t}, \ldots, \bm{\varepsilon}_{t-\Bbbk+1}\}$, and $\widetilde{\Psi}_{t}$ is the sample correlation matrix of $\{\widetilde{\bm{\varepsilon}}_{t}, \ldots, \widetilde{\bm{\varepsilon}}_{t-\Bbbk+1}\}$. 
			Thus $\bm{\varepsilon}_{t} = \widetilde{\bm{\varepsilon}}_{t}$ and $\Psi_{t} = \widetilde{\Psi}_{t}$ a.s. if $R_{t} = \widetilde{R}_{t}$ a.s., then it follows that $\bm{x}_{t} = \widetilde{\bm{x}}_{t}$ a.s..  
			To establish $\bm{x}_{t} = \widetilde{\bm{x}}_{t}$ a.s., below we prove that $R_{t} = \widetilde{R}_{t}$ a.s..  
			Recall that $\bm{\varepsilon}_{t} = (\varepsilon_{1t}, \ldots, \varepsilon_{mt})^{\prime}$ and $\widetilde{\bm{\varepsilon}}_{t} = (\widetilde{\varepsilon}_{1t}, \ldots, \widetilde{\varepsilon}_{mt})^{\prime}$. 
			Let $\bar{\bm{\varepsilon}}_{i,t-1} = (\varepsilon_{i,t-1}, \ldots, \varepsilon_{i,t-\Bbbk})^{\prime}$ and $\widetilde{\bar{\bm{\varepsilon}}}_{i,t-1} = (\widetilde{\varepsilon}_{i,t-1}, \ldots, \widetilde{\varepsilon}_{i,t-\Bbbk})^{\prime}$ for $1 \leq i \leq m$. 
			Denote $\Psi_{t} = [\Psi_{t,ij}]$ and $\widetilde{\Psi}_{t} = [\widetilde{\Psi}_{t,ij}]$. 
			For any $1 \leq i,j \leq m$, it holds that 
			$$
				\Psi_{t-1,ij} 
				= \frac{\sum_{k=1}^{\Bbbk} \varepsilon_{i,t-k} \varepsilon_{j,t-k}}{\left[\left(\sum_{k=1}^{\Bbbk} \varepsilon_{i,t-k}^{2}\right) \left(\sum_{k=1}^{\Bbbk} \varepsilon_{j,t-k}^{2}\right)\right]^{1/2}} 
				= \frac{\bar{\bm{\varepsilon}}_{i,t-1}^{\prime} \bar{\bm{\varepsilon}}_{j,t-1}}{\|\bar{\bm{\varepsilon}}_{i,t-1}\|_{2} \|\bar{\bm{\varepsilon}}_{j,t-1}\|_{2}} 
				= \left(\frac{\bar{\bm{\varepsilon}}_{i,t-1}}{\|\bar{\bm{\varepsilon}}_{i,t-1}\|_{2}}\right)^{\prime} \left(\frac{\bar{\bm{\varepsilon}}_{j,t-1}}{\|\bar{\bm{\varepsilon}}_{j,t-1}\|_{2}}\right). 
			$$
			Similarly, we have 
			$$
				\widetilde{\Psi}_{t-1,ij} 
				= \left(\frac{\widetilde{\bar{\bm{\varepsilon}}}_{i,t-1}}{\|\widetilde{\bar{\bm{\varepsilon}}}_{i,t-1}\|_{2}}\right)^{\prime} \left(\frac{\widetilde{\bar{\bm{\varepsilon}}}_{j,t-1}}{\|\widetilde{\bar{\bm{\varepsilon}}}_{j,t-1}\|_{2}}\right). 
			$$
			Notice that $(\romannumeral1)$ $|\mathbf{a}^{\prime} \mathbf{b} - \mathbf{c}^{\prime} \mathbf{d}| \leq |\mathbf{a}^{\prime} \mathbf{b} - \mathbf{a}^{\prime} \mathbf{d}| + |\mathbf{a}^{\prime} \mathbf{d} - \mathbf{c}^{\prime} \mathbf{d}| \leq \|\mathbf{a}\|_{2} \|\mathbf{b} - \mathbf{d}\|_{2} + \|\mathbf{d}\|_{2} \|\mathbf{a} - \mathbf{c}\|_{2}$ holds for any vectors $\mathbf{a}$, $\mathbf{b}$, $\mathbf{c}$ and $\mathbf{d}$, and $(\romannumeral2)$ $\|(\mathbf{u} / \|\mathbf{u}\|_{2}) - (\mathbf{v} / \|\mathbf{v}\|_{2})\|_{2} \leq 2\|\mathbf{u} - \mathbf{v}\|_{2} / \min\{\|\mathbf{u}\|_{2}, \|\mathbf{v}\|_{2}\}$ holds for any nonzero vectors $\mathbf{u}$ and $\mathbf{v}$. 
			Then under the conditions $\sum_{k=1}^{\Bbbk} \varepsilon_{i, t-\Bbbk}^{2} \geq \underline{c}$ and $\sum_{k=1}^{\Bbbk} \widetilde{\varepsilon}_{i, t-\Bbbk}^{2} \geq \underline{c}$ a.s. for $1 \leq i \leq m$, it can be shown that 
			\begin{align} \label{eq bound of |(Psi)ij - (Psitilde)ij|}
				&\left|\Psi_{t-1,ij} - \widetilde{\Psi}_{t-1,ij}\right| 
				= \left|\left(\frac{\bar{\bm{\varepsilon}}_{i,t-1}}{\|\bar{\bm{\varepsilon}}_{i,t-1}\|_{2}}\right)^{\prime} \left(\frac{\bar{\bm{\varepsilon}}_{j,t-1}}{\|\bar{\bm{\varepsilon}}_{j,t-1}\|_{2}}\right) - \left(\frac{\widetilde{\bar{\bm{\varepsilon}}}_{i,t-1}}{\|\widetilde{\bar{\bm{\varepsilon}}}_{i,t-1}\|_{2}}\right)^{\prime} \left(\frac{\widetilde{\bar{\bm{\varepsilon}}}_{j,t-1}}{\|\widetilde{\bar{\bm{\varepsilon}}}_{j,t-1}\|_{2}}\right)\right| \notag \\
				\leq &\left\|\frac{\bar{\bm{\varepsilon}}_{i,t-1}}{\|\bar{\bm{\varepsilon}}_{i,t-1}\|_{2}} - \frac{\widetilde{\bar{\bm{\varepsilon}}}_{i,t-1}}{\|\widetilde{\bar{\bm{\varepsilon}}}_{i,t-1}\|_{2}}\right\|_{2} + \left\|\frac{\bar{\bm{\varepsilon}}_{j,t-1}}{\|\bar{\bm{\varepsilon}}_{j,t-1}\|_{2}} - \frac{\widetilde{\bar{\bm{\varepsilon}}}_{j,t-1}}{\|\widetilde{\bar{\bm{\varepsilon}}}_{j,t-1}\|_{2}}\right\|_{2} \notag \\
				\leq &\frac{2\|\bar{\bm{\varepsilon}}_{i,t-1} - \widetilde{\bar{\bm{\varepsilon}}}_{i,t-1}\|_{2}}{\min\{\|\bar{\bm{\varepsilon}}_{i,t-1}\|_{2}, \|\widetilde{\bar{\bm{\varepsilon}}}_{i,t-1}\|_{2}\}} + \frac{2\|\bar{\bm{\varepsilon}}_{j,t-1} - \widetilde{\bar{\bm{\varepsilon}}}_{j,t-1}\|_{2}}{\min\{\|\bar{\bm{\varepsilon}}_{j,t-1}\|_{2}, \|\widetilde{\bar{\bm{\varepsilon}}}_{j,t-1}\|_{2}\}} \notag \\
				\leq &\frac{2}{\sqrt{\underline{c}}} \left(\|\bar{\bm{\varepsilon}}_{i,t-1} - \widetilde{\bar{\bm{\varepsilon}}}_{i,t-1}\|_{2} + \|\bar{\bm{\varepsilon}}_{j,t-1} - \widetilde{\bar{\bm{\varepsilon}}}_{j,t-1}\|_{2}\right) \;\; \text{a.s.}. 
			\end{align}
			Moreover, by $\bm{\varepsilon}_{t} = R_{t}^{1/2} \bm{\eta}_{t}$ and $\widetilde{\bm{\varepsilon}}_{t} = \widetilde{R}_{t}^{1/2} \bm{\eta}_{t}$, we have 
			$$
				|\varepsilon_{it} - \widetilde{\varepsilon}_{it}| 
				= \left|\left(R_{t}^{1/2} - \widetilde{R}_{t}^{1/2}\right)_{i \cdot} \bm{\eta}_{t}\right| 
				\leq \left\|\left(R_{t}^{1/2} - \widetilde{R}_{t}^{1/2}\right)_{i \cdot}\right\|_{2} \left\|\bm{\eta}_{t}\right\|_{2} 
				\leq \left\|R_{t}^{1/2} - \widetilde{R}_{t}^{1/2}\right\|_{F} \left\|\bm{\eta}_{t}\right\|_{2}, 
			$$
			and then 
			\begin{align} \label{eq bound of |epsilonbar - epsilonbartilde|}
				\|\bar{\bm{\varepsilon}}_{i,t-1} - \widetilde{\bar{\bm{\varepsilon}}}_{i,t-1}\|_{2} 
				= \left[\sum_{k=1}^{\Bbbk} (\varepsilon_{i,t-k} - \widetilde{\varepsilon}_{i,t-k})^{2}\right]^{1/2} 
				&\leq \left[\sum_{k=1}^{\Bbbk} \left\|R_{t-k}^{1/2} - \widetilde{R}_{t-k}^{1/2}\right\|_{F}^{2} \|\bm{\eta}_{t-k}\|_{2}^{2}\right]^{1/2} \notag \\
				&\leq \sum_{k=1}^{\Bbbk} \left\|R_{t-k}^{1/2} - \widetilde{R}_{t-k}^{1/2}\right\|_{F} \|\bm{\eta}_{t-k}\|_{2}. 
			\end{align}
			By \eqref{eq bound of |(Psi)ij - (Psitilde)ij|}--\eqref{eq bound of |epsilonbar - epsilonbartilde|}, we can obtain that 
			$$
				\left|\Psi_{t-1,ij} - \widetilde{\Psi}_{t-1,ij}\right| 
				\leq \frac{4}{\sqrt{\underline{c}}} \sum_{k=1}^{\Bbbk} \left\|R_{t-k}^{1/2} - \widetilde{R}_{t-k}^{1/2}\right\|_{F} \|\bm{\eta}_{t-k}\|_{2} \;\; \text{a.s.}, 
			$$
			and thus 
			\begin{align*}
				\|\Psi_{t-1} - \widetilde{\Psi}_{t-1}\|_{F} 
				&= \left[\sum_{i=1}^{m} \sum_{j=1}^{m} \left(\Psi_{t-1,ij} - \widetilde{\Psi}_{t-1,ij}\right)^{2}\right]^{1/2} \notag \\
				&\leq \left[m^2 \left(\frac{4}{\sqrt{\underline{c}}} \sum_{k=1}^{\Bbbk} \left\|R_{t-k}^{1/2} - \widetilde{R}_{t-k}^{1/2}\right\|_{F} \|\bm{\eta}_{t-k}\|_{2}\right)^{2}\right]^{1/2} \notag \\
				&= \frac{4m}{\sqrt{\underline{c}}} \sum_{k=1}^{\Bbbk} \left\|R_{t-k}^{1/2} - \widetilde{R}_{t-k}^{1/2}\right\|_{F} \|\bm{\eta}_{t-k}\|_{2} \;\; \text{a.s.}. 
			\end{align*}
			Furthermore, note that (a) $\|R_{t}^{1/2} - \widetilde{R}_{t}^{1/2}\|_{F} \leq (\lambda_{\text{min}}(R_{t})^{1/2} + \lambda_{\text{min}}(\widetilde{R}_{t})^{1/2})^{-1} \|R_{t} - \widetilde{R}_{t}\|_{F}$ by Theorem 6.2 of \cite{higham2008functions}, since the Frobenius norm is unitarily invariant; and (b) $\lambda_{\text{min}}(R_{t}) \geq \underline{\beta} \lambda_{\text{min}}(\underline{R}_{0})$ and $\lambda_{\text{min}}(\widetilde{R}_{t}) \geq \underline{\beta} \lambda_{\text{min}}(\underline{R}_{0})$ by Corollary 4.3.12 of \cite{horn2012matrix}, since $R_{t} = \underline{\beta} \underline{R}_{0} + \beta_{10} \Psi_{t-1} + \beta_{20} R_{t-1}$ and $\widetilde{R}_{t} = \underline{\beta} \underline{R}_{0} + \beta_{10} \widetilde{\Psi}_{t-1} + \beta_{20} \widetilde{R}_{t-1}$ with $\Psi_{t-1}$, $\widetilde{\Psi}_{t-1}$, $R_{t-1}$ and $\widetilde{R}_{t-1}$ being positive definite. 
			It then follows that 
			\begin{align*} 
				\|\Psi_{t-1} - \widetilde{\Psi}_{t-1}\|_{F} 
				&\leq \frac{4m}{\sqrt{\underline{c}}} \sum_{k=1}^{\Bbbk} \frac{1}{\lambda_{\text{min}}(R_{t-k})^{1/2} + \lambda_{\text{min}}(\widetilde{R}_{t-k})^{1/2}} \|R_{t-k} - \widetilde{R}_{t-k}\|_{F} \|\bm{\eta}_{t-k}\|_{2} \\
				&\leq \frac{4m}{\sqrt{\underline{c}}} \sum_{k=1}^{\Bbbk} \frac{1}{2 \left[\underline{\beta} \lambda_{\text{min}}(\underline{R}_{0})\right]^{1/2}} \|R_{t-k} - \widetilde{R}_{t-k}\|_{F} \|\bm{\eta}_{t-k}\|_{2} \\
				&= \frac{2m}{\sqrt{\underline{c} \underline{\beta} \lambda_{\text{min}}(\underline{R}_{0})}} \sum_{k=1}^{\Bbbk} \|R_{t-k} - \widetilde{R}_{t-k}\|_{F} \|\bm{\eta}_{t-k}\|_{2} \;\; \text{a.s.}. 
			\end{align*}
			This together with $R_{t} = \underline{\beta} \underline{R}_{0} + \beta_{10} \Psi_{t-1} + \beta_{20} R_{t-1}$ and $\widetilde{R}_{t} = \underline{\beta} \underline{R}_{0} + \beta_{10} \widetilde{\Psi}_{t-1} + \beta_{20} \widetilde{R}_{t-1}$ implies that the following inequalities hold: 
			\begin{align*}
				\|R_{t} - \widetilde{R}_{t}\|_{F} 
				&= \|\beta_{10} (\Psi_{t-1} - \widetilde{\Psi}_{t-1}) + \beta_{20} (R_{t-1} - \widetilde{R}_{t-1})\|_{F} \\
				&\leq \beta_{10} \|\Psi_{t-1} - \widetilde{\Psi}_{t-1}\|_{F} + \beta_{20} \|R_{t-1} - \widetilde{R}_{t-1}\|_{F} \\
				&\leq \frac{2m \beta_{10}}{\sqrt{\underline{c} \underline{\beta} \lambda_{\text{min}}(\underline{R}_{0})}} \sum_{k=1}^{\Bbbk} \|R_{t-k} - \widetilde{R}_{t-k}\|_{F} \|\bm{\eta}_{t-k}\|_{2} + \beta_{20} \|R_{t-1} - \widetilde{R}_{t-1}\|_{F} \\
				&\equiv \sum_{k=1}^{\Bbbk} \varpi_{t,k} \|R_{t-k} - \widetilde{R}_{t-k}\|_{F} \;\; \text{a.s.}, 
			\end{align*}
			where $\varpi_{t,1} = \frac{2m \beta_{10}}{\sqrt{\underline{c} \underline{\beta} \lambda_{\text{min}}(\underline{R}_{0})}} \|\bm{\eta}_{t-1}\|_{2} + \beta_{20}$ and $\varpi_{t,k} = \frac{2m \beta_{10}}{\sqrt{\underline{c} \underline{\beta} \lambda_{\text{min}}(\underline{R}_{0})}} \|\bm{\eta}_{t-k}\|_{2}$ for $2 \leq k \leq \Bbbk$. 
			Denote $\Delta r_{t} = \|R_{t} - \widetilde{R}_{t}\|_{F}$ and $\Delta \mathbf{r}_{t} = (\Delta r_{t}, \ldots, \Delta r_{t-\Bbbk+1})^{\prime}$. 
			Define the $\Bbbk \times \Bbbk$ matrix 
			\begin{align*} 
			  W_{t} = 
			  \left(
			  \begin{matrix}
			    \varpi_{t,1} & \varpi_{t,2} & \cdots & \varpi_{t,\Bbbk-1} & \varpi_{t,\Bbbk} \\
			    1 & 0 & \cdots & 0 & 0 \\   
			    0 & 1 & \cdots & 0 & 0 \\   
			    \vdots & \vdots & \ddots & \vdots & \vdots \\
			    0 & 0 & \cdots & 1 & 0 \\
			  \end{matrix}
			  \right). 
			\end{align*}
			Then for any positive integer $p$, it holds that 
			$$
				0 \leq \Delta \mathbf{r}_{t} \leq W_{t} \Delta \mathbf{r}_{t-1} 
				\leq \cdots \leq W_{t} W_{t-1} \cdots W_{t-p} \Delta \mathbf{r}_{t-p-1}. 
			$$
			Additionally, we have $\|W_{t} W_{t-1} \cdots W_{t-p}\|_{F}$ tends to zero a.s. when $p \to +\infty$ and for any fixed $t$ under the conditions that $\{\bm{\eta}_{t}\}$ is a sequence of $i.i.d.$ variables, $E(\ln^{+}\|W_{t}\|) < \infty$ holds for some matrix norm $\|\cdot\|$, and the top Lyapunov exponent $\varsigma \equiv \lim_{t \to +\infty} t^{-1} E[\ln(\|W_{1}W_{2} \cdots W_{t}\|)]$ is strictly negative; see also Theorem 2.3 of \cite{francq2019garch} and the proof of Theorem 3 of \cite{fermanian2017stationarity}. 
			Because $\{\Delta \mathbf{r}_{t}\}$ can be initialized arbitrarily far in the past, we obtain that $R_{t} = \widetilde{R}_{t}$ a.s.. 
			It then follows that $\bm{\varepsilon}_{t} = \widetilde{\bm{\varepsilon}}_{t}$ and $\Psi_{t} = \widetilde{\Psi}_{t}$ a.s., since $\bm{\varepsilon}_{t} = R_{t}^{1/2} \bm{\eta}_{t}$, $\widetilde{\bm{\varepsilon}}_{t} = \widetilde{R}_{t}^{1/2} \bm{\eta}_{t}$, $\Psi_{t}$ is the sample correlation matrix of $\{\bm{\varepsilon}_{t}, \ldots, \bm{\varepsilon}_{t-\Bbbk+1}\}$, and $\widetilde{\Psi}_{t}$ is the sample correlation matrix of $\{\widetilde{\bm{\varepsilon}}_{t}, \ldots, \widetilde{\bm{\varepsilon}}_{t-\Bbbk+1}\}$. 
			As a result, $\bm{x}_{t} = \widetilde{\bm{x}}_{t}$ a.s., establishing the uniqueness. 

			(Ergodicity:) 
			The ergodicity of the (now unique) solution $\{\bm{x}_{t}\}$ is a consequence of Corollary 7.17 in \citet{douc2014nonlinear}. 
			As a result, the process $\{\bm{x}_{t}\}$ admits a unique, non-anticipative, strictly stationary and ergodic solution. 
			This completes the proof of Lemma \ref{lemma Stationarity of epsilont}. 
		\end{proof}

		\begin{proof}[\textbf{Proof of Theorem \ref{thm Stationarity}}]
			By Lemma \ref{lemma Stationarity of epsilont}, the process $\{[\ovechsec(R_{t})^{\prime}, \bm{\varepsilon}_{t}^{\prime}]^{\prime}\}$ admits a unique, non-anticipative, strictly stationary and ergodic solution. 
			Based on this result, we next show that the process $\{[\ovech(H_{t})^{\prime}, \mathbf{y}_{t}^{\prime}, \ovechsec(R_{t})^{\prime}, \bm{\varepsilon}_{t}^{\prime}]^{\prime}\}$ defined in models \eqref{model Rt SGARCH(r,s)}--\eqref{model Phii in Dt} admits a unique, non-anticipative, strictly stationary and ergodic solution, which implies that the process $\{\mathbf{y}_{t}\}$ admits a unique, non-anticipative, strictly stationary and ergodic solution. 

			(Non-anticipation, strict stationarity and ergodicity:) 
			Recall that $\bm{\varepsilon}_{t} = (\varepsilon_{1t}, \ldots, \varepsilon_{mt})^{\prime} = D_{t}^{-1} \mathbf{y}_{t}$ with $\mathbf{y}_{t} = (y_{1t}, \ldots, y_{mt})^{\prime}$ and $D_{t} = \Diag\{\sqrt{h_{11,t}}, \ldots, \sqrt{h_{mm,t}}\}$. 
			It holds that $\varepsilon_{it} = y_{it} / \sqrt{h_{ii,t}}$ for $1 \leq i \leq m$ and 
			\begin{align} \label{eq lnytunderline by lnhtunderline and lnvarepsilontunderline}
				\ln\mathbf{y}^{2}_{t} 
				&= (\ln y_{1t}^{2}, \ldots, \ln y_{mt}^{2})^{\prime} = \left(\ln (h_{11,t} \varepsilon_{1t}^{2}), \ldots, \ln (h_{mm,t} \varepsilon_{mt}^{2})\right)^{\prime} \notag\\
				&= \left(\ln h_{11,t}, \ldots, \ln h_{mm,t}\right)^{\prime} + \left(\ln \varepsilon_{1t}^{2}, \ldots, \ln \varepsilon_{mt}^{2}\right)^{\prime} = \ln\bm{h}_{t} + \ln\bm{\varepsilon}^{2}_{t}, 
			\end{align}
			where $\ln\bm{\varepsilon}^{2}_{t} = (\ln \varepsilon_{1t}^{2}, \ldots, \ln \varepsilon_{mt}^{2})^{\prime}$. 
			This together with model \eqref{model Dt SGARCH(r,s)} implies that 
			\begin{align*}
				\ln\bm{h}_{t} 
				=& \underline{\bm{\omega}} + \sum_{i=1}^{\infty} \Phi_{i0} \ln\mathbf{y}^{2}_{t-i} 
				= \underline{\bm{\omega}} + \sum_{i=1}^{\infty} \Phi_{i0} \ln\bm{h}_{t-i} + \sum_{i=1}^{\infty} \Phi_{i0} \ln\bm{\varepsilon}^{2}_{t-i} \\
				=& \left(I_{m} + \sum_{k=1}^{\infty} \sum_{i_{1}, \ldots, i_{k} \geq 1} \Phi_{i_{1} 0} \cdots \Phi_{i_{k} 0}\right) \underline{\bm{\omega}} 
				+ \sum_{k=1}^{\infty} \sum_{i_{1}, \ldots, i_{k} \geq 1} \Phi_{i_{1} 0} \cdots \Phi_{i_{k} 0} \ln\bm{\varepsilon}^{2}_{t-i_{1}-\cdots-i_{k}}. 
			\end{align*}
			Let $\|\cdot\|$ be any matrix norm induced by a vector norm. 
			Then using the properties of the induced matrix norm that $\|A + B\| \leq \|A\| + \|B\|$, $\|A \bm{a}\| \leq \|A\| \|\bm{a}\|$ and $\|A B\| \leq \|A\| \|B\|$ for any vector $\bm{a}$ and matrices $A$ and $B$, it can be shown that 
			\begin{align} \label{eq ||ln htunderline|| inequality}
				\|\ln\bm{h}_{t}\| 
				&\leq 
				\|\underline{\bm{\omega}}\| 
				+ \sum_{k=1}^{\infty} \sum_{i_{1}, \ldots, i_{k} \geq 1} \left\|\Phi_{i_{1} 0} \cdots \Phi_{i_{k} 0} \underline{\bm{\omega}}\right\| 
				+ \sum_{k=1}^{\infty} \sum_{i_{1}, \ldots, i_{k} \geq 1} \left\|\Phi_{i_{1} 0} \cdots \Phi_{i_{k} 0} \ln\bm{\varepsilon}^{2}_{t-i_{1}-\cdots-i_{k}}\right\| \notag\\
				&\leq 
				\|\underline{\bm{\omega}}\| 
				+ \sum_{k=1}^{\infty} \sum_{i_{1}, \ldots, i_{k} \geq 1} \left\|\Phi_{i_{1} 0} \cdots \Phi_{i_{k} 0}\right\| \|\underline{\bm{\omega}}\| 
				+ \sum_{k=1}^{\infty} \sum_{i_{1}, \ldots, i_{k} \geq 1} \left\|\Phi_{i_{1} 0} \cdots \Phi_{i_{k} 0}\right\| \|\ln\bm{\varepsilon}^{2}_{t-i_{1}-\cdots-i_{k}}\| \notag\\
				&\leq 
				\|\underline{\bm{\omega}}\| 
				+ \sum_{k=1}^{\infty} \sum_{i_{1}, \ldots, i_{k} \geq 1} \left\|\Phi_{i_{1} 0}\right\| \cdots \left\|\Phi_{i_{k} 0}\right\| \left(\|\underline{\bm{\omega}}\| + \|\ln\bm{\varepsilon}^{2}_{t-i_{1}-\cdots-i_{k}}\|\right). 
			\end{align}
			Note that under the conditions in Theorem \ref{thm Stationarity}, we have that 
			\begin{align} \label{eq sum ||Phii|| inequality}
				\sum_{i=1}^{\infty} \|\Phi_{i0}\| 
				&= \sum_{i=1}^{\infty} \left\|\sum_{k=1}^{r} \lambda_{k0}^{i-1} G_{0,k0} 
				+ \sum_{k=1}^{s} \gamma_{k0}^{i-1} \left[\cos((i-1) \varphi_{k0}) G_{1,k0} + \sin((i-1) \varphi_{k0}) G_{2,k0}\right]\right\| \notag\\
				&\leq \sum_{i=1}^{\infty} \left[\sum_{k=1}^{r} |\lambda_{k0}|^{i-1} \|G_{0,k0}\| + \sum_{k=1}^{s} |\gamma_{k0}|^{i-1} \left(\|G_{1,k0}\| + \|G_{2,k0}\|\right)\right] \notag\\
				&= \sum_{k=1}^{r} \left(\sum_{i=1}^{\infty} |\lambda_{k0}|^{i-1}\right) \|G_{0,k0}\| 
				+ \sum_{k=1}^{s} \left(\sum_{i=1}^{\infty} |\gamma_{k0}|^{i-1}\right) \left(\|G_{1,k0}\| + \|G_{2,k0}\|\right) \notag\\
				&= \sum_{k=1}^{r} \frac{1}{1 - |\lambda_{k0}|} \|G_{0,k0}\| 
				+ \sum_{k=1}^{s} \frac{1}{1 - |\gamma_{k0}|} \left(\|G_{1,k0}\| + \|G_{2,k0}\|\right) 
				< 1. 
			\end{align}
			Thus by \eqref{eq ||ln htunderline|| inequality}--\eqref{eq sum ||Phii|| inequality}, $E\|\ln\bm{\varepsilon}^{2}_{t}\| < \infty$ and the stationarity of $\{\bm{\varepsilon}_{t}\}$, it can be obtained that 
			\begin{align} \label{eq E||lnhtunderline|| is finite}
				E\|\ln\bm{h}_{t}\| 
				&\leq \|\underline{\bm{\omega}}\| 
				+ \sum_{k=1}^{\infty} \sum_{i_{1}, \ldots, i_{k} \geq 1} \left\|\Phi_{i_{1} 0}\right\| \cdots \left\|\Phi_{i_{k} 0}\right\| E\left(\|\underline{\bm{\omega}}\| + \|\ln\bm{\varepsilon}^{2}_{t}\|\right) \notag\\
				&= \|\underline{\bm{\omega}}\| 
				+ \sum_{k=1}^{\infty} \left(\sum_{i=1}^{\infty} \|\Phi_{i0}\|\right)^{k} E\left(\|\underline{\bm{\omega}}\| + \|\ln\bm{\varepsilon}^{2}_{t}\|\right) 
				< \infty, 
			\end{align}
			which implies that $\ln\bm{h}_{t}$ is finite a.s.. 
			Then $\ln\bm{h}_{t}$ is a measurable function of the process $\{\bm{\varepsilon_{t}}\}$, and $[\ovech(H_{t})^{\prime}, \mathbf{y}_{t}^{\prime}, \ovechsec(R_{t})^{\prime}, \bm{\varepsilon}_{t}^{\prime}]^{\prime}$ is a measurable function of the process $\{[\ovechsec(R_{t})^{\prime}, \bm{\varepsilon}_{t}^{\prime}]^{\prime}\}$. 
			As a result, by Theorem 36.4 of \citet{billingsley1995}, $\{[\ovechsec(R_{t})^{\prime}, \bm{\varepsilon}_{t}^{\prime}]^{\prime}\}$ admitting a non-anticipative, strictly stationary and ergodic solution implies that $\{[\ovech(H_{t})^{\prime}, \mathbf{y}_{t}^{\prime}, \ovechsec(R_{t})^{\prime}, \bm{\varepsilon}_{t}^{\prime}]^{\prime}\}$ admits a non-anticipative, strictly stationary and ergodic solution. 

			(Uniqueness)
			Recall that $\ln\mathbf{y}^{2}_{t} = (\ln y_{1t}^{2}, \ldots, \ln y_{mt}^{2})^{\prime} = \ln\bm{h}_{t} + \ln\bm{\varepsilon}^{2}_{t}$ by \eqref{eq lnytunderline by lnhtunderline and lnvarepsilontunderline} and $\diag(D_{t}) = (\sqrt{h_{11,t}}, \ldots, \sqrt{h_{mm,t}})^{\prime}$. 
			Given the uniqueness of strictly stationary solution $\{[\ovechsec(R_{t})^{\prime}, \bm{\varepsilon}_{t}^{\prime}]^{\prime}\}$, it suffices to verify the uniqueness of the strictly stationary solution $\{\ln\bm{h}_{t}\}$. 
			Suppose that $\{\ln\bm{h}_{t}\}$ is a strictly stationary solution to model \eqref{model Dt SGARCH(r,s)}. 
			For any integer $h > 0$, by \eqref{eq lnytunderline by lnhtunderline and lnvarepsilontunderline} and successively substituting $\ln\bm{h}_{t-j}$'s $h$ times, we have that 
			$$
				\ln\bm{h}_{t} = 
				\left(I_{m} + \sum_{k=1}^{h} \sum_{j_{1}, \ldots, j_{k} \geq 1} \Phi_{j_{1} 0} \cdots \Phi_{j_{k} 0}\right) \underline{\bm{\omega}} 
				+ \sum_{k=1}^{h+1} \sum_{j_{1}, \ldots, j_{k} \geq 1} \Phi_{j_{1} 0} \cdots \Phi_{j_{k} 0} \ln\bm{\varepsilon}^{2}_{t-j_{1}-\cdots-j_{k}} 
				+ \bm{r}_{t,h},
			$$
			where $\bm{r}_{t,h} = \sum_{j_{1}, \ldots, j_{h+1} \geq 1} \Phi_{j_{1} 0} \cdots \Phi_{j_{h+1} 0} \ln\bm{h}_{t-j_{1}-\cdots-j_{h+1}}$. 
			Then we only need to show that $\bm{r}_{t,h}$ tends to zero almost surely as $h$ tends to infinity. 
			By \eqref{eq sum ||Phii|| inequality}--\eqref{eq E||lnhtunderline|| is finite}, the aforementioned properties of $\|\cdot\|$ and the stationarity of $\{\ln\bm{h}_{t}\}$, it can be shown that 
			$$
				E\|\bm{r}_{t,h}\| 
				\leq \sum_{j_{1}, \ldots, j_{h+1} \geq 1} \|\Phi_{j_{1} 0}\| \cdots \|\Phi_{j_{h+1} 0}\| E\|\ln\bm{h}_{t-j_{1}-\cdots-j_{h+1}}\| 
				= \left(\sum_{j=1}^{\infty} \|\Phi_{j0}\|\right)^{h+1} E\|\ln\bm{h}_{t}\|, 
			$$
			and then
			\begin{align*}
				\sum_{h=1}^{\infty} E\|\bm{r}_{t,h}\| \leq \sum_{h=1}^{\infty} \left(\sum_{j=1}^{\infty} \|\Phi_{j0}\|\right)^{h+1} E\|\ln\bm{h}_{t}\| < \infty. 
			\end{align*}
			Since $\sum_{n=1}^{\infty} E|X_{n} - X|^{r} < \infty$ for some $r > 0$ implies that $X_{n} \to X$ a.s., we can conclude that $\|\bm{r}_{t,h}\| \to 0$ a.s., which implies that $\bm{r}_{t,h} \to 0$ a.s.. 
			It follows that the strictly stationary solution $\{\ln\bm{h}_{t}\}$ is unique, and thus the strictly stationary solution $\{[\ovech(H_{t})^{\prime}, \mathbf{y}_{t}^{\prime}, \ovechsec(R_{t})^{\prime}, \bm{\varepsilon}_{t}^{\prime}]^{\prime}\}$ is unique. 
			This completes the proof of Theorem \ref{thm Stationarity}. 
		\end{proof}

	\subsection{Proof of Proposition \ref{propo Identification}}
		\begin{lemma} \label{lemma for identification}
			The following results hold, where $\mathbb{Z}^{+}$ denotes the set of positive integers. 
			\begin{enumerate}[($\romannumeral1$)]
				\item For any $r \in \mathbb{Z}^{+}$ and nonzero $\lambda_{k} \in (-1, 1)$ for $1 \leq k \leq r$, if $\{\lambda_{k}\}$ are distinct, then $\sum_{k=1}^{r} c_{k} \lambda_{k}^{j} = 0$ holds for all $j \in \mathbb{Z}^{+}$ if and only if $c_{k} = 0$ for all $1 \leq k \leq r$. 
				\item For any $s \in \mathbb{Z}^{+}$, and $\gamma_{k} \in (0,1)$ and $\varphi_{k} \in (0, \pi)$ for $1 \leq k \leq \max\{s, 2\}$, 
				\begin{enumerate}[(a)]
					\item if $\{\gamma_{k}\}$ are distinct, then $\sum_{k=1}^{s} \gamma_{k}^{j} [c_{k1} \cos(j \varphi_{k}) + c_{k2} \sin(j \varphi_{k})] = 0$ holds for all $j \in \mathbb{Z}^{+}$ if and only if $c_{k1} = c_{k2} = 0$ for all $1 \leq k \leq s$; 
					\item $c_{11} \cos(j \varphi_{1}) + c_{12} \sin(j \varphi_{1}) = c_{21} \cos(j \varphi_{2}) + c_{22} \sin(j \varphi_{2})$ holds for all $j \in \mathbb{Z}^{+}$ if and only if $\varphi_{1} = \varphi_{2}$, $c_{11} = c_{21}$ and $c_{12} = c_{22}$. 
				\end{enumerate}
				\item For any $r \in \mathbb{Z}^{+}$, nonzero $\lambda_{k} \in (-1, 1)$ for $1 \leq k \leq r$, $\gamma \in (0,1)$ and $\varphi \in (0, \pi)$, if $\{\lambda_{k}\}$ are distinct, then $\gamma^{j} [c_{01} \cos(j \varphi) + c_{02} \sin(j \varphi)] = \sum_{k=1}^{r} c_{k} \lambda_{k}^{j}$ holds for all $j \in \mathbb{Z}^{+}$ if and only if $c_{01} = c_{02} = c_{k} = 0$ for all $1 \leq k \leq r$. 
				\item For any $s \in \mathbb{Z}^{+}$, nonzero $\lambda \in (-1, 1)$, and $\gamma_{k} \in (0,1)$ and $\varphi_{k} \in (0, \pi)$ for $1 \leq k \leq s$, if $\{\gamma_{k}\}$ are distinct, then $c_{0} \lambda^{j} = \sum_{k=1}^{s} \gamma_{k}^{j} [c_{k1} \cos(j \varphi_{k}) + c_{k2} \sin(j \varphi_{k})]$ holds for all $j \in \mathbb{Z}^{+}$ if and only if $c_{0} = c_{k1} = c_{k2} = 0$ for all $1 \leq k \leq s$. 
			\end{enumerate}
		\end{lemma}

		\begin{proof}[\textbf{Proof of Proposition \ref{propo Identification}}]
			By Lemma \ref{lemma for identification}, we can obtain that the following equation holds if and only if $r' = r$, $s' = s$, $\lambda_{k} = \lambda_{k0}$, $\gamma_{k} = \gamma_{k0}$, $\varphi_{k} = \varphi_{k0}$, $G_{0k} = G_{0,k0}$, $G_{1k} = G_{1,k0}$ and $G_{2k} = G_{2,k0}$ for all $k$: 
			\begin{align*}
				&\sum_{k=1}^{r'} \lambda_{k}^{i-1} G_{0,k} 
				+ \sum_{k=1}^{s'} \gamma_{k}^{i-1} \left[\cos((i-1) \varphi_{k}) G_{1,k} + \sin((i-1) \varphi_{k}) G_{2,k}\right] \\
				= &\sum_{k=1}^{r} \lambda_{k0}^{i-1} G_{0,k0} 
				+ \sum_{k=1}^{s} \gamma_{k0}^{i-1} \left[\cos((i-1) \varphi_{k0}) G_{1,k0} + \sin((i-1) \varphi_{k0}) G_{2,k0}\right]. 
			\end{align*}
			Thus ($\romannumeral1$) holds. 

			To show that $\bm{\theta}_{0}$ is identifiable, that is to show that if $H_{t}(\bm{\theta}) = H_{t}(\bm{\theta}_{0})$ almost surely (a.s.), then $\bm{\theta} = \bm{\theta}_{0}$. 
			Recall that $H_{t}(\bm{\theta}) = [h_{ij,t}(\bm{\theta})] = D_{t}(\bm{\delta}) R_{t}(\bm{\theta}) D_{t}(\bm{\delta})$, $D_{t}(\bm{\delta}) = [\Diag(H_{t}(\bm{\theta}))]^{1/2}$ and $\ln\bm{h}_{t}(\bm{\delta}) = (\ln h_{11,t}(\bm{\delta}), \ldots, \ln h_{mm,t}(\bm{\delta}))^{\prime}$. 
			Hence $H_{t}(\bm{\theta}) = H_{t}(\bm{\theta}_{0})$ a.s. is equivalent to $\ln\bm{h}_{t}(\bm{\delta}) = \ln\bm{h}_{t}(\bm{\delta}_{0})$ and $R_{t}(\bm{\theta}) = R_{t}(\bm{\theta}_{0})$ a.s.. 
			Moreover, by model \eqref{model Dt SGARCH(r,s)} and the proof of ($\romannumeral1$), we have that $\ln\bm{h}_{t}(\bm{\delta}) = \ln\bm{h}_{t}(\bm{\delta}_{0})$ a.s. implies that $\bm{\delta} = \bm{\delta}_{0}$. 
			Recall that $\bm{\theta} = (\bm{\delta}^{\prime}, \bm{\beta}^{\prime})^{\prime}$, then we are left to verify that $R_{t}(\bm{\theta}) = R_{t}(\bm{\theta}_{0})$ a.s. together with $\bm{\delta} = \bm{\delta}_{0}$ implies that $\bm{\beta} = \bm{\beta}_{0}$. 
			By model \eqref{model Rt SGARCH(r,s)}, it holds that $R_{t}(\bm{\theta}) = [1 - \beta_{1}/(1 - \beta_{2})] \underline{R} + \beta_{1} \sum_{j=0}^{\infty} \beta_{2}^{j} \Psi_{t-j-1}(\bm{\delta})$. 
			Suppose that $R_{t}(\bm{\theta}) = R_{t}(\bm{\theta}_{0})$ a.s. and $\bm{\delta} = \bm{\delta}_{0}$. 
			It follows that 
			$$
				0 = R_{t}(\bm{\theta}) - R_{t}(\bm{\theta}_{0}) 
				= \left[\left(1 - \frac{\beta_{1}}{1 - \beta_{2}}\right) \underline{R} - \left(1 - \frac{\beta_{10}}{1 - \beta_{20}}\right) \underline{R}_{0}\right] + \sum_{j=0}^{\infty} \left(\beta_{1} \beta_{2}^{j} - \beta_{10} \beta_{20}^{j}\right) \Psi_{t-j-1}(\bm{\delta}_{0}) 
			$$
			a.s., which implies that $\beta_{1} = \beta_{10}$, $\beta_{2} = \beta_{20}$ and $\underline{R} = \underline{R}_{0}$, that is $\bm{\beta} = \bm{\beta}_{0}$. 
			As a result, if $H_{t}(\bm{\theta}) = H_{t}(\bm{\theta}_{0})$ a.s., then $\bm{\theta} = \bm{\theta}_{0}$. 
			And thus ($\romannumeral2$) holds. 
			The proof of this proposition is accomplished. 
		\end{proof}

	\subsection{Proof of Theorem \ref{thm Consistency}}
		Recall that $\widehat{\bm{\theta}}_{\text{G}} = \argmin_{\bm{\theta} \in \Theta} \widetilde{L}_{n}(\bm{\theta})$ with $\widetilde{L}_n(\bm{\theta})=\sum_{t=1}^{n} \widetilde{\ell}_{t}(\bm{\theta})$ and $\widetilde{\ell}_{t}(\bm{\theta}) = \frac{1}{2} \mathbf{y}_{t}^{\prime} \widetilde{H}_{t}^{-1}(\bm{\theta}) \mathbf{y}_{t} + \frac{1}{2} \ln |\widetilde{H}_{t}(\bm{\theta})|$. 
		Denote  
		\begin{align} \label{eq mathcalLntilde}
			\widetilde{\mathcal{L}}_{n}(\bm{\theta}) = \frac{1}{n} \sum_{t=1}^{n} \widetilde{\ell}_{t}(\bm{\theta}), 
		\end{align}
		and
		\begin{align} \label{eq mathcalLn}
			\mathcal{L}_{n}(\bm{\theta}) = \frac{1}{n} \sum_{t=1}^{n} \ell_{t}(\bm{\theta}) 
			\;\; \text{with} \;\; 
			\ell_{t}(\bm{\theta}) = \frac{1}{2} \mathbf{y}_{t}^{\prime} H_{t}^{-1}(\bm{\theta}) \mathbf{y}_{t} + \frac{1}{2} \ln |H_{t}(\bm{\theta})|. 
		\end{align}
		Then it holds that 
		$\widehat{\bm{\theta}}_{\text{G}} = \argmin_{\bm{\theta} \in \Theta} \widetilde{\mathcal{L}}_{n}(\bm{\theta})$. 

		\begin{lemma} \label{lemma derivatives of lnhtunderline}
			Suppose that the conditions in Assumptions \ref{assum parameters}($\romannumeral1$)--($\romannumeral2$) hold. 
			Denote the random vector $\bm{\zeta}_{t,\varrho} = 1 + \sum_{i=1}^{\infty} \varrho^{i-1} \|\ln\mathbf{y}^{2}_{t-i}\|$, where $\|\cdot\|$ is any matrix norm induced by a vector norm. 
			Then it holds that $\varrho \bm{\zeta}_{t-1,\varrho} < \bm{\zeta}_{t,\varrho} < \bm{\zeta}_{t,\varrho_{1}}$ for $0 < \varrho < \varrho_{1} < 1$, and there exists a constant $c > 0$ such that the following results hold: 
			\begin{align*}
				(\romannumeral1)& \;\; \sup_{\bm{\theta} \in \Theta} \left\|\ln\bm{h}_{t}(\bm{\delta})\right\| \leq c \bm{\zeta}_{t,\varrho}; ~~~~~~~~~~ 
				(\romannumeral2) \;\; \sup_{\bm{\theta} \in \Theta} \left\|\frac{\partial\ln\bm{h}_{t}(\bm{\delta})}{\partial\delta_{\ell}}\right\| \leq c \bm{\zeta}_{t,\varrho_{1}}; \\
				(\romannumeral3)& \;\; \sup_{\bm{\theta} \in \Theta} \left\|\frac{\partial^{2}\ln\bm{h}_{t}(\bm{\delta})}{\partial\delta_{k} \partial\delta_{\ell}}\right\| \leq c \bm{\zeta}_{t,\varrho_{1}}; ~~~~ 
				(\romannumeral4) \;\; \sup_{\bm{\theta} \in \Theta} \left\|\frac{\partial^{3}\ln\bm{h}_{t}(\bm{\delta})}{\partial\delta_{j} \partial\delta_{k} \partial\delta_{\ell}}\right\| \leq c \bm{\zeta}_{t,\varrho_{1}}, 
			\end{align*}
			where $\delta_{\ell}$ is the $\ell$-th element of $\bm{\delta}$. 
		\end{lemma}

		\begin{lemma} \label{lemma moments of derivatives of lnhtunderline}
			Suppose that the conditions in Lemma \ref{lemma derivatives of lnhtunderline} hold. 
			If $E\|\ln\mathbf{y}^{2}_{t}\| < \infty$ with $\|\cdot\|$ being any matrix norm induced by a vector norm, then the following results hold: 
			\begin{align*}
				(\romannumeral1)& \; E \sup_{\bm{\theta} \in \Theta} \left\|\ln\bm{h}_{t}(\bm{\delta})\right\| < \infty; ~~~~~~~~~
				(\romannumeral2) \; E \sup_{\bm{\theta} \in \Theta} \left\|\frac{\partial\ln\bm{h}_{t}(\bm{\delta})}{\partial\delta_{\ell}}\right\| < \infty; \\
				(\romannumeral3)& \; E \sup_{\bm{\theta} \in \Theta} \left\|\frac{\partial^{2}\ln\bm{h}_{t}(\bm{\delta})}{\partial\delta_{k} \partial\delta_{\ell}}\right\| < \infty; ~~~~
				(\romannumeral4) \; E \sup_{\bm{\theta} \in \Theta} \left\|\frac{\partial^{3}\ln\bm{h}_{t}(\bm{\delta})}{\partial\delta_{j} \partial\delta_{k} \partial\delta_{\ell}}\right\| < \infty, 
			\end{align*}
			where $\delta_{\ell}$ is the $\ell$-th element of $\bm{\delta}$. 
			Furthermore, if $E\|\ln\mathbf{y}^{2}_{t}\|^{2+\epsilon} < \infty$ for some $\epsilon > 0$, then it holds that: 
			\begin{align*}
				(\romannumeral5)& \; E \sup_{\bm{\theta} \in \Theta} \left\|\ln\bm{h}_{t}(\bm{\delta})\right\|^{2+\epsilon} < \infty; ~~~~~~~~~~
				(\romannumeral6) \; E \sup_{\bm{\theta} \in \Theta} \left\|\frac{\partial\ln\bm{h}_{t}(\bm{\delta})}{\partial\delta_{\ell}}\right\|^{2+\epsilon} < \infty; \\
				(\romannumeral7)& \; E \sup_{\bm{\theta} \in \Theta} \left\|\frac{\partial^{2}\ln\bm{h}_{t}(\bm{\delta})}{\partial\delta_{k} \partial\delta_{\ell}}\right\|^{2+\epsilon} < \infty; ~~~~
				(\romannumeral8) \; E \sup_{\bm{\theta} \in \Theta} \left\|\frac{\partial^{3}\ln\bm{h}_{t}(\bm{\delta})}{\partial\delta_{j} \partial\delta_{k} \partial\delta_{\ell}}\right\|^{2+\epsilon} < \infty. 
			\end{align*}
		\end{lemma}

		\begin{lemma} \label{lemma varrho to the power of t}
			Suppose that the conditions in Assumptions \ref{assum stationarity}--\ref{assum parameters} hold. 
			If $E\|\ln\mathbf{y}^{2}_{t}\| < \infty$ with $\|\cdot\|$ being any matrix norm induced by a vector norm, then the following results hold for large enough $t$ and some constants $0 < \rho < 1$ and $c > 0$: 
			\begin{enumerate}[($\romannumeral1$)]
				\item $\sup_{\bm{\theta} \in \Theta} \|\ln\widetilde{\bm{h}}_{t}(\bm{\delta}) - \ln\bm{h}_{t}(\bm{\delta})\| \leq c \rho^{t}$ a.s.; 
				\item $\sup_{\bm{\theta} \in \Theta} \|\widetilde{D}_{t}(\bm{\delta}) - D_{t}(\bm{\delta})\| \leq c \rho^{t}$ a.s., $\sup_{\bm{\theta} \in \Theta} \|\widetilde{D}_{t}^{-1}(\bm{\delta}) - D_{t}^{-1}(\bm{\delta})\| \leq c \rho^{t}$ a.s., and both $\sup_{\bm{\theta} \in \Theta} \|D_{t}(\bm{\delta})\|$ and $\sup_{\bm{\theta} \in \Theta} \|D_{t}^{-1}(\bm{\delta})\|$ are finite a.s.; 
				\item $\sup_{\bm{\theta} \in \Theta} \|\widetilde{R}_{t}(\bm{\theta}) - R_{t}(\bm{\theta})\| \leq c \rho^{t}$ a.s., $\sup_{\bm{\theta} \in \Theta} \|\widetilde{R}_{t}^{-1}(\bm{\theta}) - R_{t}^{-1}(\bm{\theta})\| \leq c \rho^{t}$ a.s., and both $\sup_{\bm{\theta} \in \Theta} \|R_{t}(\bm{\theta})\|$ and $\sup_{\bm{\theta} \in \Theta} \|R_{t}^{-1}(\bm{\theta})\|$ are finite; 
				\item $\sup_{\bm{\theta} \in \Theta} \|\widetilde{H}_{t}(\bm{\theta}) - H_{t}(\bm{\theta})\| \leq c \rho^{t}$ a.s., $\sup_{\bm{\theta} \in \Theta} \|\widetilde{H}_{t}^{-1}(\bm{\theta}) - H_{t}^{-1}(\bm{\theta})\| \leq c \rho^{t}$ a.s., and both $\sup_{\bm{\theta} \in \Theta} \|H_{t}(\bm{\theta})\|$ and $\sup_{\bm{\theta} \in \Theta} \|H_{t}^{-1}(\bm{\theta})\|$ are finite a.s.; 
				\item $\sup_{\bm{\theta} \in \Theta} |\ln |\widetilde{H}_{t}(\bm{\theta})| - \ln |H_{t}(\bm{\theta})|| \leq c \rho^{t}$ a.s.; 
				\item $\sup_{\bm{\theta} \in \Theta} \|{\partial\ln\widetilde{\bm{h}}_{t}(\bm{\delta})}/{\partial\delta_{\ell}} - {\partial\ln\bm{h}_{t}(\bm{\delta})}/{\partial\delta_{\ell}}\| \leq c \rho^{t}$ a.s., and $\sup_{\bm{\theta} \in \Theta} \|{\ln\bm{h}_{t}(\bm{\delta})}/{\partial\delta_{\ell}}\|$ is finite a.s.; 
				\item $\sup_{\bm{\theta} \in \Theta} \|{\partial\widetilde{R}_{t}(\bm{\theta})}/{\partial\theta_{\ell}} - {\partial R_{t}(\bm{\theta})}/{\partial\theta_{\ell}}\| \leq c \rho^{t}$ a.s., and $\sup_{\bm{\theta} \in \Theta} \|{\partial R_{t}(\bm{\theta})}/{\partial\theta_{\ell}}\|$ is finite a.s.. 
			\end{enumerate}
		\end{lemma}

		\begin{lemma} \label{lemma (ii) for consistency}
			Suppose that the conditions in Assumptions \ref{assum stationarity}--\ref{assum parameters} hold. 
			If $E\|\ln\mathbf{y}^{2}_{t}\|_{2} < \infty$, then the following results hold: 
			\begin{enumerate}[($\romannumeral1$)]
				\item $\lim_{n \to \infty} \sup_{\bm{\theta} \in \Theta} |\mathcal{L}_{n}(\bm{\theta}) - \widetilde{\mathcal{L}}_{n}(\bm{\theta})| = 0$, a.s.; 
				\item $\lim_{n \to \infty} \sup_{\bm{\theta} \in \Theta} \|n^{-1/2} \sum_{t=1}^{n} ({\partial\ell_{t}(\bm{\theta})}/{\partial\bm{\theta}} - {\partial\widetilde{\ell}_{t}(\bm{\theta})}/{\partial\bm{\theta}})\|_{2} = 0$, a.s.; 
				\item $\lim_{n \to \infty} \sup_{\bm{\theta} \in \Theta} \|n^{-1} \sum_{t=1}^{n} ({\partial^{2}\ell_{t}(\bm{\theta})}/{\partial\bm{\theta} \partial\bm{\theta}^{\prime}} - {\partial^{2}\widetilde{\ell}_{t}(\bm{\theta})}/{\partial\bm{\theta} \partial\bm{\theta}^{\prime}})\|_{2} = 0$, a.s.. 
			\end{enumerate}
		\end{lemma}

		\begin{lemma} \label{lemma (iii) for consistency}
			Suppose that the conditions in Assumption \ref{assum parameters} hold. 
			If $E\|\ln\mathbf{y}^{2}_{t}\|_{2} < \infty$, then $E|\ell_{t}(\bm{\theta}_{0})| < \infty$, and $E\ell_{t}(\bm{\theta}) > E\ell_{t}(\bm{\theta}_{0})$ holds for any $\bm{\theta} \neq \bm{\theta}_{0}$. 
		\end{lemma}

		\begin{lemma} \label{lemma (iv) for consistency}
			Suppose that the conditions in Lemma \ref{lemma (ii) for consistency} hold. Then for any $\bm{\theta} \neq \bm{\theta}_{0}$, there exists a neighborhood $V(\bm{\theta})$ such that $\liminf_{n \to \infty} \inf_{\bm{\theta}^{*} \in V(\bm{\theta})} \widetilde{\mathcal{L}}_{n}(\bm{\theta}^{*}) \geq E \ell_{1}(\bm{\theta})$ a.s.. 
		\end{lemma}

		\begin{proof}[\textbf{Proof of Theorem \ref{thm Consistency}}]
			By Proposition \ref{propo Identification}($\romannumeral2$) and Lemmas \ref{lemma (ii) for consistency}($\romannumeral1$) and \ref{lemma (iii) for consistency}--\ref{lemma (iv) for consistency}, the following results hold: 
			\begin{enumerate}[($\romannumeral1$)]
				\item $\bm{\theta}$ is identifiable; 
				\item $\lim_{n \to \infty} \sup_{\bm{\theta} \in \Theta} |\mathcal{L}_{n}(\bm{\theta}) - \widetilde{\mathcal{L}}_{n}(\bm{\theta})| = 0$, a.s.; 
				\item $E |\ell_{t}(\bm{\theta}_{0})| < \infty$, and if $\bm{\theta} \neq \bm{\theta}_{0}$ then $E \ell_{t}(\bm{\theta}) > E \ell_{t}(\bm{\theta}_{0})$; 
				\item for any $\bm{\theta} \neq \bm{\theta}_{0}$, there exists a neighborhood $V(\bm{\theta})$ such that $\liminf_{n \to \infty} \inf_{\bm{\theta}^{*} \in V(\bm{\theta})} \widetilde{\mathcal{L}}_{n}(\bm{\theta}^{*}) \geq E \ell_{1}(\bm{\theta})$ a.s.. 
			\end{enumerate}
			Recall that $\widehat{\bm{\theta}}_{\text{G}} = \argmin_{\bm{\theta} \in \Theta} \widetilde{\mathcal{L}}_{n}(\bm{\theta})$. 
			By ($\romannumeral3$) and ($\romannumeral4$), for any $\bm{\theta} \neq \bm{\theta}_{0}$, there exists a neighborhood $V(\bm{\theta})$ such that 
			\begin{align*}
				\liminf_{n \to \infty} \inf_{\bm{\theta}^{*} \in V(\bm{\theta})} \widetilde{\mathcal{L}}_{n}(\bm{\theta}^{*}) > E \ell_{1}(\bm{\theta}_{0}), \;\; \text{a.s.}. 
			\end{align*}
			Moreover, by ($\romannumeral2$) and the ergodic theorem under Assumption \ref{assum stationarity}, it can be shown that 
			\begin{align*}
				\lim_{n \to \infty} \widetilde{\mathcal{L}}_{n}(\bm{\theta}_{0}) = \lim_{n \to \infty} \mathcal{L}_{n}(\bm{\theta}_{0}) = \lim_{n \to \infty} \frac{1}{n} \sum_{t=1}^{n} \ell_{t}(\bm{\theta}_{0}) = E \ell_{1}(\bm{\theta}_{0}), \;\; \text{a.s.}. 
			\end{align*}
			As a result, by a standard compactness argument following the proof of Theorem 3.1 in \citet{francq2012qml}, we complete the proof of the strong consistency of $\widehat{\bm{\theta}}_{\text{G}}$. 
		\end{proof}

	\subsection{Proof of Theorem \ref{thm Asymptotic normality}}
		\begin{lemma} \label{lemma moments of derivatives of Psitminus1}
			Suppose that the conditions in Assumptions \ref{assum parameters}($\romannumeral1$)--($\romannumeral2$) hold. 
			If $E\|\ln\mathbf{y}^{2}_{t}\|^{2+\epsilon} < \infty$ for some $\epsilon > 0$ with $\|\cdot\|$ being any matrix norm induced by a vector norm, then the following results hold: 
			\begin{align*}
				(\romannumeral1)& \;\; E \sup_{\bm{\theta} \in \Theta} \left\|\frac{\partial\Psi_{t-1}(\bm{\delta})}{\partial\delta_{\ell}}\right\|^{2+\epsilon} < \infty; ~~~~~~ 
				(\romannumeral2) \;\; E \sup_{\bm{\theta} \in \Theta} \left\|\frac{\partial^{2}\Psi_{t-1}(\bm{\delta})}{\partial\delta_{k} \partial\delta_{\ell}}\right\|^{1+\epsilon/2} < \infty; \\
				(\romannumeral3)& \;\; E \sup_{\bm{\theta} \in \Theta} \left\|\frac{\partial^{3}\Psi_{t-1}(\bm{\delta})}{\partial\delta_{j} \partial\delta_{k} \partial\delta_{\ell}}\right\|^{1+\epsilon/2} < \infty, 
			\end{align*}
			where $\delta_{\ell}$ is the $\ell$-th element of $\bm{\delta}$. 
		\end{lemma}

		\begin{lemma} \label{lemma moments of derivatives of Rt}
			Suppose that the conditions in Lemma \ref{lemma moments of derivatives of Psitminus1} hold. 
			Then the following results hold: 
			\begin{align*}
				(\romannumeral1)& \;\; E \sup_{\bm{\theta} \in \Theta} \left\|\frac{\partial R_{t}(\bm{\theta})}{\partial\theta_{\ell}}\right\|^{2+\epsilon} < \infty; ~~~~~~~~
				(\romannumeral2) \;\; E \sup_{\bm{\theta} \in \Theta} \left\|\frac{\partial^{2} R_{t}(\bm{\theta})}{\partial\theta_{k} \partial\theta_{\ell}}\right\|^{1+\epsilon/2} < \infty; \\
				(\romannumeral3)& \;\; E \sup_{\bm{\theta} \in \Theta} \left\|\frac{\partial^{3} R_{t}(\bm{\theta})}{\partial\theta_{j} \partial\theta_{k} \partial\theta_{\ell}}\right\|^{1+\epsilon/2} < \infty, 
			\end{align*}
			where $\|\cdot\|$ is any matrix norm induced by a vector norm, and $\theta_{\ell}$ is the $\ell$-th element of $\bm{\theta}$. 
		\end{lemma}

		\begin{lemma} \label{lemma moments of Htinverse product derivatives of Ht}
			Suppose that the conditions in Assumption \ref{assum parameters} hold. 
			If $E\|\ln\mathbf{y}^{2}_{t}\|^{2+\epsilon} < \infty$ for some $\epsilon > 0$ with $\|\cdot\|$ being any matrix norm induced by a vector norm, then the following results hold: 
			\begin{align*}
				(\romannumeral1)& \;\; E \sup_{\bm{\theta} \in \Theta} \left\|D_{t}^{-1}(\bm{\delta}) \frac{\partial D_{t}(\bm{\delta})}{\partial\delta_{\ell}}\right\|^{2+\epsilon} < \infty; ~~~~~~
				E \sup_{\bm{\theta} \in \Theta} \left\|D_{t}^{-1}(\bm{\delta}) \frac{\partial^{2} D_{t}(\bm{\delta})}{\partial\delta_{k} \partial\delta_{\ell}}\right\|^{1+\epsilon/2} < \infty; ~~~ \text{and} \\
				& \;\; E \sup_{\bm{\theta} \in \Theta} \left\|D_{t}^{-1}(\bm{\delta}) \frac{\partial^{3} D_{t}(\bm{\delta})}{\partial\delta_{j} \partial\delta_{k} \partial\delta_{\ell}}\right\|^{1+\epsilon/2} < \infty; \\
				(\romannumeral2)& \;\; E \sup_{\bm{\theta} \in \Theta} \left\|R_{t}^{-1}(\bm{\theta}) \frac{\partial R_{t}(\bm{\theta})}{\partial\theta_{\ell}}\right\|^{2+\epsilon} < \infty; ~~~~~~
				E \sup_{\bm{\theta} \in \Theta} \left\|R_{t}^{-1}(\bm{\theta}) \frac{\partial^{2} R_{t}(\bm{\theta})}{\partial\theta_{k} \partial\theta_{\ell}}\right\|^{1+\epsilon/2} < \infty; ~~~ \text{and} \\
				& \;\; E \sup_{\bm{\theta} \in \Theta} \left\|R_{t}^{-1}(\bm{\theta}) \frac{\partial^{3} R_{t}(\bm{\theta})}{\partial\theta_{j} \partial\theta_{k} \partial\theta_{\ell}}\right\|^{1+\epsilon/2} < \infty; \\
				(\romannumeral3)& \;\; E \sup_{\bm{\theta} \in \Theta} \left\|H_{t}^{-1}(\bm{\theta}) \frac{\partial H_{t}(\bm{\theta})}{\partial\theta_{\ell}}\right\|^{2+\epsilon} < \infty; ~~~~~~
				E \sup_{\bm{\theta} \in \Theta} \left\|H_{t}^{-1}(\bm{\theta}) \frac{\partial^{2} H_{t}(\bm{\theta})}{\partial\theta_{k} \partial\theta_{\ell}}\right\|^{1+\epsilon/2} < \infty; ~~~ \text{and} \\
				& \;\; E \sup_{\bm{\theta} \in \Theta} \left\|H_{t}^{-1}(\bm{\theta}) \frac{\partial^{3} H_{t}(\bm{\theta})}{\partial\theta_{j} \partial\theta_{k} \partial\theta_{\ell}}\right\|^{1+\epsilon/2} < \infty, 
			\end{align*}
			where $\delta_{\ell}$ is the $\ell$-th element of $\bm{\delta}$ and $\theta_{\ell}$ is the $\ell$-th element of $\bm{\theta}$. 
		\end{lemma}

		\begin{lemma} \label{lemma moments of derivatives of lt}
			Suppose that the conditions in Assumption \ref{assum parameters} hold. 
			If $E\|\ln\mathbf{y}^{2}_{t}\|_{2}^{2+\epsilon} < \infty$ for some $\epsilon > 0$ and $E \|\bm{\eta}_{t} \bm{\eta}_{t}^{\prime}\|_{2}^{2} < \infty$, 
			then the following results hold: 
			\begin{enumerate}[($\romannumeral1$)]
				\item $E |{\partial\ell_{t}(\bm{\theta}_{0})}/{\partial\theta_{\ell}}|^{2} < \infty$, and $\sup_{\bm{\theta} \in \Theta} |{\partial\ell_{t}(\bm{\theta})}/{\partial\theta_{\ell}}|^{2}$ is finite a.s.; 
				\item $E |{\partial^{2}\ell_{t}(\bm{\theta}_{0})}/{\partial\theta_{k} \partial\theta_{\ell}}| < \infty$, and $\sup_{\bm{\theta} \in \Theta} |{\partial^{2}\ell_{t}(\bm{\theta})}/{\partial\theta_{k} \partial\theta_{\ell}}|$ is finite a.s.; 
				\item $E |{\partial^{3}\ell_{t}(\bm{\theta}_{0})}/{\partial\theta_{j} \partial\theta_{k} \partial\theta_{\ell}}| < \infty$, and $\sup_{\bm{\theta} \in \Theta} |{\partial^{3}\ell_{t}(\bm{\theta})}/{\partial\theta_{j} \partial\theta_{k} \partial\theta_{\ell}}|$ is finite a.s.; 
			\end{enumerate}
			where $\bm{\eta}_{t} = (\eta_{1t}, \ldots, \eta_{mt})^{\prime}$, and $\theta_{\ell}$ is the $\ell$-th element of $\bm{\theta}$. 
		\end{lemma}

		\begin{lemma} \label{lemma positivity of Sigma*}
			Suppose that the condition in Assumption \ref{assum parameters}($\romannumeral3$) holds. 
			Then $\Sigma_{*} = E({\partial^{2}\ell_{t}(\bm{\theta}_{0})}/{\partial\bm{\theta} \partial\bm{\theta}^{\prime}})$ is positive definite. 
		\end{lemma}

		\begin{proof}[\textbf{Proof of Theorem \ref{thm Asymptotic normality}}]
		Recall that $\widehat{\bm{\theta}}_{\text{G}} = \argmin_{\bm{\theta} \in \Theta} \widetilde{\mathcal{L}}_{n}(\bm{\theta})$, where 
		\begin{align*}
			\widetilde{\mathcal{L}}_{n}(\bm{\theta}) = \frac{1}{n} \sum_{t=1}^{n} \widetilde{\ell}_{t}(\bm{\theta}) 
			\;\; \text{with} \;\; 
			\widetilde{\ell}_{t}(\bm{\theta}) = \frac{1}{2} \mathbf{y}_{t}^{\prime} \widetilde{H}_{t}^{-1}(\bm{\theta}) \mathbf{y}_{t} + \frac{1}{2} \ln |\widetilde{H}_{t}(\bm{\theta})|. 
		\end{align*}
		By Taylor's expansion, we have 
		\begin{align*}
			0 = \frac{1}{\sqrt{n}} \sum_{t=1}^{n} \frac{\partial\widetilde{\ell}_{t}(\widehat{\bm{\theta}}_{\text{G}})}{\partial\bm{\theta}} 
			= \frac{1}{\sqrt{n}} \sum_{t=1}^{n} \frac{\partial\widetilde{\ell}_{t}(\bm{\theta}_{0})}{\partial\bm{\theta}} + \frac{1}{n} \sum_{t=1}^{n} \frac{\partial^{2}\widetilde{\ell}_{t}(\bar{\bm{\theta}}_{n})}{\partial\bm{\theta} \partial\bm{\theta}^{\prime}} \sqrt{n} \left(\widehat{\bm{\theta}}_{\text{G}} - \bm{\theta}_{0}\right), 
		\end{align*}
		where the elements of $\bar{\bm{\theta}}_{n} = (\bar{\bm{\delta}}_{n}^{\prime}, \bar{\bm{\beta}}_{n}^{\prime})^{\prime}$ lie in the segment joining the corresponding elements of $\widehat{\bm{\theta}}_{\text{G}}$ and $\bm{\theta}_{0}$. 
		It then follows that 
		\begin{align*}
			\sqrt{n} \left(\widehat{\bm{\theta}}_{\text{G}} - \bm{\theta}_{0}\right) = - \left(\frac{1}{n} \sum_{t=1}^{n} \frac{\partial^{2}\widetilde{\ell}_{t}(\bar{\bm{\theta}}_{n})}{\partial\bm{\theta} \partial\bm{\theta}^{\prime}}\right)^{-1} \frac{1}{\sqrt{n}} \sum_{t=1}^{n} \frac{\partial\widetilde{\ell}_{t}(\bm{\theta}_{0})}{\partial\bm{\theta}}. 
		\end{align*}
		Hence it suffices to show that 
		\begin{enumerate}[($\romannumeral1$)]
			\item $\frac{1}{\sqrt{n}} \sum_{t=1}^{n} \frac{\partial\widetilde{\ell}_{t}(\bm{\theta}_{0})}{\partial\bm{\theta}} \to_{d} N(\bm{0}, \Sigma)$ with $\Sigma = E\left(\frac{\partial\ell_{t}(\bm{\theta}_{0})}{\partial\bm{\theta}} \frac{\partial\ell_{t}(\bm{\theta}_{0})}{\partial\bm{\theta}^{\prime}}\right)$; 
			\item $\frac{1}{n} \sum_{t=1}^{n} \frac{\partial^{2}\widetilde{\ell}_{t}(\bar{\bm{\theta}}_{n})}{\partial\bm{\theta} \partial\bm{\theta}^{\prime}} \to_{p} \Sigma_{*}$ with $\Sigma_{*} = E\left(\frac{\partial^{2}\ell_{t}(\bm{\theta}_{0})}{\partial\bm{\theta} \partial\bm{\theta}^{\prime}}\right)$; 
			\item $\Sigma_{*}$ is positive definite. 
		\end{enumerate}

		We first show ($\romannumeral1$). 
		Note that under Assumptions \ref{assum stationarity}--\ref{assum parameters}, 
		$\{{\partial\ell_{t}(\bm{\theta}_{0})}/{\partial\bm{\theta}}\}$ is ergodic and strictly stationary, 
		${\partial\ell_{t}(\bm{\theta}_{0})}/{\partial\bm{\theta}}$ is measurable with respect to $\mathcal{F}_{t}$, 
		$E({\partial\ell_{t}(\bm{\theta}_{0})}/{\partial\bm{\theta}} \mid \mathcal{F}_{t-1}) = 0$ by \eqref{eq first derivative of lt}, 
		and $E|{\partial\ell_{t}(\bm{\theta}_{0})}/{\partial\theta_{\ell}}|^{2} < \infty$ for all $1 \leq \ell \leq d$ by Lemma \ref{lemma moments of derivatives of lt}($\romannumeral1$), 
		where $d$ is the dimension of $\bm{\theta}$. 
		Then for any $\bm{c} \in \mathbb{R}^{d}$, the sequence $\{\bm{c}^{\prime} {\partial\ell_{t}(\bm{\theta}_{0})}/{\partial\bm{\theta}}, \mathcal{F}_{t}\}$ is an ergodic, strictly stationary, and square integral martingale difference. 
		Thus by the central limit theorem of \citet{billingsley1961lindeberg}, as $n \to \infty$, it holds that 
		\begin{align*}
			\frac{1}{\sqrt{n}} \sum_{t=1}^{n} \frac{\partial\ell_{t}(\bm{\theta}_{0})}{\partial\bm{\theta}} \to_{d} N(\bm{0}, \Sigma). 
		\end{align*}
		This together with 
		\begin{align*}
			\left\|\frac{1}{\sqrt{n}} \sum_{t=1}^{n} \left(\frac{\partial\ell_{t}(\bm{\theta}_{0})}{\partial\bm{\theta}} - \frac{\partial\widetilde{\ell}_{t}(\bm{\theta}_{0})}{\partial\bm{\theta}}\right)\right\|_{2} \to_{p} 0 
		\end{align*}
		as $n \to \infty$ by Lemma \ref{lemma (ii) for consistency}($\romannumeral2$), implies that ($\romannumeral1$) holds. 

		Next we verify ($\romannumeral2$). 
		By Taylor's expansion, it holds that  
		\begin{align} \label{eq Taylor expansion of second derivative of lt}
			\frac{1}{n} \sum_{t=1}^{n} \frac{\partial^{2}\ell_{t}(\bar{\bm{\theta}}_{n})}{\partial\theta_{k} \partial\theta_{\ell}} 
			= \frac{1}{n} \sum_{t=1}^{n} \frac{\partial^{2}\ell_{t}(\bm{\theta}_{0})}{\partial\theta_{k} \partial\theta_{\ell}} + \frac{1}{n} \sum_{t=1}^{n} \frac{\partial}{\partial \bm{\theta}^{\prime}} \frac{\partial^{2}\ell_{t}(\bar{\bar{\bm{\theta}}}_{n})}{\partial\theta_{k} \partial\theta_{\ell}} \left(\bar{\bm{\theta}}_{n} - \bm{\theta}_{0}\right), 
		\end{align}
		where the elements of $\bar{\bar{\bm{\theta}}}_{n}$ lie in the segment joining the corresponding elements of $\bar{\bm{\theta}}_{n}$ and $\bm{\theta}_{0}$. 
		Then we show that the following results hold as $n \to \infty$: 
		\begin{enumerate}[(a)]
			\item $\frac{1}{n} \sum_{t=1}^{n} \frac{\partial^{2}\ell_{t}(\bm{\theta}_{0})}{\partial\bm{\theta} \partial\bm{\theta}^{\prime}} \to_{p} \Sigma_{*}$; 
			\item $\|\bar{\bm{\theta}}_{n} - \bm{\theta}_{0}\|_{2} \to 0$ a.s.;
			\item there exists a compact neighborhood $V(\bm{\theta}_{0})$ of $\bm{\theta}_{0}$ such that $\sup_{\bm{\theta} \in V(\bm{\theta}_{0})} \left\|\frac{\partial}{\partial \bm{\theta}^{\prime}} \frac{\partial^{2}\ell_{t}(\bm{\theta})}{\partial\theta_{k} \partial\theta_{\ell}}\right\|_{2}$ is finite a.s.; 
			\item $\sup_{\bm{\theta} \in \Theta} \left\|\frac{1}{n} \sum_{t=1}^{n} \left(\frac{\partial^{2}\ell_{t}(\bm{\theta})}{\partial\bm{\theta} \partial\bm{\theta}^{\prime}} - \frac{\partial^{2}\widetilde{\ell}_{t}(\bm{\theta})}{\partial\bm{\theta} \partial\bm{\theta}^{\prime}}\right)\right\|_{2} \to_{p} 0$. 
		\end{enumerate}
		Specifically, (a) holds by the ergodic theorem and $E|{\partial^{2}\ell_{t}(\bm{\theta}_{0})}/{\partial\theta_{k} \partial\theta_{\ell}}| < \infty$ from Lemma \ref{lemma moments of derivatives of lt}($\romannumeral2$), 
		(b) follows from Theorem \ref{thm Consistency}, 
		(c) holds by Lemma \ref{lemma moments of derivatives of lt}($\romannumeral3$), 
		and (d) is obtained by Lemma \ref{lemma (ii) for consistency}($\romannumeral3$). 
		Thus by \eqref{eq Taylor expansion of second derivative of lt} and (a)--(d), we can obtain that ($\romannumeral2$) holds. 

		Lastly, ($\romannumeral3$) holds by Lemma \ref{lemma positivity of Sigma*}. 
		As a result, the proof of the asymptotic normality of $\widehat{\bm{\theta}}_{\text{G}}$ is accomplished by ($\romannumeral1$)--($\romannumeral3$). 
		\end{proof}

	\subsection{Proof of Proposition \ref{propo approximates of asymptotic variances}}
		\begin{proof}[\textbf{Proof of Proposition \ref{propo approximates of asymptotic variances}}]
		Recall that $\Sigma_{\text{G}}=\Sigma_{*}^{-1} \Sigma \Sigma_{*}^{-1}$ with $\Sigma = E({\partial\ell_{t}(\bm{\theta}_{0})}/{\partial\bm{\theta}} {\partial\ell_{t}(\bm{\theta}_{0})}/{\partial\bm{\theta}^{\prime}})$ and $\Sigma_{*} = E({\partial^{2}\ell_{t}(\bm{\theta}_{0})}/{\partial\bm{\theta} \partial\bm{\theta}^{\prime}})$, 
		and $\widehat{\Sigma}_{\text{G}} = \widehat{\Sigma}_{*}^{-1} \widehat{\Sigma} \widehat{\Sigma}_{*}^{-1}$ with $\widehat{\Sigma} = n^{-1} \sum_{t=1}^{n} {\partial\widetilde{\ell}_{t}(\widehat{\bm{\theta}}_{\text{G}})}/{\partial\bm{\theta}} {\partial\widetilde{\ell}_{t}(\widehat{\bm{\theta}}_{\text{G}})}/{\partial\bm{\theta}^{\prime}}$ and $\widehat{\Sigma}_{*} = n^{-1} \sum_{t=1}^{n} {\partial^{2}\widetilde{\ell}_{t}(\widehat{\bm{\theta}}_{\text{G}})}/{\partial\bm{\theta} \partial\bm{\theta}^{\prime}}$. 
		To prove that $\widehat{\Sigma}_{\text{G}} \to_{p} \Sigma_{\text{G}}$ as $n \to \infty$, it suffices to show that the following results hold: 
		\begin{enumerate}[($\romannumeral1$)]
			\item $\widehat{\Sigma}_{*} \to_{p} \Sigma_{*}$ as $n \to \infty$; 
			\item $\widehat{\Sigma} \to_{p} \Sigma$ as $n \to \infty$. 
		\end{enumerate}

		Next we establish ($\romannumeral1$)--($\romannumeral2$) following the approach used for ($\romannumeral2$) in the proof of Theorem \ref{thm Asymptotic normality}. 
		By Taylor's expansion, it holds that  
		\begin{align}
			&\frac{1}{n} \sum_{t=1}^{n} \frac{\partial^{2}\ell_{t}(\widehat{\bm{\theta}}_{\text{G}})}{\partial\theta_{k} \partial\theta_{\ell}} 
			= \frac{1}{n} \sum_{t=1}^{n} \frac{\partial^{2}\ell_{t}(\bm{\theta}_{0})}{\partial\theta_{k} \partial\theta_{\ell}} 
			+ \frac{1}{n} \sum_{t=1}^{n} \frac{\partial}{\partial \bm{\theta}^{\prime}} \frac{\partial^{2}\ell_{t}(\bar{\bm{\theta}}_{n})}{\partial\theta_{k} \partial\theta_{\ell}} \left(\widehat{\bm{\theta}}_{\text{G}} - \bm{\theta}_{0}\right), \label{eq Taylor expansion of first and second derivative of lt I} \\
			&\frac{1}{n} \sum_{t=1}^{n} \frac{\partial\ell_{t}(\widehat{\bm{\theta}}_{\text{G}})}{\partial\theta_{k}} \frac{\partial\ell_{t}(\widehat{\bm{\theta}}_{\text{G}})}{\partial\theta_{\ell}} 
			= \frac{1}{n} \sum_{t=1}^{n} \frac{\partial\ell_{t}(\bm{\theta}_{0})}{\partial\theta_{k}} \frac{\partial\ell_{t}(\bm{\theta}_{0})}{\partial\theta_{\ell}} \notag\\
			&\phantom{\frac{1}{n} \sum_{t=1}^{n} \frac{\partial\ell_{t}(\widehat{\bm{\theta}}_{\text{G}})}{\partial\theta_{k}} \frac{\partial\ell_{t}(\widehat{\bm{\theta}}_{\text{G}})}{\partial\theta_{\ell}}}+ \frac{1}{n} \sum_{t=1}^{n} \left(\frac{\partial^{2}\ell_{t}(\bar{\bm{\theta}}_{n})}{\partial \bm{\theta}^{\prime} \partial\theta_{k}} \frac{\partial\ell_{t}(\bar{\bm{\theta}}_{n})}{\partial\theta_{\ell}} + \frac{\partial\ell_{t}(\bar{\bm{\theta}}_{n})}{\partial\theta_{k}} \frac{\partial^{2}\ell_{t}(\bar{\bm{\theta}}_{n})}{\partial \bm{\theta}^{\prime} \partial\theta_{\ell}}\right) \left(\widehat{\bm{\theta}}_{\text{G}} - \bm{\theta}_{0}\right), \label{eq Taylor expansion of first and second derivative of lt II} 
		\end{align}
		where the elements of $\bar{\bm{\theta}}_{n}$ lie in the segment joining the corresponding elements of $\widehat{\bm{\theta}}_{\text{G}}$ and $\bm{\theta}_{0}$. 
		Combining \eqref{eq Taylor expansion of first and second derivative of lt I}, $\widehat{\bm{\theta}}_{\text{G}} - \bm{\theta}_{0} = o_{p}(1)$ by Theorem \ref{thm Consistency}, and (a)--(d) in the proof of Theorem \ref{thm Asymptotic normality}, we have that ($\romannumeral1$) holds. 
		Analogous to the technique employed for (a)--(d) in the proof of Theorem \ref{thm Asymptotic normality}, we can establish that the following results hold as $n \to \infty$: 
		\begin{enumerate}[(a')]
			\item $\frac{1}{n} \sum_{t=1}^{n} \frac{\partial\ell_{t}(\bm{\theta}_{0})}{\partial\bm{\theta}} \frac{\partial\ell_{t}(\bm{\theta}_{0})}{\partial\bm{\theta}^{\prime}} \to_{p} \Sigma$; 
			\item $\|\bar{\bm{\theta}}_{n} - \bm{\theta}_{0}\|_{2} \to 0$ a.s.;
			\item there exists a compact neighborhood $V(\bm{\theta}_{0})$ of $\bm{\theta}_{0}$ such that $\sup_{\bm{\theta} \in V(\bm{\theta}_{0})} \left\|\frac{\partial^{2}\ell_{t}(\bm{\theta})}{\partial \bm{\theta}^{\prime} \partial\theta_{k}} \frac{\partial\ell_{t}(\bm{\theta})}{\partial\theta_{\ell}}\right\|_{2}$ and $\sup_{\bm{\theta} \in V(\bm{\theta}_{0})} \left\|\frac{\partial\ell_{t}(\bm{\theta})}{\partial\theta_{k}} \frac{\partial^{2}\ell_{t}(\bm{\theta})}{\partial \bm{\theta}^{\prime} \partial\theta_{\ell}}\right\|_{2}$ are finite a.s.; 
			\item $\sup_{\bm{\theta} \in \Theta} \left\|\frac{1}{n} \sum_{t=1}^{n} \left(\frac{\partial\ell_{t}(\bm{\theta})}{\partial\bm{\theta}} \frac{\partial\ell_{t}(\bm{\theta})}{\partial\bm{\theta}^{\prime}} - \frac{\partial\widetilde{\ell}_{t}(\bm{\theta})}{\partial\bm{\theta}} \frac{\partial\widetilde{\ell}_{t}(\bm{\theta})}{\partial\bm{\theta}^{\prime}}\right)\right\|_{2} \to_{p} 0$. 
		\end{enumerate}
		Then by \eqref{eq Taylor expansion of first and second derivative of lt II}, $\widehat{\bm{\theta}}_{\text{G}} - \bm{\theta}_{0} = o_{p}(1)$ by Theorem \ref{thm Consistency} and (a')--(d'), we obtain that ($\romannumeral2$) holds. 
		This accomplishes the proof of Proposition \ref{propo approximates of asymptotic variances}. 
		\end{proof}

	\subsection{Proof of Theorem \ref{thm Asymptotic normality under low rank}}
		\begin{proof}[\textbf{Proof of Theorem \ref{thm Asymptotic normality under low rank}}]
		Theorem \ref{thm Asymptotic normality under low rank} directly follows from Theorem \ref{thm Asymptotic normality} and Proposition 4.1 in \cite{shapiro1986asymptotic} for overparameterized models. 
		\end{proof}

	\subsection{Proof of Theorem \ref{thm BIC}}
		\begin{proof}[\textbf{Proof of Theorem \ref{thm BIC}}]
		We only establish the selection consistency of the proposed BIC in \eqref{BIC} using the QMLE $\widehat{\bm{\theta}}_{\text{G}}$ in \eqref{est general QMLE}, and that using the QMLE $\widehat{\bm{\theta}}_{\text{LR}}$ in \eqref{est lowrank QMLE} can be proved similarly. 
		For $(r,s)$ and $(r_{0},s_{0}) \in \Pi = \{(r,s): 1 \leq r+2s \leq K_{\text{max}}\}$, 
		denote $\Theta$ (or $\Theta_{0}$) as the parameter space of $\bm{\theta}$ with the order set to $(r,s)$ (or $(r_{0},s_{0})$). 
		Let $\widehat{\bm{\theta}}_{\text{G}}^{(r,s)}$ (or $\widehat{\bm{\theta}}_{\text{G}}^{(r_{0},s_{0})}$) be the QMLE in \eqref{est general QMLE} with the order set to $(r,s)$ (or $(r_{0},s_{0})$). 
		To prove Theorem \ref{thm BIC}, it suffices to show that the following result holds for any $(r,s) \neq (r_{0},s_{0})$: 
		\begin{align}\label{eq BIC consistency}
			\lim_{n \to \infty} P\left(\text{BIC}(r,s) - \text{BIC}(r_{0},s_{0}) > 0\right) = 1. 
		\end{align}
		By \eqref{BIC}, it holds that 
		\begin{align} \label{BIC(r,s) - BIC(r0,s0)}
			\text{BIC}(r,s) - \text{BIC}(r_{0},s_{0}) 
			= \left[\widetilde{L}_{n}\left(\widehat{\bm{\theta}}_{\text{G}}^{(r,s)}\right) - \widetilde{L}_{n}\left(\widehat{\bm{\theta}}_{\text{G}}^{(r_{0},s_{0})}\right)\right] + (d - d_{0}) \ln(n), 
		\end{align}
		where $\widetilde{L}_{n}(\bm{\theta}) = \sum_{t=1}^{n} \widetilde{\ell}_{t}(\bm{\theta})$ with $\widetilde{\ell}_{t}(\bm{\theta}) = \frac{1}{2} \mathbf{y}_{t}^{\prime} \widetilde{H}_{t}^{-1}(\bm{\theta}) \mathbf{y}_{t} + \frac{1}{2} \ln |\widetilde{H}_{t}(\bm{\theta})|$, $d = m + (r + 2s) (1 + m^2) + m(m-1)/2 + 2$ and $d_{0} = m + (r_{0} + 2s_{0}) (1 + m^2) + m(m-1)/2 + 2$. 
		By Lemma \ref{lemma (iii) for consistency}, we have that $\bm{\theta}_{0} = \argmin_{\bm{\theta} \in \Theta_{0}} E(\ell_{t}(\bm{\theta}))$ with $\ell_{t}(\bm{\theta}) = \frac{1}{2} \mathbf{y}_{t}^{\prime} H_{t}^{-1}(\bm{\theta}) \mathbf{y}_{t} + \frac{1}{2} \ln |H_{t}(\bm{\theta})|$. 
		In addition, let $\mathring{\bm{\theta}} = \argmin_{\bm{\theta} \in \Theta}E(\ell_{t}(\bm{\theta}))$. 
		To verify \eqref{eq BIC consistency}, we next consider two cases. 

		Case I (overfitting): $r \geq r_{0}$, $s \geq s_{0}$, and at least one inequality holds. 
		Note that $d - d_{0} > 0$ in Case I, which implies that $(d - d_{0}) \ln(n) \to \infty$ as $n \to \infty$. 
		Thus by \eqref{BIC(r,s) - BIC(r0,s0)}, to establish that \eqref{eq BIC consistency} holds for Case I, it suffices to show that  
		\begin{align*}
			\widetilde{L}_{n}\left(\widehat{\bm{\theta}}_{\text{G}}^{(r,s)}\right) - \widetilde{L}_{n}\left(\widehat{\bm{\theta}}_{\text{G}}^{(r_{0},s_{0})}\right) = O_{p}(1) 
			\;\; \text{as} \;\; n \to \infty. 
		\end{align*}
		Rewrite $\widetilde{L}_{n}(\widehat{\bm{\theta}}_{\text{G}}^{(r,s)}) - \widetilde{L}_{n}(\widehat{\bm{\theta}}_{\text{G}}^{(r_{0},s_{0})})$ as follows, 
		\begin{align} \label{Ltilden(thetahat,r,s) - Ltilden(thetahat,r0,s0)}
			\widetilde{L}_{n}\left(\widehat{\bm{\theta}}_{\text{G}}^{(r,s)}\right) - \widetilde{L}_{n}\left(\widehat{\bm{\theta}}_{\text{G}}^{(r_{0},s_{0})}\right) 
			= &\left[\widetilde{L}_{n}\left(\widehat{\bm{\theta}}_{\text{G}}^{(r,s)}\right) - L_{n}(\mathring{\bm{\theta}})\right] 
			+ \left[L_{n}(\mathring{\bm{\theta}}) - L_{n}(\bm{\theta}_{0})\right] \notag\\
			&+ \left[L_{n}(\bm{\theta}_{0}) - \widetilde{L}_{n}\left(\widehat{\bm{\theta}}_{\text{G}}^{(r_{0},s_{0})}\right)\right],  
		\end{align}
		where $L_{n}(\bm{\theta}) = \sum_{t=1}^{n} \ell_{t}(\bm{\theta})$. 
		Notice that the model with order $(r,s)$ in Case I corresponds to a bigger model. 
		Then it holds that $L_{n}(\bm{\theta}_{0}) = L_{n}(\mathring{\bm{\theta}})$. 
		Hence by \eqref{Ltilden(thetahat,r,s) - Ltilden(thetahat,r0,s0)}, we are left to show the following results hold as $n \to \infty$: 
		(a) $\widetilde{L}_{n}(\widehat{\bm{\theta}}_{\text{G}}^{(r_{0},s_{0})}) - L_{n}(\bm{\theta}_{0}) = O_{p}(1)$; and 
		(b) $\widetilde{L}_{n}(\widehat{\bm{\theta}}_{\text{G}}^{(r,s)}) - L_{n}(\mathring{\bm{\theta}}) = O_{p}(1)$. 

		For (a), by (a1)--(a2) in the proof of Lemma \ref{lemma (ii) for consistency}, we can obtain that $\sup_{\bm{\theta} \in \Theta_{0}} |\widetilde{L}_{n}(\bm{\theta}) - L_{n}(\bm{\theta})| = O_{p}(1)$. 
		This together with $\bm{\theta}_{0} \in \Theta_{0}$ under Assumption \ref{assum parameters}, implies that $\widetilde{L}_{n}(\bm{\theta}_{0}) - L_{n}(\bm{\theta}_{0}) = O_{p}(1)$. 
		Thus we only need to show that $\widetilde{L}_{n}(\widehat{\bm{\theta}}_{\text{G}}^{(r_{0},s_{0})}) - \widetilde{L}_{n}(\bm{\theta}_{0}) = O_{p}(1)$. 
		By the mean value theorem, we have 
		\begin{align}
			&\widetilde{L}_{n}\left(\widehat{\bm{\theta}}_{\text{G}}^{(r_{0},s_{0})}\right) - \widetilde{L}_{n}\left(\bm{\theta}_{0}\right) 
			\leq \frac{1}{\sqrt{n}} \left\|\frac{\partial \widetilde{L}_{n}\left(\check{\bm{\theta}}_{\text{G}}^{(r_{0},s_{0})}\right)}{\partial \bm{\theta}}\right\|_{2} \sqrt{n} \left\|\widehat{\bm{\theta}}_{\text{G}}^{(r_{0},s_{0})} - \bm{\theta}_{0}\right\|_{2}, \;\; \text{and} \label{eq Ltilde(thetahat) - Ltilde(theta0)}\\ 
			&\frac{1}{\sqrt{n}} \left\|\frac{\partial \widetilde{L}_{n}\left(\check{\bm{\theta}}_{\text{G}}^{(r_{0},s_{0})}\right)}{\partial \bm{\theta}} - \frac{\partial \widetilde{L}_{n}\left(\bm{\theta}_{0}\right)}{\partial \bm{\theta}}\right\|_{2} 
			\leq \frac{1}{n} \left\|\frac{\partial^{2} \widetilde{L}_{n}\left(\bar{\bm{\theta}}_{\text{G}}^{(r_{0},s_{0})}\right)}{\partial \bm{\theta} \partial \bm{\theta}^{\prime}}\right\|_{2} \sqrt{n} \left\|\check{\bm{\theta}}_{\text{G}}^{(r_{0},s_{0})} - \bm{\theta}_{0}\right\|_{2}, \label{eq dLtilde(thetahat) - dLtilde(theta0)}
		\end{align}
		where $\check{\bm{\theta}}_{\text{G}}^{(r_{0},s_{0})}$ lies between $\widehat{\bm{\theta}}_{\text{G}}^{(r_{0},s_{0})}$ and $\bm{\theta}_{0}$, and $\bar{\bm{\theta}}_{\text{G}}^{(r_{0},s_{0})}$ lies between $\check{\bm{\theta}}_{\text{G}}^{(r_{0},s_{0})}$ and $\bm{\theta}_{0}$. 
		Note that $\sqrt{n}\|\widehat{\bm{\theta}}_{\text{G}}^{(r_{0},s_{0})} - \bm{\theta}_{0}\|_{2} = O_{p}(1)$ by Theorem \ref{thm Asymptotic normality}, 
		$\|{\partial \widetilde{L}_{n}(\bm{\theta}_{0})} / {\partial \bm{\theta}}\|_{2} / \sqrt{n} = O_{p}(1)$ by ($\romannumeral1$) in the proof of Theorem \ref{thm Asymptotic normality}, and 
		$\|{\partial^{2} \widetilde{L}_{n}(\bar{\bm{\theta}}_{\text{G}}^{(r_{0},s_{0})})} / {\partial \bm{\theta} \partial \bm{\theta}^{\prime}}\|_{2} / n = O_{p}(1)$ by ($\romannumeral2$) in the proof of Theorem \ref{thm Asymptotic normality}. 
		And by \eqref{eq dLtilde(thetahat) - dLtilde(theta0)}, it follows that $\|{\partial \widetilde{L}_{n}(\check{\bm{\theta}}_{\text{G}}^{(r_{0},s_{0})})} / {\partial \bm{\theta}}\|_{2} / \sqrt{n} = O_{p}(1)$. 
		Thus by \eqref{eq Ltilde(thetahat) - Ltilde(theta0)}, $\|{\partial \widetilde{L}_{n}(\check{\bm{\theta}}_{\text{G}}^{(r_{0},s_{0})})} / {\partial \bm{\theta}}\|_{2} / \sqrt{n} = O_{p}(1)$ and $\sqrt{n}\|\widehat{\bm{\theta}}_{\text{G}}^{(r_{0},s_{0})} - \bm{\theta}_{0}\|_{2} = O_{p}(1)$, 
		we can obtain that $\widetilde{L}_{n}(\widehat{\bm{\theta}}_{\text{G}}^{(r_{0},s_{0})}) - \widetilde{L}_{n}(\bm{\theta}_{0}) = O_{p}(1)$. 
		Hence (a) holds. 

		For (b), 
		recall that $\bm{\theta} = (\bm{\delta}^{\prime}, \bm{\beta}^{\prime})^{\prime}$ with $\bm{\delta} = (\underline{\bm{\omega}}^{\prime}, \bm{\kappa}^{\prime})^{\prime}$ and $\bm{\kappa} = (\bm{\lambda}^{\prime}, \bm{\gamma}^{\prime}, \bm{\varphi}^{\prime}, \bm{g}_{0}^{\prime}, \bm{g}_{1}^{\prime}, \bm{g}_{2}^{\prime})^{\prime}$. 
		Denote $\underline{\bm{\theta}} = (\underline{\bm{\omega}}^{\prime}, \underline{\bm{\kappa}}^{\prime}, \bm{\beta}^{\prime})^{\prime}$ with $\underline{\bm{\kappa}} = (\ovec(\Phi_{1}), \ovec(\Phi_{2}), \ldots)^{\prime}$, and define $\widehat{\underline{\bm{\theta}}}_{\text{G}}^{(r,s)}$ and $\mathring{\underline{\bm{\theta}}}$ correspondingly. 
		It is noteworthy that $\bm{\theta}$ is non-identifiable in Case I due to the non-identifiability of parameters $\{\lambda_{k}, \gamma_{j}, \varphi_{j} : r_{0} < k \leq r, s_{0} < j \leq s\}$, but $\underline{\bm{\theta}}$ is still identifiable because $\Phi_{i}$ is identifiable. 
		Then we still have the consistency and asymptotic normality of the QMLEs for $\{\lambda_{k}, \gamma_{j}, \varphi_{j} : 1 \leq k \leq r_{0}, 1 \leq j \leq s_{0}\}$ and the other parameters, although $\{\lambda_{k}, \gamma_{j}, \varphi_{j} : r_{0} < k \leq r, s_{0} < j \leq s\}$ are non-identifiable. 
		It follows that 
		$\sqrt{n} \|\widehat{\underline{\bm{\omega}}}^{(r,s)} - \mathring{\underline{\bm{\omega}}}\|_{2}$, 
		$\sqrt{n} \|\widehat{\bm{\beta}}^{(r,s)} - \mathring{\bm{\beta}}\|_{2}$, 
		$\sqrt{n} \|\widehat{G}^{(r,s)}_{0,k} - \mathring{G}_{0,k}\|_{2}$ for $1 \leq k \leq r$, 
		$\sqrt{n} \|\widehat{G}^{(r,s)}_{1,k} - \mathring{G}_{1,k}\|_{2}$ and $\sqrt{n} \|\widehat{G}^{(r,s)}_{2,k} - \mathring{G}_{2,k}\|_{2}$ for $1 \leq k \leq s$, 
		$\sqrt{n} \|\widehat{\lambda}^{(r,s)}_{k} - \mathring{\lambda}_{k}\|_{2}$ for $1 \leq k \leq r_{0}$, and 
		$\sqrt{n} \|\widehat{\gamma}^{(r,s)}_{k} - \mathring{\gamma}_{k}\|_{2}$ and $\sqrt{n} \|\widehat{\varphi}^{(r,s)}_{k} - \mathring{\varphi}_{k}\|_{2}$ for $1 \leq k \leq s_{0}$ 
		are all $O_{p}(1)$, where the notation corresponds to $\widehat{\bm{\theta}}_{\text{G}}^{(r,s)}$ and $\mathring{\bm{\theta}}$. 
		Note that $\mathring{G}_{0,k} = 0$ for $r_{0} < k \leq r$, $\mathring{G}_{1,k} = 0$ and $\mathring{G}_{2,k} = 0$ for $s_{0} < k \leq s$, and $0 < |\lambda_{1}|, \ldots, |\lambda_{r}|, \gamma_{1}, \ldots, \gamma_{s} \leq \varrho$ under Assumption \ref{assum parameters}($\romannumeral2$). 
		Then by a similar argument to that for \eqref{eq sum of Phi} under Assumptions \ref{assum parameters}($\romannumeral1$)--($\romannumeral2$), we can show that 
		\begin{align*}
			&\sum_{i=1}^{\infty} \sqrt{n} \left\|\sum_{k=1}^{r} \left[\left(\widehat{\lambda}^{(r,s)}_{k}\right)^{i-1} \widehat{G}^{(r,s)}_{0,k} - \mathring{\lambda}_{k}^{i-1} \mathring{G}_{0,k}\right]\right\|_{2} \\ 
			=& \sum_{i=1}^{\infty} \sqrt{n} \left\|\sum_{k=1}^{r} \left[\left(\widehat{\lambda}^{(r,s)}_{k}\right)^{i-1} \widehat{G}^{(r,s)}_{0,k} - \left(\widehat{\lambda}^{(r,s)}_{k}\right)^{i-1} \mathring{G}_{0,k} + \left(\widehat{\lambda}^{(r,s)}_{k}\right)^{i-1} \mathring{G}_{0,k} - \mathring{\lambda}_{k}^{i-1} \mathring{G}_{0,k}\right]\right\|_{2} \\ 
			=& \sum_{i=1}^{\infty} \sqrt{n} \left\|\sum_{k=1}^{r} \left(\widehat{\lambda}^{(r,s)}_{k}\right)^{i-1} \left(\widehat{G}^{(r,s)}_{0,k} - \mathring{G}_{0,k}\right)
			+ \sum_{k=1}^{r} \left[\left(\widehat{\lambda}^{(r,s)}_{k}\right)^{i-1} - \mathring{\lambda}_{k}^{i-1}\right] \mathring{G}_{0,k}\right\|_{2} \\
			=& \sum_{i=1}^{\infty} \sqrt{n} \left\|\sum_{k=1}^{r} \left(\widehat{\lambda}^{(r,s)}_{k}\right)^{i-1} \left(\widehat{G}^{(r,s)}_{0,k} - \mathring{G}_{0,k}\right)
			+ \sum_{k=1}^{r_{0}} \left[\left(\widehat{\lambda}^{(r,s)}_{k}\right)^{i-1} - \mathring{\lambda}_{k}^{i-1}\right] \mathring{G}_{0,k}\right\|_{2} \\
			\leq& \sum_{k=1}^{r} \left(\sum_{i=1}^{\infty} \left|\widehat{\lambda}^{(r,s)}_{k}\right|^{i-1}\right) \sqrt{n} \left\|\widehat{G}^{(r,s)}_{0,k} - \mathring{G}_{0,k}\right\|_2 
			+ \sum_{k=1}^{r_{0}} \sqrt{n} \sum_{i=1}^{\infty} \left|\left(\widehat{\lambda}^{(r,s)}_{k}\right)^{i-1} - \mathring{\lambda}_{k}^{i-1}\right| \left\|\mathring{G}_{0,k}\right\|_2 \\
			\leq& \sum_{k=1}^{r} \frac{1}{1 - \left|\widehat{\lambda}^{(r,s)}_{k}\right|} \sqrt{n} \left\|\widehat{G}^{(r,s)}_{0,k} - \mathring{G}_{0,k}\right\|_2 
			+ \sum_{k=1}^{r_{0}} \sqrt{n} \frac{\left|\widehat{\lambda}^{(r,s)}_{k} - \mathring{\lambda}_{k}\right|}{\left(1 - \left|\widehat{\lambda}^{(r,s)}_{k}\right|\right) \left(1 - \left|\mathring{\lambda}_{k}\right|\right)} \left\|\mathring{G}_{0,k}\right\|_2 \\
			\leq& \sum_{k=1}^{r} \frac{1}{1 - \varrho} \sqrt{n} \left\|\widehat{G}^{(r,s)}_{0,k} - \mathring{G}_{0,k}\right\|_2 
			+ \sum_{k=1}^{r_{0}} \frac{1}{(1 - \varrho)^2} \sqrt{n} \left|\widehat{\lambda}^{(r,s)}_{k} - \mathring{\lambda}_{k}\right| \left\|\mathring{G}_{0,k}\right\|_2 \\
			=& O_{p}(1), 
		\end{align*}
		where the second inequality follows from the fact that $a^{m} - b^{m} = (a - b) \sum_{j=0}^{m-1} a^{j} b^{m-1-j}$, which implies that $\sum_{i=1}^{\infty} |a^{i-1} - b^{i-1}| \leq |a - b| \sum_{i=2}^{\infty} \sum_{j=0}^{i-2} |a|^{j} |b|^{i-j-2} = |a - b| / [(1 - |a|)(1 - |b|)]$ when $|a| < 1$ and $|b| < 1$. 
		By the same reasoning, we can obtain that 
		$\sum_{i=1}^{\infty} \sqrt{n} \|\ovec(\Phi_{i}(\widehat{\bm{\kappa}}^{(r,s)})) - \ovec(\Phi_{i}(\mathring{\bm{\kappa}}))\|_{2} = O_{p}(1)$. 
		It then follows that $\sqrt{n} \|\widehat{\underline{\bm{\theta}}}_{\text{G}}^{(r,s)} - \mathring{\underline{\bm{\theta}}}\|_{2} = O_{p}(1)$. 
		Similar to the proof of $\|{\partial \widetilde{L}_{n}(\widehat{\bm{\theta}}_{\text{G}}^{(r_{0},s_{0})})} / {\partial \bm{\theta}}\|_{2} / \sqrt{n} = O_{p}(1)$ in (a), we can verify that $\|{\partial \widetilde{L}_{n}(\widehat{\underline{\bm{\theta}}}_{\text{G}}^{(r,s)})} / {\partial \underline{\bm{\theta}}}\|_{2} / \sqrt{n} = O_{p}(1)$. 
		With analogous arguments in the proof of (a), we can show that (b) holds. 
		As a result, \eqref{eq BIC consistency} holds for Case I. 

		Case II (misspecification): $r < r_{0}$ or $s < s_{0}$. 
		Recall that $\widetilde{\mathcal{L}}_{n}(\bm{\theta}) = \widetilde{L}_{n}(\bm{\theta}) / n$. 
		Under Assumption \ref{assum stationarity}, it holds that 
		\begin{align} \label{BIC(r,s) - BIC(r0,s0), Case II}
			&\text{BIC}(r,s) - \text{BIC}(r_{0},s_{0}) 
			= \left[\widetilde{L}_{n}\left(\widehat{\bm{\theta}}_{\text{G}}^{(r,s)}\right) - \widetilde{L}_{n}\left(\widehat{\bm{\theta}}_{\text{G}}^{(r_{0},s_{0})}\right)\right] + (d - d_{0}) \ln(n) \notag\\
			=& \left[\widetilde{L}_{n}\left(\widehat{\bm{\theta}}_{\text{G}}^{(r,s)}\right) - E \left(L_{n}(\mathring{\bm{\theta}})\right)\right] 
			+ \left[E \left(L_{n}(\mathring{\bm{\theta}})\right) - E \left(L_{n}(\bm{\theta}_{0})\right)\right] 
			+ \left[E \left(L_{n}(\bm{\theta}_{0})\right) - \widetilde{L}_{n}\left(\widehat{\bm{\theta}}_{\text{G}}^{(r_{0},s_{0})}\right)\right] \notag\\
			&+ (d - d_{0}) \ln(n) \notag\\
			=& n \left\{ 
			\left[\widetilde{\mathcal{L}}_{n}\left(\widehat{\bm{\theta}}_{\text{G}}^{(r,s)}\right) - E \left(\ell_{t}(\mathring{\bm{\theta}})\right)\right] 
			+ \left[E \left(\ell_{t}(\mathring{\bm{\theta}})\right) - E \left(\ell_{t}(\bm{\theta}_{0})\right)\right] 
			+ \left[E \left(\ell_{t}(\bm{\theta}_{0})\right) - \widetilde{\mathcal{L}}_{n}\left(\widehat{\bm{\theta}}_{\text{G}}^{(r_{0},s_{0})}\right)\right] 
			\right\} \notag\\
			&+ (d - d_{0}) \ln(n). 
		\end{align}
		Thus to establish \eqref{eq BIC consistency}, it suffices to show 
		(c) $E(\ell_{t}(\mathring{\bm{\theta}})) - E(\ell_{t}(\bm{\theta}_{0})) > \delta$ for some $\delta > 0$; 
		(d) $\widetilde{\mathcal{L}}_{n}(\widehat{\bm{\theta}}_{\text{G}}^{(r_{0},s_{0})}) - E (\ell_{t}(\bm{\theta}_{0})) = o_{p}(1)$ as $n \to \infty$; and 
		(e) $\widetilde{\mathcal{L}}_{n}(\widehat{\bm{\theta}}_{\text{G}}^{(r,s)}) - E (\ell_{t}(\mathring{\bm{\theta}})) = o_{p}(1)$ as $n \to \infty$. 

		For (c), 
		denote $r_{*} = \max\{r, r_{0}\}$ and $s_{*} = \max\{s, s_{0}\}$. 
		Let $\bm{\theta}_{0 *}$ (or $\mathring{\bm{\theta}}_{*}$) be a parameter vector with the order set to $(r_{*},s_{*})$, including $\bm{\theta}_{0}$ (or $\mathring{\bm{\theta}}$) as its subvector at the corresponding locations and zeroes at the remaining locations. 
		Moreover, denote $\Theta_{*}$ as the parameter space of $\bm{\theta}$ with the order set to $(r_{*},s_{*})$, including the points $\bm{\theta}_{0 *}$ and $\mathring{\bm{\theta}}_{*}$. 
		Since $E(\ell_{t}(\bm{\theta}))$ has a unique minimum at $\bm{\theta}_{0}$ on $\Theta_{0}$ by Lemma \ref{lemma (iii) for consistency}, we can obtain that $\bm{\theta}_{0 *}$ is the unique minimizer of $E(\ell_{t}(\bm{\theta}))$ on $\Theta_{*}$. 
		This together with $\bm{\theta}_{0 *}, \mathring{\bm{\theta}}_{*} \in \Theta_{*}$, $E(\ell_{t}(\bm{\theta}_{0 *})) = E(\ell_{t}(\bm{\theta}_{0}))$ and $E(\ell_{t}(\mathring{\bm{\theta}}_{*})) = E(\ell_{t}(\mathring{\bm{\theta}}))$, implies that the following result holds for some constant $\delta > 0$: 
		\begin{align*} 
			E\left(\ell_{t}(\mathring{\bm{\theta}})\right) - E\left(\ell_{t}(\bm{\theta}_{0})\right) 
			= E\left(\ell_{t}(\mathring{\bm{\theta}}_{*})\right) - E\left(\ell_{t}(\bm{\theta}_{0 *})\right) 
			> \delta. 
		\end{align*}

		For (d), 
		we only need to show 
		(d1) $\widetilde{\mathcal{L}}_{n}(\widehat{\bm{\theta}}_{\text{G}}^{(r_{0},s_{0})}) - \widetilde{\mathcal{L}}_{n}(\bm{\theta}_{0}) = o_{p}(1)$; 
		(d2) $\widetilde{\mathcal{L}}_{n}(\bm{\theta}_{0}) - \mathcal{L}_{n}(\bm{\theta}_{0}) = o_{p}(1)$; and 
		(d3) $\mathcal{L}_{n}(\bm{\theta}_{0}) - E (\ell_{t}(\bm{\theta}_{0})) = o_{p}(1)$ as $n \to \infty$. 
		It can be obtained that 
		(d1) holds by the proof of (a), 
		(d2) holds by ($\romannumeral1$) of Lemma \ref{lemma (ii) for consistency}, and 
		(d3) holds by the ergodic theorem under Assumption \ref{assum stationarity}. 
		Hence (d) holds. 

		For (e), similar to (d), we only need to verify 
		(e1) $\widetilde{\mathcal{L}}_{n}(\widehat{\bm{\theta}}_{\text{G}}^{(r,s)}) - \widetilde{\mathcal{L}}_{n}(\mathring{\bm{\theta}}) = o_{p}(1)$; 
		(e2) $\widetilde{\mathcal{L}}_{n}(\mathring{\bm{\theta}}) - \mathcal{L}_{n}(\mathring{\bm{\theta}}) = o_{p}(1)$; and 
		(e3) $\mathcal{L}_{n}(\mathring{\bm{\theta}}) - E (\ell_{t}(\mathring{\bm{\theta}})) = o_{p}(1)$ as $n \to \infty$. 
		(e2) and (e3) can be proved with the analogous arguments in the proof of (d2) and (d3), respectively. 
		Then we are left to verify (e1). 
		Assume that $E(\ell_{t}(\bm{\theta}))$ has a unique minimum at $\mathring{\bm{\theta}}$ on $\Theta$. 
		Similar to the proof of Theorem \ref{thm Consistency}, we can obtain that $\widehat{\bm{\theta}}_{\text{G}}^{(r,s)} \to \mathring{\bm{\theta}}$ a.s. as $n \to \infty$. 
		By the mean value theorem, we have 
		\begin{align*}
			&\widetilde{\mathcal{L}}_{n}\left(\widehat{\bm{\theta}}_{\text{G}}^{(r,s)}\right) - \widetilde{\mathcal{L}}_{n}\left(\mathring{\bm{\theta}}\right) 
			\leq \left\|\frac{\partial \widetilde{\mathcal{L}}_{n}\left(\check{\bm{\theta}}_{\text{G}}^{(r,s)}\right)}{\partial \bm{\theta}}\right\|_{2} \left\|\widehat{\bm{\theta}}_{\text{G}}^{(r,s)} - \mathring{\bm{\theta}}\right\|_{2}, \;\; \text{and} \\ 
			&\left\|\frac{\partial \widetilde{\mathcal{L}}_{n}\left(\check{\bm{\theta}}_{\text{G}}^{(r,s)}\right)}{\partial \bm{\theta}} - \frac{\partial \widetilde{\mathcal{L}}_{n}\left(\mathring{\bm{\theta}}\right)}{\partial \bm{\theta}}\right\|_{2} 
			\leq \left\|\frac{\partial^{2} \widetilde{\mathcal{L}}_{n}\left(\bar{\bm{\theta}}_{\text{G}}^{(r,s)}\right)}{\partial \bm{\theta} \partial \bm{\theta}^{\prime}}\right\|_{2} \left\|\check{\bm{\theta}}_{\text{G}}^{(r,s)} - \mathring{\bm{\theta}}\right\|_{2}, 
		\end{align*}
		where $\check{\bm{\theta}}_{\text{G}}^{(r,s)}$ lies between $\widehat{\bm{\theta}}_{\text{G}}^{(r,s)}$ and $\mathring{\bm{\theta}}$, and $\bar{\bm{\theta}}_{\text{G}}^{(r,s)}$ lies between $\check{\bm{\theta}}_{\text{G}}^{(r,s)}$ and $\mathring{\bm{\theta}}$. 
		These together with $\|\widehat{\bm{\theta}}_{\text{G}}^{(r,s)} - \mathring{\bm{\theta}}\|_{2} = o_{p}(1)$ by $\widehat{\bm{\theta}}_{\text{G}}^{(r,s)} \to \mathring{\bm{\theta}}$ a.s. and 
		$\|{\partial^{2} \widetilde{\mathcal{L}}_{n}(\bar{\bm{\theta}}_{\text{G}}^{(r,s)})} / {\partial \bm{\theta} \partial \bm{\theta}^{\prime}}\|_{2} = O_{p}(1)$ by ($\romannumeral2$) in the proof of Theorem \ref{thm Asymptotic normality}, 
		imply that (e1) holds. 
		Then by (e1)--(e3), we have (e) holds. 
		Hence by \eqref{BIC(r,s) - BIC(r0,s0), Case II} and (c)--(e), we have 
		\begin{align*}
			\text{BIC}(r,s) - \text{BIC}(r_{0},s_{0}) 
			> n (o_p(1) + \delta) + O(\ln(n)) \to \infty  
			\;\; \text{as} \;\; n \to \infty. 
		\end{align*}
		As a result, \eqref{eq BIC consistency} holds for Case II. 
		This completes the proof of Theorem \ref{thm BIC}. 
		\end{proof}

	\section{Proofs of lemmas}
		\begin{proof}[Proof of Lemma \ref{lemma for identification}]
			The sufficiency of ($\romannumeral1$)--($\romannumeral4$) is obvious and thus we only need to verify their necessity. 
			Without loss of generality, suppose that $|\lambda_{1}| \geq |\lambda_{2}| \geq \ldots \geq |\lambda_{r}|$ and $\gamma_{1} > \gamma_{2} > \ldots > \gamma_{s}$. 

			For ($\romannumeral1$), we first verify the claim when $r = 2$. 
			If $c_{1} \lambda_{1}^{j} + c_{2} \lambda_{2}^{j} = 0$ holds for all $j \in \mathbb{Z}^{+}$, then we have 
			\begin{align} \label{eq1 in lemma identification}
				c_{1} = - c_{2} \left(\frac{\lambda_{2}}{\lambda_{1}}\right)^{j} 
				\;\; \text{for all} \;\; j \in \mathbb{Z}^{+}. 
			\end{align}
			Recall that $\{\lambda_{k}\}$ are distinct and $|\lambda_{1}| \geq |\lambda_{2}|$, which implies that $-1 \leq \lambda_{2} / \lambda_{1} < 1$. 
			Then as $j \to \infty$, it holds that $(\lambda_{2} / \lambda_{1})^{j} \to 0$ if $-1 < \lambda_{2} / \lambda_{1} < 1$, and $(\lambda_{2} / \lambda_{1})^{j} = \pm 1$ if $\lambda_{2} / \lambda_{1} = -1$. 
			This together with \eqref{eq1 in lemma identification} implies that $c_{1} = c_{2} = 0$. 
			Next we consider the case when $r=3$. 
			If $c_{1} \lambda_{1}^{j} + c_{2} \lambda_{2}^{j} + c_{3} \lambda_{3}^{j} = 0$ holds for all $j \in \mathbb{Z}^{+}$, then we have 
			\begin{align} \label{eq2 in lemma identification}
				c_{1} = - c_{2} \left(\frac{\lambda_{2}}{\lambda_{1}}\right)^{j} - c_{3} \left(\frac{\lambda_{3}}{\lambda_{1}}\right)^{j} 
				\;\; \text{for all} \;\; j \in \mathbb{Z}^{+}. 
			\end{align}
			Similarly, we have that $-1 \leq \lambda_{2} / \lambda_{1} < 1$ and $-1 < \lambda_{3} / \lambda_{1} < 1$. 
			If $-1 < \lambda_{2} / \lambda_{1} < 1$, it holds that $(\lambda_{2} / \lambda_{1})^{j} \to 0$ and $(\lambda_{3} / \lambda_{1})^{j} \to 0$ as $j \to \infty$. This together with \eqref{eq2 in lemma identification} implies that $c_{1} = 0$, and then $c_{2} \lambda_{2}^{j} + c_{3} \lambda_{3}^{j} = 0$ holds for all $j \in \mathbb{Z}^{+}$. 
			By the above discussion for the case $r=2$, it follows that $c_{2} = c_{3} = 0$. 
			Thus $c_{1} = c_{2} = c_{3} = 0$. 
			If $\lambda_{2} / \lambda_{1} = -1$, it holds that $(\lambda_{2} / \lambda_{1})^{j} = \pm 1$ and $(\lambda_{3} / \lambda_{1})^{j} \to 0$ as $j \to \infty$. Similarly we can also conclude that $c_{1} = c_{2} = c_{3} = 0$. 
			As a result, the claim in ($\romannumeral1$) for general $r$ can be established similarly. 

			For ($\romannumeral2$), we first show (a). 
			If $\sum_{k=1}^{s} \gamma_{k}^{j} [c_{k1} \cos(j \varphi_{k}) + c_{k2} \sin(j \varphi_{k})] = 0$ holds for all $j \in \mathbb{Z}^{+}$, then we have 
			\begin{align*}
				c_{11} \cos(j \varphi_{1}) + c_{12} \sin(j \varphi_{1}) 
				= - \sum_{k=2}^{s} \left(\frac{\gamma_{k}}{\gamma_{1}}\right)^{j} [c_{k1} \cos(j \varphi_{k}) + c_{k2} \sin(j \varphi_{k})]
				\;\; \text{for all} \;\; j \in \mathbb{Z}^{+}. 
			\end{align*}
			Note that $0 < \gamma_{k} / \gamma_{1} < 1$, we can obtain that $c_{k1} = c_{k2} = 0$ for all $k$ with analogous arguments in the proof of ($\romannumeral1$). 
			Next we verify (b). 
			By the fact that $a \cos x + b \sin x = \sqrt{a^{2} + b^{2}} \cos(x - \alpha)$ with $\alpha = \arccos(a / \sqrt{a^{2} + b^{2}})$, it is equivalent to show that if the following equation holds for all $j \in \mathbb{Z}^{+}$, 
			\begin{align} \label{eq3 in lemma identification}
				\sqrt{c_{11}^{2} + c_{12}^{2}} \cos(j \varphi_{1} - \alpha_{1}) 
				= \sqrt{c_{21}^{2} + c_{22}^{2}} \cos(j \varphi_{2} - \alpha_{2}), 
			\end{align}
			where $\alpha_{1} = \arccos(c_{11} / \sqrt{c_{11}^{2} + c_{12}^{2}})$ and $\alpha_{2} = \arccos(c_{21} / \sqrt{c_{21}^{2} + c_{22}^{2}})$, 
			then $\varphi_{1} = \varphi_{2}$, $c_{11} = c_{21}$ and $c_{12} = c_{22}$. 
			It is obvious that \eqref{eq3 in lemma identification} implies that $c_{11}^{2} + c_{12}^{2} = c_{21}^{2} + c_{22}^{2}$ and $\cos(j \varphi_{1} - \alpha_{1}) = \cos(j \varphi_{2} - \alpha_{2})$ for all $j \in \mathbb{Z}^{+}$. 
			Since $\varphi_{k} \in (0, \pi)$ for $k=1,2$, we can obtain that $\varphi_{1} = \varphi_{2}$ and $\alpha_{1} = \alpha_{2}$ by $\cos(j \varphi_{1} - \alpha_{1}) = \cos(j \varphi_{2} - \alpha_{2})$ for all $j \in \mathbb{Z}^{+}$. 
			And then $\alpha_{1} = \alpha_{2}$ together with $c_{11}^{2} + c_{12}^{2} = c_{21}^{2} + c_{22}^{2}$ implies that $c_{11} = c_{21}$ and $c_{12} = c_{22}$. 
			Thus (b) is verified. 

			For ($\romannumeral3$), if $\gamma^{j} [c_{01} \cos(j \varphi) + c_{02} \sin(j \varphi)] = \sum_{k=1}^{r} c_{k} \lambda_{k}^{j}$ holds for all $j \in \mathbb{Z}^{+}$, then we have 
			\begin{align} \label{eq4 in lemma identification}
				c_{01} \cos(j \varphi) + c_{02} \sin(j \varphi) 
				= \sum_{k=1}^{r} c_{k} \left(\frac{\lambda_{k}}{\gamma}\right)^{j} 
				\;\; \text{for all} \;\; j \in \mathbb{Z}^{+}. 
			\end{align}
			As $j \to \infty$, it holds that 
			$(\lambda_{k} / \gamma)^{j} \to 0$ if $|\lambda_{k} / \gamma| < 1$, 
			$(\lambda_{k} / \gamma)^{j} \to \infty$ if $|\lambda_{k} / \gamma| > 1$, 
			and $(\lambda_{k} / \gamma)^{j} = \pm 1$ if $|\lambda_{k} / \gamma| = 1$. 
			Note that $\cos(j \varphi)$ and $\sin (j \varphi)$ change as $j$ changes and they are both bounded. 
			These together with \eqref{eq4 in lemma identification} and ($\romannumeral1$) imply that $c_{01} = c_{02} = c_{k} = 0$ for all $k$. 

			For ($\romannumeral4$), if $c_{0} \lambda^{j} = \sum_{k=1}^{s} \gamma_{k}^{j} [c_{k1} \cos(j \varphi_{k}) + c_{k2} \sin(j \varphi_{k})]$ holds for all $j \in \mathbb{Z}^{+}$, then it holds that 
			\begin{align*}
				c_{0} = \sum_{k=1}^{s} \left(\frac{\gamma_{k}}{\lambda}\right)^{j} [c_{k1} \cos(j \varphi_{k}) + c_{k2} \sin(j \varphi_{k})]
				\;\; \text{for all} \;\; j \in \mathbb{Z}^{+}. 
			\end{align*}
			With ($\romannumeral2$) and analogous arguments in the proof of ($\romannumeral3$), it can be shown that $c_{0} = c_{k1} = c_{k2} = 0$ for all $k$. 
		\end{proof}

		\begin{proof}[Proof of Lemma \ref{lemma for stationarity}]
			We first show ($\romannumeral1$). 
			Since $\rho(\Upsilon^{\otimes k}) < 1$, we have that $I_{a} - \Upsilon^{\otimes k}$ is invertible. 
			Note that all elements of $\Upsilon$ is non-negative. 
			Then there exists a vector $\bm{\nu}_{1} > 0$ such that $[(I_{a} - \Upsilon^{\otimes k})^{\prime}]^{-1} \bm{\nu}_{1} = \sum_{i=0}^{\infty} [(\Upsilon^{\otimes k})^{\prime}]^{i} \bm{\nu}_{1} > 0$. 
			Denote $\bm{\nu} = [(I_{a} - \Upsilon^{\otimes k})^{\prime}]^{-1} \bm{\nu}_{1}$. 
			It follows that $\bm{\nu} > 0$ and $(I_{a} - \Upsilon^{\otimes k})^{\prime} \bm{\nu} = \bm{\nu}_{1} > 0$. 
			Thus ($\romannumeral1$) holds. 

			Next we establish ($\romannumeral2$). 
			Recall that $\bm{x}_{t} = \Upsilon \bm{x}_{t-1} + \bm{\xi}_{t}$ by \eqref{eq Markov chain}. 
			It can be shown that 
			\begin{align} \label{eq1 in g(xt) inequality}
				&g(\bm{x}_{t}) = 1 + |\bm{x}_{t}^{\otimes k}|^{\prime} \bm{\nu} 
				= 1 + |(\Upsilon \bm{x}_{t-1} + \bm{\xi}_{t})^{\otimes k}|^{\prime} \bm{\nu} \notag\\
				=& 1 + |(\Upsilon \bm{x}_{t-1})^{\otimes k} + (\Upsilon \bm{x}_{t-1})^{\otimes (k-1)} \otimes \bm{\xi}_{t} + \ldots + \bm{\xi}_{t}^{\otimes k}|^{\prime} \bm{\nu} \notag\\
				\leq& 1 + |\bm{x}_{t-1}^{\otimes k}|^{\prime} (\Upsilon^{\otimes k})^{\prime} \bm{\nu} + |(\Upsilon^{\otimes (k-1)} \bm{x}_{t-1}^{\otimes (k-1)}) \otimes \bm{\xi}_{t} + \ldots + \bm{\xi}_{t}^{\otimes k}|^{\prime} \bm{\nu} \notag\\
				\leq& 1 + |\bm{x}_{t-1}^{\otimes k}|^{\prime} (\Upsilon^{\otimes k})^{\prime} \bm{\nu} + \left\|(\Upsilon^{\otimes (k-1)} \bm{x}_{t-1}^{\otimes (k-1)}) \otimes \bm{\xi}_{t} + \ldots + \bm{\xi}_{t}^{\otimes k}\right\|_{2} \left\|\bm{\nu}\right\|_{2} \notag\\
				\leq& 1 + |\bm{x}_{t-1}^{\otimes k}|^{\prime} (\Upsilon^{\otimes k})^{\prime} \bm{\nu} + \left[\left\|(\Upsilon^{\otimes (k-1)} \bm{x}_{t-1}^{\otimes (k-1)}) \otimes \bm{\xi}_{t}\right\|_{2} + \ldots + \left\|\bm{\xi}_{t}^{\otimes k}\right\|_{2}\right] \left\|\bm{\nu}\right\|_{2} \notag\\
				=& 1 + |\bm{x}_{t-1}^{\otimes k}|^{\prime} (\Upsilon^{\otimes k})^{\prime} \bm{\nu} + \left(\left\|\Upsilon^{\otimes (k-1)} \bm{x}_{t-1}^{\otimes (k-1)}\right\|_{2} \left\|\bm{\xi}_{t}\right\|_{2} + \ldots + \left\|\bm{\xi}_{t}\right\|_{2}^{k}\right) \left\|\bm{\nu}\right\|_{2} \notag\\ 
				\leq& 1 + |\bm{x}_{t-1}^{\otimes k}|^{\prime} (\Upsilon^{\otimes k})^{\prime} \bm{\nu} + \left(\left\|\Upsilon^{\otimes (k-1)}\right\|_{2} \left\|\bm{x}_{t-1}\right\|_{2}^{k-1} \left\|\bm{\xi}_{t}\right\|_{2} + \ldots + \left\|\bm{\xi}_{t}\right\|_{2}^{k}\right) \left\|\bm{\nu}\right\|_{2}, 
			\end{align}
			where the four inequalities are from $(A \bm{a})^{\otimes k} = A^{\otimes k} \bm{a}^{\otimes k}$ for a matrix $A$ and a vector $\bm{a}$, the H{\"o}lder's inequality, the triangle inequality of matrix norms, and the submultiplicativity of induced matrix norms, respectively, 
			and the penultimate equation is from the fact $\|A \otimes B\|_{2} = \|A\|_{2} \|B\|_{2}$ for matrices $A$ and $B$. 
			Let $\bm{\xi}_{t} = \bm{\xi}_{t,1} + \bm{\xi}_{t,2}$ with 
			\begin{align*}
				&\bm{\xi}_{t,1} = \left((1 - \beta_{10} - \beta_{20})\ovechsec(\underline{R}_{0})^{\prime}, \ovechsec(\Psi_{t})^{\prime}, \bm{0}_{m}^{\prime}, \bm{0}_{m}^{\prime}, \ldots, \bm{0}_{m}^{\prime}, \bm{0}_{m}^{\prime}\right)^{\prime} 
				\;\; \text{and} \\
				&\bm{\xi}_{t,2} = \left(\bm{0}_{m(m-1)/2}^{\prime}, \bm{0}_{m(m-1)/2}^{\prime}, \bm{\varepsilon}_{t}^{\prime}, \bm{0}_{m}^{\prime}, \ldots, \bm{0}_{m}^{\prime}, \bm{0}_{m}^{\prime}\right)^{\prime}.  
			\end{align*}
			Since $0 < \beta_{10} + \beta_{20} < 1$ and the elements of $\underline{R}_{0}$ and $\Psi_{t}$ are between $-1$ and 1, we have that 
			\begin{small}
			\begin{equation} \label{eq2 in g(xt) inequality}
				E\left(\|\bm{\xi}_{t}\|_{2}^{k} \mid \bm{x}_{t-1} = \bm{x}\right) 
				\leq E\left[(\|\bm{\xi}_{t,1}\|_{2} + \|\bm{\xi}_{t,2}\|_{2})^{k} \mid \bm{x}_{t-1} = \bm{x}\right] 
				\leq E\left[(\sqrt{m(m-1)} + \|\bm{\varepsilon}_{t}\|_{2})^{k} \mid \bm{x}_{t-1} = \bm{x}\right]. 
			\end{equation}
			\end{small}
			Moreover, recall that $\bm{\varepsilon}_{t} = D_{t}^{-1} \mathbf{y}_{t}$ with $\mathbf{y}_{t} = H_{t}^{1/2} \bm{\eta}_{t}$ and $H_{t} = D_{t} R_{t} D_{t}$. 
			It holds that 
			\begin{align} \label{eq epsilont, a function of Rt and etat}
				\bm{\varepsilon}_{t} = D_{t}^{-1} \mathbf{y}_{t} = D_{t}^{-1} H_{t}^{1/2} \bm{\eta}_{t} = D_{t}^{-1} D_{t} R_{t}^{1/2} \bm{\eta}_{t} = R_{t}^{1/2} \bm{\eta}_{t}. 
			\end{align}
			Note that $R_{t}$ is a conditional correlation matrix. 
			Then using the submultiplicativity of induced matrix norms and the fact $\|A\|_{2} \leq \|A\|_{F}$ for a matrix $A$, we can show that 
			\begin{align} \label{eq3 in g(xt) inequality}
				&E\left(\|\bm{\varepsilon}_{t}\|_{2}^{k} \mid \bm{x}_{t-1} = \bm{x}\right) 
				= E\left(\|R_{t}^{1/2} \bm{\eta}_{t}\|_{2}^{k} \mid \bm{x}_{t-1} = \bm{x}\right) 
				\leq E\left(\|R_{t}^{1/2}\|_{2}^{k} \|\bm{\eta}_{t}\|_{2}^{k} \mid \bm{x}_{t-1} = \bm{x}\right) \notag\\
				\leq& E\left(\|R_{t}^{1/2}\|_{F}^{k} \|\bm{\eta}_{t}\|_{2}^{k} \mid \bm{x}_{t-1} = \bm{x}\right) 
				= E\left(\left[\tr(R_{t})\right]^{k/2} \|\bm{\eta}_{t}\|_{2}^{k} \mid \bm{x}_{t-1} = \bm{x}\right) 
				= m^{k/2} E\|\bm{\eta}_{t}\|_{2}^{k}. 
			\end{align}
			By \eqref{eq1 in g(xt) inequality}--\eqref{eq3 in g(xt) inequality} and $E\|\bm{\eta}_{t}\|_{2}^{k} < \infty$, the following inequality holds for some constant $c_{1} > 0$: 
			\begin{align*}
				E\left(g(\bm{x}_{t}) \mid \bm{x}_{t-1} = \bm{x}\right) 
				\leq 1 + |\bm{x}^{\otimes k}|^{\prime} (\Upsilon^{\otimes k})^{\prime} \bm{\nu} + c_{1} (1 + \|\bm{x}\|_{2} + \ldots + \|\bm{x}\|_{2}^{k-1}). 
			\end{align*}
			Let $g_{1}(\bm{x}) = c_{1} (1 + \|\bm{x}\|_{2} + \ldots + \|\bm{x}\|_{2}^{k-1})$. 
			It then follows that 
			\begin{equation} \label{eq in Eg(xt) inequality}
				E\left(g(\bm{x}_{t}) \mid \bm{x}_{t-1} = \bm{x}\right) 
				\leq 1 + |\bm{x}^{\otimes k}|^{\prime} \bm{\nu} - |\bm{x}^{\otimes k}|^{\prime} (I_{a} - \Upsilon^{\otimes k})^{\prime} \bm{\nu} + g_{1}(\bm{x}) 
				= g(\bm{x}) - |\bm{x}^{\otimes k}|^{\prime} (I_{a} - \Upsilon^{\otimes k})^{\prime} \bm{\nu} + g_{1}(\bm{x}), 
			\end{equation}
			which implies that $E\left(g(\bm{x}_{t}) \mid \bm{x}_{t-1} = \bm{x}\right) \leq g(\bm{x}) + g_{1}(\bm{x})$ holds since $(I_{a} - \Upsilon^{\otimes k})^{\prime} \bm{\nu} > 0$ by ($\romannumeral1$). 
			Furthermore, denote $c_{2} = \min\{\text{all components of} \; (I_{a} - \Upsilon^{\otimes k})^{\prime} \bm{\nu}\}$, $c_{3} = \max\{\text{all components of} \; \bm{\nu}\}$ and $c_{4} = \min\{\text{all components of} \; \bm{\nu}\}$. 
			And let $\mathcal{A} = \{\bm{x} \in \mathbb{R}^{a}: \|\bm{x}\|_{1}^{k} \leq \Delta\}$ with $\Delta > \max\{1, 1/c_{3}, [(2 c_{1} c_{3} k) / (c_{2} c_{4})]^{k}\}$. 
			Then when $\bm{x} \in \mathcal{A}^{c} = \mathbb{R}^{a} - \mathcal{A}$, it can be shown that 
			\begin{align*}
				\frac{|\bm{x}^{\otimes k}|^{\prime} (I_{a} - \Upsilon^{\otimes k})^{\prime} \bm{\nu}}{g(\bm{x})} 
				= \frac{|\bm{x}^{\otimes k}|^{\prime} (I_{a} - \Upsilon^{\otimes k})^{\prime} \bm{\nu}}{1 + |\bm{x}^{\otimes k}|^{\prime} \bm{\nu}} 
				\geq \frac{c_{2} \|\bm{x}\|_{1}^{k}}{1 + c_{3} \|\bm{x}\|_{1}^{k}} 
				\geq \frac{c_{2} \|\bm{x}\|_{1}^{k}}{2 c_{3} \|\bm{x}\|_{1}^{k}} 
				= \frac{c_{2}}{2 c_{3}}, 
			\end{align*}
			and 
			\begin{align*}
				\frac{g_{1}(\bm{x})}{g(\bm{x})} 
				&= \frac{c_{1} (1 + \|\bm{x}\|_{2} + \ldots + \|\bm{x}\|_{2}^{k-1})}{1 + |\bm{x}^{\otimes k}|^{\prime} \bm{\nu}} 
				\leq \frac{c_{1} (1 + \|\bm{x}\|_{1} + \ldots + \|\bm{x}\|_{1}^{k-1})}{c_{4} \|\bm{x}\|_{1}^{k}} \\
				&= \frac{c_{1} (\|\bm{x}\|_{1}^{1-k} + \|\bm{x}\|_{1}^{2-k} + \ldots + 1)}{c_{4} \|\bm{x}\|_{1}} 
				\leq \frac{c_{1} k}{c_{4} \Delta^{1/k}}. 
			\end{align*}
			These together with \eqref{eq in Eg(xt) inequality} implies that 
			$$
				E\left(g(\bm{x}_{t}) \mid \bm{x}_{t-1} = \bm{x}\right) 
				\leq g(\bm{x}) \left[1 - \frac{|\bm{x}^{\otimes k}|^{\prime} (I_{a} - \Upsilon^{\otimes k})^{\prime} \bm{\nu}}{g(\bm{x})} + \frac{g_{1}(\bm{x})}{g(\bm{x})}\right] 
				\leq g(\bm{x}) \left(1 - \frac{c_{2}}{2 c_{3}} + \frac{c_{1} k}{c_{4} \Delta^{1/k}}\right). 
			$$
			Thus $E\left(g(\bm{x}_{t}) \mid \bm{x}_{t-1} = \bm{x}\right) \leq (1 - c_{0}) g(\bm{x})$ for $\bm{x} \in \mathcal{A}^{c}$ with $c_{0} = c_{2} / (2 c_{3}) - (c_{1} k) / (c_{4} \Delta^{1/k}) > 0$. 
			As a result, ($\romannumeral2$) is established. 
		\end{proof}

		\begin{proof}[Proof of Lemma \ref{lemma derivatives of lnhtunderline}]
			For $\varrho < \varrho_{1} < 1$, it can be directly shown that 
			\begin{align*}
				\varrho \bm{\zeta}_{t-1,\varrho} &= \varrho + \sum_{i=1}^{\infty} \varrho^{i} \|\ln\mathbf{y}^{2}_{t-1-i}\| = \varrho + \sum_{i=2}^{\infty} \varrho^{i-1} \|\ln\mathbf{y}^{2}_{t-i}\| 
				< 1 + \sum_{i=1}^{\infty} \varrho^{i-1} \|\ln\mathbf{y}^{2}_{t-i}\| = \bm{\zeta}_{t,\varrho} \\
				&< 1 + \sum_{i=1}^{\infty} \varrho_{1}^{i-1} \|\ln\mathbf{y}^{2}_{t-i}\| = \bm{\zeta}_{t,\varrho_{1}}. 
			\end{align*}

			For ($\romannumeral1$), recall that $\ln\bm{h}_{t}(\bm{\delta}) = \underline{\bm{\omega}} + \sum_{i=1}^{\infty} \Phi_{i}(\bm{\kappa}) \ln\mathbf{y}^{2}_{t-i}$ with $\Phi_{i}(\bm{\kappa}) = \sum_{k=1}^{r} \lambda_{k}^{i-1} G_{0,k} + \sum_{k=1}^{s} \gamma_{k}^{i-1} [\cos((i-1) \varphi_{k}) G_{1,k} + \sin((i-1) \varphi_{k}) G_{2,k}]$ by model \eqref{model Dt SGARCH(r,s)}. 
			By Assumptions \ref{assum parameters}($\romannumeral1$)--($\romannumeral2$) and the properties of the induced matrix norm $\|\cdot\|$ that $\|A + B\| \leq \|A\| + \|B\|$ and $\|A \bm{a}\| \leq \|A\| \|\bm{a}\|$ for any vector $\bm{a}$ and matrices $A$ and $B$, we can obtain that 
			$$
				\sup_{\bm{\theta} \in \Theta} \|\Phi_{i}(\bm{\kappa})\| 
				\leq \sup_{\bm{\theta} \in \Theta} \left[\sum_{k=1}^{r} |\lambda_{k}|^{i-1} \|G_{0,k}\| 
				+ \sum_{k=1}^{s} |\gamma_{k}|^{i-1} \left(\|G_{1,k}\| + \|G_{2,k}\|\right)\right] 
				\leq c_{1} \varrho^{i-1}, 
			$$
			and then 
			$$
				\sup_{\bm{\theta} \in \Theta} \left\|\ln\bm{h}_{t}(\bm{\delta})\right\| 
				\leq \sup_{\bm{\theta} \in \Theta} \left(\|\underline{\bm{\omega}}\| + \sum_{i=1}^{\infty} \|\Phi_{i}(\bm{\kappa})\| \|\ln\mathbf{y}^{2}_{t-i}\|\right) 
				\leq c_{2} + c_{1} \sum_{i=1}^{\infty} \varrho^{i-1} \|\ln\mathbf{y}^{2}_{t-i}\| 
				\leq c \bm{\zeta}_{t,\varrho}, 
			$$
			for some constants $c_{1}, c_{2}, c > 0$. 
			Hence ($\romannumeral1$) holds. 

			For ($\romannumeral2$), we only show the result for ${\partial\ln\bm{h}_{t}(\bm{\delta})}/{\partial\lambda_{\ell}}$, and the other results can be verified similarly. 
			By \eqref{eq first derivative of lnhtunderline}, Assumptions \ref{assum parameters}($\romannumeral1$)--($\romannumeral2$), the aforementioned properties of the induced matrix norm $\|\cdot\|$, and the fact $i (\varrho / \varrho_{1})^{i-1} \leq c_{1}$ for $0 < \varrho < \varrho_{1} < 1$, $i \geq 1$ and some $c_{1} > 0$, it can be shown that 
			\begin{align*}
				\sup_{\bm{\theta} \in \Theta} \left\|\frac{\partial\ln\bm{h}_{t}(\bm{\delta})}{\partial\lambda_{\ell}}\right\| 
				&\leq \sup_{\bm{\theta} \in \Theta} \sum_{i=2}^{\infty} (i-1) |\lambda_{\ell}|^{i-2} \|G_{0,\ell}\| \|\ln\mathbf{y}^{2}_{t-i}\| \\
				&\leq c_{2} \sum_{i=2}^{\infty} (i-1) \varrho^{i-2} \|\ln\mathbf{y}^{2}_{t-i}\| 
				= c_{2} \sum_{i=2}^{\infty} (i-1) \left(\frac{\varrho}{\varrho_{1}}\right)^{i-2} \varrho_{1}^{i-2} \|\ln\mathbf{y}^{2}_{t-i}\| \\
				&\leq c_{1}c_{2} \sum_{i=2}^{\infty} \varrho_{1}^{i-2} \|\ln\mathbf{y}^{2}_{t-i}\| 
				= c_{1}c_{2} \sum_{i=1}^{\infty} \varrho_{1}^{i-1} \|\ln\mathbf{y}^{2}_{t-i-1}\| 
				\leq c_{1} c_{2} \bm{\zeta}_{t-1,\varrho_{1}} 
				\leq c \bm{\zeta}_{t,\varrho_{1}}, 
			\end{align*}
			for some constants $c_{1}, c_{2}, c > 0$. 
			As a result, ($\romannumeral2$) is established. 

			For ($\romannumeral3$), we only show the result for ${\partial^{2}\ln\bm{h}_{t}(\bm{\delta})}/{\partial\lambda_{\ell} \partial\lambda_{\ell}}$, and the other results can be verified similarly. 
			By \eqref{eq second derivative of lnhtunderline}, Assumptions \ref{assum parameters}($\romannumeral1$)--($\romannumeral2$), the aforementioned properties of the induced matrix norm $\|\cdot\|$, and the fact $i(i-1) (\varrho / \varrho_{1})^{i-2} \leq c_{1}$ for $0 < \varrho < \varrho_{1} < 1$, $i \geq 2$ and some $c_{1} > 0$, it can be shown that 
			\begin{align*}
				\sup_{\bm{\theta} \in \Theta} \left\|\frac{\partial^{2}\ln\bm{h}_{t}(\bm{\delta})}{\partial\lambda_{\ell} \partial\lambda_{\ell}}\right\| 
				&\leq \sup_{\bm{\theta} \in \Theta} \sum_{i=3}^{\infty} (i-1)(i-2) |\lambda_{\ell}|^{i-3} \|G_{0,\ell}\| \|\ln\mathbf{y}^{2}_{t-i}\| \\
				&\leq c_{2} \sum_{i=3}^{\infty} (i-1)(i-2) \left(\frac{\varrho}{\varrho_{1}}\right)^{i-3} \varrho_{1}^{i-3} \|\ln\mathbf{y}^{2}_{t-i}\| 
				\leq c_{1} c_{2} \sum_{i=3}^{\infty} \varrho_{1}^{i-3} \|\ln\mathbf{y}^{2}_{t-i}\| \\
				&= c_{1} c_{2} \sum_{i=1}^{\infty} \varrho_{1}^{i-1} \|\ln\mathbf{y}^{2}_{t-i-2}\| 
				\leq c_{1} c_{2} \bm{\zeta}_{t-2,\varrho_{1}}
				\leq c \bm{\zeta}_{t,\varrho_{1}},  
			\end{align*}
			for some constants $c_{1}, c_{2}, c > 0$. 
			As a result, ($\romannumeral3$) is established. 

			With analogous arguments in the proofs of ($\romannumeral2$) and ($\romannumeral3$), we can show that ($\romannumeral4$) holds. 
		\end{proof}
	
		\begin{proof}[Proof of Lemma \ref{lemma moments of derivatives of lnhtunderline}]
			By Lemma \ref{lemma derivatives of lnhtunderline}($\romannumeral1$), we have that 
			\begin{align*}
				E \sup_{\bm{\theta} \in \Theta} \left\|\ln\bm{h}_{t}(\bm{\delta})\right\| \leq c E \bm{\zeta}_{t,\varrho} = E \left(1 + \sum_{i=1}^{\infty} \varrho^{i-1} \|\ln\mathbf{y}^{2}_{t-i}\|\right), 
			\end{align*}
			and 
			\begin{align*}
				E \sup_{\bm{\theta} \in \Theta} \left\|\ln\bm{h}_{t}(\bm{\delta})\right\|^{2+\epsilon} \leq c E \bm{\zeta}_{t,\varrho}^{2+\epsilon} = E \left(1 + \sum_{i=1}^{\infty} \varrho^{i-1} \|\ln\mathbf{y}^{2}_{t-i}\|\right)^{2+\epsilon}. 
			\end{align*}
			Note that $0 < \varrho < 1$. 
			It then follows that $E \sup_{\bm{\theta} \in \Theta} \left\|\ln\bm{h}_{t}(\bm{\delta})\right\| < \infty$ holds under the condition $E\|\ln\mathbf{y}^{2}_{t}\| < \infty$, and $E \sup_{\bm{\theta} \in \Theta} \left\|\ln\bm{h}_{t}(\bm{\delta})\right\|^{2+\epsilon} < \infty$ holds under the condition $E\|\ln\mathbf{y}^{2}_{t}\|^{2+\epsilon} < \infty$. 
			Thus ($\romannumeral1$) and ($\romannumeral5$) hold. 
			Similarly, we can show that ($\romannumeral2$)--($\romannumeral4$) and ($\romannumeral6$)--($\romannumeral8$) hold. 
		\end{proof}
	
		\begin{proof}[Proof of Lemma \ref{lemma varrho to the power of t}]
			We first show ($\romannumeral1$). 
			Recall that $\ln\bm{h}_{t}(\bm{\delta}) = \underline{\bm{\omega}} + \sum_{i=1}^{\infty} \Phi_{i}(\bm{\kappa}) \ln\mathbf{y}^{2}_{t-i}$ and $\ln\widetilde{\bm{h}}_{t}(\bm{\delta}) = \underline{\bm{\omega}} + \sum_{i=1}^{t-1} \Phi_{i}(\bm{\kappa}) \ln\mathbf{y}^{2}_{t-i}$. 
			It holds that 
			\begin{align*}
				\ln\widetilde{\bm{h}}_{t}(\bm{\delta}) - \ln\bm{h}_{t}(\bm{\delta}) = - \sum_{i=t}^{\infty} \Phi_{i}(\bm{\kappa}) \ln\mathbf{y}^{2}_{t-i}. 
			\end{align*}
			Then using the properties of the induced matrix norm $\|\cdot\|$ that $\|A + B\| \leq \|A\| + \|B\|$ and $\|A \bm{a}\| \leq \|A\| \|\bm{a}\|$ for any vector $\bm{a}$ and matrices $A$ and $B$, it can be shown that 
			\begin{align} \label{eq ||lnhtunderlinetilde - lnhtunderline||}
				\sup_{\bm{\theta} \in \Theta} \|\ln\widetilde{\bm{h}}_{t}(\bm{\delta}) - \ln\bm{h}_{t}(\bm{\delta})\| 
				\leq \sum_{i=t}^{\infty} \sup_{\bm{\theta} \in \Theta} \|\Phi_{i}(\bm{\kappa})\| \|\ln\mathbf{y}^{2}_{t-i}\|. 
			\end{align}
			Under Assumptions \ref{assum parameters}($\romannumeral1$)--($\romannumeral2$), the following inequality holds for some constant $c_{1} > 0$: 
			\begin{small}
			\begin{align} \label{eq sum of Phi}
				\sum_{i=t}^{\infty} \sup_{\bm{\theta} \in \Theta} \|\Phi_{i}(\bm{\kappa})\| 
				&= \sup_{\bm{\theta} \in \Theta} \sum_{i=t}^{\infty} \left\|\sum_{k=1}^{r} \lambda_{k}^{i-1} G_{0,k} 
				+ \sum_{k=1}^{s} \gamma_{k}^{i-1} \left[\cos((i-1) \varphi_{k}) G_{1,k} + \sin((i-1) \varphi_{k}) G_{2,k}\right]\right\| \notag\\
				&\leq \sup_{\bm{\theta} \in \Theta} \sum_{k=1}^{r} \left(\sum_{i=t}^{\infty} |\lambda_{k}|^{i-1}\right) \|G_{0,k}\| 
				+ \sup_{\bm{\theta} \in \Theta} \sum_{k=1}^{s} \left(\sum_{i=t}^{\infty} |\gamma_{k}|^{i-1}\right) \left(\|G_{1,k}\| + \|G_{2,k}\|\right) \notag\\
				&= \sup_{\bm{\theta} \in \Theta} \sum_{k=1}^{r} \frac{|\lambda_{k}|^{t-1}}{1 - |\lambda_{k}|} \|G_{0,k}\| 
				+ \sup_{\bm{\theta} \in \Theta} \sum_{k=1}^{s} \frac{|\gamma_{k}|^{t-1}}{1 - |\gamma_{k}|} \left(\|G_{1,k}\| + \|G_{2,k}\|\right) \notag\\
				&\leq c_{1} \varrho^{t}. 
			\end{align} 
			\end{small}
			Moreover, since $0 < \varrho < 1$ and $\|\ln\mathbf{y}^{2}_{t}\|$ is finite a.s. by $E\|\ln\mathbf{y}^{2}_{t}\| < \infty$, it holds that $\varrho^{t/2} \|\ln\mathbf{y}^{2}_{t}\| \to 0$ a.s. as $t \to \infty$. 
			Hence $\varrho^{t/2} \|\ln\mathbf{y}^{2}_{t}\| < c_{2}$ a.s. for some $c_{2} > 0$ and large enough $t$. 
			This together with \eqref{eq ||lnhtunderlinetilde - lnhtunderline||}--\eqref{eq sum of Phi} and Assumption \ref{assum stationarity}, implies that 
			\begin{align*} 
				\sup_{\bm{\theta} \in \Theta} \|\ln\widetilde{\bm{h}}_{t}(\bm{\delta}) - \ln\bm{h}_{t}(\bm{\delta})\| 
				&\leq \sum_{i=t}^{\infty} c_{1} \varrho^{i-1} \|\ln\mathbf{y}^{2}_{t-i}\|
				\leq \sum_{i=t}^{\infty} c_{1} \varrho^{i-1-t/2} \varrho^{t/2} \|\ln\mathbf{y}^{2}_{t-i}\| \\
				&\leq c_{1} c_{2} \sum_{i=t}^{\infty} \varrho^{i-1-t/2} 
				\leq c_{1} c_{2} \frac{\varrho^{t/2 - 1}}{1 - \varrho} 
				\leq c \rho^{t} \;\; \text{a.s.} 
			\end{align*}
			holds for large enough $t$ and some constants $0 < \rho < 1$ and $c > 0$. 
			Besides, with analogous arguments, we can obtain that the following inequality holds for large enough $t$, some constants $0 < \rho < 1$ and $c > 0$, and all $1 \leq j \leq m$: 
			\begin{align} \label{eq |lnhjjttilde - lnhjjt| with varrho}
				\sup_{\bm{\theta} \in \Theta} |\ln\widetilde{h}_{jj,t}(\bm{\delta}) - \ln h_{jj,t}(\bm{\delta})| \leq c \rho^{t}. 
			\end{align}

			Next we establish ($\romannumeral2$). 
			Recall that $D_{t}(\bm{\delta}) = \Diag\{h_{11,t}^{1/2}(\bm{\delta}), \ldots, h_{mm,t}^{1/2}(\bm{\delta})\}$ and $\widetilde{D}_{t}(\bm{\delta}) = \Diag\{\widetilde{h}_{11,t}^{1/2}(\bm{\delta}), \ldots, \widetilde{h}_{mm,t}^{1/2}(\bm{\delta})\}$. 
			For all $1 \leq j \leq m$, by $|e^{x} - 1| \leq e^{|x|} - 1$ and $e^{x} - 1 \leq (e^{a} - 1) x$ for $x \in [0,a]$, it can be shown that the following inequalities hold for some constant $c$: 
			\begin{align*}
				&\left|\widetilde{h}_{jj,t}^{1/2}(\bm{\delta}) - h_{jj,t}^{1/2}(\bm{\delta})\right| 
				= \left|\frac{\widetilde{h}_{jj,t}^{1/2}(\bm{\delta})}{h_{jj,t}^{1/2}(\bm{\delta})} - 1\right| \left|h_{jj,t}^{1/2}(\bm{\delta})\right| 
				= \left|e^{\frac{1}{2} \left(\ln \widetilde{h}_{jj,t}(\bm{\delta}) -  \ln h_{jj,t}(\bm{\delta})\right)} - 1\right| \left|h_{jj,t}^{1/2}(\bm{\delta})\right| \\
				\leq& \left(e^{\frac{1}{2} \left|\ln \widetilde{h}_{jj,t}(\bm{\delta}) -  \ln h_{jj,t}(\bm{\delta})\right|} - 1\right) \left|h_{jj,t}^{1/2}(\bm{\delta})\right| 
				\leq c \left|\ln \widetilde{h}_{jj,t}(\bm{\delta}) -  \ln h_{jj,t}(\bm{\delta})\right| \left|h_{jj,t}^{1/2}(\bm{\delta})\right|, 
			\end{align*}
			and 
			\begin{align*}
				&\left|\widetilde{h}_{jj,t}^{-1/2}(\bm{\delta}) - h_{jj,t}^{-1/2}(\bm{\delta})\right| 
				= \left|\frac{\widetilde{h}_{jj,t}^{-1/2}(\bm{\delta})}{h_{jj,t}^{-1/2}(\bm{\delta})} - 1\right| \left|h_{jj,t}^{-1/2}(\bm{\delta})\right| 
				= \left|e^{\frac{1}{2} \left(\ln h_{jj,t}(\bm{\delta}) - \ln \widetilde{h}_{jj,t}(\bm{\delta})\right)} - 1\right| \left|h_{jj,t}^{-1/2}(\bm{\delta})\right| \\
				\leq& \left(e^{\frac{1}{2} \left|\ln \widetilde{h}_{jj,t}(\bm{\delta}) -  \ln h_{jj,t}(\bm{\delta})\right|} - 1\right) \left|h_{jj,t}^{-1/2}(\bm{\delta})\right| 
				\leq c \left|\ln \widetilde{h}_{jj,t}(\bm{\delta}) -  \ln h_{jj,t}(\bm{\delta})\right| \left|h_{jj,t}^{-1/2}(\bm{\delta})\right|. 
			\end{align*}
			Note that $\sup_{\bm{\theta} \in \Theta} h_{jj,t}(\bm{\delta}) = \sup_{\bm{\theta} \in \Theta} \exp\{\ln h_{jj,t}(\bm{\delta})\}$ and $\sup_{\bm{\theta} \in \Theta} h_{jj,t}^{-1}(\bm{\delta}) = \sup_{\bm{\theta} \in \Theta} \exp\{- \ln h_{jj,t}(\bm{\delta})\}$ are finite a.s. by Lemma \ref{lemma moments of derivatives of lnhtunderline}($\romannumeral1$). 
			These together with \eqref{eq |lnhjjttilde - lnhjjt| with varrho} and Assumption \ref{assum stationarity}, imply that the following inequalities hold for large enough $t$, some constants $0 < \rho < 1$ and $c > 0$, and all $1 \leq j \leq m$: 
			\begin{align*}
				\sup_{\bm{\theta} \in \Theta} |\widetilde{h}_{jj,t}^{1/2}(\bm{\delta}) - h_{jj,t}^{1/2}(\bm{\delta})| \leq c \rho^{t} \;\; \text{a.s.} \;\; \text{and} \;\; 
				\sup_{\bm{\theta} \in \Theta} |\widetilde{h}_{jj,t}^{-1/2}(\bm{\delta}) - h_{jj,t}^{-1/2}(\bm{\delta})| \leq c \rho^{t} \;\; \text{a.s.}. 
			\end{align*}
			It then follows that 
			\begin{align*} 
				\sup_{\bm{\theta} \in \Theta} \left\|\widetilde{D}_{t}(\bm{\delta}) - D_{t}(\bm{\delta})\right\| \leq c \rho^{t} \;\; \text{a.s.} \;\; \text{and} \;\; 
				\sup_{\bm{\theta} \in \Theta} \left\|\widetilde{D}_{t}^{-1}(\bm{\delta}) - D_{t}^{-1}(\bm{\delta})\right\| \leq c \rho^{t} \;\; \text{a.s.} 
			\end{align*}
			hold for large enough $t$ and some constants $0 < \rho < 1$ and $c > 0$. 
			Moreover, both $\sup_{\bm{\theta} \in \Theta} \|D_{t}(\bm{\delta})\|$ and $\sup_{\bm{\theta} \in \Theta} \|D_{t}^{-1}(\bm{\delta})\|$ are finite a.s. since both $\sup_{\bm{\theta} \in \Theta} h_{jj,t}(\bm{\delta})$ and $\sup_{\bm{\theta} \in \Theta} h_{jj,t}^{-1}(\bm{\delta})$ are finite a.s. for all $1 \leq j \leq m$. 

			We then verify ($\romannumeral3$). 
			Recall that $R_{t}(\bm{\theta}) = (1 - \beta_{1} - \beta_{2}) \underline{R} + \beta_{1} \Psi_{t-1}(\bm{\delta}) + \beta_{2} R_{t-1}(\bm{\theta})$ and $\widetilde{R}_{t}(\bm{\theta}) = (1 - \beta_{1} - \beta_{2}) \underline{R} + \beta_{1} \widetilde{\Psi}_{t-1}(\bm{\delta}) + \beta_{2} \widetilde{R}_{t-1}(\bm{\theta})$. 
			Since $R_{t}$ is a conditional correlation matrix with all elements between $-1$ and 1, we have that $\sup_{\bm{\theta} \in \Theta} \|R_{t}(\bm{\theta})\|$ is finite. 
			Moreover, note that $R_{t}(\bm{\theta})$ is positive definite under Assumption \ref{assum parameters}($\romannumeral3$) and conditions $0 <\beta_{1}, \beta_{2} < 1$ and $0 < \beta_{1} + \beta_{2} < 1$. 
			Then using Theorem 8.4.9 of \citet{Bernstein2009} and Assumption \ref{assum parameters}($\romannumeral3$), we can obtain that 
			$$
				\sup_{\bm{\theta} \in \Theta} \|R_{t}^{-1}(\bm{\theta})\|_{2} 
				= \sup_{\bm{\theta} \in \Theta} \lambda_{\max}(R_{t}^{-1}(\bm{\theta})) 
				= \sup_{\bm{\theta} \in \Theta} \lambda_{\min}^{-1}(R_{t}(\bm{\theta})) 
				\leq \sup_{\bm{\theta} \in \Theta} \lambda_{\min}^{-1}((1 - \beta_{1} - \beta_{2}) \underline{R}) 
				< \infty, 
			$$
			which implies that $\sup_{\bm{\theta} \in \Theta} \|R_{t}^{-1}(\bm{\theta})\|$ is finite for any induced matrix norm $\|\cdot\|$ as the dimension is fixed. 
			Furthermore, using the property of the induced matrix norm $\|\cdot\|$ that $\|A B\| \leq \|A\| \|B\|$ for any matrices $A$ and $B$, it holds that 
			\begin{align*}
				\sup_{\bm{\theta} \in \Theta} \|\widetilde{R}_{t}^{-1}(\bm{\theta}) - R_{t}^{-1}(\bm{\theta})\| 
				&= \sup_{\bm{\theta} \in \Theta} \|\widetilde{R}_{t}^{-1}(\bm{\theta}) (\widetilde{R}_{t}(\bm{\theta}) - R_{t}(\bm{\theta})) R_{t}^{-1}(\bm{\theta})\| \\
				&\leq \sup_{\bm{\theta} \in \Theta} \|\widetilde{R}_{t}^{-1}(\bm{\theta})\| \sup_{\bm{\theta} \in \Theta} \|\widetilde{R}_{t}(\bm{\theta}) - R_{t}(\bm{\theta})\| \sup_{\bm{\theta} \in \Theta} \|R_{t}^{-1}(\bm{\theta})\|.  
			\end{align*}
			Hence we are left to show that the following inequality holds for large enough $t$ and some constants $0 < \rho < 1$ and $c > 0$: 
			\begin{align*}
				\sup_{\bm{\theta} \in \Theta} \|\widetilde{R}_{t}(\bm{\theta}) - R_{t}(\bm{\theta})\| \leq c \rho^{t} \;\; \text{a.s.}. 
			\end{align*}
			Recall that $\Psi_{t-1}(\bm{\delta})$ and $\widetilde{\Psi}_{t-1}(\bm{\delta})$ are the sample correlation matrices of $\{\bm{\varepsilon}_{t-1}(\bm{\delta}), \ldots, \bm{\varepsilon}_{t-\Bbbk}(\bm{\delta})\}$ and $\{\widetilde{\bm{\varepsilon}}_{t-1}(\bm{\delta}), \ldots, \widetilde{\bm{\varepsilon}}_{t-\Bbbk}(\bm{\delta})\}$, respectively, with $\bm{\varepsilon}_{t}(\bm{\delta}) = D_{t}^{-1}(\bm{\delta}) \mathbf{y}_{t}$ and $\widetilde{\bm{\varepsilon}}_{t}(\bm{\delta}) = \widetilde{D}_{t}^{-1}(\bm{\delta}) \mathbf{y}_{t}$. 
			By ($\romannumeral1$), Assumption \ref{assum stationarity} and $\|\mathbf{y}_{t}\|$ is finite a.s. under the condition $E\|\ln\mathbf{y}^{2}_{t}\| < \infty$, it can be shown that 
			$$
				\sup_{\bm{\theta} \in \Theta} \|\widetilde{\bm{\varepsilon}}_{t}(\bm{\delta}) - \bm{\varepsilon}_{t}(\bm{\delta})\| 
				= \sup_{\bm{\theta} \in \Theta} \|\widetilde{D}_{t}^{-1}(\bm{\delta}) \mathbf{y}_{t} - D_{t}^{-1}(\bm{\delta}) \mathbf{y}_{t}\| 
				\leq \sup_{\bm{\theta} \in \Theta} \|\widetilde{D}_{t}^{-1}(\bm{\delta}) - D_{t}^{-1}(\bm{\delta})\| \|\mathbf{y}_{t}\| 
				\leq c \rho^{t} \;\; \text{a.s.} 
			$$
			holds for large enough $t$ and some constants $0 < \rho < 1$ and $c > 0$, 
			and 
			\begin{align*}
				\sup_{\bm{\theta} \in \Theta} \|\bm{\varepsilon}_{t}(\bm{\delta})\| 
				= \sup_{\bm{\theta} \in \Theta} \|D_{t}^{-1}(\bm{\delta}) \mathbf{y}_{t}\| 
				\leq \sup_{\bm{\theta} \in \Theta} \|D_{t}^{-1}(\bm{\delta})\| \|\mathbf{y}_{t}\| \;\; \text{is finite a.s.}. 
			\end{align*}
			Note that $\Psi_{t-1}(\bm{\delta})$ (or $\widetilde{\Psi}_{t-1}(\bm{\delta})$) is a continuous differentiable function of $\{\bm{\varepsilon}_{t-1}(\bm{\delta}), \ldots, \bm{\varepsilon}_{t-\Bbbk}(\bm{\delta})\}$ (or $\{\widetilde{\bm{\varepsilon}}_{t-1}(\bm{\delta}), \ldots, \widetilde{\bm{\varepsilon}}_{t-\Bbbk}(\bm{\delta})\}$). 
			Thus by the mean value theorem, we can obtain that 
			\begin{align*}
				\sup_{\bm{\theta} \in \Theta} \|\widetilde{\Psi}_{t-1}(\bm{\delta}) - \Psi_{t-1}(\bm{\delta})\| \leq c \rho^{t} \;\; \text{a.s.} 
			\end{align*}
			holds for large enough $t$ and some constants $0 < \rho < 1$ and $c > 0$. 
			This together with Assumption \ref{assum parameters}($\romannumeral2$) implies that 
			\begin{align*}
				\sup_{\bm{\theta} \in \Theta} \|\widetilde{R}_{t}(\bm{\theta}) - R_{t}(\bm{\theta})\| 
				&= \sup_{\bm{\theta} \in \Theta} \left\|\beta_{1} \sum_{i=0}^{\infty} \beta_{2}^{i} \left(\widetilde{\Psi}_{t-i-1}(\bm{\delta}) - \Psi_{t-i-1}(\bm{\delta})\right)\right\| \\
				&\leq \beta_{1} \sum_{i=0}^{\infty} \beta_{2}^{i} \sup_{\bm{\theta} \in \Theta} \|\widetilde{\Psi}_{t-1}(\bm{\delta}) - \Psi_{t-1}(\bm{\delta})\| 
				\leq \frac{\beta_{1}}{1 - \beta_{2}} c \rho^{t} \;\; \text{a.s.} 
			\end{align*}
			holds for large enough $t$ and some constants $0 < \rho < 1$ and $c > 0$. 
			As a result, ($\romannumeral3$) is established. 

			For ($\romannumeral4$), recall that $H_{t}(\bm{\theta}) = D_{t}(\bm{\delta}) R_{t}(\bm{\theta}) D_{t}(\bm{\delta})$ and $\widetilde{H}_{t}(\bm{\theta}) = \widetilde{D}_{t}(\bm{\delta}) \widetilde{R}_{t}(\bm{\theta}) \widetilde{D}_{t}(\bm{\delta})$. 
			Then using the aforementioned properties of the induced matrix norm $\|\cdot\|$, it holds that $\|H_{t}(\bm{\theta})\| \leq \|D_{t}(\bm{\delta})\|^{2} \|R_{t}(\bm{\theta})\|$, $\|H_{t}^{-1}(\bm{\theta})\| \leq \|D_{t}^{-1}(\bm{\delta})\|^{2} \|R_{t}^{-1}(\bm{\theta})\|$, 
			\begin{align*}
				&\|\widetilde{H}_{t}(\bm{\theta}) - H_{t}(\bm{\theta})\| 
				= \|\widetilde{D}_{t}(\bm{\delta}) \widetilde{R}_{t}(\bm{\theta}) \widetilde{D}_{t}(\bm{\delta}) - D_{t}(\bm{\delta}) R_{t}(\bm{\theta}) D_{t}(\bm{\delta})\| \\
				\leq& \|\widetilde{D}_{t}(\bm{\delta}) - D_{t}(\bm{\delta})\| \|\widetilde{R}_{t}(\bm{\theta})\| \|\widetilde{D}_{t}(\bm{\delta})\| 
				+ \|D_{t}(\bm{\delta})\| \|\widetilde{R}_{t}(\bm{\theta}) - R_{t}(\bm{\theta})\| \|\widetilde{D}_{t}(\bm{\delta})\| \\
				&+ \|D_{t}(\bm{\delta})\| \|R_{t}(\bm{\theta})\| \|\widetilde{D}_{t}(\bm{\delta}) - D_{t}(\bm{\delta})\|, \;\; \text{and} \\
				&\|\widetilde{H}_{t}^{-1}(\bm{\theta}) - H_{t}^{-1}(\bm{\theta})\| 
				= \|\widetilde{D}_{t}^{-1}(\bm{\delta}) \widetilde{R}_{t}^{-1}(\bm{\theta}) \widetilde{D}_{t}^{-1}(\bm{\delta}) - D_{t}^{-1}(\bm{\delta}) R_{t}^{-1}(\bm{\theta}) D_{t}^{-1}(\bm{\delta})\| \\
				\leq& \|\widetilde{D}_{t}^{-1}(\bm{\delta}) - D_{t}^{-1}(\bm{\delta})\| \|\widetilde{R}_{t}^{-1}(\bm{\theta})\| \|\widetilde{D}_{t}^{-1}(\bm{\delta})\| 
				+ \|D_{t}^{-1}(\bm{\delta})\| \|\widetilde{R}_{t}^{-1}(\bm{\theta}) - R_{t}^{-1}(\bm{\theta})\| \|\widetilde{D}_{t}^{-1}(\bm{\delta})\| \\
				&+ \|D_{t}^{-1}(\bm{\delta})\| \|R_{t}^{-1}(\bm{\theta})\| \|\widetilde{D}_{t}^{-1}(\bm{\delta}) - D_{t}^{-1}(\bm{\delta})\|. 
			\end{align*}
			These together with ($\romannumeral2$)--($\romannumeral3$) and Assumption \ref{assum stationarity}, imply that  
			$\sup_{\bm{\theta} \in \Theta} \|\widetilde{H}_{t}(\bm{\theta}) - H_{t}(\bm{\theta})\| \leq c \rho^{t}$ a.s. and $\sup_{\bm{\theta} \in \Theta} \|\widetilde{H}_{t}^{-1}(\bm{\theta}) - H_{t}^{-1}(\bm{\theta})\| \leq c \rho^{t}$ a.s. hold for large enough $t$ and some constants $0 < \rho < 1$ and $c > 0$, and both $\sup_{\bm{\theta} \in \Theta} \|H_{t}(\bm{\theta})\|$ and $\sup_{\bm{\theta} \in \Theta} \|H_{t}^{-1}(\bm{\theta})\|$ are finite a.s.. 

			For ($\romannumeral5$), 
			it holds that 
			\begin{align*}
				\ln |\widetilde{H}_{t}(\bm{\theta})| - \ln |H_{t}(\bm{\theta})| 
				= \ln |\widetilde{H}_{t}(\bm{\theta}) H_{t}^{-1}(\bm{\theta})| 
				= \ln |I_{m} + (\widetilde{H}_{t}(\bm{\theta}) - H_{t}(\bm{\theta})) H_{t}^{-1}(\bm{\theta})|. 
			\end{align*}
			Since $|A^{\prime}A| \leq \|A\|_{2}^{2 \min\{d_{1}, d_{2}\}}$ holds for any $d_{1} \times d_{2}$ matrix $A$, 
			$\ln|A| \leq c \ln\|A\|$ holds for any fixed dimensional square matrix $A$, any matrix norm induced by a vector norm and some constant $c > 0$. 
			These together with the aforementioned properties of the induced matrix norm $\|\cdot\|$ and the fact $\ln(1+x) \leq x$ for $x > -1$, imply that 
			\begin{align} \label{eq ln|Httilde| - ln|Ht|}
				\ln |\widetilde{H}_{t}(\bm{\theta})| - \ln |H_{t}(\bm{\theta})| 
				&\leq c \ln \|I_{m} + (\widetilde{H}_{t}(\bm{\theta}) - H_{t}(\bm{\theta})) H_{t}^{-1}(\bm{\theta})\| \notag\\
				&\leq c \ln\left(\|I_{m}\| + \|\widetilde{H}_{t}(\bm{\theta}) - H_{t}(\bm{\theta})\| \|H_{t}^{-1}(\bm{\theta})\|\right) \notag\\
				&\leq c \|\widetilde{H}_{t}(\bm{\theta}) - H_{t}(\bm{\theta})\| \|H_{t}^{-1}(\bm{\theta})\| 
			\end{align}
			holds for some constant $c > 0$. 
			Thus ($\romannumeral5$) holds by \eqref{eq ln|Httilde| - ln|Ht|}, ($\romannumeral4$) and Assumption \ref{assum stationarity}. 

			Lastly we show ($\romannumeral6$) and ($\romannumeral7$). 
			Note that for all $1 \leq i,j \leq m$, it holds that $|\Psi_{t-1,ij}(\bm{\delta})| \leq 1$, ${\varepsilon_{i,t-1}^{2}(\bm{\delta})}/{(\sum_{h=1}^{\Bbbk} \varepsilon_{i,t-h}^{2}(\bm{\delta}))} \leq 1$ and $\sup_{\bm{\theta} \in \Theta} |{\partial\ln h_{ii,t-1}(\bm{\delta})}/{\partial\delta_{\ell}}|$ is finite a.s. by Lemma \ref{lemma moments of derivatives of lnhtunderline}($\romannumeral2$). Then we have that $\sup_{\bm{\theta} \in \Theta} |{\partial\Psi_{t-1,ij}(\bm{\delta})}/{\partial\delta_{\ell}}|$ is finite a.s.. 
			As a result, with analogous arguments in the proof of ($\romannumeral1$)--($\romannumeral5$), we can establish ($\romannumeral6$) and ($\romannumeral7$). 
		\end{proof}
	
		\begin{proof}[Proof of Lemma \ref{lemma (ii) for consistency}]
			We first show ($\romannumeral1$). 
			By \eqref{eq mathcalLntilde}--\eqref{eq mathcalLn}, it holds that 
			\begin{equation} \label{eq |mathcalLn - mathcalLntilde|}
				\sup_{\bm{\theta} \in \Theta} |\mathcal{L}_{n}(\bm{\theta}) - \widetilde{\mathcal{L}}_{n}(\bm{\theta})| 
				\leq \frac{1}{2n} \sum_{t=1}^{n} \sup_{\bm{\theta} \in \Theta} \left|\mathbf{y}_{t}^{\prime} \left(H_{t}^{-1}(\bm{\theta}) - \widetilde{H}_{t}^{-1}(\bm{\theta})\right) \mathbf{y}_{t}\right| 
				+ \frac{1}{2n} \sum_{t=1}^{n} \sup_{\bm{\theta} \in \Theta} \left|\ln |H_{t}(\bm{\theta})| - \ln |\widetilde{H}_{t}(\bm{\theta})|\right|. 
			\end{equation}
			Hence it suffices to show that the following results hold for large enough $t$ and some constants $0 < \rho < 1$ and $c > 0$: 
			\begin{enumerate}
				\item[(a1)] $\sup_{\bm{\theta} \in \Theta} |\mathbf{y}_{t}^{\prime} (H_{t}^{-1}(\bm{\theta}) - \widetilde{H}_{t}^{-1}(\bm{\theta})) \mathbf{y}_{t}| \leq c \rho^{t}$ a.s.; 
				\item[(a2)] $\sup_{\bm{\theta} \in \Theta} |\ln |H_{t}(\bm{\theta})| - \ln |\widetilde{H}_{t}(\bm{\theta})|| \leq c \rho^{t}$ a.s.. 
			\end{enumerate}
			For (a1), using the facts $\tr(AB) \leq \|A\|_{F} \|B\|_{F}$ and $\|A\|_{F} \leq \sqrt{c_{1}} \|A\|_{2}$ for any matrices $A$ and $B$ with $c_{1} = \rank(A)$, it can be shown that 
			\begin{align*}
				&\sup_{\bm{\theta} \in \Theta} \left|\mathbf{y}_{t}^{\prime} \left(H_{t}^{-1}(\bm{\theta}) - \widetilde{H}_{t}^{-1}(\bm{\theta})\right) \mathbf{y}_{t}\right| 
				= \sup_{\bm{\theta} \in \Theta} \left|\tr\left[\left(H_{t}^{-1}(\bm{\theta}) - \widetilde{H}_{t}^{-1}(\bm{\theta})\right) \mathbf{y}_{t} \mathbf{y}_{t}^{\prime}\right]\right| \\
				\leq& \sup_{\bm{\theta} \in \Theta} \left\|H_{t}^{-1}(\bm{\theta}) - \widetilde{H}_{t}^{-1}(\bm{\theta})\right\|_{F} \left\|\mathbf{y}_{t} \mathbf{y}_{t}^{\prime}\right\|_{F} 
				\leq \sqrt{m} \sup_{\bm{\theta} \in \Theta} \left\|H_{t}^{-1}(\bm{\theta}) - \widetilde{H}_{t}^{-1}(\bm{\theta})\right\|_{2} \left\|\mathbf{y}_{t} \mathbf{y}_{t}^{\prime}\right\|_{F}. 
			\end{align*}
			This together with $\sup_{\bm{\theta} \in \Theta} \|\widetilde{H}_{t}^{-1}(\bm{\theta}) - H_{t}^{-1}(\bm{\theta})\|_{2} \leq c \rho^{t}$ a.s. for large enough $t$ and some constants $0 < \rho < 1$ and $c > 0$ by Lemma \ref{lemma varrho to the power of t}(\romannumeral4) and $\left\|\mathbf{y}_{t} \mathbf{y}_{t}^{\prime}\right\|_{F}$ is finite a.s. under the condition $E\|\ln\mathbf{y}^{2}_{t}\|_{2} < \infty$, implies that (a1) holds. 
			Moreover, (a2) holds by Lemma \ref{lemma varrho to the power of t}(\romannumeral5). 
			As a result, $\lim_{n \to \infty} \sup_{\bm{\theta} \in \Theta} |\mathcal{L}_{n}(\bm{\theta}) - \widetilde{\mathcal{L}}_{n}(\bm{\theta})| = 0$ a.s. follows from \eqref{eq |mathcalLn - mathcalLntilde|} and (a1)--(a2). 

			Next we show ($\romannumeral2$). 
			Recall that $\bm{\theta} = (\bm{\delta}^{\prime}, \bm{\beta}^{\prime})^{\prime}$. 
			Denote $\delta_{\ell}$ as the $\ell$-th element of $\bm{\delta}$ and $\beta_{\ell}$ as the $\ell$-th element of $\bm{\beta}$. 
			By \eqref{eq first derivative of Ht}, \eqref{eq first derivative of lt} and the fact $\tr(AB) \leq \|A\|_{F} \|B\|_{F}$ for any matrices $A$ and $B$, it can be shown that 
			\begin{small}
			\begin{align} \label{eq difference of the first derivative of lt wrt delta}
				&2 \left|\frac{\partial\ell_{t}(\bm{\theta})}{\partial\delta_{\ell}} - \frac{\partial\widetilde{\ell}_{t}(\bm{\theta})}{\partial\delta_{\ell}}\right| \notag\\
				\leq& \left|\mathbf{y}_{t}^{\prime} \left(H_{t}^{-1}(\bm{\theta}) \frac{\partial H_{t}(\bm{\theta})}{\partial\delta_{\ell}} H_{t}^{-1}(\bm{\theta}) - \widetilde{H}_{t}^{-1}(\bm{\theta}) \frac{\partial\widetilde{H}_{t}(\bm{\theta})}{\partial\delta_{\ell}} \widetilde{H}_{t}^{-1}(\bm{\theta})\right) \mathbf{y}_{t}\right| \notag\\
				&+ \left|\tr\left(H_{t}^{-1}(\bm{\theta}) \frac{\partial H_{t}(\bm{\theta})}{\partial\delta_{\ell}} - \widetilde{H}_{t}^{-1}(\bm{\theta}) \frac{\partial\widetilde{H}_{t}(\bm{\theta})}{\partial\delta_{\ell}}\right)\right| \notag\\
				\leq& 2 \left|\tr\left[\left(D_{t}^{-1}(\bm{\delta}) \frac{\partial D_{t}(\bm{\delta})}{\partial\delta_{\ell}} H_{t}^{-1}(\bm{\theta}) - \widetilde{D}_{t}^{-1}(\bm{\delta}) \frac{\partial\widetilde{D}_{t}(\bm{\delta})}{\partial\delta_{\ell}} \widetilde{H}_{t}^{-1}(\bm{\theta})\right) \mathbf{y}_{t} \mathbf{y}_{t}^{\prime}\right]\right| \notag\\
				&+ \left|\tr\left[\left(D_{t}^{-1}(\bm{\delta}) R_{t}^{-1}(\bm{\theta}) \frac{\partial R_{t}(\bm{\theta})}{\partial\delta_{\ell}} R_{t}^{-1}(\bm{\theta}) D_{t}^{-1}(\bm{\delta}) - \widetilde{D}_{t}^{-1}(\bm{\delta}) \widetilde{R}_{t}^{-1}(\bm{\theta}) \frac{\partial\widetilde{R}_{t}(\bm{\theta})}{\partial\delta_{\ell}} \widetilde{R}_{t}^{-1}(\bm{\theta}) \widetilde{D}_{t}^{-1}(\bm{\delta})\right) \mathbf{y}_{t} \mathbf{y}_{t}^{\prime}\right]\right| \notag\\
				&+ 2 \left|\tr\left(D_{t}^{-1}(\bm{\delta}) \frac{\partial D_{t}(\bm{\delta})}{\partial\delta_{\ell}} - \widetilde{D}_{t}^{-1}(\bm{\delta}) \frac{\partial\widetilde{D}_{t}(\bm{\delta})}{\partial\delta_{\ell}}\right)\right| 
				+ \left|\tr\left(R_{t}^{-1}(\bm{\theta}) \frac{\partial R_{t}(\bm{\theta})}{\partial\delta_{\ell}} - \widetilde{R}_{t}^{-1}(\bm{\theta}) \frac{\partial\widetilde{R}_{t}(\bm{\theta})}{\partial\delta_{\ell}}\right)\right| \notag\\
				\leq& 2 \left\|D_{t}^{-1}(\bm{\delta}) \frac{\partial D_{t}(\bm{\delta})}{\partial\delta_{\ell}} H_{t}^{-1}(\bm{\theta}) - \widetilde{D}_{t}^{-1}(\bm{\delta}) \frac{\partial\widetilde{D}_{t}(\bm{\delta})}{\partial\delta_{\ell}} \widetilde{H}_{t}^{-1}(\bm{\theta})\right\|_{F} \left\|\mathbf{y}_{t} \mathbf{y}_{t}^{\prime}\right\|_{F} \notag\\
				&+ \left\|D_{t}^{-1}(\bm{\delta}) R_{t}^{-1}(\bm{\theta}) \frac{\partial R_{t}(\bm{\theta})}{\partial\delta_{\ell}} R_{t}^{-1}(\bm{\theta}) D_{t}^{-1}(\bm{\delta}) - \widetilde{D}_{t}^{-1}(\bm{\delta}) \widetilde{R}_{t}^{-1}(\bm{\theta}) \frac{\partial\widetilde{R}_{t}(\bm{\theta})}{\partial\delta_{\ell}} \widetilde{R}_{t}^{-1}(\bm{\theta}) \widetilde{D}_{t}^{-1}(\bm{\delta})\right\|_{F} \left\|\mathbf{y}_{t} \mathbf{y}_{t}^{\prime}\right\|_{F} \notag\\
				&+ 2m \left\|D_{t}^{-1}(\bm{\delta}) \frac{\partial D_{t}(\bm{\delta})}{\partial\delta_{\ell}} - \widetilde{D}_{t}^{-1}(\bm{\delta}) \frac{\partial\widetilde{D}_{t}(\bm{\delta})}{\partial\delta_{\ell}}\right\|_{F} 
				+ m \left\|R_{t}^{-1}(\bm{\theta}) \frac{\partial R_{t}(\bm{\theta})}{\partial\delta_{\ell}} - \widetilde{R}_{t}^{-1}(\bm{\theta}) \frac{\partial\widetilde{R}_{t}(\bm{\theta})}{\partial\delta_{\ell}}\right\|_{F}, 
			\end{align}
			\end{small}
			and 
			\begin{small}
			\begin{align} \label{eq difference of the first derivative of lt wrt beta}
				&2 \left|\frac{\partial\ell_{t}(\bm{\theta})}{\partial\beta_{\ell}} - \frac{\partial\widetilde{\ell}_{t}(\bm{\theta})}{\partial\beta_{\ell}}\right| \notag\\
				\leq& \left|\mathbf{y}_{t}^{\prime} \left(H_{t}^{-1}(\bm{\theta}) \frac{\partial H_{t}(\bm{\theta})}{\partial\beta_{\ell}} H_{t}^{-1}(\bm{\theta}) - \widetilde{H}_{t}^{-1}(\bm{\theta}) \frac{\partial\widetilde{H}_{t}(\bm{\theta})}{\partial\beta_{\ell}} \widetilde{H}_{t}^{-1}(\bm{\theta})\right) \mathbf{y}_{t}\right| \notag\\
				&+ \left|\tr\left(H_{t}^{-1}(\bm{\theta}) \frac{\partial H_{t}(\bm{\theta})}{\partial\beta_{\ell}} - \widetilde{H}_{t}^{-1}(\bm{\theta}) \frac{\partial\widetilde{H}_{t}(\bm{\theta})}{\partial\beta_{\ell}}\right)\right| \notag\\
				\leq& \left|\tr\left[\left(D_{t}^{-1}(\bm{\delta}) R_{t}^{-1}(\bm{\theta}) \frac{\partial R_{t}(\bm{\theta})}{\partial\beta_{\ell}} R_{t}^{-1}(\bm{\theta}) D_{t}^{-1}(\bm{\delta}) - \widetilde{D}_{t}^{-1}(\bm{\delta}) \widetilde{R}_{t}^{-1}(\bm{\theta}) \frac{\partial\widetilde{R}_{t}(\bm{\theta})}{\partial\beta_{\ell}} \widetilde{R}_{t}^{-1}(\bm{\theta}) \widetilde{D}_{t}^{-1}(\bm{\delta})\right) \mathbf{y}_{t} \mathbf{y}_{t}^{\prime}\right]\right| \notag\\
				&+ \left|\tr\left(R_{t}^{-1}(\bm{\theta}) \frac{\partial R_{t}(\bm{\theta})}{\partial\beta_{\ell}} - \widetilde{R}_{t}^{-1}(\bm{\theta}) \frac{\partial\widetilde{R}_{t}(\bm{\theta})}{\partial\beta_{\ell}}\right)\right| \notag\\
				\leq& \left\|D_{t}^{-1}(\bm{\delta}) R_{t}^{-1}(\bm{\theta}) \frac{\partial R_{t}(\bm{\theta})}{\partial\beta_{\ell}} R_{t}^{-1}(\bm{\theta}) D_{t}^{-1}(\bm{\delta}) - \widetilde{D}_{t}^{-1}(\bm{\delta}) \widetilde{R}_{t}^{-1}(\bm{\theta}) \frac{\partial\widetilde{R}_{t}(\bm{\theta})}{\partial\beta_{\ell}} \widetilde{R}_{t}^{-1}(\bm{\theta}) \widetilde{D}_{t}^{-1}(\bm{\delta})\right\|_{F} \left\|\mathbf{y}_{t} \mathbf{y}_{t}^{\prime}\right\|_{F} \notag\\
				&+ m \left\|R_{t}^{-1}(\bm{\theta}) \frac{\partial R_{t}(\bm{\theta})}{\partial\beta_{\ell}} - \widetilde{R}_{t}^{-1}(\bm{\theta}) \frac{\partial\widetilde{R}_{t}(\bm{\theta})}{\partial\beta_{\ell}}\right\|_{F}. 
			\end{align}
			\end{small}
			Note that by \eqref{eq first derivative of Dt} and the facts $\|A + B\|_{F} \leq \|A\|_{F} + \|B\|_{F}$ and $\|A B\|_{F} \leq \|A\|_{F} \|B\|_{F}$ for any matrices $A$ and $B$, it holds that 
			\begin{small}
			\begin{align}
				&\left\|D_{t}^{-1}(\bm{\delta}) \frac{\partial D_{t}(\bm{\delta})}{\partial\delta_{\ell}}\right\|_{F} 
				= \frac{1}{2} \left\|\Diag\left\{\frac{\partial\ln\bm{h}_{t}(\bm{\delta})}{\partial\delta_{\ell}}\right\}\right\|_{F},  \label{eq 1 in difference of the first derivative of lt}\\ 
				&\left\|D_{t}^{-1}(\bm{\delta}) \frac{\partial D_{t}(\bm{\delta})}{\partial\delta_{\ell}} - \widetilde{D}_{t}^{-1}(\bm{\delta}) \frac{\partial\widetilde{D}_{t}(\bm{\delta})}{\partial\delta_{\ell}}\right\|_{F} 
				= \frac{1}{2} \left\|\Diag\left\{\frac{\partial\ln\bm{h}_{t}(\bm{\delta})}{\partial\delta_{\ell}}\right\} - \Diag\left\{\frac{\partial\ln\bm{h}_{t}(\bm{\delta})}{\partial\delta_{\ell}}\right\}\right\|_{F},  \label{eq 2 in difference of the first derivative of lt}\\
				&\left\|D_{t}^{-1}(\bm{\delta}) \frac{\partial D_{t}(\bm{\delta})}{\partial\delta_{\ell}} H_{t}^{-1}(\bm{\theta}) - \widetilde{D}_{t}^{-1}(\bm{\delta}) \frac{\partial\widetilde{D}_{t}(\bm{\delta})}{\partial\delta_{\ell}} \widetilde{H}_{t}^{-1}(\bm{\theta})\right\|_{F} \notag\\
				\leq& \left\|\left(D_{t}^{-1}(\bm{\delta}) \frac{\partial D_{t}(\bm{\delta})}{\partial\delta_{\ell}} - \widetilde{D}_{t}^{-1}(\bm{\delta}) \frac{\partial\widetilde{D}_{t}(\bm{\delta})}{\partial\delta_{\ell}}\right) H_{t}^{-1}(\bm{\theta})\right\|_{F} 
				+ \left\|\widetilde{D}_{t}^{-1}(\bm{\delta}) \frac{\partial\widetilde{D}_{t}(\bm{\delta})}{\partial\delta_{\ell}} \left(H_{t}^{-1}(\bm{\theta}) - \widetilde{H}_{t}^{-1}(\bm{\theta})\right)\right\|_{F} \notag\\
				\leq& \left\|D_{t}^{-1}(\bm{\delta}) \frac{\partial D_{t}(\bm{\delta})}{\partial\delta_{\ell}} - \widetilde{D}_{t}^{-1}(\bm{\delta}) \frac{\partial\widetilde{D}_{t}(\bm{\delta})}{\partial\delta_{\ell}}\right\|_{F} \left\|H_{t}^{-1}(\bm{\theta})\right\|_{F} 
				+ \left\|\widetilde{D}_{t}^{-1}(\bm{\delta}) \frac{\partial\widetilde{D}_{t}(\bm{\delta})}{\partial\delta_{\ell}}\right\|_{F} \left\|H_{t}^{-1}(\bm{\theta}) - \widetilde{H}_{t}^{-1}(\bm{\theta})\right\|_{F},  \label{eq 3 in difference of the first derivative of lt} 
			\end{align}
			\end{small}
			and 
			\begin{small}
			\begin{align} \label{eq 4 in difference of the first derivative of lt}
				&\left\|D_{t}^{-1}(\bm{\delta}) R_{t}^{-1}(\bm{\theta}) \frac{\partial R_{t}(\bm{\theta})}{\partial\theta_{\ell}} R_{t}^{-1}(\bm{\theta}) D_{t}^{-1}(\bm{\delta}) - \widetilde{D}_{t}^{-1}(\bm{\delta}) \widetilde{R}_{t}^{-1}(\bm{\theta}) \frac{\partial\widetilde{R}_{t}(\bm{\theta})}{\partial\theta_{\ell}} \widetilde{R}_{t}^{-1}(\bm{\theta}) \widetilde{D}_{t}^{-1}(\bm{\delta})\right\|_{F} \notag\\
				\leq& \left\|D_{t}^{-1}(\bm{\delta}) - \widetilde{D}_{t}^{-1}(\bm{\delta})\right\|_{F} \left\|D_{t}^{-1}(\bm{\delta})\right\|_{F} \left\|R_{t}^{-1}(\bm{\theta})\right\|_{F}^{2} \left\|\frac{\partial R_{t}(\bm{\theta})}{\partial\theta_{\ell}}\right\|_{F} \notag\\
				&+ \left\|D_{t}^{-1}(\bm{\delta}) - \widetilde{D}_{t}^{-1}(\bm{\delta})\right\|_{F} \left\|\widetilde{D}_{t}^{-1}(\bm{\delta})\right\|_{F} \left\|\widetilde{R}_{t}^{-1}(\bm{\theta})\right\|_{F}^{2} \left\|\frac{\partial \widetilde{R}_{t}(\bm{\theta})}{\partial\theta_{\ell}}\right\|_{F} \notag\\
				&+ \left\|R_{t}^{-1}(\bm{\theta}) - \widetilde{R}_{t}^{-1}(\bm{\theta})\right\|_{F} \left\|D_{t}^{-1}(\bm{\delta})\right\|_{F} \left\|\widetilde{D}_{t}^{-1}(\bm{\delta})\right\|_{F} \left\|R_{t}^{-1}(\bm{\theta})\right\|_{F} \left\|\frac{\partial R_{t}(\bm{\theta})}{\partial\theta_{\ell}}\right\|_{F} \notag\\ 
				&+ \left\|R_{t}^{-1}(\bm{\theta}) - \widetilde{R}_{t}^{-1}(\bm{\theta})\right\|_{F} \left\|D_{t}^{-1}(\bm{\delta})\right\|_{F} \left\|\widetilde{D}_{t}^{-1}(\bm{\delta})\right\|_{F} \left\|\widetilde{R}_{t}^{-1}(\bm{\theta})\right\|_{F} \left\|\frac{\partial \widetilde{R}_{t}(\bm{\theta})}{\partial\theta_{\ell}}\right\|_{F} \notag\\
				&+ \left\|\frac{\partial R_{t}(\bm{\theta})}{\partial\theta_{\ell}} - \frac{\partial\widetilde{R}_{t}(\bm{\theta})}{\partial\theta_{\ell}}\right\|_{F} \left\|D_{t}^{-1}(\bm{\delta})\right\|_{F} \left\|\widetilde{D}_{t}^{-1}(\bm{\delta})\right\|_{F} \left\|R_{t}^{-1}(\bm{\theta})\right\|_{F} \left\|\widetilde{R}_{t}^{-1}(\bm{\theta})\right\|_{F}. 
			\end{align}
			\end{small}
			Then by \eqref{eq difference of the first derivative of lt wrt delta}--\eqref{eq 4 in difference of the first derivative of lt} and Lemma \ref{lemma varrho to the power of t}, together with the facts $\|\Diag\{\bm{a}\}\|_{F} = \|\bm{a}\|_{2}$ for any vector $\bm{a}$ and $\|A\|_{F} \leq \sqrt{c_{1}} \|A\|_{2}$ for any matrix $A$ with $c_{1} = \rank(A)$, 
			we can conclude that ($\romannumeral2$) holds. 

			With analogous arguments in the proof of ($\romannumeral2$), we can establish ($\romannumeral3$). 
		\end{proof}
	
		\begin{proof}[Proof of Lemma \ref{lemma (iii) for consistency}]
			Recall that $\ell_{t}(\bm{\theta}) = \frac{1}{2} \mathbf{y}_{t}^{\prime} H_{t}^{-1}(\bm{\theta}) \mathbf{y}_{t} + \frac{1}{2} \ln |H_{t}(\bm{\theta})|$ with $\mathbf{y}_{t} = H_{t}^{1/2}(\bm{\theta}_{0}) \bm{\eta}_{t}$ and $H_{t}(\bm{\theta}) = D_{t}(\bm{\delta}) R_{t}(\bm{\theta}) D_{t}(\bm{\delta})$, 
			and $\{\bm{\eta}_{t}\}$ are $i.i.d.$ with zero mean and identity covariance matrix. 
			Using the properties that $|A^{\prime}A| \leq \|A\|_{2}^{2 \min\{d_{1}, d_{2}\}}$ holds for any $d_{1} \times d_{2}$ matrix $A$, and $\|A\|_{2} \leq \|A\|_{F}$ holds for any matrix $A$, it can be shown that 
			\begin{small}
			\begin{align} \label{eq E lt(theta0)}
				2 E \left|\ell_{t}(\bm{\theta}_{0})\right| 
				&= E \left|\mathbf{y}_{t}^{\prime} H_{t}^{-1}(\bm{\theta}_{0}) \mathbf{y}_{t} + \ln |H_{t}(\bm{\theta}_{0})|\right| 
				= E \left|\bm{\eta}_{t}^{\prime} H_{t}^{1/2}(\bm{\theta}_{0}) H_{t}^{-1}(\bm{\theta}_{0}) H_{t}^{1/2}(\bm{\theta}_{0}) \bm{\eta}_{t} + \ln |H_{t}(\bm{\theta}_{0})|\right| \notag\\
				&= E \left|\bm{\eta}_{t}^{\prime} \bm{\eta}_{t} + \ln |H_{t}(\bm{\theta}_{0})|\right| 
				\leq m + E \left|\ln |H_{t}(\bm{\theta}_{0})|\right| 
				= m + E \left|\ln |D_{t}(\bm{\delta}_{0}) R_{t}(\bm{\theta}_{0}) D_{t}(\bm{\delta}_{0})|\right| \notag\\
				&= m + E \left|\ln |D_{t}(\bm{\delta}_{0})|^{2} + \ln |R_{t}(\bm{\theta}_{0})|\right| 
				\leq m + E \left|\ln |D_{t}(\bm{\delta}_{0})|^{2}\right| + E \left|\ln |R_{t}(\bm{\theta}_{0})|\right| \notag\\
				&\leq m + E \left|\sum_{i=1}^{m} \ln h_{ii,t}(\bm{\delta}_{0})\right| + E \left|\ln \|R_{t}^{1/2}(\bm{\theta}_{0})\|_{2}^{2m}\right| \notag\\
				&\leq m + \sum_{i=1}^{m} E \left|\ln h_{ii,t}(\bm{\delta}_{0})\right| + m E \left|\ln \|R_{t}^{1/2}(\bm{\theta}_{0})\|_{F}^{2}\right|. 
			\end{align}
			\end{small}
			Since $R_{t}$ is a conditional correlation matrix, it holds that $\|R_{t}^{1/2}(\bm{\theta}_{0})\|_{F}^{2} = \tr(R_{t}(\bm{\theta}_{0})) = m$, which implies that $E \left|\ln \|R_{t}^{1/2}(\bm{\theta}_{0})\|_{F}^{2}\right| = |\ln m| < \infty$. 
			Moreover, we have that $\sum_{i=1}^{m} E |\ln h_{ii,t}(\bm{\delta}_{0})| < \infty$ by Lemma \ref{lemma moments of derivatives of lnhtunderline}($\romannumeral1$). 
			These together with \eqref{eq E lt(theta0)} imply that $E|\ell_{t}(\bm{\theta}_{0})| < \infty$ holds. 

			Next we establish that $E\ell_{t}(\bm{\theta}) > E\ell_{t}(\bm{\theta}_{0})$ holds for any $\bm{\theta} \neq \bm{\theta}_{0}$. 
			Note that 
			\begin{small}
			\begin{align*}
				&E\left(\mathbf{y}_{t}^{\prime} H_{t}^{-1}(\bm{\theta}) \mathbf{y}_{t}\right) 
				= E\left(\bm{\eta}_{t}^{\prime} H_{t}^{1/2}(\bm{\theta}_{0}) H_{t}^{-1}(\bm{\theta}) H_{t}^{1/2}(\bm{\theta}_{0}) \bm{\eta}_{t}\right) 
				= \tr\left[E\left(H_{t}^{1/2}(\bm{\theta}_{0}) H_{t}^{-1}(\bm{\theta}) H_{t}^{1/2}(\bm{\theta}_{0}) \bm{\eta}_{t} \bm{\eta}_{t}^{\prime}\right)\right] \\
				=& \tr\left\{E\left[E\left(H_{t}^{1/2}(\bm{\theta}_{0}) H_{t}^{-1}(\bm{\theta}) H_{t}^{1/2}(\bm{\theta}_{0}) \bm{\eta}_{t} \bm{\eta}_{t}^{\prime} \mid \mathcal{F}_{t-1}\right)\right]\right\} 
				= \tr\left[E\left(H_{t}^{1/2}(\bm{\theta}_{0}) H_{t}^{-1}(\bm{\theta}) H_{t}^{1/2}(\bm{\theta}_{0})\right)\right] \\
				=& E\left[\tr\left(H_{t}^{1/2}(\bm{\theta}_{0}) H_{t}^{-1}(\bm{\theta}) H_{t}^{1/2}(\bm{\theta}_{0})\right)\right] 
				= E\left[\tr\left(H_{t}(\bm{\theta}_{0}) H_{t}^{-1}(\bm{\theta})\right)\right] 
			\end{align*}
			\end{small}
			Hence it holds that 
			\begin{align} \label{eq Elt(theta) - Elt(theta0)}
				&2 \left[E\ell_{t}(\bm{\theta}) - E\ell_{t}(\bm{\theta}_{0})\right] 
				= E\left(\mathbf{y}_{t}^{\prime} H_{t}^{-1}(\bm{\theta}) \mathbf{y}_{t}\right) - E\left(\mathbf{y}_{t}^{\prime} H_{t}^{-1}(\bm{\theta}_{0}) \mathbf{y}_{t}\right) - E\left(\ln |H_{t}(\bm{\theta}_{0})| - \ln |H_{t}(\bm{\theta})|\right) \notag\\
				=& E\left[\tr\left(H_{t}(\bm{\theta}_{0}) H_{t}^{-1}(\bm{\theta})\right) - m - \ln |H_{t}(\bm{\theta}_{0}) H_{t}^{-1}(\bm{\theta})|\right]. 
			\end{align}
			Denote $\lambda_{t,1}(\bm{\theta}), \ldots, \lambda_{t,m}(\bm{\theta})$ as the eigenvalues of $H_{t}(\bm{\theta}_{0}) H_{t}^{-1}(\bm{\theta})$. 
			Since $H_{t}(\bm{\theta})$ is positive definite on $\Theta$ under Assumption \ref{assum parameters}($\romannumeral3$), $\lambda_{t,i}(\bm{\theta}) > 0$ holds for all $1 \leq i \leq m$. 
			Then by \eqref{eq Elt(theta) - Elt(theta0)} and the fact $\ln x \leq x -1$ for $x > 0$ with equality if and only if $x = 1$, we can obtain that 
			$$
				E\ell_{t}(\bm{\theta}) - E\ell_{t}(\bm{\theta}_{0}) 
				= \frac{1}{2} E \left(\sum_{i=1}^{m} \lambda_{t,i}(\bm{\theta}) - m - \ln \prod_{i=1}^{m} \lambda_{t,i}(\bm{\theta})\right) 
				= \frac{1}{2} \sum_{i=1}^{m} E \left(\lambda_{t,i}(\bm{\theta}) - 1 - \ln \lambda_{t,i}(\bm{\theta})\right) 
				\geq 0, 
			$$
			where the equality holds if and only if $\lambda_{t,i}(\bm{\theta}) = 1$ a.s. for all $1 \leq i \leq m$, which implies that $H_{t}(\bm{\theta}) = H_{t}(\bm{\theta}_{0})$ a.s. and thus $\bm{\theta} = \bm{\theta}_{0}$ by Proposition \ref{propo Identification}($\romannumeral2$). 
			As a result, $E\ell_{t}(\bm{\theta}) > E\ell_{t}(\bm{\theta}_{0})$ holds for any $\bm{\theta} \neq \bm{\theta}_{0}$. 
		\end{proof}
	
		\begin{proof}[Proof of Lemma \ref{lemma (iv) for consistency}]
			For any $\bm{\theta} \in \Theta$ and any positive integer $k$, let $V_{k}(\bm{\theta})$ be the open ball with center $\bm{\theta}$ and radius $1/k$. 
			It holds that 
			\begin{align*}
				&\liminf_{n \to \infty} \inf_{\bm{\theta}^{*} \in V_{k}(\bm{\theta}) \cap \Theta} \widetilde{\mathcal{L}}_{n}(\bm{\theta}^{*}) 
				= \liminf_{n \to \infty} \inf_{\bm{\theta}^{*} \in V_{k}(\bm{\theta}) \cap \Theta} \left[\mathcal{L}_{n}(\bm{\theta}^{*}) - \left(\mathcal{L}_{n}(\bm{\theta}^{*}) - \widetilde{\mathcal{L}}_{n}(\bm{\theta}^{*})\right)\right] \\
				\geq& \liminf_{n \to \infty} \inf_{\bm{\theta}^{*} \in V_{k}(\bm{\theta}) \cap \Theta} \mathcal{L}_{n}(\bm{\theta}^{*}) 
				- \limsup_{n \to \infty} \sup_{\bm{\theta}^{*} \in V_{k}(\bm{\theta}) \cap \Theta} \left|\mathcal{L}_{n}(\bm{\theta}^{*}) - \widetilde{\mathcal{L}}_{n}(\bm{\theta}^{*})\right| \\
				=& \liminf_{n \to \infty} \inf_{\bm{\theta}^{*} \in V_{k}(\bm{\theta}) \cap \Theta} \frac{1}{n} \sum_{t=1}^{n} \ell_{t}(\bm{\theta}^{*}) 
				- \limsup_{n \to \infty} \sup_{\bm{\theta}^{*} \in V_{k}(\bm{\theta}) \cap \Theta} \left|\mathcal{L}_{n}(\bm{\theta}^{*}) - \widetilde{\mathcal{L}}_{n}(\bm{\theta}^{*})\right| \\
				\geq& \liminf_{n \to \infty} \frac{1}{n} \sum_{t=1}^{n} \inf_{\bm{\theta}^{*} \in V_{k}(\bm{\theta}) \cap \Theta} \ell_{t}(\bm{\theta}^{*}) 
				- \limsup_{n \to \infty} \sup_{\bm{\theta}^{*} \in V_{k}(\bm{\theta}) \cap \Theta} \left|\mathcal{L}_{n}(\bm{\theta}^{*}) - \widetilde{\mathcal{L}}_{n}(\bm{\theta}^{*})\right|. 
			\end{align*}
			This together with Lemma \ref{lemma (ii) for consistency}($\romannumeral1$), implies that 
			\begin{align} \label{eq liminf inf mathcalLntilde(theta*)}
				\liminf_{n \to \infty} \inf_{\bm{\theta}^{*} \in V_{k}(\bm{\theta}) \cap \Theta} \widetilde{\mathcal{L}}_{n}(\bm{\theta}^{*}) 
				\geq \liminf_{n \to \infty} \frac{1}{n} \sum_{t=1}^{n} \inf_{\bm{\theta}^{*} \in V_{k}(\bm{\theta}) \cap \Theta} \ell_{t}(\bm{\theta}^{*}) \;\; \text{a.s.}. 
			\end{align}
			Next we use the following ergodic theorem under Assumption \ref{assum stationarity}: 
			if $\{X_{t}\}$ is a stationary and ergodic process such that $E X_{1} \in \mathbb{R} \cup \{\infty\}$, then $n^{-1} \sum_{t=1}^{n} X_{t}$ converges a.s. to $E X_{1}$ when $n \to \infty$ (see \citet{billingsley1995}, pages 284 and 495; or the proof of Theorem 2.1 in \citet{francq2004maximum}). 
			To apply this theorem to $\{\inf_{\bm{\theta}^{*} \in V_{k}(\bm{\theta}) \cap \Theta} \ell_{t}(\bm{\theta}^{*})\}$, we need to verify that $E \ell_{t}^{-}(\bm{\theta}) = \max\{- E \ell_{t}(\bm{\theta}), 0\} < \infty$ on $\Theta$. 
			Recall that $\ell_{t}(\bm{\theta}) = \frac{1}{2} \mathbf{y}_{t}^{\prime} H_{t}^{-1}(\bm{\theta}) \mathbf{y}_{t} + \frac{1}{2} \ln |H_{t}(\bm{\theta})|$ with $H_{t}(\bm{\theta}) = D_{t}(\bm{\delta}) R_{t}(\bm{\theta}) D_{t}(\bm{\delta})$ and $D_{t}(\bm{\delta}) = \Diag\{h_{11,t}^{1/2}(\bm{\delta}), \ldots, h_{mm,t}^{1/2}(\bm{\delta})\}$. 
			Hence by Lemma \ref{lemma moments of derivatives of lnhtunderline}($\romannumeral1$) and the positive definiteness of $R_{t}(\bm{\theta})$ under Assumption \ref{assum parameters}($\romannumeral3$), it holds that 
			\begin{align*}
				- 2 E \ell_{t}(\bm{\theta}) 
				\leq& - E \ln |H_{t}(\bm{\theta})|
				= - E \ln |D_{t}(\bm{\delta}) R_{t}(\bm{\theta}) D_{t}(\bm{\delta})| 
				= - E \ln |D_{t}^{2}(\bm{\delta})| - E \ln |R_{t}(\bm{\theta})| \\
				=& - E \ln \prod_{i=1}^{m} h_{ii,t}(\bm{\delta}) - E \ln |R_{t}(\bm{\theta})| 
				= - \sum_{i=1}^{m} E \ln h_{ii,t}(\bm{\delta}) - E \ln |R_{t}(\bm{\theta})| < \infty. 
			\end{align*}
			And it follows that $E \ell_{t}^{-}(\bm{\theta}) = \max\{- E \ell_{t}(\bm{\theta}), 0\} < \infty$ on $\Theta$. 
			Then applying the aforementioned ergodic theorem to $\{\inf_{\bm{\theta}^{*} \in V_{k}(\bm{\theta}) \cap \Theta} \ell_{t}(\bm{\theta}^{*})\}$, we can obtain that 
			\begin{align*}
				\liminf_{n \to \infty} \frac{1}{n} \sum_{t=1}^{n} \inf_{\bm{\theta}^{*} \in V_{k}(\bm{\theta}) \cap \Theta} \ell_{t}(\bm{\theta}^{*}) = E \inf_{\bm{\theta}^{*} \in V_{k}(\bm{\theta}) \cap \Theta} \ell_{1}(\bm{\theta}^{*}) \;\; \text{a.s.}. 
			\end{align*}
			This together with \eqref{eq liminf inf mathcalLntilde(theta*)} implies that 
			\begin{align*}
				\liminf_{n \to \infty} \inf_{\bm{\theta}^{*} \in V_{k}(\bm{\theta}) \cap \Theta} \widetilde{\mathcal{L}}_{n}(\bm{\theta}^{*}) 
				\geq E \inf_{\bm{\theta}^{*} \in V_{k}(\bm{\theta}) \cap \Theta} \ell_{1}(\bm{\theta}^{*}) \;\; \text{a.s.}. 
			\end{align*}
			Furthermore, by the Beppo-Levi theorem, we have $E \inf_{\bm{\theta}^{*} \in V_{k}(\bm{\theta}) \cap \Theta} \ell_{1}(\bm{\theta}^{*})$ increases to $E \ell_{1}(\bm{\theta})$ as $k$ increases to the infinity. 
			As a result, we have 
			\begin{align*}
				\liminf_{n \to \infty} \inf_{\bm{\theta}^{*} \in V_{k}(\bm{\theta}) \cap \Theta} \widetilde{\mathcal{L}}_{n}(\bm{\theta}^{*}) 
				\geq E \ell_{1}(\bm{\theta}) \;\; \text{a.s.}, 
			\end{align*}
			that is the statement in this lemma holds. 
		\end{proof}
	
		\begin{proof}[Proof of Lemma \ref{lemma moments of derivatives of Psitminus1}]
			We first show ($\romannumeral1$). 
			It suffices to show that $E \sup_{\bm{\theta} \in \Theta} |{\partial\Psi_{t-1,ij}(\bm{\delta})}/{\partial\delta_{\ell}}|^{2+\epsilon} < \infty$ for all $1 \leq i,j \leq m$. 
			Note that $|\Psi_{t-1,ij}(\bm{\delta})| \leq 1$ on $\Theta$. 
			By \eqref{eq first derivative of Psitminus1}, we only need to show that 
			\begin{small}
			$$
				(\text{a1}) \;
				E \sup_{\bm{\theta} \in \Theta} \left|\frac
				{\frac{\partial\varepsilon_{i,t-1}(\bm{\delta})}{\partial\delta_{\ell}} \varepsilon_{j,t-1}(\bm{\delta})}
				{\left[\left(\sum_{h=1}^{\Bbbk} \varepsilon_{i,t-h}^{2}(\bm{\delta})\right) \left(\sum_{h=1}^{\Bbbk} \varepsilon_{j,t-h}^{2}(\bm{\delta})\right)\right]^{1/2}}\right|^{2+\epsilon} < \infty; ~
				(\text{a2}) \;
				E \sup_{\bm{\theta} \in \Theta} \left|\frac
				{\varepsilon_{i,t-1}(\bm{\delta}) \frac{\partial\varepsilon_{i,t-1}(\bm{\delta})}{\partial\delta_{\ell}}}
				{\sum_{h=1}^{\Bbbk} \varepsilon_{i,t-h}^{2}(\bm{\delta})}\right|^{2+\epsilon} < \infty. 
			$$
			\end{small}
			By \eqref{eq first derivative of varepsilon} and Lemma \ref{lemma moments of derivatives of lnhtunderline}($\romannumeral6$) under the condition $E\|\ln\mathbf{y}^{2}_{t}\|^{2+\epsilon} < \infty$, it can be shown that 
			\begin{align*}
				&E \sup_{\bm{\theta} \in \Theta} \left|\frac
				{\frac{\partial\varepsilon_{i,t-1}(\bm{\delta})}{\partial\delta_{\ell}} \varepsilon_{j,t-1}(\bm{\delta})}
				{\left[\left(\sum_{h=1}^{\Bbbk} \varepsilon_{i,t-h}^{2}(\bm{\delta})\right) \left(\sum_{h=1}^{\Bbbk} \varepsilon_{j,t-h}^{2}(\bm{\delta})\right)\right]^{1/2}}\right|^{2+\epsilon} \\
				=& E \sup_{\bm{\theta} \in \Theta} \left|\frac
				{-\frac{1}{2} \frac{\partial\ln h_{ii,t-1}(\bm{\delta})}{\partial\delta_{\ell}} \varepsilon_{i,t-1}(\bm{\delta}) \varepsilon_{j,t-1}(\bm{\delta})}
				{\left[\left(\sum_{h=1}^{\Bbbk} \varepsilon_{i,t-h}^{2}(\bm{\delta})\right) \left(\sum_{h=1}^{\Bbbk} \varepsilon_{j,t-h}^{2}(\bm{\delta})\right)\right]^{1/2}}\right|^{2+\epsilon} \\
				=& \frac{1}{4^{1+\epsilon/2}} E \sup_{\bm{\theta} \in \Theta} 
				\left[\left(\frac{\varepsilon_{i,t-1}^{2}(\bm{\delta})}{\sum_{h=1}^{\Bbbk} \varepsilon_{i,t-h}^{2}(\bm{\delta})}\right)^{1+\epsilon/2}
				\left(\frac{\varepsilon_{j,t-1}^{2}(\bm{\delta})}{\sum_{h=1}^{\Bbbk} \varepsilon_{j,t-h}^{2}(\bm{\delta})}\right)^{1+\epsilon/2}
				\left|\frac{\partial\ln h_{ii,t-1}(\bm{\delta})}{\partial\delta_{\ell}}\right|^{2+\epsilon}\right] \\
				\leq& \frac{1}{4^{1+\epsilon/2}} E \sup_{\bm{\theta} \in \Theta} \left|\frac{\partial\ln h_{ii,t-1}(\bm{\delta})}{\partial\delta_{\ell}}\right|^{2+\epsilon} 
				< \infty, 
			\end{align*}
			and 
			\begin{align*}
				&E \sup_{\bm{\theta} \in \Theta} \left|\frac
				{\varepsilon_{i,t-1}(\bm{\delta}) \frac{\partial\varepsilon_{i,t-1}(\bm{\delta})}{\partial\delta_{\ell}}}
				{\sum_{h=1}^{\Bbbk} \varepsilon_{i,t-h}^{2}(\bm{\delta})}\right|^{2+\epsilon} 
				= E \sup_{\bm{\theta} \in \Theta} \left|\frac
				{-\frac{1}{2} \frac{\partial\ln h_{ii,t-1}(\bm{\delta})}{\partial\delta_{\ell}} \varepsilon_{i,t-1}^{2}(\bm{\delta})}
				{\sum_{h=1}^{\Bbbk} \varepsilon_{i,t-h}^{2}(\bm{\delta})}\right|^{2+\epsilon} \\
				=& \frac{1}{4^{1+\epsilon/2}} E \sup_{\bm{\theta} \in \Theta} \left(\left|\frac
				{\varepsilon_{i,t-1}^{2}(\bm{\delta})}
				{\sum_{h=1}^{\Bbbk} \varepsilon_{i,t-h}^{2}(\bm{\delta})}\right|^{2+\epsilon} 
				\left|\frac{\partial\ln h_{ii,t-1}(\bm{\delta})}{\partial\delta_{\ell}}\right|^{2+\epsilon}\right) \\
				\leq& \frac{1}{4^{1+\epsilon/2}} E \sup_{\bm{\theta} \in \Theta} 
				\left|\frac{\partial\ln h_{ii,t-1}(\bm{\delta})}{\partial\delta_{\ell}}\right|^{2+\epsilon} < \infty. 
			\end{align*}
			Thus (a1) and (a2) hold, and then ($\romannumeral1$) holds. 

			Next we show ($\romannumeral2$). 
			Similarly, it suffices to verify that $E \sup_{\bm{\theta} \in \Theta} |{\partial^{2}\Psi_{t-1,ij}(\bm{\delta})}/{\partial\delta_{k} \partial\delta_{\ell}}|^{1+\epsilon/2} < \infty$ for all $1 \leq i,j \leq m$. 
			Then by \eqref{eq second derivative of Psitminus1}, $|\Psi_{t-1,ij}(\bm{\delta})| \leq 1$, 
			($\romannumeral1$), (a1)--(a2) 
			and the Cauchy-Schwarz inequality, we only need to show the following results hold: 
			\begin{align*}
				&(\text{b1}) \;\;
				E \sup_{\bm{\theta} \in \Theta} \left|\frac
				{\frac{\partial\varepsilon_{i,t-1}(\bm{\delta})}{\partial\delta_{\ell}} \frac{\partial\varepsilon_{j,t-1}(\bm{\delta})}{\partial\delta_{k}}}
				{\left[\left(\sum_{h=1}^{\Bbbk} \varepsilon_{i,t-h}^{2}(\bm{\delta})\right) \left(\sum_{h=1}^{\Bbbk} \varepsilon_{j,t-h}^{2}(\bm{\delta})\right)\right]^{1/2}}\right|^{1+\epsilon/2} < \infty; \\
				&(\text{b2}) \;\;
				E \sup_{\bm{\theta} \in \Theta} \left|\frac
				{\frac{\partial^{2}\varepsilon_{i,t-1}(\bm{\delta})}{\partial\delta_{k} \partial\delta_{\ell}} \varepsilon_{j,t-1}(\bm{\delta})}
				{\left[\left(\sum_{h=1}^{\Bbbk} \varepsilon_{i,t-h}^{2}(\bm{\delta})\right) \left(\sum_{h=1}^{\Bbbk} \varepsilon_{j,t-h}^{2}(\bm{\delta})\right)\right]^{1/2}}\right|^{1+\epsilon/2} < \infty. 
			\end{align*}
			For (b1), by \eqref{eq first derivative of varepsilon} and Lemma \ref{lemma moments of derivatives of lnhtunderline}($\romannumeral6$), we have that 
			\begin{small}
			\begin{align*}
				&E \sup_{\bm{\theta} \in \Theta} \left|\frac{\frac{\partial\varepsilon_{i,t-1}(\bm{\delta})}{\partial\delta_{\ell}}}{\left(\sum_{h=1}^{\Bbbk} \varepsilon_{i,t-h}^{2}(\bm{\delta})\right)^{1/2}}\right|^{2+\epsilon} 
				= E \sup_{\bm{\theta} \in \Theta} \left|\frac{-\frac{1}{2} \frac{\partial\ln h_{ii,t-1}(\bm{\delta})}{\partial\delta_{\ell}} \varepsilon_{i,t-1}(\bm{\delta})}{\left(\sum_{h=1}^{\Bbbk} \varepsilon_{i,t-h}^{2}(\bm{\delta})\right)^{1/2}}\right|^{2+\epsilon} \\
				=& \frac{1}{4^{1+\epsilon/2}} E \sup_{\bm{\theta} \in \Theta} \left(\left|\frac{\varepsilon_{i,t-1}^{2}(\bm{\delta})}{\sum_{h=1}^{\Bbbk} \varepsilon_{i,t-h}^{2}(\bm{\delta})}\right|^{1+\epsilon/2} \left|\frac{\partial\ln h_{ii,t-1}(\bm{\delta})}{\partial\delta_{\ell}}\right|^{2+\epsilon}\right) 
				\leq \frac{1}{4^{1+\epsilon/2}} E \sup_{\bm{\theta} \in \Theta} \left|\frac{\partial\ln h_{ii,t-1}(\bm{\delta})}{\partial\delta_{\ell}}\right|^{2+\epsilon} 
				< \infty. 
			\end{align*}
			\end{small}
			It then follows that (b1) holds by the Cauchy-Schwarz inequality. 
			For (b2), by \eqref{eq second derivative of varepsilon}, it can be shown that 
			\begin{align*}
				&E \sup_{\bm{\theta} \in \Theta} \left|\frac
				{\frac{\partial^{2}\varepsilon_{i,t-1}(\bm{\delta})}{\partial\delta_{k} \partial\delta_{\ell}}}
				{\left(\sum_{h=1}^{\Bbbk} \varepsilon_{i,t-h}^{2}(\bm{\delta})\right)^{1/2}}\right|^{1+\epsilon/2} \\
				=& E \sup_{\bm{\theta} \in \Theta} \left|\frac
				{-\frac{1}{2} \frac{\partial^{2}\ln h_{ii,t-1}(\bm{\delta})}{\partial\delta_{k} \partial\delta_{\ell}} \varepsilon_{i,t-1}(\bm{\delta}) + \frac{1}{4} \frac{\partial\ln h_{ii,t-1}(\bm{\delta})}{\partial\delta_{\ell}} \frac{\partial\ln h_{ii,t-1}(\bm{\delta})}{\partial\delta_{k}} \varepsilon_{i,t-1}(\bm{\delta})}
				{\left(\sum_{h=1}^{\Bbbk} \varepsilon_{i,t-h}^{2}(\bm{\delta})\right)^{1/2}}\right|^{1+\epsilon/2} \\
				=& E \sup_{\bm{\theta} \in \Theta} \left|-\frac{1}{2} \frac{\partial^{2}\ln h_{ii,t-1}(\bm{\delta})}{\partial\delta_{k} \partial\delta_{\ell}} + \frac{1}{4} \frac{\partial\ln h_{ii,t-1}(\bm{\delta})}{\partial\delta_{\ell}} \frac{\partial\ln h_{ii,t-1}(\bm{\delta})}{\partial\delta_{k}}\right|^{1+\epsilon/2} \left(\frac{\varepsilon_{i,t-1}^{2}(\bm{\delta})}{\sum_{h=1}^{\Bbbk} \varepsilon_{i,t-h}^{2}(\bm{\delta})}\right)^{1/2+\epsilon/4} \\
				\leq& E \sup_{\bm{\theta} \in \Theta} \left|-\frac{1}{2} \frac{\partial^{2}\ln h_{ii,t-1}(\bm{\delta})}{\partial\delta_{k} \partial\delta_{\ell}} + \frac{1}{4} \frac{\partial\ln h_{ii,t-1}(\bm{\delta})}{\partial\delta_{\ell}} \frac{\partial\ln h_{ii,t-1}(\bm{\delta})}{\partial\delta_{k}}\right|^{1+\epsilon/2} \\
				\leq& E \sup_{\bm{\theta} \in \Theta} \left(\left|\frac{1}{2} \frac{\partial^{2}\ln h_{ii,t-1}(\bm{\delta})}{\partial\delta_{k} \partial\delta_{\ell}}\right| + \left|\frac{1}{4} \frac{\partial\ln h_{ii,t-1}(\bm{\delta})}{\partial\delta_{\ell}} \frac{\partial\ln h_{ii,t-1}(\bm{\delta})}{\partial\delta_{k}}\right|\right)^{1+\epsilon/2}. 
			\end{align*}
			This together with $E \sup_{\bm{\theta} \in \Theta} \left|{\partial\ln h_{ii,t-1}(\bm{\delta})}/{\partial\delta_{\ell}}\right|^{2+\epsilon} < \infty$ and $E \sup_{\bm{\theta} \in \Theta} \left|{\partial^{2}\ln h_{ii,t-1}(\bm{\delta})}/{\partial\delta_{k} \partial\delta_{\ell}}\right|^{1+\epsilon/2} < \infty$ by Lemmas \ref{lemma moments of derivatives of lnhtunderline}($\romannumeral6$)--($\romannumeral7$), the Minkowski inequality and the Cauchy-Schwarz inequality, implies that 
			\begin{align*}
				E \sup_{\bm{\theta} \in \Theta} \left|\frac
				{\frac{\partial^{2}\varepsilon_{i,t-1}(\bm{\delta})}{\partial\delta_{k} \partial\delta_{\ell}}}
				{\left(\sum_{h=1}^{\Bbbk} \varepsilon_{i,t-h}^{2}(\bm{\delta})\right)^{1/2}}\right|^{1+\epsilon/2} < \infty. 
			\end{align*}
			Then we can obtain that 
			\begin{align*}
				&E \sup_{\bm{\theta} \in \Theta} \left|\frac
				{\frac{\partial^{2}\varepsilon_{i,t-1}(\bm{\delta})}{\partial\delta_{k} \partial\delta_{\ell}} \varepsilon_{j,t-1}(\bm{\delta})}
				{\left[\left(\sum_{h=1}^{\Bbbk} \varepsilon_{i,t-h}^{2}(\bm{\delta})\right) \left(\sum_{h=1}^{\Bbbk} \varepsilon_{j,t-h}^{2}(\bm{\delta})\right)\right]^{1/2}}\right|^{1+\epsilon/2} \\
				=& E \sup_{\bm{\theta} \in \Theta} \left[\left|\frac
				{\frac{\partial^{2}\varepsilon_{i,t-1}(\bm{\delta})}{\partial\delta_{k} \partial\delta_{\ell}}}
				{\left(\sum_{h=1}^{\Bbbk} \varepsilon_{i,t-h}^{2}(\bm{\delta})\right)^{1/2}}\right|^{1+\epsilon/2} 
				\left(\frac{\varepsilon_{j,t-1}^{2}(\bm{\delta})}{\sum_{h=1}^{\Bbbk} \varepsilon_{j,t-h}^{2}(\bm{\delta})}\right)^{1/2+\epsilon/4}\right] \\
				\leq& E \sup_{\bm{\theta} \in \Theta} \left|\frac
				{\frac{\partial^{2}\varepsilon_{i,t-1}(\bm{\delta})}{\partial\delta_{k} \partial\delta_{\ell}}}
				{\left(\sum_{h=1}^{\Bbbk} \varepsilon_{i,t-h}^{2}(\bm{\delta})\right)^{1/2}}\right|^{1+\epsilon/2} 
				< \infty. 
			\end{align*}
			Thus (b2) holds. 
			As a result, ($\romannumeral2$) is established. 

			With analogous arguments in the proofs of ($\romannumeral1$) and ($\romannumeral2$), we can show that ($\romannumeral3$) holds. 
		\end{proof}

		\begin{proof}[Proof of Lemma \ref{lemma moments of derivatives of Rt}]
			We first show ($\romannumeral1$). 
			Recall that $\bm{\theta} = (\bm{\delta}^{\prime}, \bm{\beta}^{\prime})^{\prime}$ with $\bm{\beta} = (\beta_{1}, \beta_{2}, \underline{\bm{r}}^{\prime})^{\prime}$ and $\underline{\bm{r}} = \ovechsec(\underline{R})$. 
			Let $\delta_{\ell}$ be the $\ell$-th element of $\bm{\delta}$, and $\underline{R} = [\underline{R}_{ij}]$. 
			By \eqref{eq first derivative of Rt} and the triangle inequality of matrix norms, it holds that 
			\begin{small}
			\begin{align*}
				E \sup_{\bm{\theta} \in \Theta} \left\|\frac{\partial R_{t}(\bm{\theta})}{\partial\delta_{\ell}}\right\|^{2+\epsilon} 
				= E \sup_{\bm{\theta} \in \Theta} \left\|\beta_{1} \sum_{h=0}^{\infty} \beta_{2}^{h} \frac{\partial\Psi_{t-h-1}(\bm{\delta})}{\partial\delta_{\ell}}\right\|^{2+\epsilon} 
				\leq E \sup_{\bm{\theta} \in \Theta} \left(\beta_{1} \sum_{h=0}^{\infty} \beta_{2}^{h} \left\|\frac{\partial\Psi_{t-h-1}(\bm{\delta})}{\partial\delta_{\ell}}\right\|\right)^{2+\epsilon} 
			\end{align*}
			\end{small}
			This together with $0 < \beta_{1} < 1$, $0 < \beta_{2} < 1$ and $E \sup_{\bm{\theta} \in \Theta} \|{\partial\Psi_{t-1}(\bm{\delta})}/{\partial\delta_{\ell}}\|^{2+\epsilon} < \infty$ by Lemma \ref{lemma moments of derivatives of Psitminus1}($\romannumeral1$), implies that 
			\begin{align*}
				E \sup_{\bm{\theta} \in \Theta} \left\|\frac{\partial R_{t}(\bm{\theta})}{\partial\delta_{\ell}}\right\|^{2+\epsilon} < \infty. 
			\end{align*}
			Moreover, since $0 < \beta_{1} < 1$, $0 < \beta_{2} < 1$, and all elements of $\underline{R}$ and $\Psi_{t}(\bm{\delta})$ are between $-1$ and 1, it can be easily shown that $E \sup_{\bm{\theta} \in \Theta} \|{\partial R_{t}(\bm{\theta})}/{\partial\beta_{1}}\|^{2+\epsilon} < \infty$, $E \sup_{\bm{\theta} \in \Theta} \|{\partial R_{t}(\bm{\theta})}/{\partial\beta_{2}}\|^{2+\epsilon} < \infty$ and $E \sup_{\bm{\theta} \in \Theta} \|{\partial R_{t}(\bm{\theta})}/{\partial\underline{R}_{ij}}\|^{2+\epsilon} < \infty$ by \eqref{eq first derivative of Rt}. 
			As a result, ($\romannumeral1$) is established. 

			Similarly, using \eqref{eq second derivative of Rt} and Lemma \ref{lemma moments of derivatives of Psitminus1}($\romannumeral2$), we can obtain that $E \sup_{\bm{\theta} \in \Theta} \|{\partial^{2} R_{t}(\bm{\theta})}/{\partial\theta_{k} \partial\theta_{\ell}}\|^{1+\epsilon/2} < \infty$. Thus ($\romannumeral2$) holds. 
			And we can also establish ($\romannumeral3$) with analogous arguments. 
		\end{proof}

		\begin{proof}[Proof of Lemma \ref{lemma moments of Htinverse product derivatives of Ht}]
			We first show ($\romannumeral1$). 
			By \eqref{eq first derivative of Dt} and $E \sup_{\bm{\theta} \in \Theta} \|{\partial\ln\bm{h}_{t}(\bm{\delta})}/{\partial\delta_{\ell}}\|^{2+\epsilon} < \infty$ by Lemma \ref{lemma moments of derivatives of lnhtunderline}($\romannumeral6$), we have that 
			\begin{align*}
				&E \sup_{\bm{\theta} \in \Theta} \left\|D_{t}^{-1}(\bm{\delta}) \frac{\partial D_{t}(\bm{\delta})}{\partial\delta_{\ell}}\right\|^{2+\epsilon} 
				= E \sup_{\bm{\theta} \in \Theta} \left\|D_{t}^{-1}(\bm{\delta}) \frac{1}{2} D_{t}(\bm{\delta}) \Diag\left\{\frac{\partial\ln\bm{h}_{t}(\bm{\delta})}{\partial\delta_{\ell}}\right\}\right\|^{2+\epsilon} \\
				=& E \sup_{\bm{\theta} \in \Theta} \left\|\frac{1}{2} \Diag\left\{\frac{\partial\ln\bm{h}_{t}(\bm{\delta})}{\partial\delta_{\ell}}\right\}\right\|^{2+\epsilon} 
				< \infty, 
			\end{align*}
			Moreover, by \eqref{eq second derivative of Dt} and the properties of the induced matrix norm $\|\cdot\|$ that $\|A + B\| \leq \|A\| + \|B\|$ and $\|A B\| \leq \|A\| \|B\|$ for any matrices $A$ and $B$, it holds that 
			\begin{small}
			\begin{align*}
				&E \sup_{\bm{\theta} \in \Theta} \left\|D_{t}^{-1}(\bm{\delta}) \frac{\partial^{2} D_{t}(\bm{\delta})}{\partial\delta_{k} \partial\delta_{\ell}}\right\|^{1+\epsilon/2} \\
				=& E \sup_{\bm{\theta} \in \Theta} \left\|D_{t}^{-1}(\bm{\delta}) \left(\frac{1}{4} D_{t}(\bm{\delta}) \Diag\left\{\frac{\partial\ln\bm{h}_{t}(\bm{\delta})}{\partial\delta_{k}}\right\} \Diag\left\{\frac{\partial\ln\bm{h}_{t}(\bm{\delta})}{\partial\delta_{\ell}}\right\} 
				+ \frac{1}{2} D_{t}(\bm{\delta}) \Diag\left\{\frac{\partial^{2}\ln\bm{h}_{t}(\bm{\delta})}{\partial\delta_{k} \partial\delta_{\ell}}\right\}\right)\right\|^{1+\epsilon/2} \\
				=& E \sup_{\bm{\theta} \in \Theta} \left\|\frac{1}{4} \Diag\left\{\frac{\partial\ln\bm{h}_{t}(\bm{\delta})}{\partial\delta_{k}}\right\} \Diag\left\{\frac{\partial\ln\bm{h}_{t}(\bm{\delta})}{\partial\delta_{\ell}}\right\} 
				+ \frac{1}{2} \Diag\left\{\frac{\partial^{2}\ln\bm{h}_{t}(\bm{\delta})}{\partial\delta_{k} \partial\delta_{\ell}}\right\}\right\|^{1+\epsilon/2} \\
				\leq& E \sup_{\bm{\theta} \in \Theta} \left(\frac{1}{4} \left\|\Diag\left\{\frac{\partial\ln\bm{h}_{t}(\bm{\delta})}{\partial\delta_{k}}\right\}\right\| \left\|\Diag\left\{\frac{\partial\ln\bm{h}_{t}(\bm{\delta})}{\partial\delta_{\ell}}\right\}\right\| 
				+ \frac{1}{2} \left\|\Diag\left\{\frac{\partial^{2}\ln\bm{h}_{t}(\bm{\delta})}{\partial\delta_{k} \partial\delta_{\ell}}\right\}\right\|\right)^{1+\epsilon/2}. 
			\end{align*}
			\end{small}
			This together with $E \sup_{\bm{\theta} \in \Theta} \|{\partial\ln\bm{h}_{t}(\bm{\delta})}/{\partial\delta_{\ell}}\|^{2+\epsilon} < \infty$ and $E \sup_{\bm{\theta} \in \Theta} \left\|{\partial^{2}\ln\bm{h}_{t}(\bm{\delta})}/{\partial\delta_{k} \partial\delta_{\ell}}\right\|^{1+\epsilon/2} < \infty$ by Lemmas \ref{lemma moments of derivatives of lnhtunderline}($\romannumeral6$)--($\romannumeral7$), the Minkowski inequality and the Cauchy-Schwarz inequality, implies that 
			\begin{align*}
				E \sup_{\bm{\theta} \in \Theta} \left\|D_{t}^{-1}(\bm{\delta}) \frac{\partial^{2} D_{t}(\bm{\delta})}{\partial\delta_{k} \partial\delta_{\ell}}\right\|^{1+\epsilon/2} < \infty. 
			\end{align*}
			With analogous arguments, we can also verify $E \sup_{\bm{\theta} \in \Theta} \|D_{t}^{-1}(\bm{\delta}) {\partial^{3} D_{t}(\bm{\delta})}/{\partial\delta_{j} \partial\delta_{k} \partial\delta_{\ell}}\|^{1+\epsilon/2} < \infty$. 
			Thus ($\romannumeral1$) holds. 

			Next we show ($\romannumeral2$). 
			Note that $R_{t}(\bm{\theta}) = (1 - \beta_{1} - \beta_{2}) \underline{R} + \beta_{1} \Psi_{t-1}(\bm{\delta}) + \beta_{2} R_{t-1}(\bm{\theta})$, with $0 <\beta_{1}, \beta_{2} < 1$ and $0 < \beta_{1} + \beta_{2} < 1$, is positive definite under Assumption \ref{assum parameters}($\romannumeral3$). 
			Then using Theorem 8.4.9 of \citet{Bernstein2009} and Assumption \ref{assum parameters}($\romannumeral3$), we can obtain that 
			$$
				\sup_{\bm{\theta} \in \Theta} \|R_{t}^{-1}(\bm{\theta})\|_{2} 
				= \sup_{\bm{\theta} \in \Theta} \lambda_{\max}(R_{t}^{-1}(\bm{\theta})) 
				= \sup_{\bm{\theta} \in \Theta} \lambda_{\min}^{-1}(R_{t}(\bm{\theta})) 
				\leq \sup_{\bm{\theta} \in \Theta} \lambda_{\min}^{-1}((1 - \beta_{1} - \beta_{2}) \underline{R}) 
				< \infty, 
			$$
			which implies that $\sup_{\bm{\theta} \in \Theta} \|R_{t}^{-1}(\bm{\theta})\| < \infty$ for any induced matrix norm $\|\cdot\|$ as the dimension is fixed. 
			This together with $\|A B\| \leq \|A\| \|B\|$ for any matrices $A$ and $B$ and Lemma \ref{lemma moments of derivatives of Rt}, implies that 
			\begin{align*}
				&E \sup_{\bm{\theta} \in \Theta} \left\|R_{t}^{-1}(\bm{\theta}) \frac{\partial R_{t}(\bm{\theta})}{\partial\theta_{\ell}}\right\|^{2+\epsilon} 
				\leq E \sup_{\bm{\theta} \in \Theta} \left\|\frac{\partial R_{t}(\bm{\theta})}{\partial\theta_{\ell}}\right\|^{2+\epsilon} 
				< \infty, \\
				&E \sup_{\bm{\theta} \in \Theta} \left\|R_{t}^{-1}(\bm{\theta}) \frac{\partial^{2} R_{t}(\bm{\theta})}{\partial\theta_{k} \partial\theta_{\ell}}\right\|^{1+\epsilon/2} 
				\leq E \sup_{\bm{\theta} \in \Theta} \left\|\frac{\partial^{2} R_{t}(\bm{\theta})}{\partial\theta_{k} \partial\theta_{\ell}}\right\|^{1+\epsilon/2} 
				< \infty, \;\; \text{and} \\
				&E \sup_{\bm{\theta} \in \Theta} \left\|R_{t}^{-1}(\bm{\theta}) \frac{\partial^{3} R_{t}(\bm{\theta})}{\partial\theta_{j} \partial\theta_{k} \partial\theta_{\ell}}\right\|^{1+\epsilon/2} 
				\leq E \sup_{\bm{\theta} \in \Theta} \left\|\frac{\partial^{3} R_{t}(\bm{\theta})}{\partial\theta_{j} \partial\theta_{k} \partial\theta_{\ell}}\right\|^{1+\epsilon/2} 
				< \infty. 
			\end{align*}
			Hence ($\romannumeral2$) holds. 

			Recall that $H_{t}(\bm{\theta}) = D_{t}(\bm{\delta}) R_{t}(\bm{\theta}) D_{t}(\bm{\delta})$. 
			By \eqref{eq first derivative of Ht} and \eqref{eq second derivative of Ht}, ($\romannumeral3$) directly follows from ($\romannumeral1$)--($\romannumeral2$) and the aforementioned properties of the induced matrix norm $\|\cdot\|$. 
		\end{proof}
	
		\begin{proof}[Proof of Lemma \ref{lemma moments of derivatives of lt}]
			We first show ($\romannumeral1$). 
			By \eqref{eq first derivative of lt}, it suffice to show that 
			\begin{enumerate}
				\item[(a1)] $E \left|\mathbf{y}_{t}^{\prime} H_{t}^{-1}(\bm{\theta}_{0}) \frac{\partial H_{t}(\bm{\theta}_{0})}{\partial\theta_{\ell}} H_{t}^{-1}(\bm{\theta}_{0}) \mathbf{y}_{t}\right|^{2} < \infty$ and $\sup_{\bm{\theta} \in \Theta} \left|\mathbf{y}_{t}^{\prime} H_{t}^{-1}(\bm{\theta}) \frac{\partial H_{t}(\bm{\theta})}{\partial\theta_{\ell}} H_{t}^{-1}(\bm{\theta}) \mathbf{y}_{t}\right|^{2}$ is finite a.s.; 
				\item[(a2)] $E \sup_{\bm{\theta} \in \Theta} \left|\tr\left(D_{t}^{-1}(\bm{\delta}) \frac{\partial D_{t}(\bm{\delta})}{\partial\delta_{\ell}}\right)\right|^{2} < \infty$ \;\; \text{and} \;\; 
				$E \sup_{\bm{\theta} \in \Theta} \left|\tr\left(R_{t}^{-1}(\bm{\theta}) \frac{\partial R_{t}(\bm{\theta})}{\partial\theta_{\ell}}\right)\right|^{2} < \infty$. 
			\end{enumerate}
			For (a1), recall that $\mathbf{y}_{t} = H_{t}^{1/2}(\bm{\theta}_{0}) \bm{\eta}_{t}$ and $\{\bm{\eta}_{t}\}$ are $i.i.d.$ with zero mean and identity covariance matrix. 
			Using the fact $\tr(AB) = \tr(BA) \leq \|A\|_{F} \|B\|_{F}$, we have that 
			\begin{align} \label{eq ytprime Htinverse dHt Htinverse yt}
				&\left|\mathbf{y}_{t}^{\prime} H_{t}^{-1}(\bm{\theta}) \frac{\partial H_{t}(\bm{\theta})}{\partial\theta_{\ell}} H_{t}^{-1}(\bm{\theta}) \mathbf{y}_{t}\right|^{2} 
				= \left|\tr\left(\mathbf{y}_{t}^{\prime} H_{t}^{-1}(\bm{\theta}) \frac{\partial H_{t}(\bm{\theta})}{\partial\theta_{\ell}} H_{t}^{-1}(\bm{\theta}) \mathbf{y}_{t}\right)\right|^{2} \notag\\
				=& \left|\tr\left(\bm{\eta}_{t}^{\prime} H_{t}^{1/2}(\bm{\theta}_{0}) H_{t}^{-1}(\bm{\theta}) \frac{\partial H_{t}(\bm{\theta})}{\partial\theta_{\ell}} H_{t}^{-1}(\bm{\theta}) H_{t}^{1/2}(\bm{\theta}_{0}) \bm{\eta}_{t}\right)\right|^{2} \notag\\
				=& \left|\tr\left(H_{t}^{1/2}(\bm{\theta}_{0}) H_{t}^{-1}(\bm{\theta}) \frac{\partial H_{t}(\bm{\theta})}{\partial\theta_{\ell}} H_{t}^{-1}(\bm{\theta}) H_{t}^{1/2}(\bm{\theta}_{0}) \bm{\eta}_{t} \bm{\eta}_{t}^{\prime}\right)\right|^{2} \notag\\
				\leq& \left\|H_{t}^{1/2}(\bm{\theta}_{0}) H_{t}^{-1}(\bm{\theta}) \frac{\partial H_{t}(\bm{\theta})}{\partial\theta_{\ell}} H_{t}^{-1}(\bm{\theta}) H_{t}^{1/2}(\bm{\theta}_{0})\right\|_{F}^{2} \left\|\bm{\eta}_{t} \bm{\eta}_{t}^{\prime}\right\|_{F}^{2}. 
			\end{align}
			Moreover, by \eqref{eq first derivative of Ht}, $\|A\|_{F}^{2} = \tr(A^{\prime} A)$ and $\tr(AB) = \tr(BA) \leq \|A\|_{F} \|B\|_{F}$, it holds that 
			\begin{align} \label{eq sec0 about Dtinverse Dt}
				&\left\|H_{t}^{1/2}(\bm{\theta}_{0}) H_{t}^{-1}(\bm{\theta}) \frac{\partial H_{t}(\bm{\theta})}{\partial\theta_{\ell}} H_{t}^{-1}(\bm{\theta}) H_{t}^{1/2}(\bm{\theta}_{0})\right\|_{F}^{2} \notag\\
				=& \tr\left(H_{t}^{1/2}(\bm{\theta}_{0}) H_{t}^{-1}(\bm{\theta}) \frac{\partial H_{t}(\bm{\theta})}{\partial\theta_{\ell}} H_{t}^{-1}(\bm{\theta}) H_{t}(\bm{\theta}_{0}) H_{t}^{-1}(\bm{\theta}) \frac{\partial H_{t}(\bm{\theta})}{\partial\theta_{\ell}} H_{t}^{-1}(\bm{\theta}) H_{t}^{1/2}(\bm{\theta}_{0})\right) \notag\\
				=& \tr\left(\frac{\partial H_{t}(\bm{\theta})}{\partial\theta_{\ell}} H_{t}^{-1}(\bm{\theta}) H_{t}(\bm{\theta}_{0}) H_{t}^{-1}(\bm{\theta}) \frac{\partial H_{t}(\bm{\theta})}{\partial\theta_{\ell}} H_{t}^{-1}(\bm{\theta}) H_{t}(\bm{\theta}_{0}) H_{t}^{-1}(\bm{\theta})\right), 
			\end{align}
			with 
			\begin{align} \label{eq sec1 about Dtinverse Dt}
				\tr&\left(\frac{\partial H_{t}(\bm{\theta})}{\partial\delta_{\ell}} H_{t}^{-1}(\bm{\theta}) H_{t}(\bm{\theta}_{0}) H_{t}^{-1}(\bm{\theta}) \frac{\partial H_{t}(\bm{\theta})}{\partial\delta_{\ell}} H_{t}^{-1}(\bm{\theta}) H_{t}(\bm{\theta}_{0}) H_{t}^{-1}(\bm{\theta})\right) \notag\\
				= \tr&\left(D_{t}^{-1}(\bm{\delta}) \frac{\partial D_{t}(\bm{\delta})}{\partial\delta_{\ell}} D_{t}^{-1}(\bm{\delta}) D_{t}(\bm{\delta}_{0}) R_{t}(\bm{\theta}_{0}) D_{t}(\bm{\delta}_{0}) D_{t}^{-1}(\bm{\delta}) R_{t}^{-1}(\bm{\theta})\right. \notag\\ &\left.D_{t}^{-1}(\bm{\delta}) \frac{\partial D_{t}(\bm{\delta})}{\partial\delta_{\ell}} D_{t}^{-1}(\bm{\delta}) D_{t}(\bm{\delta}_{0}) R_{t}(\bm{\theta}_{0}) D_{t}(\bm{\delta}_{0}) D_{t}^{-1}(\bm{\delta}) R_{t}^{-1}(\bm{\theta})\right) \notag\\
				+ \tr&\left(\frac{\partial R_{t}(\bm{\theta})}{\partial\delta_{\ell}} R_{t}^{-1}(\bm{\theta}) D_{t}^{-1}(\bm{\delta}) D_{t}(\bm{\delta}_{0}) R_{t}(\bm{\theta}_{0}) D_{t}(\bm{\delta}_{0}) D_{t}^{-1}(\bm{\delta}) R_{t}^{-1}(\bm{\theta})\right. \notag\\ &\left.\frac{\partial R_{t}(\bm{\theta})}{\partial\delta_{\ell}} R_{t}^{-1}(\bm{\theta}) D_{t}^{-1}(\bm{\delta}) D_{t}(\bm{\delta}_{0}) R_{t}(\bm{\theta}_{0}) D_{t}(\bm{\delta}_{0}) D_{t}^{-1}(\bm{\delta}) R_{t}^{-1}(\bm{\theta})\right) \notag\\
				+ \tr&\left(\frac{\partial D_{t}(\bm{\delta})}{\partial\delta_{\ell}} D_{t}^{-1}(\bm{\delta}) R_{t}^{-1}(\bm{\theta}) D_{t}^{-1}(\bm{\delta}) D_{t}(\bm{\delta}_{0}) R_{t}(\bm{\theta}_{0}) D_{t}(\bm{\delta}_{0}) D_{t}^{-1}(\bm{\delta})\right. \notag\\ &\left.\frac{\partial D_{t}(\bm{\delta})}{\partial\delta_{\ell}} D_{t}^{-1}(\bm{\delta}) R_{t}^{-1}(\bm{\theta}) D_{t}^{-1}(\bm{\delta}) D_{t}(\bm{\delta}_{0}) R_{t}(\bm{\theta}_{0}) D_{t}(\bm{\delta}_{0}) D_{t}^{-1}(\bm{\delta})\right) \notag\\
				\leq\mathrel{\phantom{\tr}}& \left\|D_{t}^{-1}(\bm{\delta}) \frac{\partial D_{t}(\bm{\delta})}{\partial\delta_{\ell}}\right\|_{F}^{2} \|D_{t}^{-1}(\bm{\delta}) D_{t}(\bm{\delta}_{0})\|_{F}^{2} \|R_{t}(\bm{\theta}_{0})\|_{F}^{2} \|D_{t}(\bm{\delta}_{0}) D_{t}^{-1}(\bm{\delta})\|_{F}^{2} \|R_{t}^{-1}(\bm{\theta})\|_{F}^{2} \notag\\
				+\mathrel{\phantom{\tr}}& \left\|\frac{\partial R_{t}(\bm{\theta})}{\partial\delta_{\ell}} R_{t}^{-1}(\bm{\theta})\right\|_{F}^{2} \|D_{t}^{-1}(\bm{\delta}) D_{t}(\bm{\delta}_{0})\|_{F}^{2} \|R_{t}(\bm{\theta}_{0})\|_{F}^{2} \|D_{t}(\bm{\delta}_{0}) D_{t}^{-1}(\bm{\delta})\|_{F}^{2} \|R_{t}^{-1}(\bm{\theta})\|_{F}^{2} \notag\\
				+\mathrel{\phantom{\tr}}& \left\|\frac{\partial D_{t}(\bm{\delta})}{\partial\delta_{\ell}} D_{t}^{-1}(\bm{\delta})\right\|_{F}^{2} \|R_{t}^{-1}(\bm{\theta})\|_{F}^{2} \|D_{t}^{-1}(\bm{\delta}) D_{t}(\bm{\delta}_{0})\|_{F}^{2} \|R_{t}(\bm{\theta}_{0})\|_{F}^{2} \|D_{t}(\bm{\delta}_{0}) D_{t}^{-1}(\bm{\delta})\|_{F}^{2}, 
			\end{align}
			and 
			\begin{align} \label{eq sec2 about Dtinverse Dt}
				\tr&\left(\frac{\partial H_{t}(\bm{\theta})}{\partial\beta_{\ell}} H_{t}^{-1}(\bm{\theta}) H_{t}(\bm{\theta}_{0}) H_{t}^{-1}(\bm{\theta}) \frac{\partial H_{t}(\bm{\theta})}{\partial\beta_{\ell}} H_{t}^{-1}(\bm{\theta}) H_{t}(\bm{\theta}_{0}) H_{t}^{-1}(\bm{\theta})\right) \notag\\
				= \tr&\left(\frac{\partial R_{t}(\bm{\theta})}{\partial\beta_{\ell}} R_{t}^{-1}(\bm{\theta}) D_{t}^{-1}(\bm{\delta}) D_{t}(\bm{\delta}_{0}) R_{t}(\bm{\theta}_{0}) D_{t}(\bm{\delta}_{0}) D_{t}^{-1}(\bm{\delta}) R_{t}^{-1}(\bm{\theta})\right. \notag\\ &\left.\frac{\partial R_{t}(\bm{\theta})}{\partial\beta_{\ell}} R_{t}^{-1}(\bm{\theta}) D_{t}^{-1}(\bm{\delta}) D_{t}(\bm{\delta}_{0}) R_{t}(\bm{\theta}_{0}) D_{t}(\bm{\delta}_{0}) D_{t}^{-1}(\bm{\delta}) R_{t}^{-1}(\bm{\theta})\right) \notag\\
				\leq\mathrel{\phantom{\tr}}& \left\|\frac{\partial R_{t}(\bm{\theta})}{\partial\beta_{\ell}} R_{t}^{-1}(\bm{\theta})\right\|_{F}^{2} \|D_{t}^{-1}(\bm{\delta}) D_{t}(\bm{\delta}_{0})\|_{F}^{2} \|R_{t}(\bm{\theta}_{0})\|_{F}^{2} \|D_{t}(\bm{\delta}_{0}) D_{t}^{-1}(\bm{\delta})\|_{F}^{2} \|R_{t}^{-1}(\bm{\theta})\|_{F}^{2}. 
			\end{align}
			Note that by $|e^{x}| \leq e^{|x|}$, we have that 
			\begin{align*}
				&\sup_{\bm{\theta} \in \Theta} \left|\frac{h_{ii,t}(\bm{\delta}_{0})}{h_{ii,t}(\bm{\delta})}\right| 
				= \sup_{\bm{\theta} \in \Theta} \left|\frac{\exp\{\ln h_{ii,t}(\bm{\delta}_{0})\}}{\exp\{\ln h_{ii,t}(\bm{\delta})\}}\right| 
				= \sup_{\bm{\theta} \in \Theta} \left|\exp\left\{\ln h_{ii,t}(\bm{\delta}_{0}) - \ln h_{ii,t}(\bm{\delta})\right\}\right| \\
				\leq& \sup_{\bm{\theta} \in \Theta} \exp\left\{\left|\ln h_{ii,t}(\bm{\delta}_{0}) - \ln h_{ii,t}(\bm{\delta})\right|\right\} 
				= \exp\left\{\sup_{\bm{\theta} \in \Theta} \left|\ln h_{ii,t}(\bm{\delta}_{0}) - \ln h_{ii,t}(\bm{\delta})\right|\right\}. 
			\end{align*}
			This together with Lemma \ref{lemma moments of derivatives of lnhtunderline}($\romannumeral1$) implies that 
			\begin{align} \label{eq Dtinverse Dt finite}
				\sup_{\bm{\theta} \in \Theta} \|D_{t}^{-1}(\bm{\delta}) D_{t}(\bm{\delta}_{0})\|_{F}^{2} \;\; \text{is finite a.s.}. 
			\end{align}
			Then by \eqref{eq sec0 about Dtinverse Dt}--\eqref{eq sec2 about Dtinverse Dt}, 
			together with \eqref{eq Dtinverse Dt finite}, $\sup_{\bm{\theta} \in \Theta} \|R_{t}(\bm{\theta})\|$ and $\sup_{\bm{\theta} \in \Theta} \|R_{t}^{-1}(\bm{\theta})\|$ are finite by Lemma \ref{lemma varrho to the power of t}($\romannumeral3$), $E \sup_{\bm{\theta} \in \Theta}$ $\|D_{t}^{-1}(\bm{\delta}) \frac{\partial D_{t}(\bm{\delta})}{\partial\delta_{\ell}}\|^{2+\epsilon} < \infty$ and $E \sup_{\bm{\theta} \in \Theta} \|R_{t}^{-1}(\bm{\theta}) \frac{\partial R_{t}(\bm{\theta})}{\partial\theta_{\ell}}\|^{2+\epsilon} < \infty$ by Lemma \ref{lemma moments of Htinverse product derivatives of Ht}, $\|A\|_{F} \leq \sqrt{c_{1}} \|A\|_{2}$ for any matrix $A$ with $c_{1} = \rank(A)$, 
			it can be shown that 
			\begin{align} \label{eq Esec0 about Dtinverse Dt}
				E \left\|H_{t}^{1/2}(\bm{\theta}_{0}) H_{t}^{-1}(\bm{\theta}_{0}) \frac{\partial H_{t}(\bm{\theta}_{0})}{\partial\theta_{\ell}} H_{t}^{-1}(\bm{\theta}_{0}) H_{t}^{1/2}(\bm{\theta}_{0})\right\|_{F}^{2} < \infty, 
			\end{align}
			and 
			\begin{align} \label{eq Ht Htinverse dHt Htinver Ht finite}
				\sup_{\bm{\theta} \in \Theta} \left\|H_{t}^{1/2}(\bm{\theta}_{0}) H_{t}^{-1}(\bm{\theta}) \frac{\partial H_{t}(\bm{\theta})}{\partial\theta_{\ell}} H_{t}^{-1}(\bm{\theta}) H_{t}^{1/2}(\bm{\theta}_{0})\right\|_{F}^{2} \;\; \text{is finite a.s.}. 
			\end{align}
			Since $H_{t}(\bm{\theta})$ is measurable with respect to $\mathcal{F}_{t-1}$, by \eqref{eq ytprime Htinverse dHt Htinverse yt}, \eqref{eq Esec0 about Dtinverse Dt}, $\|A\|_{F} \leq \sqrt{c_{1}} \|A\|_{2}$ for any matrix $A$ with $c_{1} = \rank(A)$ and $E \left\|\bm{\eta}_{t} \bm{\eta}_{t}^{\prime}\right\|_{2}^{2} < \infty$, we can obtain that 
			\begin{align*}
				&E \left|\mathbf{y}_{t}^{\prime} H_{t}^{-1}(\bm{\theta}_{0}) \frac{\partial H_{t}(\bm{\theta}_{0})}{\partial\theta_{\ell}} H_{t}^{-1}(\bm{\theta}_{0}) \mathbf{y}_{t}\right|^{2} \\
				\leq& E \left(\left\|H_{t}^{1/2}(\bm{\theta}_{0}) H_{t}^{-1}(\bm{\theta}_{0}) \frac{\partial H_{t}(\bm{\theta}_{0})}{\partial\theta_{\ell}} H_{t}^{-1}(\bm{\theta}_{0}) H_{t}^{1/2}(\bm{\theta}_{0})\right\|_{F}^{2} \left\|\bm{\eta}_{t} \bm{\eta}_{t}^{\prime}\right\|_{F}^{2}\right) \\
				\leq& E \left\|H_{t}^{1/2}(\bm{\theta}_{0}) H_{t}^{-1}(\bm{\theta}_{0}) \frac{\partial H_{t}(\bm{\theta}_{0})}{\partial\theta_{\ell}} H_{t}^{-1}(\bm{\theta}_{0}) H_{t}^{1/2}(\bm{\theta}_{0})\right\|_{F}^{2} \cdot 
				E \left\|\bm{\eta}_{t} \bm{\eta}_{t}^{\prime}\right\|_{2}^{2} 
				< \infty, 
			\end{align*}
			and 
			\begin{align*}
				&\sup_{\bm{\theta} \in \Theta} \left|\mathbf{y}_{t}^{\prime} H_{t}^{-1}(\bm{\theta}) \frac{\partial H_{t}(\bm{\theta})}{\partial\theta_{\ell}} H_{t}^{-1}(\bm{\theta}) \mathbf{y}_{t}\right|^{2} \\
				\leq& \sup_{\bm{\theta} \in \Theta} \left\|H_{t}^{1/2}(\bm{\theta}_{0}) H_{t}^{-1}(\bm{\theta}) \frac{\partial H_{t}(\bm{\theta})}{\partial\theta_{\ell}} H_{t}^{-1}(\bm{\theta}) H_{t}^{1/2}(\bm{\theta}_{0})\right\|_{F}^{2} \left\|\bm{\eta}_{t} \bm{\eta}_{t}^{\prime}\right\|_{2}^{2} \;\; \text{is finite a.s.}. 
			\end{align*}
			Thus (a1) holds. 
			For (a2), by Lemmas \ref{lemma moments of Htinverse product derivatives of Ht}($\romannumeral1$)--($\romannumeral2$) together with the facts $\tr(AB) = \tr(BA) \leq \|A\|_{F} \|B\|_{F}$ and $\|A\|_{F} \leq \sqrt{c_{1}} \|A\|_{2}$ for any matrices $A$ and $B$, it can be shown that 
			$$
				E \sup_{\bm{\theta} \in \Theta} \left|\tr\left(D_{t}^{-1}(\bm{\delta}) \frac{\partial D_{t}(\bm{\delta})}{\partial\delta_{\ell}}\right)\right|^{2} 
				\leq \sqrt{m} E \sup_{\bm{\theta} \in \Theta} \left\|D_{t}^{-1}(\bm{\delta}) \frac{\partial D_{t}(\bm{\delta})}{\partial\delta_{\ell}}\right\|_{F}^{2} 
				\leq m E \sup_{\bm{\theta} \in \Theta} \left\|D_{t}^{-1}(\bm{\delta}) \frac{\partial D_{t}(\bm{\delta})}{\partial\delta_{\ell}}\right\|_{2}^{2} 
				< \infty, 
			$$
			and 
			$$
				E \sup_{\bm{\theta} \in \Theta} \left|\tr\left(R_{t}^{-1}(\bm{\theta}) \frac{\partial R_{t}(\bm{\theta})}{\partial\theta_{\ell}}\right)\right|^{2} 
				\leq \sqrt{m} E \sup_{\bm{\theta} \in \Theta} \left\|R_{t}^{-1}(\bm{\theta}) \frac{\partial R_{t}(\bm{\theta})}{\partial\theta_{\ell}}\right\|_{F}^{2} 
				\leq m E \sup_{\bm{\theta} \in \Theta} \left\|R_{t}^{-1}(\bm{\theta}) \frac{\partial R_{t}(\bm{\theta})}{\partial\theta_{\ell}}\right\|_{2}^{2} 
				< \infty. 
			$$
			Hence (a2) holds. 
			As a result, ($\romannumeral1$) holds by (a1)--(a2). 

			Next we show ($\romannumeral2$). 
			By \eqref{eq second derivative of lt}, it suffice to show that 
			\begin{enumerate}
				\item[(b1)] $E \left|\mathbf{y}_{t}^{\prime} \left(2 H_{t}^{-1}(\bm{\theta}_{0}) \frac{\partial H_{t}(\bm{\theta}_{0})}{\partial\theta_{k}} H_{t}^{-1}(\bm{\theta}_{0}) \frac{\partial H_{t}(\bm{\theta}_{0})}{\partial\theta_{\ell}} H_{t}^{-1}(\bm{\theta}_{0}) - H_{t}^{-1}(\bm{\theta}_{0}) \frac{\partial^{2} H_{t}(\bm{\theta}_{0})}{\partial\theta_{k} \partial\theta_{\ell}} H_{t}^{-1}(\bm{\theta}_{0})\right) \mathbf{y}_{t}\right| < \infty$ and $\sup_{\bm{\theta} \in \Theta} \left|\mathbf{y}_{t}^{\prime} \left(2 H_{t}^{-1}(\bm{\theta}) \frac{\partial H_{t}(\bm{\theta})}{\partial\theta_{k}} H_{t}^{-1}(\bm{\theta}) \frac{\partial H_{t}(\bm{\theta})}{\partial\theta_{\ell}} H_{t}^{-1}(\bm{\theta}) - H_{t}^{-1}(\bm{\theta}) \frac{\partial^{2} H_{t}(\bm{\theta})}{\partial\theta_{k} \partial\theta_{\ell}} H_{t}^{-1}(\bm{\theta})\right) \mathbf{y}_{t}\right|$ is finite a.s.; 
				\item[(b2)] $E \sup_{\bm{\theta} \in \Theta} \left|\tr\left(- H_{t}^{-1}(\bm{\theta}) \frac{\partial H_{t}(\bm{\theta})}{\partial\theta_{k}} H_{t}^{-1}(\bm{\theta}) \frac{\partial H_{t}(\bm{\theta})}{\partial\theta_{\ell}} + H_{t}^{-1}(\bm{\theta}) \frac{\partial^{2} H_{t}(\bm{\theta})}{\partial\theta_{k} \partial\theta_{\ell}}\right)\right| < \infty$. 
			\end{enumerate}
			For (b1), similarly it holds that 
			\begin{align*}
				&\left|\mathbf{y}_{t}^{\prime} \left(2 H_{t}^{-1}(\bm{\theta}) \frac{\partial H_{t}(\bm{\theta})}{\partial\theta_{k}} H_{t}^{-1}(\bm{\theta}) \frac{\partial H_{t}(\bm{\theta})}{\partial\theta_{\ell}} H_{t}^{-1}(\bm{\theta}) - H_{t}^{-1}(\bm{\theta}) \frac{\partial^{2} H_{t}(\bm{\theta})}{\partial\theta_{k} \partial\theta_{\ell}} H_{t}^{-1}(\bm{\theta})\right) \mathbf{y}_{t}\right| \\
				\leq& 2 \left|\tr\left(H_{t}^{1/2}(\bm{\theta}_{0}) H_{t}^{-1}(\bm{\theta}) \frac{\partial H_{t}(\bm{\theta})}{\partial\theta_{k}} H_{t}^{-1}(\bm{\theta}) \frac{\partial H_{t}(\bm{\theta})}{\partial\theta_{\ell}} H_{t}^{-1}(\bm{\theta}) H_{t}^{1/2}(\bm{\theta}_{0}) \bm{\eta}_{t} \bm{\eta}_{t}^{\prime}\right)\right| \\
				&+ \left|\tr\left(H_{t}^{1/2}(\bm{\theta}_{0}) H_{t}^{-1}(\bm{\theta}) \frac{\partial^{2} H_{t}(\bm{\theta})}{\partial\theta_{k} \partial\theta_{\ell}} H_{t}^{-1}(\bm{\theta}) H_{t}^{1/2}(\bm{\theta}_{0}) \bm{\eta}_{t} \bm{\eta}_{t}^{\prime}\right)\right| \\
				\leq& 2 \left\|H_{t}^{1/2}(\bm{\theta}_{0}) H_{t}^{-1}(\bm{\theta}) \frac{\partial H_{t}(\bm{\theta})}{\partial\theta_{k}} H_{t}^{-1}(\bm{\theta}) \frac{\partial H_{t}(\bm{\theta})}{\partial\theta_{\ell}} H_{t}^{-1}(\bm{\theta}) H_{t}^{1/2}(\bm{\theta}_{0})\right\|_{F} \left\|\bm{\eta}_{t} \bm{\eta}_{t}^{\prime}\right\|_{F} \\
				&+ \left\|H_{t}^{1/2}(\bm{\theta}_{0}) H_{t}^{-1}(\bm{\theta}) \frac{\partial^{2} H_{t}(\bm{\theta})}{\partial\theta_{k} \partial\theta_{\ell}} H_{t}^{-1}(\bm{\theta}) H_{t}^{1/2}(\bm{\theta}_{0})\right\|_{F} \left\|\bm{\eta}_{t} \bm{\eta}_{t}^{\prime}\right\|_{F}.  
			\end{align*}
			Then with analogous arguments in the proof of (a1), we can show that (b1) holds. 
			Moreover, (b2) can be verified similarly as (a2). 
			As a result, ($\romannumeral2$) holds by (b1)--(b2). 

			($\romannumeral3$) can be established with similar arguments in the proof of ($\romannumeral1$) and ($\romannumeral2$). 
		\end{proof}

		\begin{proof}[Proof of Lemma \ref{lemma positivity of Sigma*}] 
			By \eqref{eq second derivative of lt}, we have that 
			\begin{align*}
				E\left(\frac{\partial^{2}\ell_{t}(\bm{\theta}_{0})}{\partial\theta_{k} \partial\theta_{\ell}} \mid \mathcal{F}_{t-1}\right) 
				=& \frac{1}{2} \tr\left(H_{t}^{-1}(\bm{\theta}_{0}) \frac{\partial H_{t}(\bm{\theta}_{0})}{\partial\theta_{k}} H_{t}^{-1}(\bm{\theta}_{0}) \frac{\partial H_{t}(\bm{\theta}_{0})}{\partial\theta_{\ell}}\right) \\
				=& \frac{1}{2} \tr\left(H_{t}^{-1/2}(\bm{\theta}_{0}) \frac{\partial H_{t}(\bm{\theta}_{0})}{\partial\theta_{k}} H_{t}^{-1/2}(\bm{\theta}_{0}) H_{t}^{-1/2}(\bm{\theta}_{0}) \frac{\partial H_{t}(\bm{\theta}_{0})}{\partial\theta_{\ell}} H_{t}^{-1/2}(\bm{\theta}_{0})\right). 
			\end{align*}
			Let $\dot{\bm{h}}_{\ell} = \ovec({\partial H_{t}(\bm{\theta}_{0})}/{\partial\theta_{\ell}})$, 
			$\overline{\bm{h}}_{\ell} = \ovec(H_{t}^{-1/2}(\bm{\theta}_{0}) ({\partial H_{t}(\bm{\theta}_{0})}/{\partial\theta_{\ell}}) H_{t}^{-1/2}(\bm{\theta}_{0}))$, 
			and $\mathbb{H} = H_{t}^{-1/2}(\bm{\theta}_{0}) \otimes H_{t}^{-1/2}(\bm{\theta}_{0})$. 
			Then using the facts $\tr(A^{\prime} B) = (\ovec(A))^{\prime} \ovec(B)$ and $\ovec(A B C) = (C^{\prime} \otimes A) \ovec(B)$ for any matrices $A$, $B$ and $C$, 
			it holds that  
			\begin{align*}
				E\left(\frac{\partial^{2}\ell_{t}(\bm{\theta}_{0})}{\partial\theta_{k} \partial\theta_{\ell}} \mid \mathcal{F}_{t-1}\right) = \frac{1}{2} \overline{\bm{h}}_{k}^{\prime} \overline{\bm{h}}_{\ell} 
				\;\;\;\; \text{and} \;\;\;\; 
				\overline{\bm{h}}_{\ell} = \mathbb{H} \dot{\bm{h}}_{\ell}. 
			\end{align*}
			Denote $d = m + (r + 2s) (1 + m^2) + 2 + m(m-1)/2$ as the dimension of $\bm{\theta}$. 
			Let $\overline{\mathbb{H}} = (\overline{\bm{h}}_{1}, \ldots, \overline{\bm{h}}_{d})$ 
			and $\dot{\mathbb{H}} = (\dot{\bm{h}}_{1}, \ldots, \dot{\bm{h}}_{d})$. 
			Then we can obtain that $\overline{\mathbb{H}} = \mathbb{H} \dot{\mathbb{H}}$, and 
			\begin{align*}
				E\left(\frac{\partial^{2}\ell_{t}(\bm{\theta}_{0})}{\partial\bm{\theta} \partial\bm{\theta}} \mid \mathcal{F}_{t-1}\right) 
				= \frac{1}{2} \overline{\mathbb{H}}^{\prime} \overline{\mathbb{H}} 
				= \frac{1}{2} \dot{\mathbb{H}}^{\prime} \mathbb{H}^{\prime} \mathbb{H} \dot{\mathbb{H}} 
				= \frac{1}{2} \dot{\mathbb{H}}^{\prime} \mathbb{H}^{2} \dot{\mathbb{H}}. 
			\end{align*}
			And it follows that 
			\begin{align*}
				\Sigma_{*} = E\left(\frac{\partial^{2}\ell_{t}(\bm{\theta}_{0})}{\partial\bm{\theta} \partial\bm{\theta}}\right) 
				= \frac{1}{2} E\left(\dot{\mathbb{H}}^{\prime} \mathbb{H}^{2} \dot{\mathbb{H}}\right). 
			\end{align*}

			Suppose that $\Sigma_{*}$ is singular. 
			Then there exists a nonzero constant vector $\bm{c} \in \mathbb{R}^{d}$ such that $\bm{c}^{\prime} \Sigma_{*} \bm{c} = \frac{1}{2} E(\bm{c}^{\prime} \dot{\mathbb{H}}^{\prime} \mathbb{H}^{2} \dot{\mathbb{H}} \bm{c}) = 0$. 
			Note that $\dot{\mathbb{H}}^{\prime} \mathbb{H}^{2} \dot{\mathbb{H}}$ is positive semi-definite a.s., thus we have $\bm{c}^{\prime} \dot{\mathbb{H}}^{\prime} \mathbb{H}^{2} \dot{\mathbb{H}} \bm{c} = 0$ a.s.. 
			This together with that $\mathbb{H}$ is positive definite a.s., implies that 
			\begin{align} \label{eq dotmathbbH c equal to 0}
				\dot{\mathbb{H}} \bm{c} = 0 \;\; \text{a.s.}. 
			\end{align}
			Recall that $\bm{\theta} = (\bm{\delta}^{\prime}, \bm{\beta}^{\prime})^{\prime}$. 
			Denote $\bm{c} = (\bm{c}_{1}^{\prime}, \bm{c}_{2}^{\prime})^{\prime}$ with $\bm{c}_{1} = (c_{11}, \ldots, c_{1 d_{\delta}})^{\prime}$ and $\bm{c}_{2} = (c_{21}, \ldots, c_{2 d_{\beta}})^{\prime}$, 
			where $d_{\delta}$ and $d_{\beta}$ are the dimensions of $\bm{\delta}$ and $\bm{\beta}$, respectively. 
			Firstly, for $1 \leq i \leq m$, the $[(i-1)m + i]$-th equation in \eqref{eq dotmathbbH c equal to 0} is 
			\begin{align*}
				0 = \sum_{\ell=1}^{d_{\delta}} c_{1 \ell} \frac{\partial h_{ii,t}(\bm{\delta_{0}})}{\partial\delta_{\ell}} = \sum_{\ell=1}^{d_{\delta}} c_{1 \ell} h_{ii,t}(\bm{\delta_{0}}) \frac{\partial\ln h_{ii,t}(\bm{\delta_{0}})}{\partial\delta_{\ell}} \;\; \text{a.s.}, 
			\end{align*}
			which implies that $D_{t}^{2}(\bm{\delta}_{0}) \sum_{\ell=1}^{d_{\delta}} c_{1 \ell} {\partial\ln\bm{h}_{t}(\bm{\delta}_{0})}/{\partial\delta_{\ell}} = 0$ a.s., and then $\sum_{\ell=1}^{d_{\delta}} c_{1 \ell} {\partial\ln\bm{h}_{t}(\bm{\delta}_{0})}/{\partial\delta_{\ell}} = 0$ a.s.. 
			Note that by \eqref{eq first derivative of lnhtunderline} and Lemma \ref{lemma for identification}, $\{{\partial\ln\bm{h}_{t}(\bm{\delta}_{0})}/{\partial\delta_{\ell}}, 1 \leq \ell \leq d_{\delta}\}$ are linearly independent. 
			Thus we have $\bm{c}_{1} = 0$. 
			Secondly, by \eqref{eq dotmathbbH c equal to 0}, $\bm{c}_{1} = 0$ and $H_{t}(\bm{\theta}_{0}) = D_{t}(\bm{\delta}_{0}) R_{t}(\bm{\theta}_{0}) D_{t}(\bm{\delta}_{0})$, it can be obtained that 
			\begin{align*}
				0 = \sum_{\ell=1}^{d_{\beta}} c_{2 \ell} \frac{\partial \ovec(H_{t}(\bm{\theta}_{0}))}{\partial\beta_{\ell}} 
				=& \sum_{\ell=1}^{d_{\beta}} c_{2 \ell} \frac{\partial \left[\left(D_{t}(\bm{\delta}_{0}) \otimes D_{t}(\bm{\delta}_{0})\right) \ovec(R_{t}(\bm{\theta}_{0}))\right]}{\partial\beta_{\ell}} \\
				=& \left(D_{t}(\bm{\delta}_{0}) \otimes D_{t}(\bm{\delta}_{0})\right) \sum_{\ell=1}^{d_{\beta}} c_{2 \ell} \frac{\partial\ovec(R_{t}(\bm{\theta}_{0}))}{\partial\beta_{\ell}} \;\; \text{a.s.}, 
			\end{align*}
			and then $\sum_{\ell=1}^{d_{\beta}} c_{2 \ell} {\partial\ovec(R_{t}(\bm{\theta}_{0}))}/{\partial\beta_{\ell}} = 0$ a.s.. 
			Since $\{{\partial\ovec(R_{t}(\bm{\theta}_{0}))}/{\partial\beta_{\ell}}, 1 \leq \ell \leq d_{\beta}\}$ are linearly independent by \eqref{eq first derivative of Rt}, we can conclude that $\bm{c}_{2} = 0$. 
			Above all, we have that $\bm{c} = 0$, which is in contradiction with that $\bm{c}$ is nonzero. 
			As a result, $\Sigma_{*}$ is non-singular. 
		\end{proof}

	\section{Extension to a general multivariate log-GARCH model} \label{section SGARCH(q,r*,s*)}
		For general orders $p$ and $q$, denote 
		$$
			\bm{h}_{t}^{*} = 
			\left(
			\begin{matrix}
				\bm{h}_{t} \\ \bm{h}_{t-1} \\ \vdots \\\bm{h}_{t-p+2} \\ \bm{h}_{t-p+1} 
			\end{matrix}
			\right), \;\; 
			B_{*} = 
			\left(
			\begin{matrix}
				B_{1} & B_{2} & \cdots & B_{p-1} & B_{p} \\
				I_{m} & 0_{m} & \cdots & 0_{m} & 0_{m} \\
				\vdots & \vdots & & \vdots & \vdots \\
				0_{m} & 0_{m} & \cdots & 0_{m} & 0_{m} \\
				0_{m} & 0_{m} & \cdots & I_{m} & 0_{m} 
			\end{matrix}
			\right)
			\;\; \text{and} \;\; 
			\underline{\bm{c}}_{t} = 
			\left(
			\begin{matrix}
				\bm{\omega} + \sum_{i=1}^{q} A_{i} \mathbf{y}^{2}_{t-i} \\
				\bm{0}_{m} \\ \vdots \\ \bm{0}_{m} \\ \bm{0}_{m} 
			\end{matrix}
			\right). 
		$$
		Supposing $\rho(B_{*}) < 1$, 
		model \eqref{eq volatility in DCC} can be rewritten as follows, 
		$$
			\bm{h}_{t}^{*} = \sum_{i=0}^{\infty} B_{*}^{i} \underline{\bm{c}}_{t-i}. 
		$$
		Then the multivariate log-ARCH($\infty$) form for model \eqref{eq volatility in DCC} is 
		$$
			\bm{h}_{t} = \sum_{i=0}^{\infty} \mathcal{I} B_{*}^{i} \mathcal{I}^{\prime} \bm{\omega} + \sum_{i=1}^{\infty} \sum_{j=1}^{\min\{i,q\}} \mathcal{I} B_{*}^{i-j} \mathcal{I}^{\prime} A_{j} \underline{\bm{y}}_{t-i}, 
		$$
		where $\mathcal{I} = (I_{m}, 0_{m \times m(p-1)})$ is an $m \times mp$ known matrix. 
		We assume that $B_{*}$ is diagonalizable over the complex field with $r$ nonzero real eigenvalues and $s$ conjugate pairs of nonzero complex eigenvalues. 
		Then $B_{*}$ can be decomposed by $B_{*} = P_{*} J_{*} P_{*}^{-1}$, 
		where 
		$P_{*}$ is an $mp \times mp$ invertible matrix, 
		and $J_{*} = \Diag\{\lambda_{1}, \ldots, \lambda_{r}, C_{1}, \ldots, C_{s}, \bm{0}_{mp-r-2s}\}$ is an $mp \times mp$ real block diagonal matrix. 
		Hence it holds that 
		\begin{align*}
			\bm{h}_{t} 
			&= \sum_{i=0}^{\infty} \mathcal{I} B_{*}^{i} \mathcal{I}^{\prime} \bm{\omega} 
			+ \sum_{i=1}^{q-1} \sum_{j=1}^{i} \mathcal{I} B_{*}^{i-j} \mathcal{I}^{\prime} A_{j} \mathbf{y}^{2}_{t-i} 
			+ \sum_{i=q}^{\infty} \sum_{j=1}^{q} \mathcal{I} B_{*}^{i-j} \mathcal{I}^{\prime} A_{j} \mathbf{y}^{2}_{t-i} \\
			&= \sum_{i=0}^{\infty} \mathcal{I} B_{*}^{i} \mathcal{I}^{\prime} \bm{\omega} 
			+ \sum_{i=1}^{q-1} \sum_{j=1}^{i} \mathcal{I} B_{*}^{i-j} \mathcal{I}^{\prime} A_{j} \mathbf{y}^{2}_{t-i} 
			+ \sum_{i=q}^{\infty} \mathcal{I} B_{*}^{i-q} \sum_{j=1}^{q} B_{*}^{q-j} \mathcal{I}^{\prime} A_{j} \mathbf{y}^{2}_{t-i} \\
			&= \sum_{i=0}^{\infty} \mathcal{I} B_{*}^{i} \mathcal{I}^{\prime} \bm{\omega} 
			+ \sum_{i=1}^{q-1} \sum_{j=1}^{i} \mathcal{I} B_{*}^{i-j} \mathcal{I}^{\prime} A_{j} \mathbf{y}^{2}_{t-i} 
			+ \sum_{i=q}^{\infty} \mathcal{I} P_{*} J_{*}^{i-q} P_{*}^{-1} \sum_{j=1}^{q} B_{*}^{q-j} \mathcal{I}^{\prime} A_{j} \mathbf{y}^{2}_{t-i} \\
			&= \sum_{i=0}^{\infty} \mathcal{I} B_{*}^{i} \mathcal{I}^{\prime} \bm{\omega} 
			+ \sum_{i=1}^{q-1} \Phi^{*}_{i} \mathbf{y}^{2}_{t-i} 
			+ \sum_{i=q}^{\infty} \underline{B} J_{*}^{i-q} \underline{A} \mathbf{y}^{2}_{t-i}, 
		\end{align*}
		where $\Phi^{*}_{i} = \sum_{j=1}^{i} \mathcal{I} B_{*}^{i-j} \mathcal{I}^{\prime} A_{j}$, $\underline{B} = \mathcal{I} P_{*}$ and $\underline{A} = P_{*}^{-1} \sum_{j=1}^{q} B_{*}^{q-j} \mathcal{I}^{\prime} A_{j}$. 
		Consequently, with the similar arguments in Section \ref{section connection between DCC and SGARCH}, 
		we propose a general multivariate log-GARCH model as described in Remark \ref{remark general SGARCH}: 
		\begin{align}
			&\mathbf{y}_{t} = H_{t}^{1/2} \bm{\eta}_{t}, \;\;
			H_{t} = D_{t} R_{t} D_{t}, \;\;
			R_{t} = (1 - \beta_{1} - \beta_{2}) \underline{R} + \beta_{1} \Psi_{t-1} + \beta_{2} R_{t-1}, \notag\\
			&\ln\bm{h}_{t} 
			= \underline{\bm{\omega}} 
			+ \sum_{i=1}^{\infty} \Phi_{i} \ln\mathbf{y}^{2}_{t-i}, \label{model 2 Dt SGARCH(q,r*,s*)}
		\end{align}
		where $\mathbf{y}_{t}$, $\bm{\eta}_{t}$, $H_{t}$, $R_{t}$, $D_{t}$, $\beta_{1}$, $\beta_{2}$, $\underline{R}$, $\Psi_{t}$, $\ln\bm{h}_{t}$, $\underline{\bm{\omega}}$ and $\ln\mathbf{y}^{2}_{t}$ are defined as in models \eqref{model Rt SGARCH(r,s)}--\eqref{model Dt SGARCH(r,s)}. 
		Here, the $m \times m$ coefficient matrix $\Phi_{i}$ is defined as follows, 
		\begin{align*}
			\Phi_{i} = & 
			\sum_{k=1}^{q-1} I(i=k) G_{k} 
			+ \sum_{k=1}^{r} I(i\geq q) \lambda_{k}^{i-q} G_{0,k} \\ 
			&+ \sum_{k=1}^{s} I(i\geq q) \gamma_{k}^{i-q} 
			\left[\cos((i-q) \varphi_{k}) G_{1,k} 
			+ \sin((i-q) \varphi_{k}) G_{2,k}\right], 
		\end{align*}
		where $r$ and $s$ are integers that satisfy $r+2s \leq mp$, 
		$G_{k}$'s are $m \times m$ parameter matrices, 
		and parameters $\lambda_{k}$'s, $\gamma_{k}$'s, $\varphi_{k}$'s, $G_{0,k}$'s, $G_{1,k}$'s and $G_{2,k}$'s are defined as in \eqref{model Phii in Dt}.
		It is obvious that if $p=q=1$, then model \eqref{model 2 Dt SGARCH(q,r*,s*)} will degenerate into model \eqref{model Dt SGARCH(r,s)}.

\newpage

\begin{figure}[htbp]
\centering
\includegraphics[width=4.4in]{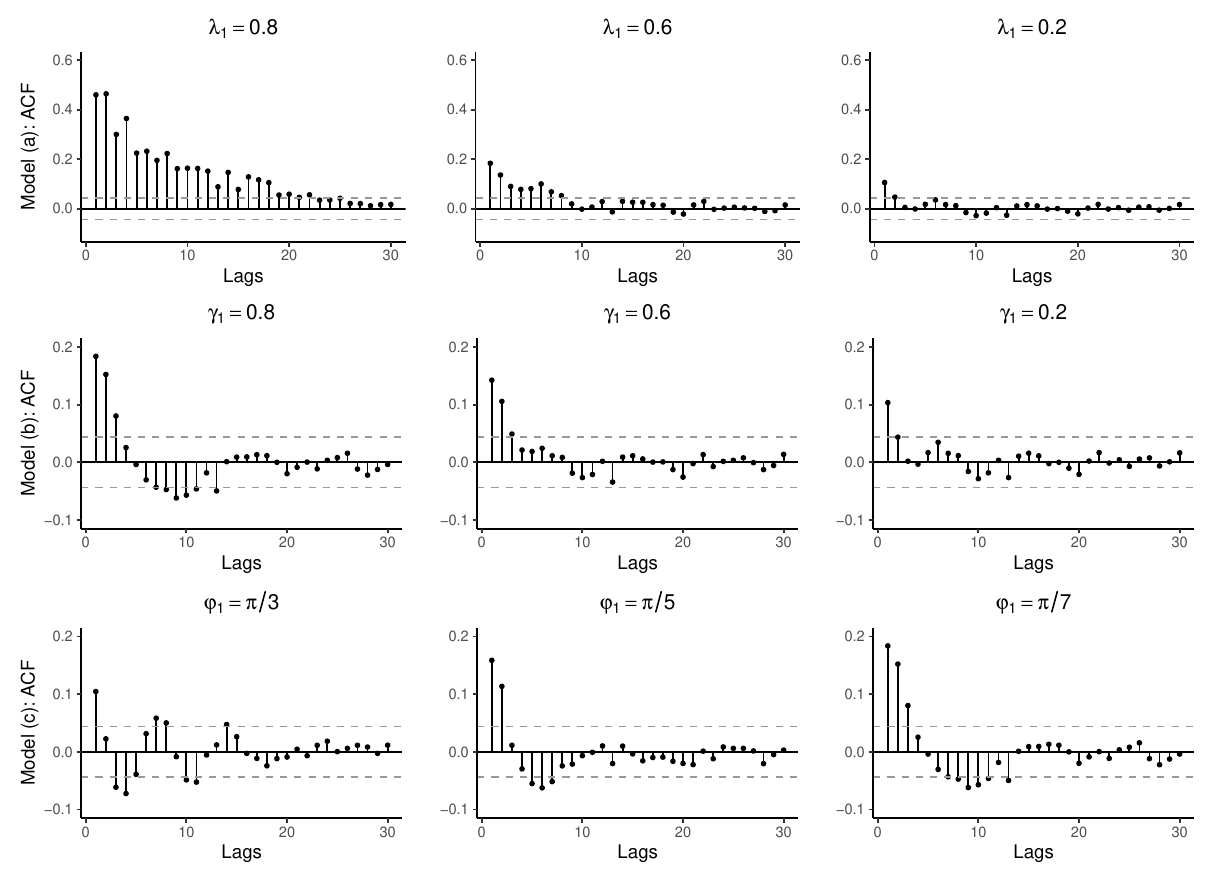}
\caption{\label{figure ACF_squared_return} ACFs of squared returns $\{y_{1t}^2\}$ under model (a) with $\lambda_{1} = 0.8, 0.6$ and $0.2$; model (b) with $\gamma_{1} = 0.8, 0.6$ and $0.2$; and model (c) with $\varphi_{1} = \pi/3, \pi/5$ and $\pi/7$.}
\end{figure}

\begin{figure}[htbp]
\centering
\includegraphics[width=4.4in]{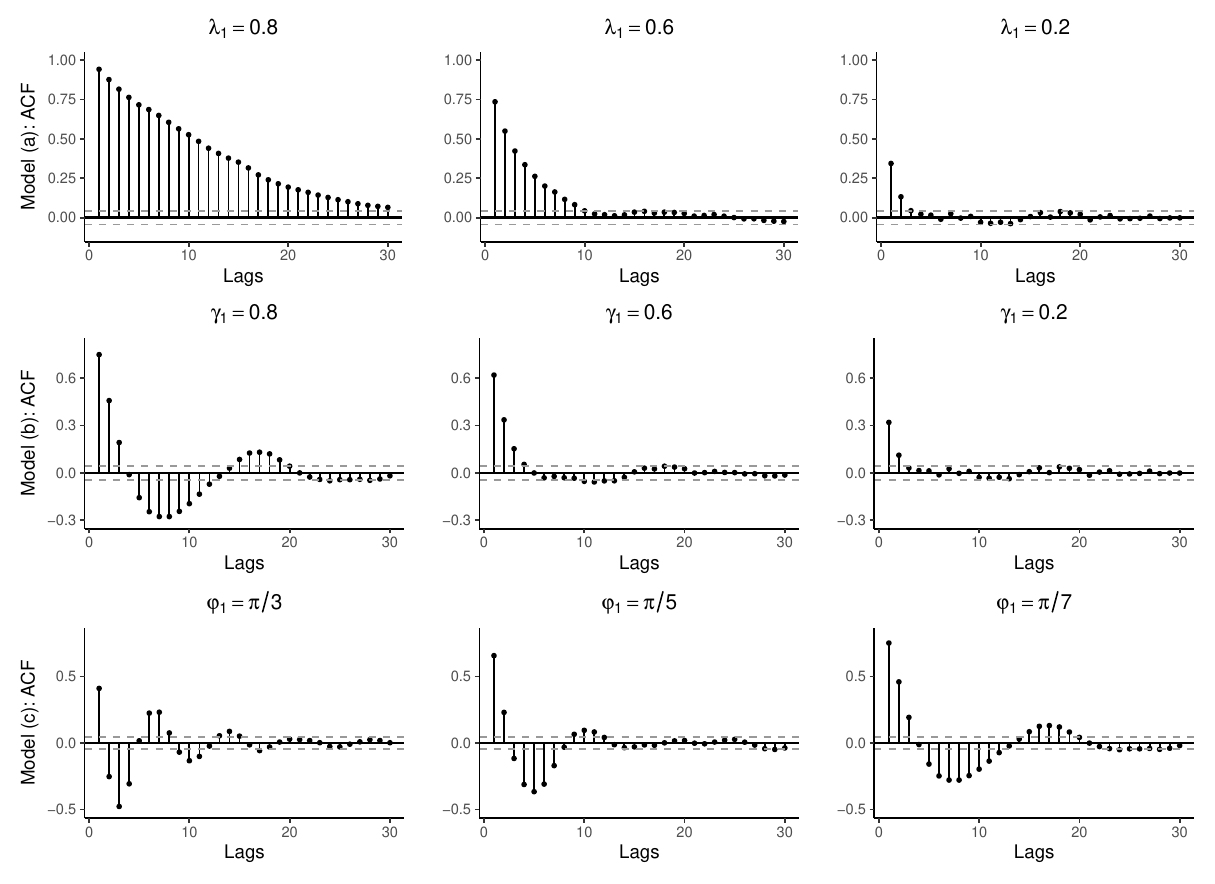}
\caption{\label{figure ACF_squared_volatility} ACFs of square volatilities $\{h_{11,t}\}$ under model (a) with $\lambda_{1} = 0.8, 0.6$ and $0.2$; model (b) with $\gamma_{1} = 0.8, 0.6$ and $0.2$; and model (c) with $\varphi_{1} = \pi/3, \pi/5$ and $\pi/7$.}
\end{figure}

\begin{table}[htbp]
\renewcommand\arraystretch{1.2}
\caption{\label{table general QMLE m2r2s0Bkm} Biases, ESDs, and ASDs of $\widehat{\bm{\theta}}_{\text{G}}$ in \eqref{est general QMLE} under DGP2 when $m=2$, $(r,s) = (2,0)$, and the innovations $\{\bm{\eta}_{t}\}$ follow a multivariate normal or Student's $t_5$ distribution.}
\begin{center}
  \begin{tabular}{crrrrrrrrrr}
    \hline\hline
    &&&\multicolumn{3}{c}{Normal}&&\multicolumn{3}{c}{$t_5$}\\
    \cline{4-6}\cline{8-10}
    &\multicolumn{1}{c}{$n$}&&\multicolumn{1}{c}{Bias}&\multicolumn{1}{c}{ESD}&\multicolumn{1}{c}{ASD}&&\multicolumn{1}{c}{Bias}&\multicolumn{1}{c}{ESD}&\multicolumn{1}{c}{ASD}\\
    \hline
    $\underline{\omega}_{1}$ & 1000 && 0.0416 & 0.0515 & 0.0497 && 0.0522 & 0.0661 & 0.0706 \\ 
                             & 2000 && 0.0281 & 0.0352 & 0.0351 && 0.0315 & 0.0401 & 0.0517 \\ 
    $\underline{\omega}_{2}$ & 1000 && 0.0413 & 0.0520 & 0.0498 && 0.0528 & 0.0653 & 0.0702 \\ 
                             & 2000 && 0.0270 & 0.0342 & 0.0350 && 0.0328 & 0.0413 & 0.0516 \\ 
    $\lambda_{1}$ & 1000 && 0.0359 & 0.0461 & 0.0432 && 0.0553 & 0.0767 & 0.0674 \\ 
                  & 2000 && 0.0234 & 0.0294 & 0.0283 && 0.0379 & 0.0486 & 0.0455 \\ 
    $\lambda_{2}$ & 1000 && 0.0514 & 0.0697 & 0.0667 && 0.0971 & 0.1499 & 0.1030 \\ 
                  & 2000 && 0.0339 & 0.0436 & 0.0420 && 0.0614 & 0.0922 & 0.0746 \\ 
    $(G_{0,1})_{11}$ & 1000 && 0.0103 & 0.0129 & 0.0124 && 0.0160 & 0.0223 & 0.0192 \\ 
                    & 2000 && 0.0071 & 0.0088 & 0.0084 && 0.0113 & 0.0144 & 0.0133 \\ 
    $(G_{0,1})_{21}$ & 1000 && 0.0107 & 0.0136 & 0.0124 && 0.0168 & 0.0396 & 0.0200 \\ 
                    & 2000 && 0.0069 & 0.0087 & 0.0084 && 0.0115 & 0.0145 & 0.0133 \\ 
    $(G_{0,1})_{12}$ & 1000 && 0.0105 & 0.0132 & 0.0124 && 0.0171 & 0.0366 & 0.0198 \\ 
                    & 2000 && 0.0072 & 0.0090 & 0.0085 && 0.0114 & 0.0145 & 0.0133 \\ 
    $(G_{0,1})_{22}$ & 1000 && 0.0104 & 0.0132 & 0.0124 && 0.0177 & 0.0374 & 0.0197 \\ 
                    & 2000 && 0.0066 & 0.0084 & 0.0085 && 0.0109 & 0.0139 & 0.0132 \\ 
    $(G_{0,2})_{11}$ & 1000 && 0.0119 & 0.0148 & 0.0146 && 0.0192 & 0.0274 & 0.0227 \\ 
                    & 2000 && 0.0082 & 0.0101 & 0.0099 && 0.0133 & 0.0167 & 0.0157 \\ 
    $(G_{0,2})_{21}$ & 1000 && 0.0114 & 0.0143 & 0.0147 && 0.0195 & 0.0436 & 0.0231 \\ 
                    & 2000 && 0.0079 & 0.0100 & 0.0100 && 0.0132 & 0.0173 & 0.0161 \\ 
    $(G_{0,2})_{12}$ & 1000 && 0.0116 & 0.0147 & 0.0146 && 0.0208 & 0.0384 & 0.0234 \\ 
                    & 2000 && 0.0077 & 0.0097 & 0.0100 && 0.0135 & 0.0178 & 0.0160 \\ 
    $(G_{0,2})_{22}$ & 1000 && 0.0118 & 0.0151 & 0.0146 && 0.0207 & 0.0427 & 0.0232 \\ 
                    & 2000 && 0.0081 & 0.0101 & 0.0100 && 0.0128 & 0.0163 & 0.0157 \\ 
    $\beta_{1}$ & 1000 && 0.0194 & 0.0244 & 0.0235 && 0.0200 & 0.0258 & 0.0253 \\ 
                & 2000 && 0.0130 & 0.0165 & 0.0161 && 0.0139 & 0.0177 & 0.0176 \\ 
    $\beta_{2}$ & 1000 && 0.0505 & 0.0623 & 0.0596 && 0.0548 & 0.0814 & 0.0652 \\ 
                & 2000 && 0.0319 & 0.0404 & 0.0390 && 0.0343 & 0.0440 & 0.0428 \\ 
    $(\underline{R})_{21}$ & 1000 && 0.0389 & 0.0477 & 0.0459 && 0.0399 & 0.0506 & 0.0464 \\ 
                           & 2000 && 0.0275 & 0.0345 & 0.0332 && 0.0279 & 0.0357 & 0.0331 \\ 
   \hline
  \end{tabular}
\end{center}
\end{table}

\begin{table}[htbp]
\renewcommand\arraystretch{1.2}
\caption{\label{table general QMLE m2r0s1Bkm} Biases, ESDs, and ASDs of $\widehat{\bm{\theta}}_{\text{G}}$ in \eqref{est general QMLE} under DGP3 when $m=2$, $(r,s) = (0,1)$, and the innovations $\{\bm{\eta}_{t}\}$ follow a multivariate normal or Student's $t_5$ distribution.}
\begin{center}
  \begin{tabular}{crrrrrrrrrr}
    \hline\hline
    &&&\multicolumn{3}{c}{Normal}&&\multicolumn{3}{c}{$t_5$}\\
    \cline{4-6}\cline{8-10}
    &\multicolumn{1}{c}{$n$}&&\multicolumn{1}{c}{Bias}&\multicolumn{1}{c}{ESD}&\multicolumn{1}{c}{ASD}&&\multicolumn{1}{c}{Bias}&\multicolumn{1}{c}{ESD}&\multicolumn{1}{c}{ASD}\\
    \hline
    $\underline{\omega}_{1}$ & 1000 && 0.0367 & 0.0457 & 0.0448 && 0.0495 & 0.0617 & 0.0672 \\ 
                             & 2000 && 0.0255 & 0.0319 & 0.0317 && 0.0301 & 0.0378 & 0.0500 \\ 
    $\underline{\omega}_{2}$ & 1000 && 0.0374 & 0.0460 & 0.0449 && 0.0504 & 0.0617 & 0.0670 \\ 
                             & 2000 && 0.0242 & 0.0308 & 0.0316 && 0.0316 & 0.0397 & 0.0498 \\ 
    $\gamma_{1}$ & 1000 && 0.0408 & 0.0531 & 0.0543 && 0.0693 & 0.0978 & 0.0906 \\ 
                 & 2000 && 0.0245 & 0.0311 & 0.0333 && 0.0476 & 0.0681 & 0.0606 \\ 
    $\varphi_{1}$ & 1000 && 0.0483 & 0.0658 & 0.0659 && 0.0901 & 0.1271 & 0.1084 \\ 
                  & 2000 && 0.0292 & 0.0371 & 0.0403 && 0.0560 & 0.0753 & 0.0772 \\ 
    $(G_{1,1})_{11}$ & 1000 && 0.0127 & 0.0160 & 0.0160 && 0.0196 & 0.0246 & 0.0241 \\ 
                    & 2000 && 0.0085 & 0.0106 & 0.0109 && 0.0144 & 0.0182 & 0.0172 \\ 
    $(G_{1,1})_{21}$ & 1000 && 0.0130 & 0.0165 & 0.0161 && 0.0201 & 0.0248 & 0.0238 \\ 
                    & 2000 && 0.0086 & 0.0108 & 0.0109 && 0.0136 & 0.0176 & 0.0169 \\ 
    $(G_{1,1})_{12}$ & 1000 && 0.0124 & 0.0155 & 0.0158 && 0.0189 & 0.0243 & 0.0235 \\ 
                    & 2000 && 0.0082 & 0.0104 & 0.0108 && 0.0141 & 0.0185 & 0.0171 \\ 
    $(G_{1,1})_{22}$ & 1000 && 0.0133 & 0.0168 & 0.0159 && 0.0195 & 0.0248 & 0.0245 \\ 
                    & 2000 && 0.0086 & 0.0108 & 0.0108 && 0.0139 & 0.0177 & 0.0169 \\ 
    $(G_{2,1})_{11}$ & 1000 && 0.0186 & 0.0243 & 0.0237 && 0.0328 & 0.0595 & 0.0476 \\ 
                    & 2000 && 0.0121 & 0.0152 & 0.0154 && 0.0204 & 0.0273 & 0.0284 \\ 
    $(G_{2,1})_{21}$ & 1000 && 0.0185 & 0.0240 & 0.0240 && 0.0374 & 0.0830 & 0.0572 \\ 
                    & 2000 && 0.0117 & 0.0147 & 0.0154 && 0.0226 & 0.0403 & 0.0315 \\ 
    $(G_{2,1})_{12}$ & 1000 && 0.0194 & 0.0248 & 0.0242 && 0.0358 & 0.0633 & 0.0512 \\ 
                    & 2000 && 0.0120 & 0.0152 & 0.0155 && 0.0223 & 0.0352 & 0.0316 \\ 
    $(G_{2,1})_{22}$ & 1000 && 0.0186 & 0.0245 & 0.0244 && 0.0354 & 0.0720 & 0.0526 \\ 
                    & 2000 && 0.0121 & 0.0153 & 0.0156 && 0.0214 & 0.0314 & 0.0313 \\ 
    $\beta_{1}$ & 1000 && 0.0193 & 0.0241 & 0.0235 && 0.0201 & 0.0256 & 0.0251 \\ 
                & 2000 && 0.0129 & 0.0164 & 0.0160 && 0.0139 & 0.0176 & 0.0174 \\ 
    $\beta_{2}$ & 1000 && 0.0486 & 0.0612 & 0.0592 && 0.0536 & 0.0790 & 0.0632 \\ 
                & 2000 && 0.0313 & 0.0402 & 0.0384 && 0.0347 & 0.0438 & 0.0420 \\ 
    $(\underline{R})_{21}$ & 1000 && 0.0382 & 0.0472 & 0.0462 && 0.0398 & 0.0501 & 0.0465 \\ 
                           & 2000 && 0.0270 & 0.0342 & 0.0333 && 0.0276 & 0.0353 & 0.0332 \\ 
   \hline
  \end{tabular}
\end{center}
\end{table}

\begin{table}[htbp]
  \renewcommand\arraystretch{1.2}
  \caption{\label{table general QMLE m5r1s0Bkm} Biases, ESDs, and ASDs of $\widehat{\bm{\theta}}_{\text{G}}$ in \eqref{est general QMLE} under DGP4 when $m=5$, $(r,s) = (1,0)$, and the innovations $\{\bm{\eta}_{t}\}$ follow a multivariate normal or Student's $t_5$ distribution.}
  \begin{center}
    \begin{tabular}{crrrrrrrrrr}
      \hline\hline
      &&&\multicolumn{3}{c}{Normal}&&\multicolumn{3}{c}{$t_5$}\\
      \cline{4-6}\cline{8-10}
      &\multicolumn{1}{c}{$n$}&&\multicolumn{1}{c}{Bias}&\multicolumn{1}{c}{ESD}&\multicolumn{1}{c}{ASD}&&\multicolumn{1}{c}{Bias}&\multicolumn{1}{c}{ESD}&\multicolumn{1}{c}{ASD}\\
    \hline
    $\underline{\omega}_{1}$ & 1000 && 0.0454 & 0.0570 & 0.0540 && 0.0490 & 0.0615 & 0.0648 \\ 
                             & 2000 && 0.0319 & 0.0397 & 0.0379 && 0.0312 & 0.0384 & 0.0472 \\ 
    $\underline{\omega}_{2}$ & 1000 && 0.0456 & 0.0568 & 0.0540 && 0.0488 & 0.0614 & 0.0637 \\ 
                             & 2000 && 0.0317 & 0.0395 & 0.0378 && 0.0299 & 0.0374 & 0.0469 \\ 
    $\underline{\omega}_{3}$ & 1000 && 0.0459 & 0.0583 & 0.0539 && 0.0488 & 0.0611 & 0.0645 \\ 
                             & 2000 && 0.0311 & 0.0390 & 0.0378 && 0.0315 & 0.0396 & 0.0473 \\ 
    $\underline{\omega}_{4}$ & 1000 && 0.0448 & 0.0559 & 0.0537 && 0.0491 & 0.0612 & 0.0642 \\ 
                             & 2000 && 0.0307 & 0.0380 & 0.0377 && 0.0316 & 0.0409 & 0.0468 \\ 
    $\underline{\omega}_{5}$ & 1000 && 0.0466 & 0.0581 & 0.0535 && 0.0491 & 0.0624 & 0.0644 \\ 
                             & 2000 && 0.0308 & 0.0384 & 0.0376 && 0.0318 & 0.0400 & 0.0471 \\ 
    $\lambda_{1}$ & 1000 && 0.0207 & 0.0250 & 0.0245 && 0.0303 & 0.0386 & 0.0359 \\ 
                  & 2000 && 0.0134 & 0.0160 & 0.0161 && 0.0203 & 0.0261 & 0.0247 \\ 
    $(G_{0,1})_{11}$ & 1000 && 0.0086 & 0.0107 & 0.0106 && 0.0135 & 0.0169 & 0.0152 \\ 
                    & 2000 && 0.0060 & 0.0075 & 0.0074 && 0.0101 & 0.0127 & 0.0112 \\ 
    $(G_{0,1})_{12}$ & 1000 && 0.0081 & 0.0103 & 0.0106 && 0.0138 & 0.0177 & 0.0152 \\ 
                    & 2000 && 0.0058 & 0.0073 & 0.0074 && 0.0097 & 0.0123 & 0.0111 \\ 
    $(G_{0,1})_{13}$ & 1000 && 0.0091 & 0.0113 & 0.0106 && 0.0135 & 0.0171 & 0.0153 \\ 
                    & 2000 && 0.0060 & 0.0076 & 0.0074 && 0.0095 & 0.0122 & 0.0111 \\ 
    $(G_{0,1})_{14}$ & 1000 && 0.0087 & 0.0110 & 0.0106 && 0.0137 & 0.0174 & 0.0152 \\ 
                    & 2000 && 0.0060 & 0.0075 & 0.0074 && 0.0096 & 0.0123 & 0.0112 \\ 
    $(G_{0,1})_{15}$ & 1000 && 0.0086 & 0.0107 & 0.0106 && 0.0138 & 0.0176 & 0.0152 \\ 
                    & 2000 && 0.0060 & 0.0076 & 0.0074 && 0.0101 & 0.0129 & 0.0112 \\ 
    $\beta_{1}$ & 1000 && 0.0089 & 0.0112 & 0.0109 && 0.0100 & 0.0123 & 0.0119 \\ 
                & 2000 && 0.0063 & 0.0078 & 0.0076 && 0.0071 & 0.0087 & 0.0084 \\ 
    $\beta_{2}$ & 1000 && 0.0228 & 0.0237 & 0.0241 && 0.0243 & 0.0274 & 0.0263 \\ 
                & 2000 && 0.0144 & 0.0164 & 0.0165 && 0.0156 & 0.0185 & 0.0184 \\ 
    $(\underline{R})_{21}$ & 1000 && 0.0337 & 0.0425 & 0.0409 && 0.0350 & 0.0437 & 0.0415 \\ 
                           & 2000 && 0.0235 & 0.0300 & 0.0294 && 0.0244 & 0.0309 & 0.0302 \\ 
    $(\underline{R})_{31}$ & 1000 && 0.0351 & 0.0442 & 0.0411 && 0.0353 & 0.0443 & 0.0415 \\ 
                           & 2000 && 0.0242 & 0.0306 & 0.0294 && 0.0253 & 0.0315 & 0.0300 \\ 
    $(\underline{R})_{41}$ & 1000 && 0.0342 & 0.0429 & 0.0409 && 0.0342 & 0.0430 & 0.0414 \\ 
                           & 2000 && 0.0244 & 0.0306 & 0.0294 && 0.0246 & 0.0310 & 0.0301 \\ 
    $(\underline{R})_{51}$ & 1000 && 0.0345 & 0.0439 & 0.0408 && 0.0366 & 0.0450 & 0.0414 \\ 
                           & 2000 && 0.0241 & 0.0301 & 0.0294 && 0.0249 & 0.0312 & 0.0301 \\ 
   \hline
    \end{tabular}
  \end{center}
\end{table}

\begin{table}[htbp]
  \renewcommand\arraystretch{1.2}
  \caption{\label{table lowrank QMLE m5r1s0Bkm} Biases, ESDs, and ASDs of $\widehat{\bm{\theta}}_{\text{LR}}$ in \eqref{est lowrank QMLE} under DGP4 when $m=5$, $(r,s) = (1,0)$, and the innovations $\{\bm{\eta}_{t}\}$ follow a multivariate normal or Student's $t_5$ distribution.}
  \begin{center}
    \begin{tabular}{crrrrrrrrrr}
      \hline\hline
      &&&\multicolumn{3}{c}{Normal}&&\multicolumn{3}{c}{$t_5$}\\
      \cline{4-6}\cline{8-10}
      &\multicolumn{1}{c}{$n$}&&\multicolumn{1}{c}{Bias}&\multicolumn{1}{c}{ESD}&\multicolumn{1}{c}{ASD}&&\multicolumn{1}{c}{Bias}&\multicolumn{1}{c}{ESD}&\multicolumn{1}{c}{ASD}\\
    \hline
    $\underline{\omega}_{1}$ & 1000 && 0.0444 & 0.0557 & 0.0533 && 0.0487 & 0.0614 & 0.0657 \\ 
                             & 2000 && 0.0317 & 0.0395 & 0.0377 && 0.0306 & 0.0382 & 0.0478 \\ 
    $\underline{\omega}_{2}$ & 1000 && 0.0453 & 0.0559 & 0.0533 && 0.0475 & 0.0608 & 0.0647 \\ 
                             & 2000 && 0.0315 & 0.0392 & 0.0375 && 0.0293 & 0.0373 & 0.0475 \\ 
    $\underline{\omega}_{3}$ & 1000 && 0.0458 & 0.0579 & 0.0532 && 0.0477 & 0.0605 & 0.0656 \\ 
                             & 2000 && 0.0310 & 0.0387 & 0.0376 && 0.0311 & 0.0395 & 0.0481 \\ 
    $\underline{\omega}_{4}$ & 1000 && 0.0444 & 0.0549 & 0.0529 && 0.0484 & 0.0607 & 0.0653 \\ 
                             & 2000 && 0.0304 & 0.0377 & 0.0375 && 0.0310 & 0.0403 & 0.0475 \\ 
    $\underline{\omega}_{5}$ & 1000 && 0.0457 & 0.0568 & 0.0527 && 0.0479 & 0.0624 & 0.0655 \\ 
                             & 2000 && 0.0308 & 0.0385 & 0.0374 && 0.0315 & 0.0400 & 0.0477 \\ 
    $\lambda_{1}$ & 1000 && 0.0200 & 0.0239 & 0.0232 && 0.0294 & 0.0390 & 0.0341 \\ 
                  & 2000 && 0.0131 & 0.0158 & 0.0156 && 0.0200 & 0.0263 & 0.0242 \\ 
    $(G_{0,1})_{11}$ & 1000 && 0.0054 & 0.0067 & 0.0065 && 0.0088 & 0.0111 & 0.0098 \\ 
                    & 2000 && 0.0036 & 0.0045 & 0.0045 && 0.0062 & 0.0078 & 0.0072 \\ 
    $(G_{0,1})_{12}$ & 1000 && 0.0054 & 0.0067 & 0.0065 && 0.0090 & 0.0114 & 0.0099 \\ 
                    & 2000 && 0.0038 & 0.0048 & 0.0045 && 0.0063 & 0.0080 & 0.0073 \\ 
    $(G_{0,1})_{13}$ & 1000 && 0.0054 & 0.0068 & 0.0066 && 0.0087 & 0.0111 & 0.0099 \\ 
                    & 2000 && 0.0036 & 0.0046 & 0.0046 && 0.0062 & 0.0079 & 0.0072 \\ 
    $(G_{0,1})_{14}$ & 1000 && 0.0054 & 0.0068 & 0.0066 && 0.0088 & 0.0115 & 0.0100 \\ 
                    & 2000 && 0.0037 & 0.0046 & 0.0046 && 0.0061 & 0.0077 & 0.0073 \\ 
    $(G_{0,1})_{15}$ & 1000 && 0.0055 & 0.0069 & 0.0067 && 0.0091 & 0.0118 & 0.0100 \\ 
                    & 2000 && 0.0038 & 0.0048 & 0.0046 && 0.0066 & 0.0083 & 0.0074 \\ 
    $\beta_{1}$ & 1000 && 0.0089 & 0.0112 & 0.0109 && 0.0101 & 0.0125 & 0.0121 \\ 
                & 2000 && 0.0063 & 0.0078 & 0.0076 && 0.0072 & 0.0089 & 0.0086 \\ 
    $\beta_{2}$ & 1000 && 0.0227 & 0.0236 & 0.0240 && 0.0245 & 0.0278 & 0.0268 \\ 
                & 2000 && 0.0144 & 0.0164 & 0.0165 && 0.0160 & 0.0189 & 0.0187 \\ 
    $(\underline{R})_{21}$ & 1000 && 0.0336 & 0.0423 & 0.0408 && 0.0352 & 0.0439 & 0.0416 \\ 
                           & 2000 && 0.0234 & 0.0300 & 0.0294 && 0.0247 & 0.0310 & 0.0303 \\ 
    $(\underline{R})_{31}$ & 1000 && 0.0350 & 0.0442 & 0.0410 && 0.0350 & 0.0440 & 0.0417 \\ 
                           & 2000 && 0.0242 & 0.0306 & 0.0293 && 0.0252 & 0.0314 & 0.0302 \\ 
    $(\underline{R})_{41}$ & 1000 && 0.0342 & 0.0430 & 0.0409 && 0.0344 & 0.0431 & 0.0416 \\ 
                           & 2000 && 0.0245 & 0.0307 & 0.0293 && 0.0247 & 0.0311 & 0.0302 \\ 
    $(\underline{R})_{51}$ & 1000 && 0.0346 & 0.0439 & 0.0408 && 0.0364 & 0.0449 & 0.0417 \\ 
                           & 2000 && 0.0242 & 0.0302 & 0.0294 && 0.0251 & 0.0313 & 0.0302 \\ 
   \hline
    \end{tabular}
  \end{center}
\end{table}

\begin{figure}[htbp]
\centering
\includegraphics[width=6in]{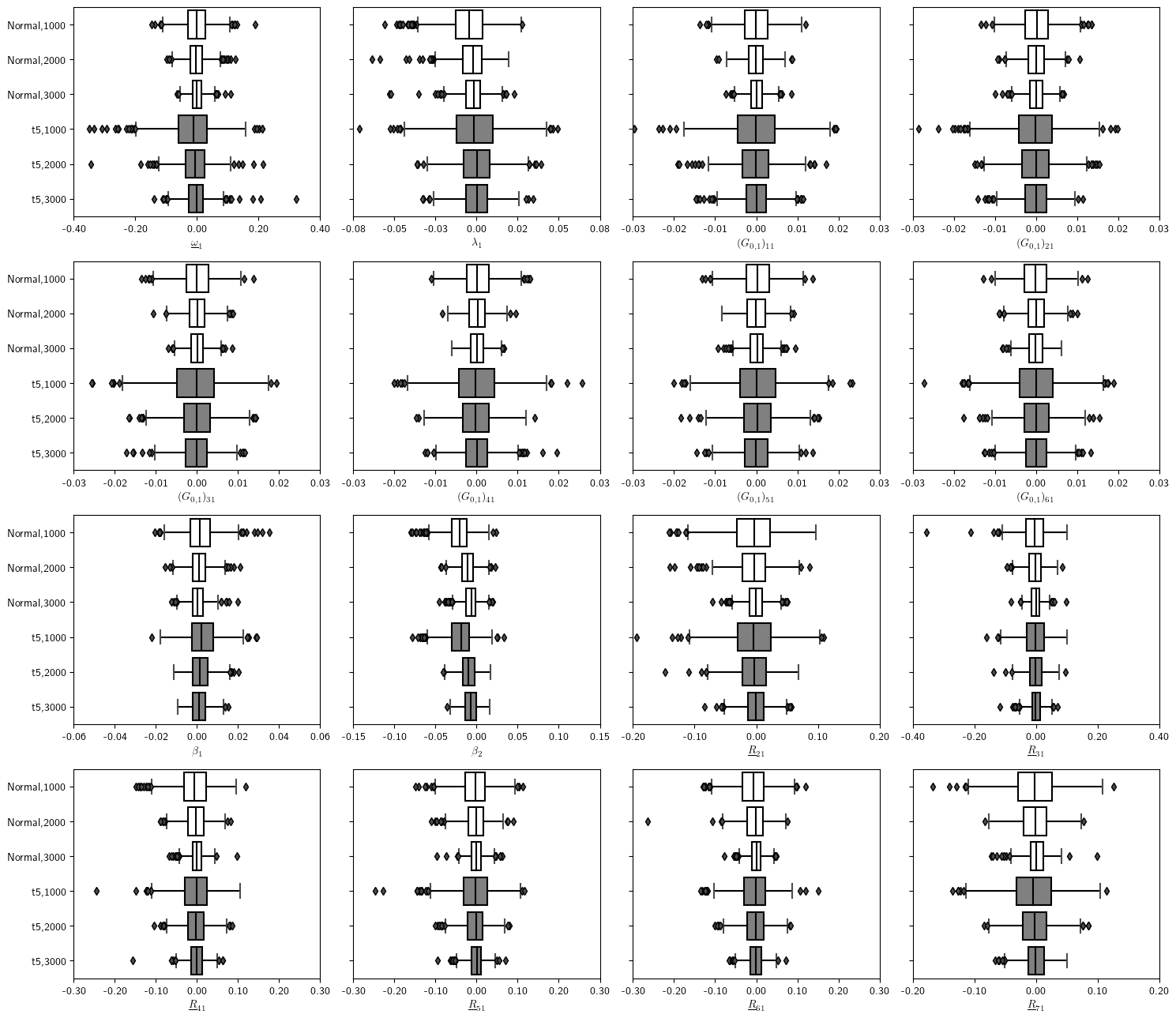}
\caption{\label{figure simulation_m20} Box plots of $\widehat{\bm{\theta}}_{\text{LR}}-\bm{\theta}_{0}$ under DGP5 when $m=20$, $(r,s) = (1,0)$, and the innovations $\{\bm{\eta}_{t}\}$ follow a multivariate normal or Student's $t_5$ distribution.}
\end{figure}

\begin{table}[htbp]
\renewcommand\arraystretch{1.2}
\caption{\label{table BIC} Percentages of underfitted, correct selected, and overfitted cases by the BIC using the QMLEs $\widehat{\bm{\theta}}_{\text{G}}$ in \eqref{est general QMLE} and $\widehat{\bm{\theta}}_{\text{LR}}$ in \eqref{est lowrank QMLE} under DGP1--DGP4, when the innovations $\{\bm{\eta}_{t}\}$ follow a multivariate normal or Student's $t_5$ distribution.}
\begin{center}
 \begin{tabular}{cccccrrrrrrrr}
   \hline\hline
   &&&&&&\multicolumn{3}{c}{Normal}&&\multicolumn{3}{c}{$t_5$}\\
   \cline{7-9}\cline{11-13}
   \multicolumn{1}{c}{DGP}&\multicolumn{1}{c}{$m$}&\multicolumn{1}{c}{$(r_{0},s_{0})$}&\multicolumn{1}{c}{QMLE}&\multicolumn{1}{c}{$n$}&&\multicolumn{1}{c}{Under}&\multicolumn{1}{c}{Exact}&\multicolumn{1}{c}{Over}&&\multicolumn{1}{c}{Under}&\multicolumn{1}{c}{Exact}&\multicolumn{1}{c}{Over}\\
   \hline
   DGP1 & 2 & $(1,0)$ & $\widehat{\bm{\theta}}_{\text{G}}$ 
   & 1000 && 0.000 & 1.000 & 0.000 && 0.022 & 0.924 & 0.054 \\
   &&&                 
   & 2000 && 0.000 & 1.000 & 0.000 && 0.021 & 0.930 & 0.049 \\
   \cline{4-13}
   &&& $\widehat{\bm{\theta}}_{\text{LR}}$ 
   & 1000 && 0.000 & 1.000 & 0.000 && 0.003 & 0.946 & 0.051 \\
   &&&                
   & 2000 && 0.000 & 1.000 & 0.000 && 0.000 & 0.951 & 0.049 \\
   \hline
   DGP2 & 2 & $(2,0)$ & $\widehat{\bm{\theta}}_{\text{G}}$ 
   & 1000 && 0.144 & 0.856 & \multicolumn{1}{c}{-} && 0.125 & 0.875 & \multicolumn{1}{c}{-} \\
   &&&                 
   & 2000 && 0.000 & 1.000 & \multicolumn{1}{c}{-} && 0.004 & 0.996 & \multicolumn{1}{c}{-} \\
   \cline{4-13}
   &&& $\widehat{\bm{\theta}}_{\text{LR}}$ 
   & 1000 && 0.114 & 0.886 & \multicolumn{1}{c}{-} && 0.140 & 0.860 & \multicolumn{1}{c}{-} \\
   &&&                
   & 2000 && 0.000 & 1.000 & \multicolumn{1}{c}{-} && 0.002 & 0.998 & \multicolumn{1}{c}{-} \\
   \hline
   DGP3 & 2 & $(0,1)$ & $\widehat{\bm{\theta}}_{\text{G}}$ 
   & 1000 && 0.169 & 0.831 & \multicolumn{1}{c}{-} && 0.121 & 0.879 & \multicolumn{1}{c}{-} \\
   &&&                 
   & 2000 && 0.001 & 0.999 & \multicolumn{1}{c}{-} && 0.005 & 0.995 & \multicolumn{1}{c}{-} \\
   \cline{4-13}
   &&& $\widehat{\bm{\theta}}_{\text{LR}}$ 
   & 1000 && 0.989 & 0.011 & \multicolumn{1}{c}{-} && 0.838 & 0.162 & \multicolumn{1}{c}{-} \\
   &&&                
   & 2000 && 0.546 & 0.454 & \multicolumn{1}{c}{-} && 0.343 & 0.657 & \multicolumn{1}{c}{-} \\
   \hline
   DGP4 & 5 & $(1,0)$ & $\widehat{\bm{\theta}}_{\text{G}}$ 
   & 1000 && 0.000 & 1.000 & 0.000 && 0.001 & 0.994 & 0.005 \\
   &&&                
   & 2000 && 0.000 & 1.000 & 0.000 && 0.000 & 0.997 & 0.003 \\
   \cline{4-13}
   &&& $\widehat{\bm{\theta}}_{\text{LR}}$ 
   & 1000 && 0.000 & 1.000 & 0.000 && 0.001 & 0.932 & 0.067 \\
   &&&                
   & 2000 && 0.000 & 1.000 & 0.000 && 0.000 & 0.940 & 0.060 \\
   \hline
 \end{tabular}
\end{center}
\end{table}

\begin{figure}[htbp]
\centering
\includegraphics[width=6.5in]{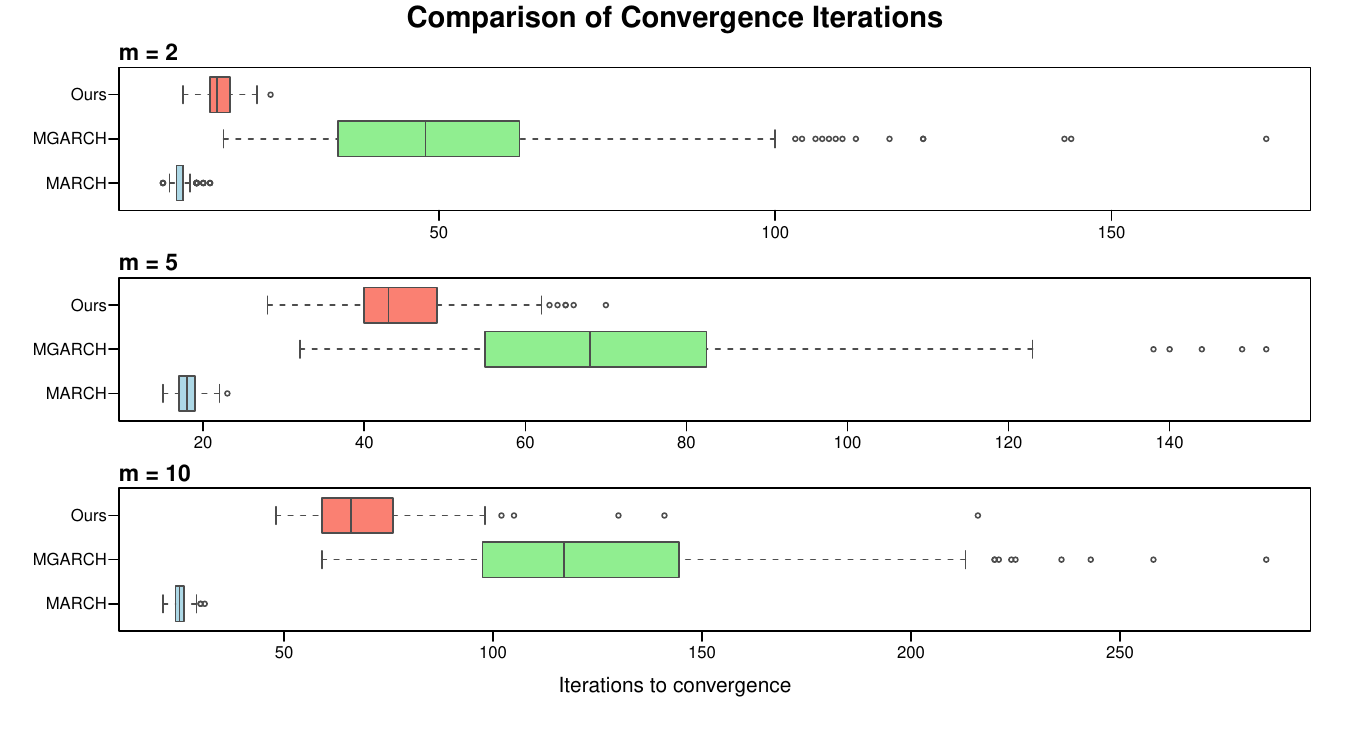}
\caption{\label{figure convergence_iterations_r1s0} Box plots for the number of iterations required for convergence of the full QMLE algorithm over 200 replications when $J = J^{(a)}$. ``MARCH'', ``MGARCH'', and ``Ours'' correspond to models \eqref{eq ccc-log-arch(1)}, \eqref{eq ccc-log-garch(1,1)}, and \eqref{eq ccc-log-arch(infty)}, respectively.}
\end{figure}

\begin{figure}[htbp]
\centering
\includegraphics[width=6.5in]{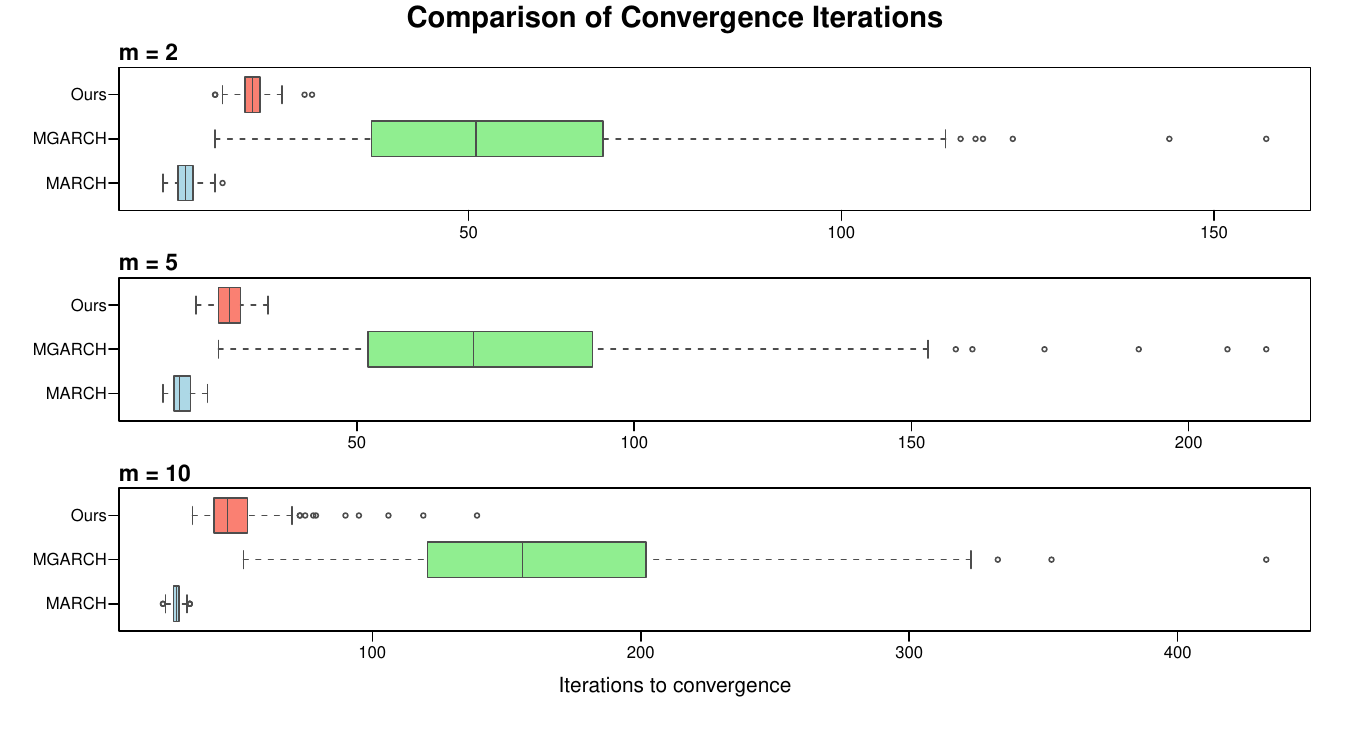}
\caption{\label{figure convergence_iterations_r0s1} Box plots for the number of iterations required for convergence of the full QMLE algorithm over 200 replications when $J = J^{(b)}$. ``MARCH'', ``MGARCH'', and ``Ours'' correspond to models \eqref{eq ccc-log-arch(1)}, \eqref{eq ccc-log-garch(1,1)}, and \eqref{eq ccc-log-arch(infty)}, respectively.}
\end{figure}

\begin{figure}[htbp]
\centering
\includegraphics[width=6in]{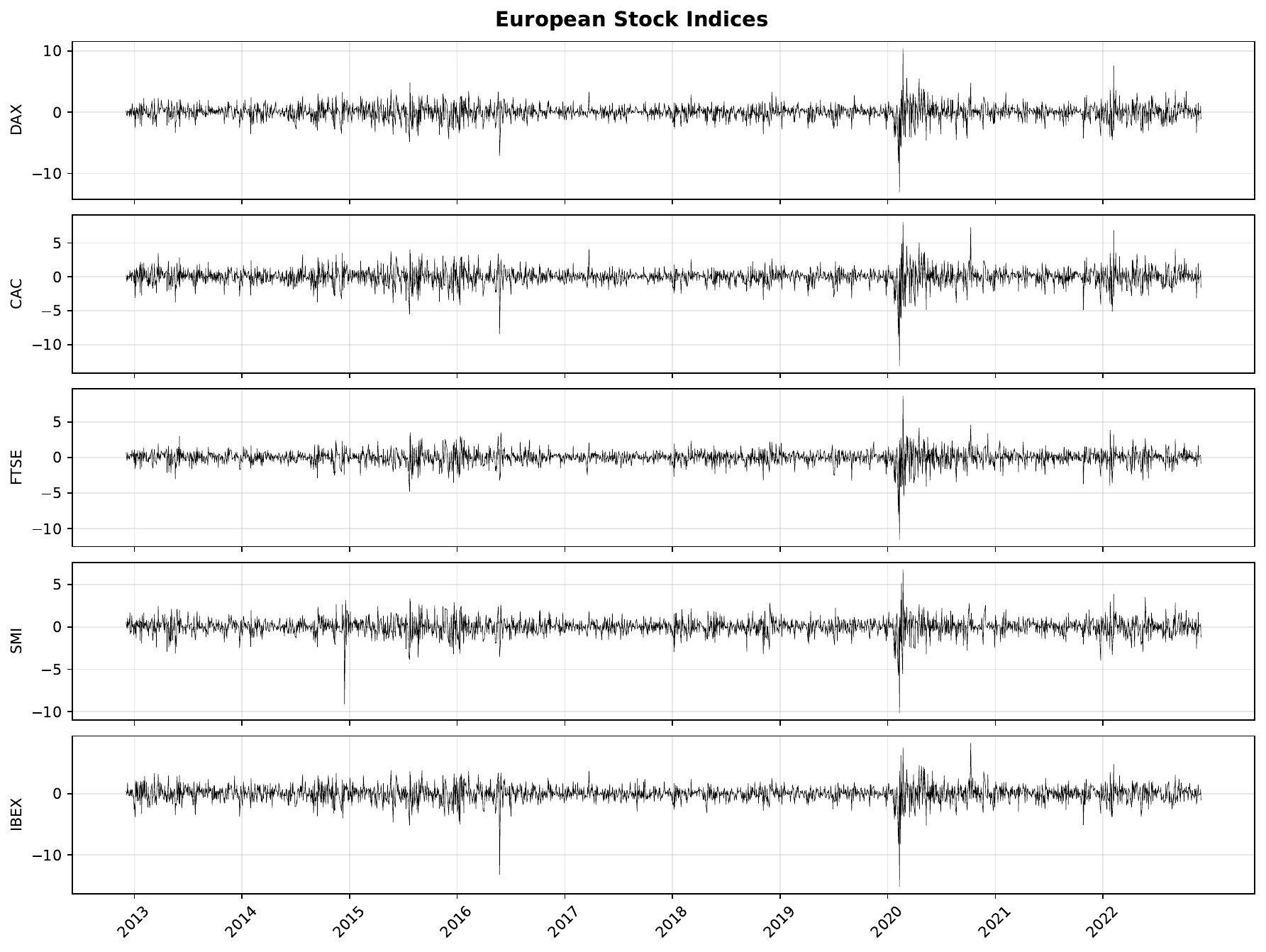}
\caption{\label{figure time plot} Time plot for centered daily log-returns in percentage of 
Deutsche Aktien Index (DAX), French CAC 40 (CAC), Financial Times Stock Exchange 100 Index (FTSE), Swiss Market Index (SMI), and Spanish IBEX 35 (IBEX) from January 4, 2013, to December 30, 2022.}
\end{figure}

\begin{table}[htbp]
\renewcommand\arraystretch{1.2}
  \caption{\label{table major stocks summary statistics} Summary statistics for $\{\mathbf{y}_{t}\}$.}
  \begin{center}
    \begin{tabular}{crrrrrrr}
      \hline\hline
			& Mean & Median & Min & Max & Std.Dev. & Skewness & Kurtosis \\
			\hline
			DAX    & 0.000 & 0.056 & -13.079 & 10.398 & 1.260 & -0.542 & 9.391 \\
		    CAC   & 0.000 & 0.049 & -13.121 & 8.034 & 1.237 & -0.762 & 9.994 \\
		    FTSE   & 0.000 & 0.049 & -11.520 & 8.658 & 1.017 & -0.832 & 12.236 \\
		    SMI  & 0.000 & 0.043 & -10.151 & 6.763 & 0.987 & -0.992 & 10.683 \\
		    IBEX   & 0.000 & 0.051 & -15.150 & 8.226 & 1.313 & -1.166 & 14.110 \\
      \hline
    \end{tabular}
  \end{center}
\end{table}

\begin{table}[htbp]
  \renewcommand\arraystretch{1.2}
  \caption{\label{table major stocks fitted coefficients} Summary information of fitted coefficients for model \eqref{fitted model general}.} 
  \begin{center}
    \begin{tabular}{crrrr}
      \hline\hline
      &\multicolumn{1}{c}{Estimate}&\multicolumn{1}{c}{Std. Error}&\multicolumn{1}{c}{$z$ statistic}&\multicolumn{1}{c}{$p$-value}\\ 
      \hline
		$\underline{\omega}_{1}$ & 0.994  & 0.078  & 12.790  & 0.000  \\
		$\underline{\omega}_{2}$ & 0.915  & 0.075  & 12.138  & 0.000  \\
		$\underline{\omega}_{3}$ & 0.674  & 0.083  & 8.159  & 0.000  \\
		$\underline{\omega}_{4}$ & 0.695  & 0.085  & 8.134  & 0.000  \\
		$\underline{\omega}_{5}$ & 1.048  & 0.109  & 9.602  & 0.000  \\
		$\lambda_{1}$ & 0.867  & 0.021  & 41.994  & 0.000  \\
		$\beta_{1}$ & 0.028  & 0.008  & 3.728  & 0.000  \\
		$\beta_{2}$ & 0.847  & 0.047  & 17.990  & 0.000  \\
      \hline
    \end{tabular}
  \end{center}
\end{table}

\begin{figure}
	\centering
	\includegraphics[width=6.3in]{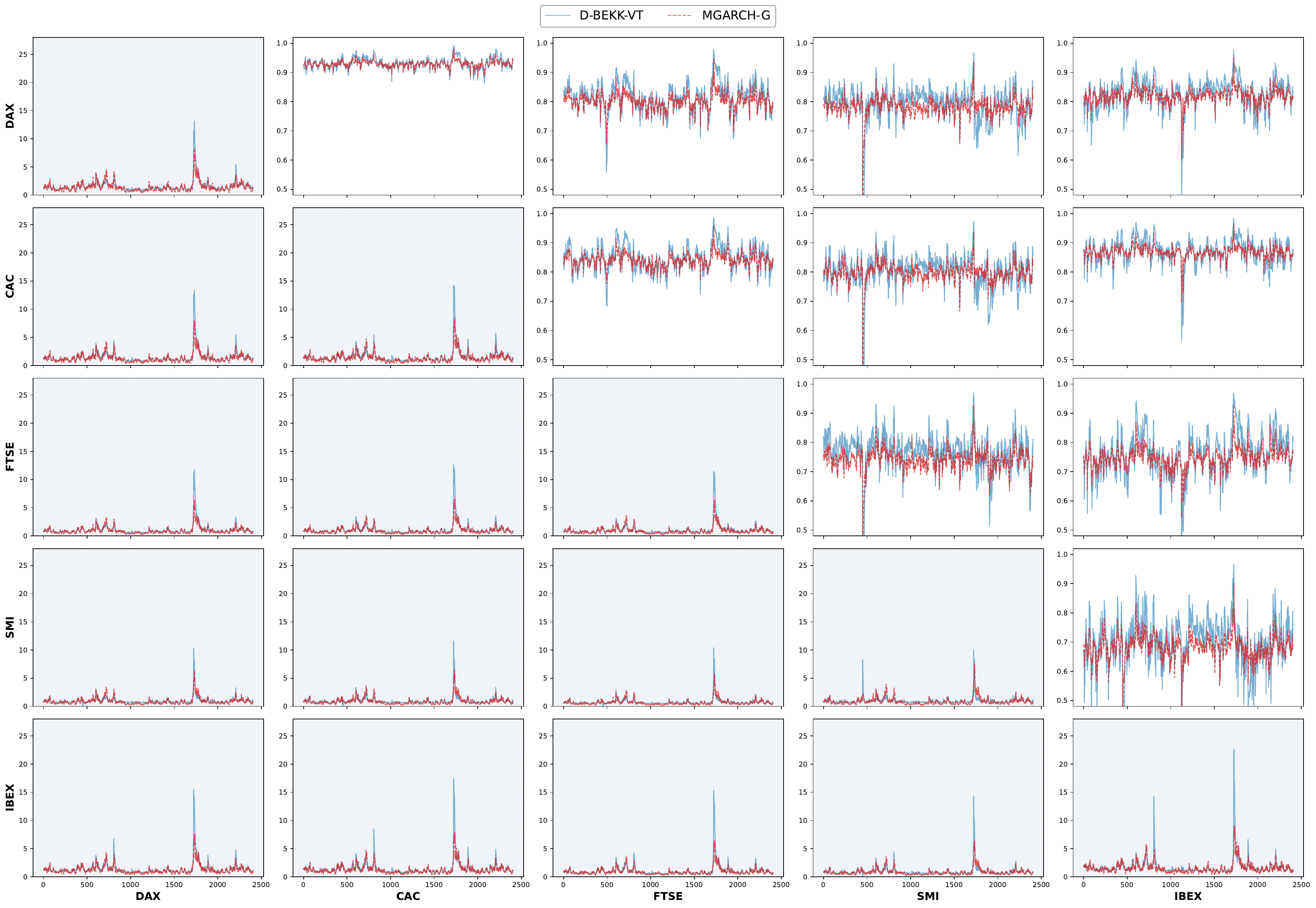}
	\caption{\label{figure_comparison}
	The entries in $H_{t}$ (lower-left triangle) and $R_{t}$ (upper-right triangle) fitted by the proposed model (dashed red lines) and the benchmark diagonal BEKK model (solid blue lines). 
}
\end{figure}

\begin{table}[htbp]
  \centering
  \renewcommand\arraystretch{1.2}
  \caption{\label{table major stocks comparison methods}The methods for comparison.}
  \begin{threeparttable}
  \begin{tabular}{ll}
      \toprule\toprule %
      \multicolumn{1}{l}{Model}         & \multicolumn{1}{l}{Description}     \\  
      \midrule
      \multicolumn{2}{c}{Variance-correlation models} \\ \hline
      CCC  & CCC model of \cite{bollerslev1990ModellingTC} \\
      DCC-T & DCC model of \cite{tse2002multivariate} with scalar coefficients in $R_{t}$ \\
      DCC-E & DCC model of \cite{engle2002dynamic} with scalar coefficients in $R_{t}$\\
      DCC-A & Corrected DCC model of \cite{aielli2013dynamic} with scalar coefficients in $R_{t}$\\
      D-DCC-T & Extension of DCC-T model with diagonal coefficient matrices in $R_{t}$ \\
      D-DCC-E & Extension of DCC-E model with diagonal coefficient matrices in $R_{t}$ \\
      D-DCC-A & Extension of DCC-A model with diagonal coefficient matrices in $R_{t}$ \\
      \midrule
      \multicolumn{2}{c}{Covariance models}\\ \hline
      S-BEKK-F & Scalar BEKK model estimated by full QMLE \\
      D-BEKK-F & Diagonal BEKK model estimated by full QMLE \\
      S-BEKK-VT & Scalar BEKK model estimated by VT method \\
      D-BEKK-VT & Diagonal BEKK model estimated by VT method \\
      \bottomrule
  \end{tabular} 
  \begin{tablenotes}
    \footnotesize
    \item $R_{t}$ is the conditional correlation matrix of $\mathbf{y}_{t}$ given $\mathcal{F}_{t-1}$. 
    \item The scalar and diagonal BEKK models are defined by \cite{ding2001large}. 
  \end{tablenotes}
  \end{threeparttable}
\end{table}%

\begin{table}[htbp]
  \centering
  \caption{\label{table stock stat loss} Average MSEs and QLIKE losses for one-step-ahead covariance matrix forecasts by 12 competing models in the first example, where MGARCH-G denotes the proposed model in \eqref{model Rt SGARCH(r,s)}--\eqref{model Phii in Dt} fitted by the QMLE $\widehat{\bm{\theta}}_{\text{G}}$ in \eqref{est general QMLE}. Boldface values indicate the models included in the 90\% MCS.}
  \begin{tabular}{lcccccc}
    \toprule
    Metric & MGARCH-G & CCC & DCC-T & D-DCC-T & DCC-E & D-DCC-E \\
    \midrule
	QLIKE & \textbf{-0.0496} & 0.2784 & \textbf{0.1120} & \textbf{0.1339} & \textbf{0.0454} & 0.2580 \\
    MSE & \textbf{5.8412} & \textbf{5.7423} & \textbf{5.7469} & \textbf{5.7480} & \textbf{5.7669} & \textbf{5.8022} \\
    \midrule
    Metric & DCC-A & D-DCC-A & S-BEKK-F & D-BEKK-F & S-BEKK-VT & D-BEKK-VT \\
    \midrule
	QLIKE & 0.5605 & 0.8888 & 2.7620 & 1.6665 & \textbf{0.0179} & \textbf{-0.0223} \\
    MSE & \textbf{5.7590} & \textbf{5.7914} & 7.7269 & \textbf{5.8486} & \textbf{5.9810} & \textbf{5.7725} \\
    \bottomrule
  \end{tabular}
\end{table}

\begin{table}[htbp]
    \centering
    \renewcommand\arraystretch{1.2}
    \caption{\label{table major stocks VaR EW} Out-of-sample VaR forecasting performance for the EW portfolio in the first example, where MGARCH-G denotes the proposed model in \eqref{model Rt SGARCH(r,s)}--\eqref{model Phii in Dt} fitted by the QMLE $\widehat{\bm{\theta}}_{\text{G}}$ in \eqref{est general QMLE}. The smallest values of PE, $\text{Loss}_{\mathrm{Q}}$, and $\text{Loss}_{\mathrm{C}}$ in each column are shown in boldface. For ECR, $\dagger$ and $\ddagger$ indicate the models that pass one or both of the CC and DQ tests at the 5\% significance level, respectively. For $\text{Loss}_{\mathrm{Q}}$ and $\text{Loss}_{\mathrm{C}}$, $\dagger$ indicates that the corresponding model is included in the 90\% MCS at the 10\% significance level.}
    \resizebox{\textwidth}{!}{
    \begin{tabular}{l *{12}{r}}
        \toprule\toprule
        & \multicolumn{4}{c}{$\tau=1\%$} & \multicolumn{4}{c}{$\tau=2.5\%$} & \multicolumn{4}{c}{$\tau=5\%$} \\
        \cmidrule(lr){2-5} \cmidrule(lr){6-9} \cmidrule(lr){10-13} 
        Method & ECR & PE & $\text{Loss}_{\mathrm{Q}}$ & $\text{Loss}_{\mathrm{C}}$ & ECR & PE & $\text{Loss}_{\mathrm{Q}}$ & $\text{Loss}_{\mathrm{C}}$ & ECR & PE & $\text{Loss}_{\mathrm{Q}}$ & $\text{Loss}_{\mathrm{C}}$ \\
        \midrule
        MGARCH-G & $0.82^{\ddagger}$ & $0.39$  & $3.68^{\dagger}$ & $0.82^{\dagger}$ & $2.47^{\ddagger}$ & $\mathbf{0.04}$ & $7.21^{\dagger}$ & $2.41^{\dagger}$ & $4.73^{\ddagger}$ & $0.27$  & $11.85^{\dagger}$ & $\mathbf{4.51}^{\dagger}$ \\
    CCC   & $1.03^{\ddagger}$ & $\mathbf{0.06}$ & $3.77^{\dagger}$ & $1.02^{\dagger}$ & $2.26^{\ddagger}$ & $0.33$  & $7.20^{\dagger}$ & $\mathbf{2.21}^{\dagger}$ & $4.94^{\ddagger}$ & $\mathbf{0.06}$ & $11.92^{\dagger}$ & $4.69^{\dagger}$ \\
    DCC-T & $0.82^{\ddagger}$ & $0.39$  & $3.84$  & $0.82^{\dagger}$ & $2.26^{\ddagger}$ & $0.33$  & $7.22^{\dagger}$ & $\mathbf{2.21}^{\dagger}$ & $4.94^{\ddagger}$ & $\mathbf{0.06}$ & $11.93^{\dagger}$ & $4.69^{\dagger}$ \\
    D-DCC-T & $0.82^{\ddagger}$ & $0.39$  & $3.83$  & $0.82^{\dagger}$ & $2.26^{\ddagger}$ & $0.33$  & $7.22^{\dagger}$ & $\mathbf{2.21}^{\dagger}$ & $4.94^{\ddagger}$ & $\mathbf{0.06}$ & $11.93^{\dagger}$ & $4.69^{\dagger}$ \\
    DCC-E & $0.82^{\ddagger}$ & $0.39$  & $3.86$  & $0.82^{\dagger}$ & $2.26^{\ddagger}$ & $0.33$  & $7.28^{\dagger}$ & $\mathbf{2.21}^{\dagger}$ & $4.73^{\ddagger}$ & $0.27$  & $11.95^{\dagger}$ & $\mathbf{4.51}^{\dagger}$ \\
    D-DCC-E & $0.82^{\ddagger}$ & $0.39$  & $3.84$  & $0.82^{\dagger}$ & $2.67^{\ddagger}$ & $0.25$  & $7.27^{\dagger}$ & $2.60^{\dagger}$ & $5.14^{\ddagger}$ & $0.15$  & $11.97^{\dagger}$ & $4.88^{\dagger}$ \\
    DCC-A & $1.03^{\ddagger}$ & $\mathbf{0.06}$ & $3.81^{\dagger}$ & $1.02^{\dagger}$ & $2.67^{\ddagger}$ & $0.25$  & $7.28^{\dagger}$ & $2.60^{\dagger}$ & $5.56^{\ddagger}$ & $0.56$  & $11.97^{\dagger}$ & $5.25^{\dagger}$ \\
    D-DCC-A & $1.65^{\ddagger}$ & $1.43$  & $3.84^{\dagger}$ & $1.62^{\dagger}$ & $3.70^{\ddagger}$ & $1.70$  & $7.35^{\dagger}$ & $3.58$  & $5.97^{\ddagger}$ & $0.98$  & $12.02^{\dagger}$ & $5.62$ \\
    S-BEKK-VT & $0.62^{\ddagger}$ & $0.85$  & $3.60^{\dagger}$ & $\mathbf{0.61}^{\dagger}$ & $2.67^{\ddagger}$ & $0.25$  & $7.29^{\dagger}$ & $2.60^{\dagger}$ & $5.14^{\ddagger}$ & $0.15$  & $11.97^{\dagger}$ & $4.88^{\dagger}$ \\
    S-BEKK-F & $0.62^{\ddagger}$ & $0.85$  & $3.57^{\dagger}$ & $\mathbf{0.61}^{\dagger}$ & $2.47^{\ddagger}$ & $\mathbf{0.04}$ & $\mathbf{7.04}^{\dagger}$ & $2.41^{\dagger}$ & $4.94^{\ddagger}$ & $\mathbf{0.06}$ & $\mathbf{11.75}^{\dagger}$ & $4.69^{\dagger}$ \\
    D-BEKK-VT & $0.62^{\ddagger}$ & $0.85$  & $\mathbf{3.56}^{\dagger}$ & $\mathbf{0.61}^{\dagger}$ & $2.26^{\ddagger}$ & $0.33$  & $7.05^{\dagger}$ & $\mathbf{2.21}^{\dagger}$ & $5.14^{\ddagger}$ & $0.15$  & $11.76^{\dagger}$ & $4.88^{\dagger}$ \\
    D-BEKK-F & $0.62^{\ddagger}$ & $0.85$  & $3.64^{\dagger}$ & $\mathbf{0.61}^{\dagger}$ & $2.47^{\ddagger}$ & $\mathbf{0.04}$ & $7.11^{\dagger}$ & $2.41^{\dagger}$ & $5.35^{\ddagger}$ & $0.35$  & $11.83^{\dagger}$ & $5.06^{\dagger}$ \\
        \midrule
        & \multicolumn{4}{c}{$\tau=95\%$} & \multicolumn{4}{c}{$\tau=97.5\%$} & \multicolumn{4}{c}{$\tau=99\%$} \\
        \cmidrule(lr){2-5} \cmidrule(lr){6-9} \cmidrule(lr){10-13} 
        Method & ECR & PE & $\text{Loss}_{\mathrm{Q}}$ & $\text{Loss}_{\mathrm{C}}$ & ECR & PE & $\text{Loss}_{\mathrm{Q}}$ & $\text{Loss}_{\mathrm{C}}$ & ECR & PE & $\text{Loss}_{\mathrm{Q}}$ & $\text{Loss}_{\mathrm{C}}$ \\
        \midrule
        MGARCH-G & $94.44^{\ddagger}$ & $0.56$  & $10.23^{\dagger}$ & $5.25^{\dagger}$ & $97.12^{\ddagger}$ & $0.54$  & $6.05^{\dagger}$ & $2.80^{\dagger}$ & $98.77^{\ddagger}$ & $0.52$  & $3.15^{\dagger}$ & $1.22^{\dagger}$ \\
    CCC   & $94.86^{\ddagger}$ & $\mathbf{0.15}$ & $9.90^{\dagger}$ & $\mathbf{4.88}^{\dagger}$ & $97.12^{\ddagger}$ & $0.54$  & $5.81^{\dagger}$ & $2.80^{\dagger}$ & $98.56^{\ddagger}$ & $0.98$  & $2.83^{\dagger}$ & $1.42^{\dagger}$ \\
    DCC-T & $94.86^{\ddagger}$ & $\mathbf{0.15}$ & $\mathbf{9.85}^{\dagger}$ & $\mathbf{4.88}^{\dagger}$ & $97.12^{\ddagger}$ & $0.54$  & $\mathbf{5.79}^{\dagger}$ & $2.80^{\dagger}$ & $98.35^{\ddagger}$ & $1.43$  & $\mathbf{2.79}^{\dagger}$ & $1.62^{\dagger}$ \\
    D-DCC-T & $94.86^{\ddagger}$ & $\mathbf{0.15}$ & $9.86^{\dagger}$ & $\mathbf{4.88}^{\dagger}$ & $97.12^{\ddagger}$ & $0.54$  & $\mathbf{5.79}^{\dagger}$ & $2.80^{\dagger}$ & $98.35^{\ddagger}$ & $1.43$  & $2.80^{\dagger}$ & $1.62^{\dagger}$ \\
    DCC-E & $94.86^{\ddagger}$ & $\mathbf{0.15}$ & $9.86^{\dagger}$ & $\mathbf{4.88}^{\dagger}$ & $97.12^{\ddagger}$ & $0.54$  & $5.80^{\dagger}$ & $2.80^{\dagger}$ & $98.35^{\ddagger}$ & $1.43$  & $2.81^{\dagger}$ & $1.62^{\dagger}$ \\
    D-DCC-E & $94.86^{\ddagger}$ & $\mathbf{0.15}$ & $9.89^{\dagger}$ & $\mathbf{4.88}^{\dagger}$ & $96.91^{\ddagger}$ & $0.83$  & $5.84^{\dagger}$ & $2.99^{\dagger}$ & $98.15^{\dagger}$ & $1.89$  & $2.85^{\dagger}$ & $1.82^{\dagger}$ \\
    DCC-A & $94.03^{\ddagger}$ & $0.98$  & $9.87^{\dagger}$ & $5.62^{\dagger}$ & $96.50^{\ddagger}$ & $1.41$  & $5.87^{\dagger}$ & $3.39^{\dagger}$ & $97.74$ & $2.80$  & $2.88^{\dagger}$ & $2.23^{\dagger}$ \\
    D-DCC-A & $92.59$ & $2.44$  & $10.00^{\dagger}$ & $6.92$  & $95.47$ & $2.86$  & $6.00^{\dagger}$ & $4.36$  & $97.53$ & $3.26$  & $3.02^{\dagger}$ & $2.43$ \\
    S-BEKK-VT & $94.44^{\ddagger}$ & $0.56$  & $10.48^{\dagger}$ & $5.25^{\dagger}$ & $97.53^{\ddagger}$ & $\mathbf{0.04}$ & $6.14^{\dagger}$ & $\mathbf{2.41}^{\dagger}$ & $98.97^{\ddagger}$ & $\mathbf{0.06}$ & $3.18^{\dagger}$ & $\mathbf{1.02}^{\dagger}$ \\
    S-BEKK-F & $94.24^{\ddagger}$ & $0.77$  & $10.17^{\dagger}$ & $5.44^{\dagger}$ & $97.12^{\ddagger}$ & $0.54$  & $5.99^{\dagger}$ & $2.80^{\dagger}$ & $98.56^{\ddagger}$ & $0.98$  & $2.91^{\dagger}$ & $1.42^{\dagger}$ \\
    D-BEKK-VT & $93.83^{\dagger}$ & $1.19$  & $10.30$ & $5.81^{\dagger}$ & $96.91^{\ddagger}$ & $0.83$  & $6.00^{\dagger}$ & $2.99^{\dagger}$ & $98.35^{\ddagger}$ & $1.43$  & $3.08^{\dagger}$ & $1.62^{\dagger}$ \\
    D-BEKK-F & $94.65^{\ddagger}$ & $0.35$  & $10.05^{\dagger}$ & $5.06^{\dagger}$ & $96.50^{\ddagger}$ & $1.41$  & $5.92^{\dagger}$ & $3.39^{\dagger}$ & $98.35^{\ddagger}$ & $1.43$  & $2.96^{\dagger}$ & $1.62^{\dagger}$ \\
        \bottomrule
    \end{tabular}
	}
\end{table}

\begin{table}[htbp]
  \centering
  \caption{\label{table portfolio stat loss} Average MSEs and QLIKE losses for one-step-ahead covariance matrix forecasts by 12 competing models in the second example, where MGARCH-LR denotes the proposed model in \eqref{model Rt SGARCH(r,s)}--\eqref{model Phii in Dt} fitted by the QMLE $\widehat{\bm{\theta}}_{\text{LR}}$ in \eqref{est lowrank QMLE}. Boldface values indicate the models included in the 90\% MCS.}
  \begin{tabular}{lcccccc}
    \toprule
    Metric & MGARCH-LR & CCC & DCC-T & D-DCC-T & DCC-E & D-DCC-E \\
    \midrule
	QLIKE & \textbf{11.5743} & 12.1996 & \textbf{11.5804} & \textbf{11.6282} & \textbf{11.2058} & 15.1266 \\
    MSE & \textbf{18.5293} & \textbf{18.2468} & \textbf{18.1809} & \textbf{18.1850} & \textbf{18.1811} & \textbf{18.5174} \\
    \midrule
    Metric & DCC-A & D-DCC-A & S-BEKK-F & D-BEKK-F & S-BEKK-VT & D-BEKK-VT \\
    \midrule
	QLIKE & 14.6781 & 14.7736 & 12.8096 & 12.9957 & 33.6543 & \textbf{11.0877} \\
    MSE & \textbf{18.2729} & \textbf{18.2987} & 29.4899 & \textbf{18.5747} & \textbf{18.7173} & \textbf{18.1112}\\
    \bottomrule
  \end{tabular}
\end{table}

\begin{table}[htbp]
    \centering
    \renewcommand\arraystretch{1.2}
    \caption{\label{table portfolio VaR MV} Out-of-sample VaR forecasting performance for the MV portfolio in the second example, where MGARCH-LR denotes the proposed model in \eqref{model Rt SGARCH(r,s)}--\eqref{model Phii in Dt} fitted by the QMLE $\widehat{\bm{\theta}}_{\text{LR}}$ in \eqref{est lowrank QMLE}. The smallest values of PE, $\text{Loss}_{\mathrm{Q}}$, and $\text{Loss}_{\mathrm{C}}$ in each column are shown in boldface. For ECR, $\dagger$ and $\ddagger$ indicate the models that pass one or both of the CC and DQ tests at the 5\% significance level, respectively. For $\text{Loss}_{\mathrm{Q}}$ and $\text{Loss}_{\mathrm{C}}$, $\dagger$ indicates that the corresponding model is included in the 90\% MCS at the 10\% significance level.}
    \resizebox{\textwidth}{!}{
    \begin{tabular}{l *{12}{r}}
        \toprule\toprule
        & \multicolumn{4}{c}{$\tau=1\%$} & \multicolumn{4}{c}{$\tau=2.5\%$} & \multicolumn{4}{c}{$\tau=5\%$} \\
        \cmidrule(lr){2-5} \cmidrule(lr){6-9} \cmidrule(lr){10-13} 
        Method & ECR & PE & $\text{Loss}_{\mathrm{Q}}$ & $\text{Loss}_{\mathrm{C}}$ & ECR & PE & $\text{Loss}_{\mathrm{Q}}$ & $\text{Loss}_{\mathrm{C}}$ & ECR & PE & $\text{Loss}_{\mathrm{Q}}$ & $\text{Loss}_{\mathrm{C}}$ \\
        \midrule
        MGARCH-LR & $0.99^\dagger$ & $\mathbf{0.01}$ & $2.53^\dagger$ & $0.98^\dagger$ & $3.18^\ddagger$ & $\mathbf{0.98}$ & $5.45^\dagger$ & $\mathbf{3.08}^\dagger$ & $7.36^\ddagger$ & $\mathbf{2.42}$ & $9.67^\dagger$ & $\mathbf{6.87}^\dagger$ \\
    CCC   & $1.79^\ddagger$ & $1.78$  & $3.36^\dagger$ & $1.76^\dagger$ & $3.78^\ddagger$ & $1.83$ & $6.09^\dagger$ & $3.65^\dagger$ & $9.94$  & $5.08$  & $10.22^\dagger$ & $9.20^\dagger$ \\
    DCC-T & $1.79^\ddagger$ & $1.78$  & $3.33^\dagger$ & $1.76^\dagger$ & $4.77$  & $3.26$  & $6.22^\dagger$ & $4.60^\dagger$ & $9.94$  & $5.08$  & $10.37^\dagger$ & $9.20^\dagger$ \\
    D-DCC-T & $1.79^\ddagger$ & $1.78$  & $3.35^\dagger$ & $1.76^\dagger$ & $4.57$  & $2.98$  & $6.23^\dagger$ & $4.41^\dagger$ & $10.14$ & $5.29$  & $10.36^\dagger$ & $9.38^\dagger$ \\
    DCC-E & $2.58$  & $3.57$  & $3.42$  & $2.54$  & $4.57$  & $2.98$  & $6.51^\dagger$ & $4.41^\dagger$ & $9.54$  & $4.67$  & $10.53^\dagger$ & $8.84^\dagger$ \\
    D-DCC-E & $2.98$  & $4.47$  & $3.41^\dagger$ & $2.93$  & $7.16$  & $6.69$  & $6.93^\dagger$ & $6.86$  & $10.74$ & $5.90$  & $11.63$ & $9.91^\dagger$ \\
    DCC-A & $3.78$  & $6.26$  & $4.19$  & $3.71$  & $7.16$  & $6.69$  & $8.16$  & $6.86$  & $11.13$ & $6.31$  & $13.05$ & $10.27^\dagger$ \\
    D-DCC-A & $4.17$  & $7.16$  & $4.37$  & $4.10$  & $7.55$  & $7.26$  & $8.42$  & $7.24$  & $13.52$ & $8.77$  & $13.60$ & $12.42$ \\
    S-BEKK-VT & $1.79^\ddagger$ & $1.78$  & $2.68^\dagger$ & $1.76^\dagger$ & $5.37$  & $4.12$  & $5.97^\dagger$ & $5.16^\dagger$ & $11.13$ & $6.31$  & $10.65^\dagger$ & $10.27^\dagger$ \\
    S-BEKK-F & $0.40^\ddagger$ & $1.36$  & $\mathbf{2.33}^\dagger$ & $\mathbf{0.40}^\dagger$ & $4.37^\dagger$ & $2.69$  & $\mathbf{5.05}^\dagger$ & $4.22^\dagger$ & $8.35$  & $3.45$ & $\mathbf{9.13}^\dagger$ & $7.76^\dagger$ \\
    D-BEKK-VT & $0.80^\ddagger$ & $0.46$ & $2.65^\dagger$ & $0.79^\dagger$ & $4.37$  & $2.69$  & $5.39^\dagger$ & $4.22^\dagger$ & $9.34$  & $4.47$  & $9.48^\dagger$ & $8.66^\dagger$ \\
    D-BEKK-F & $0.80^\ddagger$ & $0.46$ & $2.44^\dagger$ & $0.79^\dagger$ & $4.97$  & $3.55$  & $5.42^\dagger$ & $4.78^\dagger$ & $9.34$  & $4.47$  & $9.61^\dagger$ & $8.66^\dagger$ \\
        \midrule
        & \multicolumn{4}{c}{$\tau=95\%$} & \multicolumn{4}{c}{$\tau=97.5\%$} & \multicolumn{4}{c}{$\tau=99\%$} \\
        \cmidrule(lr){2-5} \cmidrule(lr){6-9} \cmidrule(lr){10-13} 
        Method & ECR & PE & $Loss_Q$ & $Loss_C$ & ECR & PE & $Loss_Q$ & $Loss_C$ & ECR & PE & $Loss_Q$ & $Loss_C$ \\
        \midrule
        MGARCH-LR & $95.83$ & $\mathbf{0.85}$ & $9.43^\dagger$ & $\mathbf{4.01}^\dagger$ & $98.01^\ddagger$ & $\mathbf{0.74}$ & $5.46^\dagger$ & $\mathbf{1.95}^\dagger$ & $99.40^\ddagger$ & $0.91$ & $2.53^\dagger$ & $\mathbf{0.59}^\dagger$ \\
    CCC   & $92.45^\dagger$ & $2.63$ & $9.25^\dagger$ & $7.05^\dagger$ & $95.83^\ddagger$ & $2.41$ & $5.44^\dagger$ & $4.03^\dagger$ & $97.81^\ddagger$ & $2.68$  & $2.67^\dagger$ & $2.15^\dagger$ \\
    DCC-T & $91.65$ & $3.45$  & $9.41^\dagger$ & $7.76^\dagger$ & $95.63^\ddagger$ & $2.69$  & $5.52^\dagger$ & $4.22^\dagger$ & $98.21^\ddagger$ & $1.78$  & $2.70^\dagger$ & $1.76^\dagger$ \\
    D-DCC-T & $91.85^\dagger$ & $3.24$  & $9.37^\dagger$ & $7.59^\dagger$ & $95.63^\ddagger$ & $2.69$  & $5.49^\dagger$ & $4.22^\dagger$ & $98.21^\ddagger$ & $1.78$  & $2.66^\dagger$ & $1.76^\dagger$ \\
    DCC-E & $92.25$ & $2.83$  & $9.66^\dagger$ & $7.23^\dagger$ & $95.43$ & $2.98$  & $5.63^\dagger$ & $4.41^\dagger$ & $97.42$ & $3.57$  & $2.71^\dagger$ & $2.54^\dagger$ \\
    D-DCC-E & $92.05$ & $3.04$  & $9.96^\dagger$ & $7.41^\dagger$ & $94.83$ & $3.83$  & $5.86^\dagger$ & $4.97^\dagger$ & $96.82$ & $4.92$  & $2.66^\dagger$ & $3.13^\dagger$ \\
    DCC-A & $89.07$ & $6.11$  & $11.76$ & $10.09$ & $93.04$ & $6.40$  & $7.12$  & $6.67$  & $96.62$ & $5.36$  & $3.55^\dagger$ & $3.32$ \\
    D-DCC-A & $88.67$ & $6.52$  & $11.97$ & $10.45$ & $92.84$ & $6.69$  & $7.19$  & $6.86$  & $96.22$ & $6.26$  & $3.62$  & $3.71$ \\
    S-BEKK-VT & $91.45$ & $3.65$  & $10.30^\dagger$ & $7.94^\dagger$ & $94.83$ & $3.83$  & $5.94^\dagger$ & $4.97^\dagger$ & $97.81^\ddagger$ & $2.68$  & $2.75^\dagger$ & $2.15^\dagger$ \\
    S-BEKK-F & $91.65$ & $3.45$  & $9.42^\dagger$ & $7.76^\dagger$ & $95.83^\dagger$ & $2.41$ & $5.25^\dagger$ & $4.03^\dagger$ & $98.61^\ddagger$ & $\mathbf{0.88}$ & $\mathbf{2.39}^\dagger$ & $1.37^\dagger$ \\
    D-BEKK-VT & $92.45$ & $2.63$ & $9.69^\dagger$ & $7.05^\dagger$ & $95.63$ & $2.69$  & $5.70^\dagger$ & $4.22^\dagger$ & $98.21^\ddagger$ & $1.78$  & $2.70^\dagger$ & $1.76^\dagger$ \\
    D-BEKK-F & $92.45$ & $2.63$ & $\mathbf{9.14}^\dagger$ & $7.05^\dagger$ & $95.83^\dagger$ & $2.41$ & $\mathbf{5.22}^\dagger$ & $4.03^\dagger$ & $98.21^\dagger$ & $1.78$  & $2.45^\dagger$ & $1.76^\dagger$ \\
        \bottomrule
    \end{tabular}
    }
\end{table}

\begin{table}[htbp]
    \centering
    \renewcommand\arraystretch{1.2}
    \caption{\label{table portfolio VaR EW} Out-of-sample VaR forecasting performance for the EW portfolio in the second example, where MGARCH-LR denotes the proposed model in \eqref{model Rt SGARCH(r,s)}--\eqref{model Phii in Dt} fitted by the QMLE $\widehat{\bm{\theta}}_{\text{LR}}$ in \eqref{est lowrank QMLE}. The smallest values of PE, $\text{Loss}_{\mathrm{Q}}$, and $\text{Loss}_{\mathrm{C}}$ in each column are shown in boldface. For ECR, $\dagger$ and $\ddagger$ indicate the models that pass one or both of the CC and DQ tests at the 5\% significance level, respectively. For $\text{Loss}_{\mathrm{Q}}$ and $\text{Loss}_{\mathrm{C}}$, $\dagger$ indicates that the corresponding model is included in the 90\% MCS at the 10\% significance level.}
    \resizebox{\textwidth}{!}{
    \begin{tabular}{l *{12}{r}}
        \toprule\toprule
        & \multicolumn{4}{c}{$\tau=1\%$} & \multicolumn{4}{c}{$\tau=2.5\%$} & \multicolumn{4}{c}{$\tau=5\%$} \\
        \cmidrule(lr){2-5} \cmidrule(lr){6-9} \cmidrule(lr){10-13} 
        Method & ECR & PE & $\text{Loss}_{\mathrm{Q}}$ & $\text{Loss}_{\mathrm{C}}$ & ECR & PE & $\text{Loss}_{\mathrm{Q}}$ & $\text{Loss}_{\mathrm{C}}$ & ECR & PE & $\text{Loss}_{\mathrm{Q}}$ & $\text{Loss}_{\mathrm{C}}$ \\
        \midrule
        MGARCH-LR & $0.20^\ddagger$ & $1.81$  & $4.20^\dagger$ & $\mathbf{0.20}^\dagger$ & $1.39^\dagger$ & $1.59$  & $8.71^\dagger$ & $\mathbf{1.38}^\dagger$ & $5.57^\ddagger$ & $\mathbf{0.58}$ & $\mathbf{15.28}^\dagger$ & $\mathbf{5.26}^\dagger$ \\
    CCC   & $0.60^\ddagger$ & $0.91$  & $4.04^\dagger$ & $0.59^\dagger$ & $2.39^\ddagger$ & $0.16$ & $8.78^\dagger$ & $2.33^\dagger$ & $6.96^\dagger$ & $2.02$  & $15.52^\dagger$ & $6.51^\dagger$ \\
    DCC-T & $0.40^\ddagger$ & $1.36$  & $4.02^\dagger$ & $0.40^\dagger$ & $2.39^\ddagger$ & $0.16$ & $8.73^\dagger$ & $2.33^\dagger$ & $6.56^\dagger$ & $1.61$  & $15.49^\dagger$ & $6.15^\dagger$ \\
    D-DCC-T & $0.40^\ddagger$ & $1.36$  & $4.02^\dagger$ & $0.40^\dagger$ & $2.39^\ddagger$ & $0.16$ & $8.73^\dagger$ & $2.33^\dagger$ & $6.56^\dagger$ & $1.61$  & $15.49^\dagger$ & $6.15^\dagger$ \\
    DCC-E & $0.40^\ddagger$ & $1.36$  & $3.99^\dagger$ & $0.40^\dagger$ & $2.78^\ddagger$ & $0.41$  & $8.73^\dagger$ & $2.71^\dagger$ & $6.76^\dagger$ & $1.81$  & $15.58^\dagger$ & $6.33^\dagger$ \\
    D-DCC-E & $0.40^\ddagger$ & $1.36$  & $3.97^\dagger$ & $0.40^\dagger$ & $2.98^\ddagger$ & $0.69$  & $8.76^\dagger$ & $2.90^\dagger$ & $6.36^\dagger$ & $1.40$ & $15.66^\dagger$ & $5.98^\dagger$ \\
    DCC-A & $0.60^\ddagger$ & $0.91$  & $3.94^\dagger$ & $0.59^\dagger$ & $2.58^\ddagger$ & $\mathbf{0.12}$ & $\mathbf{8.69}^\dagger$ & $2.52^\dagger$ & $6.96^\dagger$ & $2.02$  & $15.65^\dagger$ & $6.51^\dagger$ \\
    D-DCC-A & $0.60^\ddagger$ & $0.91$  & $3.94^\dagger$ & $0.59^\dagger$ & $2.98^\ddagger$ & $0.69$  & $8.71^\dagger$ & $2.90^\dagger$ & $7.36$  & $2.42$  & $15.71^\dagger$ & $6.87^\dagger$ \\
    S-BEKK-VT & $1.39^\dagger$ & $0.88$ & $4.07^\dagger$ & $1.37$  & $3.78^\ddagger$ & $1.83$  & $9.13^\dagger$ & $3.65$  & $8.75$  & $3.86$  & $16.13^\dagger$ & $8.12$ \\
    S-BEKK-F & $1.39^\dagger$ & $0.88$ & $4.20^\dagger$ & $1.37$  & $3.58^\ddagger$ & $1.55$  & $9.36^\dagger$ & $3.46^\dagger$ & $8.55$  & $3.65$  & $16.17^\dagger$ & $7.94^\dagger$ \\
    D-BEKK-VT & $1.39^\dagger$ & $0.88$ & $4.40^\dagger$ & $1.37$  & $2.98^\ddagger$ & $0.69$  & $9.34^\dagger$ & $2.90^\dagger$ & $7.75$  & $2.83$  & $15.96^\dagger$ & $7.23^\dagger$ \\
    D-BEKK-F & $1.19^\dagger$ & $\mathbf{0.43}$ & $\mathbf{3.87}^\dagger$ & $1.18^\dagger$ & $3.38^\ddagger$ & $1.26$  & $8.91^\dagger$ & $3.27$  & $8.55$  & $3.65$  & $16.01^\dagger$ & $7.94$ \\
        \midrule
        & \multicolumn{4}{c}{$\tau=95\%$} & \multicolumn{4}{c}{$\tau=97.5\%$} & \multicolumn{4}{c}{$\tau=99\%$} \\
        \cmidrule(lr){2-5} \cmidrule(lr){6-9} \cmidrule(lr){10-13} 
        Method & ECR & PE & $Loss_Q$ & $Loss_C$ & ECR & PE & $Loss_Q$ & $Loss_C$ & ECR & PE & $Loss_Q$ & $Loss_C$ \\
        \midrule
        MGARCH-LR & $93.24^\ddagger$ & $\mathbf{1.81}$ & $14.89^\dagger$ & $\mathbf{6.33}^\dagger$ & $96.62^\dagger$ & $\mathbf{1.26}$ & $8.31^\dagger$ & $\mathbf{3.27}^\dagger$ & $99.60^\ddagger$ & $1.36$  & $3.85^\dagger$ & $\mathbf{0.40}^\dagger$ \\
    CCC   & $92.84^\ddagger$ & $2.22$ & $14.52^\dagger$ & $6.69^\dagger$ & $96.62^\ddagger$ & $\mathbf{1.26}$ & $8.02^\dagger$ & $\mathbf{3.27}^\dagger$ & $98.81^\ddagger$ & $\mathbf{0.43}$ & $3.62^\dagger$ & $1.18^\dagger$ \\
    DCC-T & $92.45^\dagger$ & $2.63$  & $14.55^\dagger$ & $7.05^\dagger$ & $96.42^\ddagger$ & $1.55$  & $8.04^\dagger$ & $3.46^\dagger$ & $98.81^\ddagger$ & $\mathbf{0.43}$ & $3.65^\dagger$ & $1.18^\dagger$ \\
    D-DCC-T & $92.45^\dagger$ & $2.63$  & $14.55^\dagger$ & $7.05^\dagger$ & $96.42^\ddagger$ & $1.55$  & $8.04^\dagger$ & $3.46^\dagger$ & $98.81^\ddagger$ & $\mathbf{0.43}$ & $3.65^\dagger$ & $1.18^\dagger$ \\
    DCC-E & $92.84^\ddagger$ & $2.22$ & $\mathbf{14.51}^\dagger$ & $6.69^\dagger$ & $96.22^\ddagger$ & $1.83$  & $\mathbf{7.97}^\dagger$ & $3.65^\dagger$ & $98.81^\ddagger$ & $\mathbf{0.43}$ & $\mathbf{3.57}^\dagger$ & $1.18^\dagger$ \\
    D-DCC-E & $92.05$ & $3.04$  & $14.82^\dagger$ & $7.41^\dagger$ & $95.83^\dagger$ & $2.41$  & $8.00^\dagger$ & $4.03^\dagger$ & $98.81^\ddagger$ & $\mathbf{0.43}$ & $3.60^\dagger$ & $1.18^\dagger$ \\
    DCC-A & $92.05^\dagger$ & $3.04$  & $14.82^\dagger$ & $7.41^\dagger$ & $95.43^\dagger$ & $2.98$  & $8.18^\dagger$ & $4.41^\dagger$ & $98.81^\ddagger$ & $\mathbf{0.43}$ & $3.65^\dagger$ & $1.18^\dagger$ \\
    D-DCC-A & $91.45$ & $3.65  $ & $14.90^\dagger$ & $7.94^\dagger$ & $95.23^\dagger$ & $3.26$  & $8.24^\dagger$ & $4.60^\dagger$ & $98.61^\ddagger$ & $0.88$  & $3.67^\dagger$ & $1.37^\dagger$ \\
    S-BEKK-VT & $90.46$ & $4.67$  & $15.74$ & $8.84$  & $94.83$ & $3.83$  & $8.40^\dagger$ & $4.97^\dagger$ & $99.20^\dagger$ & $0.46$  & $3.59^\dagger$ & $0.79^\dagger$ \\
    S-BEKK-F & $90.66$ & $4.47$  & $15.89$ & $8.66^\dagger$ & $95.23^\dagger$ & $3.26$  & $8.48^\dagger$ & $4.60^\dagger$ & $99.60^\ddagger$ & $1.36$  & $3.91^\dagger$ & $\mathbf{0.40}^\dagger$ \\
    D-BEKK-VT & $91.45$ & $3.65$  & $15.25^\dagger$ & $7.94^\dagger$ & $95.63^\dagger$ & $2.69$  & $8.13^\dagger$ & $4.22^\dagger$ & $99.20^\ddagger$ & $0.46$  & $3.59^\dagger$ & $0.79^\dagger$ \\
    D-BEKK-F & $90.26$ & $4.88$  & $15.52$ & $9.02$  & $94.83$ & $3.83$  & $8.24^\dagger$ & $4.97$  & $98.61^\dagger$ & $0.88$  & $3.60^\dagger$ & $1.37^\dagger$ \\
        \bottomrule
    \end{tabular}
    }
\end{table}

	\newpage
	\bibliography{references}

@article{tse2002multivariate,
  title={A multivariate generalized autoregressive conditional heteroscedasticity model with time-varying correlations},
  author={Tse, Yiu Kuen and Tsui, Albert K C},
  journal={Journal of Business \& Economic Statistics},
  volume={20},
  pages={351--362},
  year={2002}
}

@article{engle2002dynamic,
  title={Dynamic conditional correlation: {A} simple class of multivariate generalized autoregressive conditional heteroskedasticity models},
  author={Engle, Robert},
  journal={Journal of Business \& Economic Statistics},
  volume={20},
  pages={339--350},
  year={2002}
}

@article{aielli2013dynamic,
  title={Dynamic conditional correlation: on properties and estimation},
  author={Aielli, Gian Piero},
  journal={Journal of Business \& Economic Statistics},
  volume={31},
  pages={282--299},
  year={2013}
}

@article{jeantheau1998strong,
  title={Strong consistency of estimators for multivariate {ARCH} models},
  author={Jeantheau, Thierry},
  journal={Econometric Theory},
  volume={14},
  pages={70--86},
  year={1998}
}

@article{huang2024sarma,
  title={{SARMA}: Scalable Low-Rank High-Dimensional Autoregressive Moving Averages via Tensor Decomposition},
  author={Huang, Feiqing and Lu, Kexin and Zheng, Yao},
  journal={arXiv preprint arXiv:2405.00626},
  year={2024}
}

@book{andreescu2016essential,
  title={Essential linear algebra with applications},
  author={Andreescu, Titu},
  year={2014},
  publisher={Springer}
}

@article{francq2017equation,
  title={An equation-by-equation estimator of a multivariate {log-GARCH-X} model of financial returns},
  author={Francq, Christian and Sucarrat, Genaro},
  journal={Journal of Multivariate Analysis},
  volume={153},
  pages={16--32},
  year={2017}
}

@article{ling2003asymptotic,
  title={{Asymptotic theory for a vector ARMA-GARCH model}},
  author={Ling, Shiqing and McAleer, Michael},
  journal={Econometric Theory},
  volume={19},
  pages={280--310},
  year={2003}
}

@article{fermanian2017stationarity,
  title={On the stationarity of dynamic conditional correlation models},
  author={Fermanian, Jean-David and Malongo, Hassan},
  journal={Econometric Theory},
  volume={33},
  pages={636--663},
  year={2017}
}

@book{francq2019garch,
  title={{GARCH models: structure, statistical inference and financial applications}},
  author={Francq, Christian and Zakoian, Jean-Michel},
  year={2019},
  publisher={John Wiley \& Sons}
}

@book{douc2014nonlinear,
  title={Nonlinear time series: {Theory}, methods and applications with {R} examples},
  author={Douc, Randal and Moulines, Eric and Stoffer, David},
  year={2014},
  publisher={CRC press}
}

@book{billingsley1995,
  title={Probability and Measure},
  author={Billingsley, Patrick},
  year={1995},
  publisher={John Wiley \& Sons}
}

@article{francq2012qml,
  title={{QML} estimation of a class of multivariate asymmetric {GARCH} models},
  author={Francq, Christian and Zako{\"\i}an, Jean-Michel},
  journal={Econometric Theory},
  volume={28},
  pages={179--206},
  year={2012}
}

@article{francq2004maximum,
  title={Maximum likelihood estimation of pure {GARCH} and {ARMA-GARCH} processes},
  author={Francq, Christian and Zakoian, Jean-Michel},
  journal={Bernoulli},
  volume={10},
  pages={605--637},
  year={2004}
}

@book{Bernstein2009,
  title={Matrix mathematics: {T}heory, facts, and formulas},
  author={Dennis S. Bernstein},
  year={2009},
  publisher={Princeton University Press}
}

@article{billingsley1961lindeberg,
  title={The {Lindeberg-Levy} theorem for martingales},
  author={Billingsley, Patrick},
  journal={Proceedings of the American Mathematical Society},
  volume={12},
  pages={788--792},
  year={1961}
}

@article{zheng2022interpretable,
  title = {{An Interpretable and Efficient Infinite-Order Vector Autoregressive Model for High-Dimensional Time Series}},
  author = {Yao Zheng},
  journal = {Journal of the American Statistical Association},
  volume={120},
  pages={212-225},
  year = {2025}
}

@article{tweedie1988invariant,
  title={Invariant measures for {Markov} chains with no irreducibility assumptions},
  author={Tweedie, RL},
  journal={Journal of Applied Probability},
  volume={25},
  pages={275--285},
  year={1988}
}

@article{Fryzlewicz_SubbaRao2011,
  title={{Mixing properties of {ARCH} and time-varying {ARCH} processes}},
  author={Fryzlewicz, P. and Subba Rao, S.},
  journal={Bernoulli},
  volume={17},
  pages={320--346},
  year={2011}
}

@article{caporin2012we,
  title={Do we really need both {BEKK} and {DCC}? {A} tale of two multivariate {GARCH} models},
  author={Caporin, Massimiliano and McAleer, Michael},
  journal={Journal of Economic Surveys},
  volume={26},
  pages={736--751},
  year={2012}
}

@article{caporin2013ten,
  title={Ten things you should know about the dynamic conditional correlation representation},
  author={Caporin, Massimiliano and McAleer, Michael},
  journal={Econometrics},
  volume={1},
  pages={115--126},
  year={2013}
}

@article{hafner2009generalized,
  title={A generalized dynamic conditional correlation model: simulation and application to many assets},
  author={Hafner, Christian M and Franses, Philip Hans},
  journal={Econometric Reviews},
  volume={28},
  pages={612--631},
  year={2009}
}

@article{hafner2009asymptotic,
  title={On asymptotic theory for multivariate GARCH models},
  author={Hafner, Christian M and Preminger, Arie},
  journal={Journal of Multivariate Analysis},
  volume={100},
  pages={2044--2054},
  year={2009},
  publisher={Elsevier}
}

@article{shapiro1986asymptotic,
  title={Asymptotic theory of overparameterized structural models},
  author={Shapiro, Alexander},
  journal={Journal of the American Statistical Association},
  volume={81},
  pages={142--149},
  year={1986}
}

@article{bollerslev1990ModellingTC,
	title={{Modelling the coherence in short-run nominal exchange rates: A multivariate generalized ARCH model}},
	author={Bollerslev, Tim},
	journal={The Review of Economics and Statistics},
	pages={498--505},
	volume={72},
	year={1990},
	publisher={JSTOR}
}

@article{bauwens2006multivariate,
  title={{Multivariate GARCH models: a survey}},
  author={Bauwens, Luc and Laurent, S{\'e}bastien and Rombouts, Jeroen VK},
  journal={Journal of Applied Econometrics},
  volume={21},
  pages={79--109},
  year={2006},
  publisher={Wiley Online Library}
}

@article{caporin2014robust,
  title={{Robust ranking of multivariate GARCH models by problem dimension}},
  author={Caporin, Massimiliano and McAleer, Michael},
  journal={Computational Statistics \& Data Analysis},
  volume={76},
  pages={172--185},
  year={2014},
  publisher={Elsevier}
}

@article{engle1995MultivariateSG,
	title={{Multivariate simultaneous generalized ARCH}},
	author={Engle, Robert F and Kroner, Kenneth F},
	journal={Econometric Theory},
	volume={11},
	pages={122--150},
	year={1995},
	publisher={Cambridge University Press}
}

@article{mcaleer2005automated,
  title={Automated inference and learning in modeling financial volatility},
  author={McAleer, Michael},
  journal={Econometric Theory},
  volume={21},
  pages={232--261},
  year={2005},
  publisher={Cambridge University Press}
}

@article{comte2003asymptotic,
  title={Asymptotic theory for multivariate {GARCH} processes},
  author={Comte, Fabienne and Lieberman, Offer},
  journal={Journal of Multivariate Analysis},
  volume={84},
  pages={61--84},
  year={2003},
  publisher={Elsevier}
}

@article{bollerslev1988capital,
	title={A capital asset pricing model with time-varying covariances},
	author={Bollerslev, Tim and Engle, Robert F and Wooldridge, Jeffrey M},
	journal={Journal of Political Economy},
	volume={96},
	pages={116--131},
	year={1988},
	publisher={The University of Chicago Press}
}

@article{ding2001large,
  title={Large scale conditional covariance matrix modeling, estimation and testing},
  author={Ding, Zhuanxin and Engle, Robert F},
  year={2001},
  publisher={NYU working paper No. Fin-01-029}
}

@article{cappiello2006asymmetric,
  title={Asymmetric dynamics in the correlations of global equity and bond returns},
  author={Cappiello, Lorenzo and Engle, Robert F and Sheppard, Kevin},
  journal={Journal of Financial Econometrics},
  volume={4},
  pages={537--572},
  year={2006},
  publisher={Oxford University Press}
}

@article{hafner2022identification,
  title={Identification of structural multivariate {GARCH} models},
  author={Hafner, Christian M and Herwartz, Helmut and Maxand, Simone},
  journal={Journal of Econometrics},
  volume={227},
  pages={212--227},
  year={2022},
  publisher={Elsevier}
}

@article{billio2005multivariate,
  title={Multivariate Markov switching dynamic conditional correlation {GARCH} representations for contagion analysis},
  author={Billio, Monica and Caporin, Massimiliano},
  journal={Statistical Methods and Applications},
  volume={14},
  pages={145--161},
  year={2005}
}

@Article{Christoffersen_1998,
  author    = {Christoffersen, P. F.},
  journal   = {International Economic Review},
  title     = {Evaluating interval forecasts},
  year      = {1998},
  publisher = {JSTOR},
  volume={39},
  pages={841-862}
}

@Article{Engle2004,
  author    = {Engle, R. F. and Manganelli, S.},
  journal   = {Journal of Business \& Economic Statistics},
  title     = {{CAViaR: conditional autoregressive value at risk by regression quantiles}},
  year      = {2004},
  pages     = {367--381},
  volume    = {22},
  publisher = {Taylor \& Francis},
}

@article{engle1996garch,
  title={{GARCH for groups: A round-up of recent developments in garch techniques for estimating correlation}},
  author={Engle, Robert and Mezrich, Joseph},
  journal={Risk-London-Risk Magazine Limited},
  volume={9},
  pages={36--40},
  year={1996},
  publisher={RISK MAGAZINE LIMITED}
}

@article{de2018mgarch,
	title={MGARCH models: Trade-off between feasibility and flexibility},
	author={De Almeida, Daniel and Hotta, Luiz K and Ruiz, Esther},
	journal={International Journal of Forecasting},
	volume={34},
	pages={45--63},
	year={2018},
	publisher={Elsevier}
}

@article{hirshleifer2020mood,
	title={Mood beta and seasonalities in stock returns},
	author={Hirshleifer, David and Jiang, Danling and DiGiovanni, Yuting Meng},
	journal={Journal of Financial Economics},
	volume={137},
	pages={272--295},
	year={2020},
	publisher={Elsevier}
}

@article{fama1997industry,
	title={Industry costs of equity},
	author={Fama, Eugene F and French, Kenneth R},
	journal={Journal of Financial Economics},
	volume={43},
	pages={153--193},
	year={1997},
	publisher={Elsevier}
}

@article{behr2012using,
	title={Using industry momentum to improve portfolio performance},
	author={Behr, Patrick and Guettler, Andre and Truebenbach, Fabian},
	journal={Journal of Banking \& Finance},
	volume={36},
	pages={1414--1423},
	year={2012},
	publisher={Elsevier}
}

@Article{Nielsen_Noel2021,
  author = {Nielsen, Morten {\O}rregaard and No{\"e}l, Antoine L.},
  title   = {{To infinity and beyond: Efficient computation of ARCH($\infty$) models}},
  journal = {Journal of Time Series Analysis},
  year    = {2021},
  volume  = {42},
  pages   = {338--354}
}

@book{higham2008functions,
  title={Functions of matrices: theory and computation},
  author={Higham, Nicholas J},
  year={2008},
  publisher={SIAM}
}

@book{horn2012matrix,
  title={Matrix analysis},
  author={Horn, Roger A and Johnson, Charles R},
  year={2013},
  publisher={Cambridge university press}
}

@book{wonham1979linear,
  title={Linear Multivariable Control: A Geometric Approach},
  author={Wonham, Walter Murray},
  edition={2nd},
  year={1979},
  publisher={Springer-Verlag}
}

@article{engle1998autoregressive,
  title={Autoregressive conditional duration: a new model for irregularly spaced transaction data},
  author={Engle, Robert F and Russell, Jeffrey R},
  journal={Econometrica},
  pages={1127--1162},
  year={1998},
  publisher={JSTOR}
}

@article{sheppard2007positive,
  title={Positive semi-definite matrix multiplicative error models},
  author={Sheppard, Kevin},
  journal={University of Oxford},
  year={2007}
}

@article{Geweke01011986,
author = {John Geweke},
title = {Modelling the Persistence of Conditional Variance: {A} Comment},
journal = {Econometric Reviews},
volume = {5},
pages = {57--61},
year = {1986},
publisher = {Taylor \& Francis}
}

@article{Pantula01011986,
author = {Sastry G. Pantula},
title = {Modelling the Persistence of Conditional Variance: {A} Comment},
journal = {Econometric Reviews},
volume = {5},
pages = {71--74},
year = {1986},
publisher = {Taylor \& Francis}
}

@techreport{milhoj1987multiplicative,
  title       = {A Multiplicative Parametrization of {ARCH} Models},
  author       = {Milh{\o}j, Anders},
  year         = {1987},
  institution  = {University of Copenhagen, Institute of Statistics},
  address      = {Copenhagen},
  type         = {Research Report}
}

@article{francq2013garch,
  title={{GARCH} models without positivity constraints: {E}xponential or log {GARCH}?},
  author={Francq, Christian and Wintenberger, Olivier and Zakoian, Jean-Michel},
  journal={Journal of Econometrics},
  volume={177},
  pages={34--46},
  year={2013},
  publisher={Elsevier}
}

@article{francq2018exponential,
  title={{An exponential chi-squared QMLE for log-GARCH models via the ARMA representation}},
  author={Francq, Christian and Sucarrat, Genaro},
  journal={Journal of Financial Econometrics},
  volume={16},
  pages={129--154},
  year={2018},
  publisher={Oxford University Press}
}

@article{sucarrat2016estimation,
  title={Estimation and inference in univariate and multivariate log-{GARCH}-{X} models when the conditional density is unknown},
  author={Sucarrat, Genaro and Gr{\o}nneberg, Steffen and Escribano, Alvaro},
  journal={Computational statistics \& data analysis},
  volume={100},
  pages={582--594},
  year={2016},
  publisher={Elsevier}
}

@article{escribano2018equation,
  title={Equation-by-equation estimation of multivariate periodic electricity price volatility},
  author={Escribano, Alvaro and Sucarrat, Genaro},
  journal={Energy Economics},
  volume={74},
  pages={287--298},
  year={2018},
  publisher={Elsevier}
}

@article{hansen2011model,
  title={The model confidence set},
  author={Hansen, Peter R and Lunde, Asger and Nason, James M},
  journal={Econometrica},
  volume={79},
  pages={453--497},
  year={2011},
  publisher={Wiley Online Library}
}

@book{Trefethen2005,
 title = {Spectra and Pseudospectra: The Behavior of Nonnormal Matrices and Operators},
 author = {Lloyd N. Trefethen and Mark Embree},
 year = {2005},
publisher = {Princeton University Press}
}

@article{wang2024high,
  title={High-dimensional low-rank tensor autoregressive time series modeling},
  author={Wang, Di and Zheng, Yao and Li, Guodong},
  journal={Journal of Econometrics},
  volume={238},
  pages={105544},
  year={2024},
  publisher={Elsevier}
}

@article{sucarrat2022risk,
  title={Risk estimation with a time-varying probability of zero returns},
  author={Sucarrat, Genaro and Gr{\o}nneberg, Steffen},
  journal={Journal of Financial Econometrics},
  volume={20},
  pages={278--309},
  year={2022},
  publisher={Oxford University Press}
}
\end{spacing}

\end{document}